\documentclass[letterpaper,11pt]{article}
\usepackage[utf8]{inputenc}
\usepackage[round,authoryear]{natbib} 
\usepackage[margin=1in]{geometry}
\PassOptionsToPackage{hyphens}{url}
\usepackage[bookmarks, colorlinks=true, plainpages=false,hypertexnames=false,
            citecolor=red,
            linkcolor=blue,
            anchorcolor=red,
            urlcolor=blue]{hyperref}
\usepackage{url}
\usepackage{amsmath,amsfonts,amsthm,amssymb,bm, verbatim,dsfont,mathtools}
\usepackage{xcolor,graphicx,mathrsfs,appendix}
\usepackage[normalem]{ulem}
\usepackage{etoolbox}
\usepackage{booktabs}

\usepackage{array}
\usepackage{multirow}
\usepackage[T1]{fontenc}

\usepackage{float}
\usepackage{tikz}
\usetikzlibrary{arrows.meta,positioning}
\usetikzlibrary{matrix,arrows,calc,shapes,backgrounds}
\usetikzlibrary{shapes.callouts,decorations.text} 

\usetikzlibrary{shapes.misc}

\tikzset{cross/.style={cross out, draw=black, minimum size=2*(#1-\pgflinewidth), inner sep=0pt, outer sep=0pt},
cross/.default={1pt}}

\tikzstyle{int}=[draw, fill=blue!20, minimum size=2em]
\tikzstyle{dot}=[circle, draw, fill=blue!20, minimum size=2em]
\tikzstyle{dotred}=[circle, draw, fill=red!20, minimum size=2em]
\tikzstyle{init} = [pin edge={to-,thin,black}]
\tikzstyle{initred} = [pin edge={to-,thin,red}]
\tikzstyle{plan}=[draw, fill=blue!20, minimum size=2em, text width=5em, rounded corners,align=center]
\tikzstyle{planwide}=[draw, fill=blue!20, minimum size=2em, text width=8em, rounded corners,align=center]
\tikzstyle{nodedot}=[circle, draw, fill=white, minimum size=0.3cm,inner sep=0pt]
\tikzstyle{nodedot}=[circle, draw, fill=white, minimum size=3,inner sep=0pt]
\tikzstyle{Medge}=[green!60!black, thick]
\tikzstyle{Bedge}=[red, thick]
\tikzstyle{Cedge}=[blue, thick]
\tikzstyle{Sedge}=[black, thick]
\tikzstyle{Mgiantedge}=[green!60!black, line width=3.0pt]
\tikzstyle{Bgiantedge}=[red,line width=3.0pt]
\tikzstyle{Cgiantedge}=[blue,line width=3.0pt]
\tikzstyle{Sgiantedge}=[black,line width=3.0pt]
\tikzstyle{shadedgiantnode}=[circle, draw, fill=black!10, minimum size=1cm, inner sep=0pt]
\tikzstyle{unshadedgiantnode}=[circle, draw, fill=white, minimum size=1cm, inner sep=0pt]
\makeatletter
\tikzset{my loop/.style =  {to path={
  \pgfextra{}
  [looseness=5,min distance=10mm]
  \tikz@to@curve@path},font=\sffamily\small
  }}  
\makeatletter 
\newcolumntype{C}[1]{>{\centering\arraybackslash}p{#1}}

\tikzstyle{vertexdot}=[circle, draw, fill=white, minimum size=3,inner sep=0pt]
\tikzstyle{vertexdotsolid}=[circle, draw, fill=black, minimum size=3,inner sep=0pt]

\usepackage{xr,xspace}
\usepackage{paralist}
\usepackage{enumitem}

\usepackage{caption,subcaption,soul}
\usepackage{algorithm}%
\usepackage{algpseudocode}
\usepackage{algorithmicx}%
\makeatletter
\usepackage{romanbar}

\theoremstyle{plain}
\newtheorem{theorem}{Theorem}
\newtheorem{lemma}{Lemma}
\newtheorem{proposition}{Proposition}
\newtheorem{assumption}{Assumption}
\newtheorem{corollary}{Corollary}
\theoremstyle{definition}
\newtheorem{definition}{Definition}

\newtheorem{problem}{Problem}

\newtheorem{remark}{Remark}

\newtheorem*{remark*}{Remark}
\newtheorem*{theorem*}{Theorem}
\newtheorem{example}{Example}

\newcommand{\ceil}[1]{\left\lceil #1 \right\rceil}

\newcommand{\Hyper}{\text{Hypergeometric}}

\usepackage{xspace,prettyref}

\newcommand{\diverge}{\to\infty}

\newcommand{\iiddistr}{{\stackrel{\text{\iid}}{\sim}}}

\newcommand{\reals}{{\mathbb{R}}}

\newcommand{\naturals}{{\mathbb{N}}}

\newcommand{\Expect}{\mathbb{E}}
\newcommand{\expect}[1]{\mathbb{E}\left[ #1 \right]}

\newcommand{\Prob}{\mathbb{P}}

\newcommand{\prob}[1]{ \mathbb{P}\left\{ #1 \right\} }

\newcommand{\Bern}{{\rm Bern}}
\newcommand{\Binom}{{\rm Binom}}

\newcommand{\iid}{i.i.d.\xspace}
\newrefformat{eq}{(\ref{#1})}
\newrefformat{chap}{Chapter~\ref{#1}}
\newrefformat{sec}{Section~\ref{#1}}
\newrefformat{alg}{Algorithm~\ref{#1}}
\newrefformat{fig}{Figure~\ref{#1}}
\newrefformat{tab}{Table~\ref{#1}}
\newrefformat{rmk}{Remark~\ref{#1}}
\newrefformat{clm}{Claim~\ref{#1}}
\newrefformat{claim}{Claim~\ref{#1}}
\newrefformat{def}{Definition~\ref{#1}}
\newrefformat{cor}{Corollary~\ref{#1}}
\newrefformat{lmm}{Lemma~\ref{#1}}
\newrefformat{prop}{Proposition~\ref{#1}}
\newrefformat{prob}{Problem~\ref{#1}}
\newrefformat{app}{Appendix~\ref{#1}}
\newrefformat{hyp}{Hypothesis~\ref{#1}}
\newrefformat{thm}{Theorem~\ref{#1}}
\newrefformat{fac}{Fact~\ref{#1}}
\newrefformat{ex}{Example~\ref{#1}}

\newcommand{\indc}[1]{{\mathbf{1}_{\left\{{#1}\right\}}}}

\newcommand{\sfX}{{\mathsf{X}}}

\newcommand{\polylog}{\mathsf{polylog}}

\newcommand{\calA}{{\mathcal{A}}}

\newcommand{\calC}{{\mathcal{C}}}

\newcommand{\calE}{{\mathcal{E}}}

\newcommand{\calJ}{{\mathcal{J}}}

\newcommand{\calN}{{\mathcal{N}}}

\newcommand{\calP}{{\mathcal{P}}}
\newcommand{\calQ}{{\mathcal{Q}}}

\newcommand{\calT}{{\mathcal{T}}}

\newcommand{\calV}{{\mathcal{V}}}
\newcommand{\calW}{{\mathcal{W}}}

\newcommand{\Side}{\calP}
\newcommand{\Set}{\calW}

\newcommand{\foverodd}{f_{\overline{O}_{\mathrm{odd}}}}
\newcommand{\fovereven}{f_{\overline{O}_{\mathrm{even}}}}
\newcommand{\fundereven}{f_{\underline{O}_{\mathrm{even}}}}
\newcommand{\funderodd}{f_{\underline{O}_{\mathrm{odd}}}}
\newcommand{\doverodd}{\overline{O}_{\mathrm{odd}}}
\newcommand{\dunderodd}{\underline{O}_{\mathrm{odd}}}
\newcommand{\dovereven}{\overline{O}_{\mathrm{even}}}
\newcommand{\dundereven}{\underline{O}_{\mathrm{even}}}

\newcommand{\fone}{\xi_1}
\newcommand{\ftwo}{\xi_2}
\newcommand{\paraone}{\eta_1}
\newcommand{\paratwo}{\eta_2}

\newcommand{\Short}{\calA}
\newcommand{\Long}{\calJ}

\newcommand{\ess}{s}
\newcommand{\Post}{A}
\newcommand{\Pre}{B}
\newcommand{\DistA}{\mathbb{A}}
\newcommand{\DistB}{\mathbb{B}}

\algrenewcommand\alglinenumber[1]{\footnotesize #1. }
\algrenewcommand\algorithmiccomment[1]{\hfill\textcolor{blue}{$\triangleright$ #1}}
\newrefformat{assump}{Assumption~\ref{#1}}
\makeatother

\begin{document}

\title{From signaling to interviews in random matching markets}

\author{Maxwell Allman, Itai Ashlagi, Amin Saberi, and Sophie H.\ Yu\thanks{
{I. Ashlagi and A.\ Saberi} are with the Department of Management Science and Engineering, Stanford University, Stanford CA, USA
\texttt{\{iashlagi,saberi\}@stanford.edu}.
S.\ H.\ Yu is with The Wharton School of Business, University of Pennsylvania, Philadelphia PA, USA,  \texttt{hysophie@wharton.upenn.edu}.
I. Ashlagi is supported in part by NSF Award CCF2312156. A. Saberi is supported in part by NSF Awards CCF2209520, CCF2312156, and a gift from CISCO.
}}

\date{}
\maketitle

\begin{abstract}
In many two-sided labor  markets, interviews are conducted before matches are formed. An increase in the number of interviews in the market for medical residencies  raised the demand for signaling mechanisms, in which applicants can send a limited number of signals to  communicate   interest. We study the role of  signaling mechanisms in reducing the number of interviews in centralized random matching markets with  post-interview shocks. For the market to clear we focus on 
\emph{interim stability}, which extends the notion of stability to  ensure that agents do not regret not interviewing with each other. A matching is \emph{almost interim stable} if it is  interim stable after removing a vanishingly small fraction of agents.

We first study signaling mechanisms in random matching markets with $n$ agents when agents on the short side, long side, or both sides signal their top~$d$ preferred partners. {Interview graphs are formed} by including all pairs where at least one party has signaled the other. {We show that when the probability of a nonnegative interview shock is bounded away from zero and $d = \omega(1)$, short-side signaling leads to almost interim stable matchings.} 
{
Long-side signaling can fail in strongly imbalanced markets when pre-interview scores outweigh interview shocks.
}%
 Conversely, {when the market is balanced and interview shocks are absent} and $d = o(\log n)$, both-side signaling fails to achieve almost interim stability.
{For larger $d \ge \Omega(\log^2 n)$, under the same probability condition, short-side signaling achieves perfect interim stability,} 
{
while long-side signaling can fail in imbalanced markets.
}%
 We build on our findings to propose a signaling mechanism for  multi-tiered random markets. Our analysis identifies conditions under which signaling mechanisms are incentive compatible. 
{
A technical contribution is the analysis of a \emph{message-passing algorithm} for stable matchings on trees, together with local availability certificates used to study interim stability on general graphs.
}%
\end{abstract}

\section{Introduction}

In many two-sided matching markets frictions  arise as market participants search and learn about their preferences through costly interactions before matches are formed. A notable example is the residency market in the US, where medical graduates are matched every year to residency programs through a clearinghouse organized by the {National Resident Matching Program (NRMP)}. The match occurs after an interview season, and a large surge in the number of interviews in recent years \citep{watson2017burden,gadepalli2015effort,melcher2018matching} has raised the demand for mechanisms to alleviate interview congestion. As a result, {residency and fellowship markets} are increasingly adopting signaling mechanisms: candidates can send a limited number of signals to different programs, which assist them in deciding whom to invite for interviews.\footnote{Other proposals included capping applications or interviews \citep{morgan2021case} or an interview match \citep{melcher2018matching}.}
 Signaling mechanisms are used in the academic hiring market for economists and even dating apps,\footnote{Users can send a few ``special'' messages to other users.} and studies explain why such signals have the potential to increase match efficiency ~{\citep{chang2021preference,lee2015propose,coles2013preference,jagadeesan2018varying}}.\footnote{{In a survey of otolaryngology applicants during the 2023--2024 application cycle, when applicants could send 25 preference signals, 84.4\% of respondents' interview offers came from programs they had signaled {\citep{yousef2024impact}}.}} Little is known about how to design the number of signals and which signals are effective in reducing the number of interviews. Medical residencies vary substantially in the number of signals (e.g., Family Medicine allows 5 signals, and Orthopaedic Surgery allows 30 signals).\footnote{\url{https://students-residents.aamc.org/applying-residencies-eras/program-signaling-2025-myeras-application-season}}

This paper studies the effectiveness of signaling mechanisms in reducing congestion in two-sided matching markets. We focus on markets similar to the medical match, where a centralized clearinghouse forms matches after an interview season. Agents in the market initially have {latent preferences over each potential match, which account for only prior information} about the value generated by the interview. A signaling mechanism prescribes to each agent which potential partners to interview with. Following these interviews agents update their preferences, and can ultimately  match only with those they have interviewed with.

A desirable property of the final match is stability \citep{gale1962college}. Stability, which can be viewed as an equilibrium concept for two-sided markets, requires that no pair of agents prefer being matched with each other over their current partners. This notion is well defined when restricting attention to pairs of agents who have interviewed with each other.
We expand this notion to further preclude mutual regret by pairs who did not interview with each other; a matching is \emph{interim stable} if no pair of agents prefers to match with each other with respect to their \emph{interim preferences}, which reflect their preferences after the interview season - incorporating post-interview utilities for those they interviewed with and pre-interview utilities for those they did not.

{
We study when one-side and both-side signaling yield interim stability, and how the required number of signals depends on population imbalance, tier structure, and the distribution of interview shocks.

}%

\subsection{Model overview}

Our model generalizes large two-sided random markets to allow for both pre- and post-interview preferences and tier structures. The market consists of applicants and firms, with each side partitioned into commonly known tiers. An agent in a higher tier has a higher intrinsic value than an agent in a lower tier, making them universally more desirable.

{
Each agent's pre-interview utility is the sum of the potential partner's intrinsic value and an idiosyncratic pre-interview score. If an interview occurs, an additional post-interview score is revealed and added to this utility, as specified in \prettyref{sec:model_setup}. Unless specified otherwise, these scores are drawn independently, and the probability that the post-interview score is non-negative is positive.
}%

The matching  process occurs in two stages. First, a signaling mechanism forms interviews based on agents' pre-interview preferences. Each agent can signal up to $d$ potential partners, and interviews occur between pairs where at least one party signaled the other. 
 During this stage, agents learn their post-interview scores. Then, a stable one-to-one (final) matching is formed in the market induced by the set of interviews. Agents can match only with someone they interviewed with, and stability implies that there are no two agents who interviewed with each other but prefer each other over their assigned partners.

{
We also study \emph{almost interim stability}, which permits removal of a vanishing fraction of agents before requiring interim stability.
}%

\subsection{Contributions}

We begin by analyzing single-tiered markets, where all agents share the same intrinsic value and there is no hierarchical structure. In these markets, agents can be unevenly distributed across the two sides. Let $n_{\Short}$ and $n_{\Long}$ denote the number of applicants and firms, respectively, and let $n\triangleq n_{\Short}+ n_{\Long}$. We say that a market is balanced if $n_\Short = n_\Long$, weakly imbalanced if $ \left(1- o(1)\right)n_\Long \le n_{\Short} \le \left(1+ o(1)\right)n_\Long$, and strongly imbalanced if $n_{\Short} \le (1 - \Omega(1))n_{\Long}$ or $n_{\Short} \ge (1 + \Omega(1))n_{\Long}$.

We first consider a \emph{one-side signaling mechanism} where agents from one side of the market signal their most preferred potential partners based on pre-interview utilities. Each agent can signal up to $d$ potential partners, and we distinguish between cases where the short-side or long-side of the market sends signals.

{
We distinguish sparse signaling, where almost interim stability is the objective, from dense signaling, where we seek perfect interim stability.

}%
We assume  pre-interview scores are $\iid$ drawn from a continuous distribution, and post-interview scores are $\iid$ drawn from another distribution, where $p$ denotes the probability of a post-interview score being non-negative.\footnote{Continuity of the pre-interview score distribution ensures that agents have strict interim preferences. {See Assumption~\ref{assump:general}.}} The main results characterize conditions for achieving interim stability:
\begin{theorem*}[Informal]\label{inf:single_signal_sparse_summary} 
Let $H$ denote an interview graph constructed by one-side signaling. 
\begin{itemize}
    \item
{
 \emph{(Sparse regime)} with $\omega(1) \le d \le O\left(\polylog n\right)$ and $p\ge \Omega(1)$:
}%

    \begin{itemize}
    \item For weakly imbalanced markets  where $ \left(1- o(1)\right)n_\Long \le n_{\Short} \le \left(1+ o(1)\right)n_\Long$: 
    \begin{itemize}
        \item Under either short-side or long-side signaling, with high probability, every stable matching on $H$ is almost interim stable.
    \end{itemize}
    \item For strongly imbalanced markets where $n_{\Short} \le \left(1-\Omega(1)\right) n_{\Long}$ or $n_{\Short} \ge (1 + \Omega(1))n_{\Long}$:
    \begin{itemize}
        \item Under short-side signaling, with high probability, every stable matching on $H$ is almost interim stable;
        \item
{
 Under long-side signaling in markets with comparable populations, if pre-interview scores outweigh post-interview scores,
}%
\footnote{See \prettyref{def:outweigh} for the formal definition of pre-interview scores outweighing post-interview scores.} with high probability, no stable matching on $H$ is almost interim stable. 
    \end{itemize}
    \end{itemize}
\item \emph{(Dense regime)} with $d \ge \Omega\left(\log^2 n / p\right)$:  
\begin{itemize}
\item For imbalanced markets where  $n_\Short \neq n_\Long$: 
\begin{itemize}
    \item Under short-side signaling, with high probability, every stable matching on $H$ is perfect interim stable;

{
    \item Under long-side signaling and the degree and score conditions of \prettyref{thm:single_firm_signal_dense}, if the population gap $|n_\Short-n_\Long|$ tends to infinity, with high probability, no stable matching on $H$ is perfect interim stable.
}%
 
\end{itemize}
\item For balanced markets with $n_\Short = n_\Long$:
\begin{itemize}
    \item Under applicant-signaling (resp. firm-signaling),  with high probability, the applicant-optimal (resp. firm-optimal) stable matching on $H$ is perfect interim stable.
\end{itemize}
\end{itemize}
\end{itemize}
\end{theorem*}

{
For detailed statements and extensions of the above theorem on sparse signaling regime, we refer to \prettyref{sec:single_sparse}, and for dense signaling regime, we refer to \prettyref{sec:single_dense}.
}%

We also examine a \emph{both-side signaling} mechanism, where all agents signal their top $d$ preferred partners based on pre-interview preferences. Interestingly, in balanced markets, both-side signaling exhibits contrasting behaviors: it fails to achieve almost interim stability with $d=o(\log n)$ when interviews have effectively no influence on preferences, while succeeding with $d=\omega(1)$ when interviews have full impact. We further explore how heterogeneity in post-interview shocks, arising from pre-interview non-observable attributes, affects interim stability by relaxing the i.i.d. assumption. Building on insights from single-tiered markets, we also extend our analysis to multi-tiered markets. 

{
Finally, in single-tiered markets we establish an exact truthful-signaling best response {for the side proposing in deferred acceptance (DA).} For tiered markets with constant general imbalance and fixed bounded score laws, truthful signaling is an approximate Bayes--Nash equilibrium in the ex ante expected-payoff sense specified in \prettyref{sec:incentive}.
}%

{

A methodological contribution is a message-passing algorithm that computes the unique stable matching on a finite tree. On general interview graphs, truncated neighborhoods and comparative-statics arguments provide local availability certificates. These certificates can imply interim stability; the random-market conclusions additionally require probability bounds for their occurrence. This local analysis complements deferred-acceptance methods based on balls-into-bins processes and rejection chains~{\citep{gale1962college,im2005,ashlagi2017unbalanced,kanoria2023competition,potukuchi2024unbalanced}}.
}%

\subsection{Literature review }

\paragraph{Interview dynamics and market Design.}

There is an emerging literature on information acquisition and interviews in two-sided matching markets. Several papers find benefits in  how interview formation affects outcomes:~\cite{lee2017interviewing} show that ``interview overlap'' can improve match rates and ~\cite{manjunath2021interview} find benefits in balancing the number of interviews in random markets and  \cite{skancke2021welfare,beyhaghi2021randomness} demonstrate the effectiveness of limiting interview numbers to reduce costs. Our paper instead looks at large markets and examines how a small number of interviews, guided by signals, can clear the market. 

{
Related interview-formation models include \citet{allman2023interviewing}; our focus is interim stability.
}%

The literature on signaling in matching markets emerged to reduce congestion ~{\cite{lee2015propose,coles2013preference,jagadeesan2018varying}}, with applications in residency and fellowship markets~{\citep{melcher2019signal, pletcher2022otolaryngology,irwin2024use,yousef2024impact}}. These papers study primarily  how  strategic signals from doctors to hospitals can improve efficiency by indicating special interest. Our paper looks instead at two-sided large markets, quantifies the number of signals and interprets signals as informative rather than strategic.

Beyond congestion reduction, other papers have analyzed interview decision-making through different lenses:~{\cite{drummond2013elicitation} study preference elicitation with partial information, while~\cite{kadam2021interviewing} studies strategic interview decisions and their effects on matching outcomes. Related algorithmic work considers preference elicitation and interview minimization~\cite{drummond2014preference,rastegari2013two}.} 
Our paper is closely  related  to \citet{ashlagi2020clearing}, which identifies how little communication can help to reach a stable matching in large markets when agents know their preferences. Our paper expands  their multi-tiered market model and their (simultaneous) signaling protocol to incorporate incomplete information about interviewing scores.    Also related is \cite{ashlagi2025stablematchinginterviews}, which develops  adaptive and non-adaptive algorithms for generating interviews. In contrast, we focus on  simple and practical decentralized signaling mechanisms for  generating interviews and ask when such signaling mechanisms can clear the market and have good incentive properties  as a function of the market structure. 

\paragraph{Competition in random two-sided matching markets}

{
Random matching models characterize market clearing, welfare, and competition under complete and partial preferences \citep{pittel1989,im2005,pittel2019likely,cai2022short,ashlagi2023welfare,kanoria2023competition,potukuchi2024unbalanced,arnosti2023lottery}. \citet{ashlagi2017unbalanced} establish the short-side advantage in imbalanced complete markets. For partially connected markets with small imbalance, \citet{kanoria2023competition} identify a transition near $d=\log^2n$ between weak and strong competition. Our results instead ask when the interviews generated by signaling preclude blocking pairs, including pairs that did not interview. More general utility models have also been used to study match rates and welfare {\citep{menzel2015large,pkeski2017large,lee2016incentive,ashlagi2023welfare,agarwal2023stable}}.
}%

\paragraph{Message passing algorithm.}
The message-passing algorithm, also known as belief propagation, has been widely applied across statistical physics~{\citep{mezard2009information}}, computer science~\citep{mezard2002analytic}, artificial intelligence~\citep{pearl2014probabilistic, pearl2022fusion}, and computer vision~\citep{freeman2000learning}. {A pioneering application to analyze complex market dynamics} in matching markets was undertaken by~\cite{immorlica2022matching}, who leveraged message-passing algorithms to investigate information deadlocks in markets with costly compatibility inspections. 

\subsection{Notation and paper organization}

For any graph $H$, let $\calV(H)$ denote its vertex set and $\calE(H)$ denote its edge set. For each node $i$ in $H$, let $\calN(i) \subset \calV(H)$ represent its neighbors in $H$, and $d_i$ denote its degree, i.e., $d_i = |\calN(i)|$. A \emph{bipartite} graph is a graph whose vertices can be divided into two independent sets, $\Short$ and $\Long$, such that every edge $(u,v)$ connects a vertex $u \in \Short$ to a vertex $v \in \Long$. For any graph $H$ and a subset of its vertices $\calV \subseteq \calV(H)$, the vertex-induced subgraph of $H$ on $\calV$ is the graph with vertex set $\calV$ and edge set $\{(u,v) \in \calE(H) : u,v \in \calV\}$. For any graph $H$ and a subset of its edges $\calE \subseteq \calE(H)$, the edge-induced subgraph of $H$ on $\calE$ is the graph with vertex set $\{u,v \in \calV(H) : (u,v) \in \calE\}$ and edge set $\calE$. For any two graphs $H_1$ and $H_2$, the graph union $H_1 \cup H_2$ is the graph with vertex set $\calV(H_1) \cup \calV(H_2)$ and edge set $\calE(H_1) \cup \calE(H_2)$. 
{
The tree excess of a connected graph $H$ is defined as
}%
 $|\calE(H)| - |\calV(H)| + 1$, which is the maximum number of edges that can be deleted from the induced subgraph on $H$ while keeping $H$ connected.

A \emph{rooted} graph is a graph where one vertex is designated as the root. Given a rooted graph $H$ with root $\rho$, for any $m \in \naturals$, 
we define the depth of each node on $H$ as the number of edges in the shortest path from $\rho$ to that node, and we define $H_m(\rho)$ as the vertex-induced subgraph of $H$ on the set of nodes at depths less than or equal to $m$, which is also known as the $m$-hop neighborhood of $\rho$ on $H$.  The depth of a tree is the total number of edges from the root node to any leaf node in the longest path.

Let $X \overset{\mathrm{s.t.}}{\succeq} Y$ denote that random variable $X$ has first-order stochastic dominance over random variable $Y$ if $\prob{X \ge x} \ge \prob{Y \ge x}$ for any $x \in \reals$.  
If $X$ and $Y$ are two independent random variables with probability distributions $\mathbb{D}_1$ and $\mathbb{D}_2$ respectively, then the distribution of the sum $X + Y$ is given by the convolution $\mathbb{D}_1 * \mathbb{D}_2$. 
{For any $K, N, t \in \naturals_+$ with $K,t\le N$, let $X \sim \mathrm{Hypergeometric}(K, N, t)$} denote a random variable $X$ following the hypergeometric distribution with parameters $K$, $N$, and $t$. The probability mass function of $X$ is given by $\prob{X = x} = \binom{K}{x} \binom{N-K}{t-x}/\binom{N}{t}$ for $ x \in \left \{0, 1, \ldots, \min\{t, K\}  \right \}$. 

For two real numbers $x$ and $y$, we let $x \vee y \triangleq \max\{x, y\}$
and $x\wedge y \triangleq \min\{x, y\}$. 
We use standard asymptotic notation: for two positive sequences $\{x_n\}$ and $\{y_n\}$, we write $x_n = O(y_n)$ or $x_n \lesssim y_n$, if $x_n \le C y_n$ for an absolute constant $C$ and for all $n$; $x_n = \Omega(y_n)$ or $x_n \gtrsim y_n$, if $y_n = O(x_n)$; $x_n = \Theta(y_n)$ or $x_n \asymp y_n$, if $x_n = O(y_n)$ and $x_n = \Omega(y_n)$; 
$x_n = o(y_n)$ or $y_n = \omega(x_n)$, if $x_n / y_n \to 0$ as $n\diverge$. 

{
\prettyref{sec:model_setup} defines the model, and \prettyref{sec:main} presents the main results and incentive analysis. \prettyref{sec:message} outlines the local availability method, with its full analysis in \prettyref{sec:local}. \prettyref{sec:numerical} reports simulations, and \prettyref{sec:conclusions} discusses further questions. Proofs and supplementary results appear in the appendices.
}%

\section{Model setup and key concepts} \label{sec:model_setup}

Let $\Short$ denote the set of applicants and $\Long$ denote the set of firms. Let $n_{\Short}$ and $n_{\Long}$ denote the number of applicants and firms, respectively, and let $n\triangleq n_{\Short}+ n_{\Long}$.

Each applicant $a \in \Short$ has a pre-interview utility for a firm $j \in \Long$, denoted $U_{a,j}^B = v_j + \Pre_{a,j}$, where  $v_j$ is an intrinsic value of firm $j$ and $\Pre_{a,j}$ is a pre-interview idiosyncratic score of $a$ towards $j$. 
The post-interview utility an applicant $a$ has for firm $j$ is given by $U_{a,j}^A = U_{a,j}^B + \Post_{a,j}$, where $\Post_{a,j}$ is a post-interview idiosyncratic score of $a$ towards $j$. 
Similarly, for a firm $j \in \Long$, its pre-interview utility for an applicant $a\in\Short$ is $U_{j,a}^B = v_a + \Pre_{j,a}$, and its post-interview utility is $U_{j,a}^A = v_a + \Pre_{j,a} + \Post_{j,a}$, where $v_a$ is an intrinsic value of applicant $a$ and  $\Pre_{j,a}$ and $\Post_{j,a}$  are pre- and post-interview idiosyncratic scores of $j$ towards $a$, respectively.  Unless  otherwise specified we make the following assumption.  
\begin{assumption}[Pre-interview and post-interview scores]
\label{assump:general}
For every applicant $a \in \Short$ and every  $j \in \Long$, pre-interview scores $\Pre_{a,j}$ and $\Pre_{j,a}$ are $\iid$ drawn from a continuous distribution $\DistB$ and  post-interview scores $\Post_{a,j}$ and $\Post_{j,a}$ are $\iid$ drawn from a distribution $\DistA$. {The pre-interview and post-interview score arrays are independent of each other.} Let $p$ denote the probability of a post-interview score being non-negative, where $p>0$.
\end{assumption}
Continuity of the pre-interview score distribution $\DistB$ ensures that agents have strict interim preferences. The probability $p$ can be interpreted as {the chance that an interview maintains or improves} upon the pre-interview impression.

The set of  interviews between applicants and firms can be represented by a graph.
\begin{definition}[Interview Graph]
 An interview graph is a bipartite graph $H$ that connects applicants $\Short$ to firms $\Long$. An applicant $a \in \Short$ is said to interview with a firm $j \in \Long$ if and only if $(a,j) \in \calE(H)$.
\end{definition}

Consider an interview graph $H$.  We denote $U_{a,j}^H$  and $U_{j,a}^H$  as the interim utilities  induced by  $H$.  
That is, if $a$ interviewed with $j$,  $U_{a,j}^H = U_{a,j}^A$ and $U_{j,a}^H = U_{j,a}^A$ and otherwise $U_{a,j}^H = U_{a,j}^B$ and $U_{j,a}^H = U_{j,a}^B$.

Following interviews, a one-to-one matching $\Phi$ will be formed between applicants and firms. We assume that an applicant can match with a firm only if they interview with each other. For every $a \in \Short$ and $j \in \Long$, we say that $a$ is matched to $j$ in $\Phi$ if $(a,j) \in \Phi$, denoted as $\phi(a) = j$ and $\phi(j) = a$. For $i\in \Short\cup\Long$, if $i$ is unmatched, write $\phi(i) = \emptyset$.

For any applicant $a \in \Short$ evaluating two firms $j_1, j_2 \in \Long$, $a$ strictly (resp. weakly) prefers $j_1$ over $j_2$, represented as $j_1 \succ_a j_2$ if $U_{a,j_1}^H > U_{a,j_2}^H$ (resp. $j_1\succeq_a j_2$ if ${U_{a,j_1}^H}\ge U_{a,j_2}^H$). Similarly, a firm $j \in \Long$ strictly (resp. weakly) prefers applicant $a_1$ over $a_2$, denoted as $a_1 \succ_j a_2$ if $U_{j,a_1}^H > U_{j,a_2}^H$ (resp. $a_1\succeq_j a_2$ if $U_{j,a_1}^H \ge U_{j,a_2}^H $). 
Unless specified otherwise, we  assume  that every agent is acceptable to every other agent, that is: every agent  prefers  being matched with any agent over remaining unmatched.\footnote{Without loss of generality, we can assume that the $U_{i,\emptyset}^H = -\infty$ for any $i\in\Short \cup \Long$.}  

\begin{definition}[Stable matching]
    A stable matching in a bipartite graph $H$ is a matching where there are no applicant-firm blocking pairs $(a,j)\in \calE(H)$ such that both $a \in \Short$ and $j \in \Long$ prefer each other over their current match in $\Phi$, indicated as $j \succ_a \phi(a)$ and $a \succ_j \phi(j)$.
\end{definition}

\begin{definition}[Interim blocking pair]
    In a matching $\Phi$ on $H$, an interim blocking pair is formed if an applicant $a$ and a firm $j$ mutually strictly prefer each other over their respective matches in $\Phi$, indicated as $j \succ_a \phi(a)$ and $a \succ_j \phi(j)$, irrespective of whether $(a, j)$ are connected in $H$.
\end{definition}

\begin{definition}[Interim stability]\label{def:stable}
A stable matching $\Phi$ is perfect interim stable if it does not have any interim blocking pairs. A stable matching is considered almost interim stable if it becomes perfect interim stable when a vanishingly small fraction of agents is excluded.
\end{definition}

{
Thus interim stability is stability on the complete bipartite market with the utilities induced by $H$; stability on $H$ alone excludes only interviewed blocking pairs.
}%

Interviews in our setup are formed through a signaling mechanism as follows. A signaling mechanism prescribes {which agents send signals, and each such agent sends} $d$ signals to agents on the other side of the market, based on their pre-interview preferences. We assume that an interview between an  applicant and a firm occurs if at least one of them signaled the other.  

We are interested in whether simple signaling mechanisms are able to attain almost or perfect interim stability. The signaling mechanisms we consider differ by which side of the market can send signals  (or both sides) in single-tiered markets. In multi-tiered markets, agents in each tier will send signals to a target tier, which will be defined in the sequel.

\section{Main results and discussions} \label{sec:main}

\subsection{Single-tiered market}\label{sec:single}

We first focus on single-tiered markets: in these markets all agents have the same intrinsic value. We investigate the effectiveness of short-side signaling, long-side signaling, and both-side signaling in achieving almost interim stability and perfect interim stability under various market conditions and signaling regimes. %

In the case of one-side signaling, we consider the applicant-signaling mechanism, i.e., applicants are the ones who send out signals. By varying the relative sizes of the two sides of the market, we can then analyze both short-side signaling ($n_{\Short} < n_{\Long}$) and long-side signaling ($n_{\Short} > n_{\Long}$). For both-side signaling, both applicants and firms signal their top $d$ preferred partners.

\subsubsection{Single-tiered market with sparse signals}\label{sec:single_sparse}

In this subsection, we focus on the sparse signaling regime. The following theorem first establishes  that every stable matching on the resulting interview graph $H$ under the applicant-signaling mechanism is almost interim stable with high probability if either applicants are on the short side, or there are not many more applicants than firms. %
\begin{theorem}[Effectiveness on one-side-signaling with sparse signals]\label{thm:single_signal_sparse}
Consider a single-tiered two-sided market with applicants $\Short$ and firms $\Long$ with $n_{\Short} \le \left(1+o(1)\right) n_{\Long}$. Let $H$ denote an interview graph constructed based on the applicant-signaling mechanism with $\omega(1) \le d \le O\left(\polylog n\right)$. If $p \ge \Omega (1)$, every stable matching on $H$ is almost interim stable with high probability.
\end{theorem}

The following remark provides more explicit conditions under which the result of Theorem \ref{thm:single_signal_sparse} holds.
\begin{remark}\label{rmk:single_signal_sparse}
Building upon Theorem \ref{thm:single_signal_sparse}, every stable matching on $H$ is almost interim stable with high probability if any of the following conditions hold:
\begin{itemize}
    \item $\omega(1) \le d \le O\left(\polylog n\right)$, $n_{\Short} \le \left(1+d^{-\lambda} \right)n_\Long$ and $p = \omega \left( \frac{1}{\left(1\wedge \lambda \right)\log d} \right)$ 
{
for any $\lambda>0$;
}%

    \item $\omega(1) \le d \le o\left(\log^2 n\right)$, $n_{\Short} \le \left(1+n^{-\lambda}\right)n_\Long$ and $p = \omega \left( \frac{1}{\sqrt{d}} \right)$ for any $\lambda \ge \Omega(1)$. 
    \item $\omega(1) \le d \le O\left(\polylog n\right)$, $n_{\Short} \le \left(1- \Omega(1)\right)n_\Long$ and $p = \omega \left( \frac{1}{d}  \right)$.
\end{itemize}
\end{remark}

\begin{remark}\label{rmk:single_signal_sparse_identify}
Building upon Theorem \ref{thm:single_signal_sparse}, there exists a subset $\Short'\subset \Short$ such that every stable matching on the vertex-induced subgraph of $H$ on $\left(\Short\backslash \Short'\right)\cup \Long$ is perfect interim stable with high probability, where $|\Short'| = o \left( n_{\Short} \right)$.  
\end{remark}

{
The failure results use the following comparison of pre-interview and post-interview scores.
}%

\begin{definition}[Pre-interview scores outweighing post-interview scores]\label{def:outweigh}

{
Let $\DistA$ and $\DistB$ denote the distributions of post-interview and pre-interview scores, respectively. For a distribution $\mathbb D$ with cumulative distribution function $F_{\mathbb D}$, write
\[
 F_{\mathbb D}^{-1}(u)=\inf\{x\in\reals:F_{\mathbb D}(x)\ge u\},
 \qquad 0<u<1.
\]
For any positive parameters $\kappa_1,\kappa_2>0$, we say that pre-interview scores outweigh post-interview scores in the $(\kappa_1,\kappa_2)$ range if
\[
 F_{\DistA*\DistB}^{-1}\!\left(\frac{\kappa_1}{\kappa_1+1}\right)
 <F_{\DistB}^{-1}\!\left(\frac{\kappa_2}{\kappa_2+1}\right).
\]

{
Here $\DistA*\DistB$ denotes convolution.
}%

}%
\end{definition}

The following theorem highlights the limitations of long-side signaling in achieving almost interim stability in strongly imbalanced markets, where $n_{\Short} \ge \left(1+\Omega(1)\right) n_{\Long}$, and when  pre-interview scores have a stronger influence than post-interview scores. 

\begin{theorem}[Failure on one-side-signaling with sparse signals] \label{thm:single_firm_signal_sparse}

{
Suppose that $n_{\Short} \ge \left(1+\Omega(1)\right)n_{\Long}$ and $n_{\Short}\le Cn_{\Long}$ for some fixed constant $C<\infty$. Assume also that pre-interview scores outweigh the post-interview scores in the range $(\kappa_1,\kappa_2)$, where
}%
 $\kappa_1= \ceil{2 d n_\Short /n_\Long}$ and $\kappa_2 = n_\Short -n_\Long - \ceil{2 d n_\Short /n_\Long}$. Let $H$ denote an interview graph constructed based on the applicant-signaling mechanism with $\omega(1) \le d \le O\left(\polylog n\right)$. Then,  with high probability, no stable matching on the interview graph $H$ is almost interim stable. 

{
Two examples where the $(\kappa_1,\kappa_2)$-condition is satisfied are: 
\begin{itemize}
\item $\DistB$ is any continuous distribution and $\DistA$ is a degenerate distribution at zero ($\DistA=\boldsymbol{\delta}_0$).
\item 
{
$\DistB$ is a fixed nondegenerate normal distribution and $\DistA$ is any fixed bounded distribution.
}%

\end{itemize}
}%

\end{theorem}

{
Under these conditions, with high probability there are $\Omega(n)$ vertex-disjoint interim blocking pairs common to all stable matchings on $H$. Each pair joins an unmatched applicant to a firm that prefers that applicant to every interviewed partner. At least a positive fraction of the market must therefore be removed to eliminate these pairs from the restricted matching.
}%

Next, we explore the effectiveness of both-side signaling. %
We  show that both-side signaling will fail even when the market is balanced and {interviews have no impact on agents' preferences}. %
\begin{theorem}[Failure on both-side signaling with sparse signals]
\label{thm:single_both_signal_sparse}
Suppose $n_{\Short}= n_{\Long}$ and the post-interview scores are absent, i.e., $\DistA = \boldsymbol\delta_0$. Let $H$ denote an interview graph constructed based on both-side signaling with $\omega(1)\le d\le o\left(\log n\right)$. Then, no stable matching on $H$ is almost interim stable with high probability.
\end{theorem}

To provide some intuition on why both-side signaling fails in the special case when post-interview scores are absent, we consider running the applicant-proposing DA algorithm on the interview graph. Since the proposing order does not change the final outcome, and agents strictly prefer those they signal to over those who signal to them, the applicant-proposing DA algorithm can be considered as the following two-stage procedure:

{
In the first stage, the applicant-proposing DA algorithm is run on the interview subgraph induced by the signals sent out by the applicants. In the stated sparse regime, the number of unmatched applicants after this stage diverges with high probability, although their fraction tends to zero.
}%

In the second stage, 
the applicant-proposing DA algorithm continues based on the resulting matching from the first stage and the interview subgraph induced by the signals received by the applicants: all unmatched applicants from the first stage will continue proposing to firms that signaled to them. 

It's worth noting that firms prefer applicants they have signaled to over applicants that have signaled to them. This preference structure can lead to a long rejection chain in the second stage, resulting in a constant fraction of applicants being matched to firms that signaled to them and a constant fraction of firms being matched to applicants that signaled to them. These two groups of agents may prefer each other, even though they have not interviewed each other.%

On the other extreme, when pre-interview scores are silent and only  post-interview scores are present, there is no harm in both-side signaling.%
\begin{remark}[Effectiveness on both-side signaling with sparse signals]\label{rmk:single_both_signal_sparse_2}
Suppose $n_{\Short} \le \left(1+o(1)\right)n_{\Long}$ and the pre-interview scores are absent, i.e., $\DistB = \boldsymbol\delta_0$, while the post-interview scores follow a continuous distribution $\DistA$.\footnote{When $\DistB = \boldsymbol\delta_0$, each agent randomly selects $d$ partners to signal on both sides, given that the pre-interview utilities are the same across all agents. The continuous post-interview scores $\DistA$ guarantee the strict preferences of every agent with respect to the agents they interviewed with.} Let $H$ denote an interview graph constructed based on both-side signaling with $\omega(1)\le d\le O\left(\polylog n \right)$. If $p \ge \Omega (1)$, then every stable matching on $H$ is almost interim stable with high probability.
\end{remark}

The intuition behind this result is that when pre-interview scores are absent, agents' preferences over interviewed partners are $\iid$ according to the continuous distribution $\DistA$. This setting closely resembles the case of one-side signaling, where almost interim stability is achieved when post-interview scores play a significant role in shaping preferences.

Based on the insights gained from these extreme cases, we conjecture that there exists a threshold for the magnitude of post-interview scores relative to pre-interview scores that determines the achievability of almost interim stability under both-side signaling with sparse signals. When the magnitude is small, and when pre-interview scores have a stronger influence than post-interview scores, agents are more likely to prefer partners they signal to over those who signal to them, and both-side signaling fails to achieve almost interim stability. However, as the magnitude of post-interview scores increases, post-interview utilities gradually approach uniform preferences over interviewed partners. In this case, both-side signaling can achieve almost interim stability, as the impact of signaling behavior on preferences diminishes, and the market more closely resembles a setting with $\iid$ preferences.

The proofs of Theorems \ref{thm:single_signal_sparse}, \ref{thm:single_firm_signal_sparse} and \ref{thm:single_both_signal_sparse}, along with Remarks \ref{rmk:single_signal_sparse}, \ref{rmk:single_signal_sparse_identify}, and \ref{rmk:single_both_signal_sparse_2} are deferred to \prettyref{sec:single_sparse_proof}.

\subsubsection{Single-tiered market with dense signals}\label{sec:single_dense}

In this subsection, we shift the focus to the dense signaling regime.
The next theorem  establishes conditions on the number of signals  for achieving perfect interim stability {as a function of the imbalance in the market}. %

\begin{theorem}[Effectiveness on one-side signaling with dense signals] \label{thm:single_signal_dense}

{
Consider a single-tiered two-sided market with applicants $\Short$ and firms $\Long$, where $n_{\Short}= \delta n_{\Long}$ for some $0< \delta \le 1$, and $n_\Short\to\infty$. Let $H$ denote an interview graph constructed based on the applicant-signaling mechanism. Suppose $d \ge \frac{8}{ \delta p} \log \left( \frac{1}{1-\delta + \delta^2/ n_\Long} \right) \log n_{\Short}$.
}%
 
\begin{itemize}
\item If $\delta<1$, every stable matching on $H$ is perfect interim stable with high probability.
\item If $\delta = 1$, the applicant-optimal stable matching on $H$ is perfect interim stable with high probability.
\end{itemize}
\end{theorem}
This theorem highlights the relation between competition and the number of  signals required for perfect interim stability. As the  market becomes more imbalanced, fewer signals are required to achieve this stability. For example, if $\delta \ge 1- O\left(1/n\right)$, we need $\Omega\left(\log^2 n / p\right)$ signals to achieve perfect interim stability. On the other hand, if $\delta \le 1-\Omega(1)$, we only need $\Theta\left(\log n / p\right)$ signals to achieve perfect interim stability.

{
The next remark gives a matching lower bound up to a constant factor
when the market is nearly balanced and $p=\Theta(1)$.
\begin{remark}\label{rmk:insuffcient_signal}
Consider a single-tiered applicant-signaling market with
$\delta=n_\Short/n_\Long\le1$ and $\delta=1-o(1)$. If, for a fixed $\epsilon>0$,
\[
 d\le\frac{1-\epsilon}{\delta}
 \log\!\left(\frac{1}{1-\delta+\delta^2/n_\Long}\right)
 \log n_\Short,
\]
then, with high probability, no stable matching on $H$ is perfect
interim stable.
\end{remark}
}%

The following remark highlights a key distinction between balanced and imbalanced markets. %
\begin{remark}\label{rmk:firm_optimal_signal}
Building on \prettyref{thm:single_signal_dense}, in balanced markets ($n_\Short = n_{\Long}$), {the firm-optimal stable matching  (under applicant signaling)} may fail to be perfect interim stable with non-vanishing probability.
\end{remark}

The following theorem highlights the limitations of long-side signaling in achieving perfect interim stability in imbalanced markets with dense signals, when the pre-interview scores outweigh the post-interview scores. 
{
The proof uses unmatched applicants and independent unused firm scores to bound the probability of a blocking pair; it yields high-probability failure when the applicant surplus diverges.
}%
 
\begin{theorem}[Failure on one-side signaling with dense signals] \label{thm:single_firm_signal_dense}
Let $H$ denote an interview graph constructed based on the applicant-signaling mechanism. Suppose that $ n _{\Long} < n_{\Short} \le C n _{\Long}$ 
{
for some fixed constant $C>1$
}%
, and that pre-interview scores outweigh the post-interview scores in the $(\kappa_1,\kappa_2)$ range with $\kappa_1 = \left( 2  d  n_\Short / n_\Long \right) \vee \log^2 n $ and $\kappa_2 =n_{\Long}-d-1$. 
{
If $d \le \frac{1}{4 C} n_{\Long}$, then, with probability at least $1-\exp\!\left[-(n_{\Short}-n_{\Long})/12\right]-o(1)$, no stable matching on $H$ is perfect interim stable. In particular, this failure holds with high probability when $n_{\Short}-n_{\Long}\to\infty$.
}%
 
Two examples where the $(\kappa_1,\kappa_2)$-condition is satisfied are:
\begin{itemize}
\item $\DistB$ is any continuous distribution and $\DistA$ is a degenerate distribution at zero ($\DistA=\boldsymbol{\delta}_0$).
\item 
{
For $d \le n^{\alpha}$ with any constant $\alpha<1$, take $\DistB$ to be a fixed nondegenerate normal distribution and $\DistA$ to be any fixed bounded distribution.
}%

\end{itemize}
\end{theorem}

{

{
Next we consider dense both-side signaling with log-concave post-interview shocks. The theorem also retains two sufficient special cases: an upper atom and a near-complete interview budget.
}%

}%

\begin{theorem}[Effectiveness on both-side signaling with dense signals] \label{thm:single_both_signal_dense}

{
Suppose $n_\Short=n_\Long$ and let $H$ be the interview graph constructed by both-side signaling with $d\ge\frac8p\log^2 n_\Short$. Assume at least one of the following:
\begin{enumerate}
\item $p\ge p_0$ for some fixed $p_0>0$, and $\DistA$ has a log-concave density;
\item for some $c\ge0$, $A\le c$ almost surely and $\Pr\{A=c\}\ge p/2$;
\item $(n_\Short-d)^2\log^2 n_\Short=o(n_\Short)$.
\end{enumerate}
Then both the applicant-optimal and the firm-optimal stable matchings on $H$ are perfect interim stable simultaneously with high probability.
}%
\end{theorem}

{

{
Log-concavity means that the shock density $f_A$ satisfies
$f_A(\lambda x+(1-\lambda)y)\ge f_A(x)^\lambda f_A(y)^{1-\lambda}$
for $0<\lambda<1$. Normal, uniform, and Laplace shocks satisfy this
condition. The continuous pre-interview score law $\DistB$ is unrestricted,
and both laws may vary with market size if $p\ge p_0>0$.
The proof uses the likelihood-ratio order induced by log-concavity to
compare the different conditional score laws on outgoing and incoming-only
edges.
}%

{
The upper-atom alternative includes absent shocks, nonpositive shocks
with positive mass at zero, and two-point laws such as the Rademacher
distribution. The near-complete-budget alternative permits general
additive shock distributions.
}%

}%
{The proofs of Theorems~\ref{thm:single_signal_dense}, \ref{thm:single_firm_signal_dense}, and \ref{thm:single_both_signal_dense}, and \prettyref{rmk:firm_optimal_signal} are deferred} to \prettyref{sec:single_dense_proof}. 
{
The proof of \prettyref{rmk:insuffcient_signal} is also given in \prettyref{sec:single_dense_proof}.
}%

\subsubsection{Single-tiered market with correlated post-interview scores}\label{sec:correlated}
In the previous subsections, we assumed that post-interview scores are independently and identically distributed ($\iid$) for all agents. However, in some scenarios, Assumption~\ref{assump:general} may not hold due to heterogeneity in post-interview scores among applicants and firms.

This heterogeneity can be attributed to pre-interview non-observable attributes, which can significantly influence post-interview outcomes. For instance, some applicants might excel in interviews due to innate charisma or strong communication skills, resulting in consistently higher post-interview evaluations from firms. Similarly, certain firms may have unadvertised benefits or a particularly positive work environment that consistently elicits more favorable responses from interviewees than their pre-interview expectations suggested. These pre-interview non-observable attributes can lead to correlated post-interview scores, where an applicant or firm consistently generates positive (or negative) impressions across multiple interviews, deviating from the $\iid$ assumption.

To illustrate the impact of heterogeneous post-interview scores on the interim stability of matchings, we present the following example, which considers a market with two types of applicants: good interviewers and bad interviewers. Despite this heterogeneity, we show that the results from \prettyref{thm:single_signal_sparse} and \prettyref{thm:single_signal_dense} still hold under certain conditions.

\begin{example}\label{ex:good_bad_interviewers}
Consider a two-sided market with applicants $\Short$ and firms $\Long$, where $n_{\Short} \le n_{\Long}$. The applicants are categorized into two types: good interviewers and bad interviewers. 
{
Applicant types are mutually independent and independent of all pre-interview scores and applicant-side post-interview shocks. An applicant is a good interviewer with probability $\alpha$ and a bad interviewer with probability $1-\alpha$, for a fixed $0<\alpha<1$.
}%
 Each agent does not know their own type or the types of other agents.

{
Suppose that Assumption~\ref{assump:general} holds, except that if an applicant $a$ is a good interviewer,
}%
 the post-interview score of any firm $j$ with respect to $a$ is $\iid$ sampled from distribution $\DistA_g$; if $a$ is a bad interviewer, the post-interview score of any firm $j$ with respect to $a$ is $\iid$ sampled from distribution $\DistA_b$. We have 
\[
A_{j,a}\iiddistr
\begin{cases}
\DistA_g & \text{if $a$ is a good interviewer}\\
\DistA_b & \text{o.w.}
\end{cases} \,, \quad \text{where $\DistA_g \overset{\mathrm{s.t.}}{\succeq} \DistA_b$} \,,
\]
where $\DistA_g$ first-order stochastically dominates $\DistA_b$. Let $p_{\DistA_g}$ (resp. $p_{\DistA_b}$) denote the probability of a post-interview score from $\DistA_g$ (resp.  $\DistA_b$) being non-negative. 
{
Then, $p_{\DistA_g} \ge p_{\DistA_b}$. The signalers' own nonnegative-shock probability remains $p=\Pr\{A_{a,j}\ge0\}$; it is this probability that enters the following conditions. 
}%

Let $H$ denote the interview graph constructed by the applicant-signaling mechanism such that every applicant signals its top $d$ preferred firms based on its pre-interview utilities. Then, we have:

{
\begin{itemize}
\item If $d=\omega(1)$ and $p\ge\Omega(1)$, every stable matching on $H$ is almost interim stable with high probability.
\item There is an absolute constant $C>0$ such that, if
$d\ge \frac{C}{p(1-\alpha)^2}\log^2n$, the applicant-optimal stable matching on $H$ is perfect interim stable with high probability.
\end{itemize}
}%

{
The sparse conclusion has the same form as \prettyref{thm:single_signal_sparse}, and the dense conclusion retains the $\log^2n$ order for fixed type shares. Its sufficient constant now accounts for heterogeneous competition. The proofs condition on type labels and use independence of the score arrays within each type. The sparse proof controls all stable matchings by a first-moment calculation, while the dense proof bounds an applicant's proposals after inserting it last in deferred acceptance. 
}%

\end{example}

To generalize our analysis, we consider a two-sided market with $m$ types of applicants and $\ell$ types of firms, where agents differ in their non-observable attributes while sharing the same observable intrinsic values.
The probability of an applicant belonging to type $\ess$ is $\alpha_\ess$ for $1\le \ess \le m$, with $\sum_{\ess=1}^m \alpha_\ess =1$. Similarly, the probability of a firm belonging to type $\kappa$ is $\beta_\kappa$ for $1\le \kappa \le \ell$, with $\sum_{\kappa=1}^{\ell}\beta_\kappa =1$. For an applicant $a$ of type $\ess$, the post-interview score of any firm $j$ with respect to $a$ is $\iid$ according to distribution $\DistA_\ess$, where $p_{\DistA_\ess}$ denotes the probability of a post-interview score from $\DistA_\ess$ being non-negative. 
For a firm $j$ of type $\kappa$, the post-interview score of any applicant $a$ with respect to $j$ is $\iid$ sampled from distribution $\DistA_\kappa'$, where $p_{\DistA_\kappa'}$ denotes the probability of a post-interview score from $\DistA_\kappa'$ being non-negative. 

{

{
Type labels are independent of the pre-interview scores. Conditional
on all type labels, the directed interview shocks are mutually
independent with the specified laws, and independent of the
pre-interview scores.

{
The following remark extends the dense bilateral conclusion to
type-dependent post-interview scores in imbalanced markets. Its
quantitative imbalance condition allows $n_\Short/n_\Long\to1$.
}%

}%

}%

{
\begin{remark}\label{rmk:correlated}
Under the relaxation of Assumption~\ref{assump:general}, consider a
two-sided market with correlated post-interview scores with applicants
$\Short$ and firms $\Long$, where $n_{\Short}\le n_{\Long}$. Let $H$
denote the interview graph constructed by the both-side-signaling
mechanism. Suppose that $\{\alpha_\ess\}_{\ess=1}^m,
\{\beta_\kappa\}_{\kappa=1}^\ell\ge\gamma$ and
$\{p_{\DistA_\ess}\}_{\ess=1}^m,
\{p_{\DistA_\kappa'}\}_{\kappa=1}^\ell\ge p$ for some constants
$\gamma,p>0$. For a fixed $\epsilon>0$, suppose additionally that
\[
 pd\left(1-\frac{n_\Short}{n_\Long}\right)
 \ge(1+\epsilon)\log n.
\]
Then, if $d\ge\frac{8+\epsilon}{\gamma p}\log^2n$, every stable
matching on $H$ is perfect interim stable with high probability.
\end{remark}
}%

{The proofs of \prettyref{ex:good_bad_interviewers} and \prettyref{rmk:correlated} are deferred} to \prettyref{sec:correlated_proof}.

\subsection{Multi-tiered market}\label{sec:multi}

We now shift the focus to a multi-tiered market model. Unlike the single-tiered market with correlated post-interview scores discussed in \prettyref{sec:correlated}, where heterogeneity arises from non-observable intrinsic attributes, the multi-tiered model captures observable heterogeneity through its tier structure. In this model, applicants and firms stratify into distinct hierarchical tiers based on their observable intrinsic attributes and desirability.

Formally, we partition applicants $\Short$ and firms $\Long$ 
{
into distinct hierarchical tiers
}%
 based on their intrinsic attributes and desirability. We denote these tiers by $\Short = \bigcup_{\ess=1}^m \Short_\ess$ for applicants and $\Long = \bigcup_{\kappa=1}^\ell \Long_\kappa$ for firms, where $m$ and $\ell$ are the number of tiers in the applicant and firm markets, respectively. 

The tier number is positively correlated with desirability, meaning agents in higher-ranked tiers are considered more desirable. In particular, 
for any $1\le \ess\le m$, every applicant $a\in \Short_\ess$ has an intrinsic value $v_a$ equal to $s$, i.e., $v_a \equiv s$. Similarly, for any $ 1\le \kappa\le \ell$, every firm $j\in \Long_\kappa$ has an intrinsic value $v_j$ equal to $\kappa$, i.e., $v_j \equiv \kappa$. This framework can be simplified to a single-tiered market scenario when $m=\ell=1$.

For any $1\le \ess \le m$, the proportion of applicants in tier $\Short_\ess$ is denoted by $\alpha_\ess$, where $\boldsymbol{\alpha}=\left(\alpha_\ess\right)_{\ess=1}^m$, 
$|\Short_\ess| = \alpha_\ess n_{\Short}$, and $\sum_{\ess=1}^m \alpha_\ess = 1$. Similarly, for any $1\le \kappa \le \ell$, the proportion of firms in tier $\Long_\kappa$ is denoted by $\beta_\kappa$, where $\boldsymbol{\beta}=\left(\beta_\kappa\right)_{\kappa=1}^\ell$,
$|\Long_\kappa| = \beta_\kappa n_{\Long}$, and $\sum_{\kappa=1}^\ell \beta_\kappa = 1$. 
To complement Assumption~\ref{assump:general}, we introduce two additional assumptions. 

\begin{assumption}[Non-vanishing tier size]\label{assump:non_vanishing}
We assume $ \{ \alpha_{\ess}\}_{\ess=1}^{m} , \{\beta_{\kappa}\}_{\kappa=1}^{\ell} \ge \Omega(1)$.
\end{assumption}

\begin{assumption}[Boundedness]\label{assump:bounded}
$\DistA$ and $\DistB$ are bounded distributions with $\mathrm{Support}(\DistA), \mathrm{Support}(\DistB) \subset \left[-M,M\right]$, where $0\le M < \frac{1}{4}$.
\end{assumption}

Throughout this subsection, we assume that Assumptions \ref{assump:general}, \ref{assump:non_vanishing} and \ref{assump:bounded} hold unless otherwise specified.
The non-vanishing Assumption \ref{assump:non_vanishing} ensures that as the market scales up, the relative sizes of the tiers remain stable and do not become negligibly small as the market grows. The boundedness assumption \ref{assump:bounded} ensures that the tier market structure is maintained, as it guarantees that applicants and firms consistently prefer counterparts in higher tiers, regardless of whether an interview has taken place or not.\footnote{For example, for any applicant $a\in\Short$, and firms $j_1 \in \Long_{\kappa_1}$ and $j_2\in \Long_{\kappa_2}$ with $\kappa_1 > \kappa_2$, we have $j_1 \succ_a j_2$, {since each idiosyncratic interim score is either $B_{a,j}$ or $B_{a,j}+A_{a,j}$ and hence lies in $[-2M,2M]$. Therefore $U_{a,j_1}^H-U_{a,j_2}^H\ge(\kappa_1-\kappa_2)-4M\ge1-4M>0$ by Assumption~\ref{assump:bounded}.}}   Later, we also extend our results by relaxing this boundedness Assumption \ref{assump:bounded}. 

For this model, we study a  \emph{multi-tiered signaling mechanism} \citep{ashlagi2020clearing} where each agent signals its top $d$ preferred partners within its target tier. Target tiers are determined by \emph{dominance relationships}: 
an applicant tier $\Short_\ess$ dominates a firm tier $\Long_\kappa$,  
if the total number of applicants in $\Short_\ess$ and above is less than or equal to the total number of jobs in $\Long_\kappa$ and above, i.e., $\sum_{\ess'=\ess}^{m} |\Short_{\ess'}| \le \sum_{\kappa'=\kappa}^{\ell} |\Long_{\kappa'}|$. Conversely, a job tier $\Long_\kappa$ dominates an applicant tier $\Short_\ess$ if $\sum_{\ess'=\ess}^{m} |\Short_{\ess'}| \ge \sum_{\kappa'=\kappa}^{\ell} |\Long_{\kappa'}|$.
For an applicant tier $\Short_\ess$, its target tier $\calT(\Short_\ess)$ is the highest-ranked firm tier that it dominates; for a firm tier $\Long_\kappa$, its target tier $\calT(\Long_\kappa)$ is the highest-ranked applicant tier that it dominates. If a tier does not dominate any tier on the opposite side, we denote its target tier as $\calT(\cdot) = \emptyset$.
Figure \ref{fig:target_tier} illustrates an example of the multi-tiered signaling mechanism.

  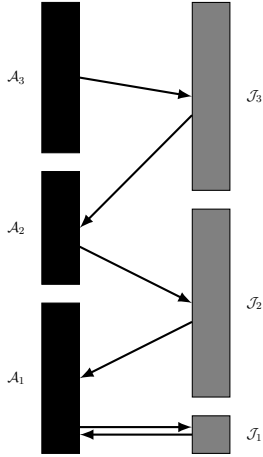
\begin{figure}[ht]
    \centering
    \begin{tikzpicture}[>=latex, arrows={-Triangle[length=3mm, width=2mm]}, scale=0.5, transform shape]
    \tikzstyle{leftbox} = [rectangle, fill=black, minimum width=1cm, draw=black]
    \tikzstyle{rightbox} = [rectangle, fill=gray, minimum width=1cm, draw=black]
    \tikzstyle{arrow} = [->, thick]

    \node[leftbox, minimum height=4cm] (A3) at (0,0) {};
    \node[leftbox, minimum height=3cm] (A2) at (0,-4 cm) {};
    \node[leftbox, minimum height=4cm] (A1) at (0,-8 cm) {};

    \node[rightbox, minimum height=5cm] (J3) at (4,-0.5) {};
    \node[rightbox, minimum height=5cm] (J2) at (4,-6 cm) {};
    \node[rightbox, minimum height=1cm] (J1) at (4,-9.5 cm) {};

    \node (A3Label) [left= 3mm of A3] {$\mathcal{A}_3$};
    \node (A2Label) [left=3mm of A2] {$\mathcal{A}_2$};
    \node (A1Label) [left=3mm of A1] {$\mathcal{A}_1$};

    \node (J2Label) [right=3mm of J3] {$\mathcal{J}_3$};
    \node (J2Label) [right=3mm of J2] {$\mathcal{J}_2$};
    \node (J1Label) [right=3mm of J1] {$\mathcal{J}_1$};

    \draw[arrow] ([yshift=0 cm]A3.east) -- ([yshift=0 cm]J3.west);
    \draw[arrow] ([yshift=-0.5 cm]J3.west) -- ([yshift=0 cm]A2.east);
    \draw[arrow] ([yshift=-0.5 cm]A2.east) -- ([yshift=0 cm]J2.west);
    \draw[arrow] ([yshift=-0.5 cm]J2.west) -- ([yshift=0 cm]A1.east);
    \draw[arrow] ([yshift= -1.3 cm]A1.east) -- ([yshift= 0.2 cm]J1.west);
    \draw[arrow] ([yshift= 0 cm]J1.west) -- ([yshift=-1.5cm]A1.east);
    
    \end{tikzpicture}
    \caption{This bipartite graph illustrates an example of the multi-tiered signaling mechanism, with the applicant tiers $(\mathcal{A}_1, \mathcal{A}_2, \mathcal{A}_3)$ and firm tiers $(\mathcal{J}_1, \mathcal{J}_2, \mathcal{J}_3)$ arranged in order of desirability. The length of each box represents the size of the respective tier, with $\mathcal{A}_3$ being the highest-ranked applicant tier and $\mathcal{J}_3$ being the highest-ranked firm tier.  Arrows indicate target tiers. 
    }\label{fig:target_tier}
\end{figure}

For any applicant tier $\Short_\ess$ and firm tier $\Long_\kappa$, we define their cumulative counts difference as:
$$
\left|\left|\cup_{\ess'=\ess}^{m} \Short_{\ess'} \right| -  |\cup_{\kappa'=\kappa}^{\ell} \Long_{\kappa'}|\right| =
\left|\sum_{\ess'=\ess}^{m} \alpha_{\ess'} n_{\Short} - \sum_{\kappa'=\kappa}^{\ell} \beta_{\kappa'} n_{\Long}\right|\,.
$$
A market is said to be generally imbalanced if, for any applicant tier $\Short_\ess$ and firm tier $\Long_\kappa$, their cumulative counts difference is always positive. It is equivalent to saying that there does not exist any pair of an applicant tier and a firm tier that are the target tiers of each other, or no applicant tier and firm tier simultaneously dominate each other. A market is $\gamma$-generally imbalanced if, for any applicant tier $\Short_\ess$ and firm tier $\Long_\kappa$, their cumulative counts difference is always lower bounded by $\gamma n$ for some $0<\gamma <1$.\footnote{
{
In a generally imbalanced market, every nonzero cumulative-count difference is an integer and hence at least $1$. Therefore $\gamma$ can be chosen at least $1/n$; smaller positive values are also admissible lower bounds.
}%
}

\subsubsection{Multi-tiered market with sparse signals}\label{sec:multi_sparse}

\begin{theorem}[Effectiveness on multi-tiered signaling with sparse signals]
\label{thm:multi_signal_sparse}
Consider a multi-tiered two-sided market with applicants $\Short = \bigcup_{\ess=1}^m \Short_{\ess}$ and firms $\Long = \bigcup_{\kappa=1}^\ell \Long_\kappa$ with $n_{\Short} \le n_{\Long}$.
Let $H$ denote an interview graph constructed based on the multi-tiered signaling mechanism with $\omega (1) \le d\le O\left(\polylog n \right)$. Then, if the market is generally imbalanced and $p = \omega\left(1/\log d\right)$, every stable matching on $H$ is almost interim stable with high probability.

\end{theorem}

{
The following remark extends \prettyref{rmk:single_signal_sparse_identify} to the multi-tiered market.
}%

\begin{remark}\label{rmk:multi_signal_sparse_identify}
    {Building upon \prettyref{thm:multi_signal_sparse}, there exist $\Short'\subset\Short$ and $\Long'\subset \Long$, where $|\Short'|=o\left(n_\Short\right)$ and $|\Long'|=o\left(n_\Long\right)$, such that every stable matching} on the vertex-induced subgraph of $H$ on $\left(\Short\backslash \Short' \right)\cup \left(\Long\backslash \Long'\right)$ is perfect interim stable with high probability.  
\end{remark}

To address the potential uncertainty in the sparse case (see e.g., \prettyref{thm:single_both_signal_sparse}), where both-side signaling could possibly make both parties worse off, we introduce the ``restricted'' multi-tiered signaling mechanism.
The ``restricted'' multi-tiered signaling mechanism is identical to the original mechanism, except in cases where a firm tier and an applicant tier are each other's target tiers, in which case only one tier is allowed to signal the other. 
\begin{remark}\label{rmk:restricted_sparse}
Extending \prettyref{thm:multi_signal_sparse}, let $H$ denote an interview graph constructed based on the ``restricted'' multi-tiered signaling mechanism with $\omega (1) \le d\le O\left(\polylog n \right)$. Then, if $p = \omega\left(1/\log d\right)$, every stable matching on $H$ is almost interim stable with high probability.
\end{remark}

It is possible that the pre-interview scores and post-interview scores do not preserve the tier structure, meaning that an agent might prefer another agent from a lower tier over an agent from a higher tier. To address this scenario, we consider a relaxation of Assumption~\ref{assump:bounded}. Suppose $\DistA$ and $\DistB$ are bounded distributions with maximum values of their supports denoted by $M_\DistA$ and $M_\DistB$, respectively. We define $q$ as the probability that the sum of independently drawn random variables $A \sim \DistA$ and $B \sim \DistB$ exceeds $M_\DistA + M_\DistB - 1$, i.e.,
\begin{align}
q \triangleq \prob{A + B > M_\DistA + M_\DistB - 1}, \quad \text{where } A \sim \DistA\,,\, B \sim \DistB \text{ independently} \, . \label{eq:q_AB}
\end{align}

\begin{remark}\label{rmk:relaxed_bounded_sparse}

{
Extending \prettyref{thm:multi_signal_sparse}, allow bounded pre-interview and post-interview scores that need not preserve the tier ordering. Suppose that $p\log d\to\infty$, $q\log d\to\infty$, and the market is $\gamma$-generally imbalanced with $\gamma dq/\log d\to\infty$. Then every stable matching is almost interim stable with high probability. This assertion concerns restrictions of the original stable matchings; it does not extend the common-reduced-market or the restricted-signaling conclusions to arbitrary preference reversals.
}%

\end{remark}

The proofs of \prettyref{thm:multi_signal_sparse}, and Remarks \ref{rmk:multi_signal_sparse_identify}, \ref{rmk:restricted_sparse} and \ref{rmk:relaxed_bounded_sparse} are deferred to \prettyref{sec:multi_sparse_proof}. 

\subsubsection{Multi-tiered market with dense signals}\label{sec:multi_dense}

{

{
We now analyze the dense signaling regime. When cumulative tier counts coincide, the following theorem guarantees both prescribed optimal matchings for log-concave shocks with $p$ bounded away from zero, as well as under the upper-atom alternative. Assumption~\ref{assump:bounded} continues to apply. In generally imbalanced markets, the all-stable-matchings conclusion requires neither additional distributional condition. The required number of signals depends on the market's general imbalance.
}%

}%

\begin{theorem}[Effectiveness on multi-tiered signaling with dense signals]
\label{thm:multi_signal_dense}
Consider a multi-tiered two-sided market with applicants $\Short = \bigcup_{\ess=1}^m \Short_{\ess}$ and firms $\Long = \bigcup_{\kappa=1}^\ell \Long_\kappa$ with $n_{\Short}= \delta n_{\Long}$, where $\Omega(1) \le \delta \le 1$.
Let $H$ denote an interview graph constructed based on the multi-tiered signaling mechanism with $d \ge \underline{d}/p$, for some $\underline{d}$ that only depends on $\boldsymbol{\alpha},\boldsymbol{\beta},\delta$ and $n$.
\begin{itemize}

{
\item Suppose the market is not generally imbalanced. Assume additionally either that $p\ge p_0$ for some fixed $p_0>0$ and $\DistA$ has a log-concave density, or that, for some $c\ge0$, $A\le c$ almost surely with $\Pr\{A=c\}\ge p/2$. Then $\underline{d}=\Theta\left(\log^2n\right)$, and both the applicant-optimal and the firm-optimal stable matchings are perfect interim stable with high probability.
}%

\item If the market is $\gamma$-generally imbalanced with some $0<\gamma <1$, and $\underline{d} = \Theta \left( \log ( 1/\gamma ) \log n \right)$, every stable matching on $H$ is perfect interim stable with high probability. 
\end{itemize}
\end{theorem}

{

{
The following guarantee permits bounded scores that do not preserve tier ordering when one firm tier can receive all applicants with a fixed surplus. The counterexample in \prettyref{sec:relaxed_bounded_dense} shows that the signal budget $d\gg\log^2n/(p\wedge q)$ alone is insufficient.
}%

}%

\begin{remark}\label{rmk:relaxed_bounded_dense}

{
Allow bounded pre-interview and post-interview scores that need not preserve tier ordering. Suppose $n_\Short\le(1-\varepsilon)|\Long_\ell|$ for a fixed $\varepsilon>0$, so every applicant targets the highest firm tier and no firm signals. If $d\ge(8/\varepsilon)\log n/(p\wedge q)$, where $q$ is defined in \prettyref{eq:q_AB}, every stable matching is perfect interim stable with high probability.
}%

\end{remark}

The proofs of \prettyref{thm:multi_signal_dense} and \prettyref{rmk:relaxed_bounded_dense} are deferred to \prettyref{sec:multi_dense_proof}. 

\subsection{Incentive compatibility}\label{sec:incentive}

{
We investigate incentives in the signaling stage, with applicant-proposing deferred acceptance fixed as the post-interview matching rule. Signals are simultaneous and are chosen from an agent's own pre-interview scores before any interview shocks or other agents' signals are observed. Post-interview rankings are reported truthfully. An agent may deviate to any, possibly randomized, set of at most $d$ recipients; truthful signaling selects its top $d$ partners in the prescribed target tier. Throughout this subsection, Assumptions~\ref{assump:general}, \ref{assump:non_vanishing}, and \ref{assump:bounded} hold, and the unmatched payoff is a fixed finite value below every matched payoff. The ordinal convention $U_{i,\emptyset}=-\infty$ is not used for expected-payoff comparisons.
For the approximate conclusion, we use ex ante expected payoffs: writing $\sigma^*$ for prescribed signaling and $V_i$ for payoff, the assertion means
\[
 \max_i\sup_{\sigma_i}
 \mathbb E\!\left[V_i(\sigma_i,\sigma_{-i}^*)
                    -V_i(\sigma_i^*,\sigma_{-i}^*)\right]\le\epsilon_n,
 \qquad \epsilon_n\to0.
\]
The expectation includes the deviator's own private type; the bound is uniform over agents and admissible strategies, not over every realized score row.
}%

{
\begin{theorem}[Incentive compatibility for truthful signaling]\label{thm:incentive}
Let $H_n$ be a sequence of random multi-tiered markets with applicants $\Short = \bigcup_{\ess=1}^m \Short_{\ess}$ and firms $\Long = \bigcup_{\kappa=1}^\ell \Long_\kappa$ with $n_{\Short} \le n_{\Long}$, paired with a multi-tiered signaling mechanism. Suppose each agent can send at most $d$ signals and has only distributional information about other agents' preferences. 

{
Let $H$ denote an interview graph constructed based on the multi-tiered signaling mechanism. 
}%

\begin{itemize}
\item If $\ell=m=1$, truthful signaling is an exact best response for every applicant when the other agents follow the signaling prescription. If only applicants are permitted to signal, truthful signaling is a Bayes--Nash equilibrium. For firm-proposing deferred acceptance, signaling the top $\min\{d,n_{\Short}\}$ applicants is likewise an exact best response for every firm; if only firms may signal, this top-choice profile is a Bayes--Nash equilibrium.
\item 
{
If $d\to\infty$, $\DistA$ and $\DistB$ are fixed bounded laws, and the market is $\gamma$-generally imbalanced with $\gamma\ge\Omega(1)$,
}%
 truthful targeted signaling is an $\epsilon_n$-Bayes--Nash equilibrium in the ex ante sense above, with $\epsilon_n=o(1)$.
\end{itemize}
\end{theorem}
}%

{

{
The exact conclusion concerns the proposing side; the counterexample in \prettyref{sec:thm_incentive} shows why allowing both sides to signal need not give an exact equilibrium. The approximate conclusion uses fixed score laws and constant general imbalance. Nearly balanced blocks and score laws varying with market size require separate deviation-payoff estimates.
}%

}%

The proof of \prettyref{thm:incentive} is deferred to \prettyref{sec:incentive_proof}. 

\section{{Proof sketch using local neighborhoods}} \label{sec:message}

{
Local availability certificates underlie the sparse-market results. We outline the method here; \prettyref{sec:local} gives the full analysis.
}%

{
For an agent $\rho$ and a neighbor $i$ in a bipartite graph $H$ with strict preferences, $i$ is \emph{available} to $\rho$ if $\rho\succeq_i\phi(i)$ in every stable matching $\Phi$ on $H$. Stability then implies $\phi(\rho)\succeq_\rho i$ in every such matching.
}%

{
To illustrate the certificate used in \prettyref{thm:single_signal_sparse}, consider applicant signaling and let $\Long_a$ be the top $d$ firms signaled by applicant $a$. If some $j\in\Long_a$ is available to $a$ and $A_{a,j}\ge0$, then in every stable matching $\Phi$ on $H$,
\[
 U_{a,\phi(a)}^A\ge U_{a,j}^A
 =U_{a,j}^B+A_{a,j}\ge U_{a,j}^B
 >\max_{j'\notin\Long_a}U_{a,j'}^B.
\]
The first inequality follows from stability and availability; the last follows from the strict pre-interview ordering between signaled and unsignaled firms. When $a$ interviews every firm, there is no off-graph alternative to check. Thus a non-negative shock suffices, including a zero shock.

Let $\Short_{\rm bad}$ contain applicants without this certificate. If each applicant fails the certificate with probability $o(1)$ uniformly in the applicant, then $\expect{|\Short_{\rm bad}|}=o(n_{\Short})$, and Markov's inequality yields $|\Short_{\rm bad}|=o(n_{\Short})$ with high probability. For any stable matching $\Phi$, remove $\Short_{\rm bad}$ together with their matched firms. At most $2|\Short_{\rm bad}|$ agents are removed, every surviving match is unchanged, and every surviving applicant rejects every off-graph alternative. Stability on interviewed pairs is also preserved. This proves the almost-stability conclusion once the uniform certificate probability has been established; simultaneous certification of all applicants is a stronger requirement relevant to perfect interim stability.
}%

{
We obtain these certificates from truncated neighborhoods.
}%

\paragraph{Truncation on local neighborhood.}

{
Recall that $H_m(\rho)$ is the vertex-induced $m$-hop neighborhood of $\rho$ in $H$.
}%

{
By~\cite[Theorems 1 and 2]{crawford1991}, when agents are removed from one side of a bipartite graph, the remaining agents on that side are weakly better off and agents on the opposite side are weakly worse off when comparing the corresponding applicant-optimal matchings, or the corresponding firm-optimal matchings, before and after removal.
}%

{
For either fixed choice of proposing side, compare its deferred-acceptance outcome in $H_m(\rho)$ with its outcome in $H$. If $m$ is even, $\rho$ is weakly worse off in $H_m(\rho)$; if $m$ is odd, $\rho$ is weakly better off there (\prettyref{lmm:local}).
}%
 To see this, note that $H_m(\rho)$ can be obtained from $H$ by removing all agents at depth $m+1$ and considering the connected component containing $\rho$ in the remaining graph. Since $H$ is a bipartite graph, the agents removed at depth $m+1$ are on the opposite side of the market from $\rho$ when $m$ is even and on the same side when $m$ is odd, from which the claim follows.
This claim leads to a key observation: for any agent $\rho$ and its neighbor $i$ on $H$, if $m$ is even and $i$ is available to $\rho$ in $H_m(\rho)$, then $i$ must also be available to $\rho$ in $H$; if $m$ is odd and $i$ is not available to $\rho$ in $H_m(\rho)$, then $i$ must also not be available to $\rho$ in $H$ (\prettyref{lmm:local_available}). This observation is further illustrated in \prettyref{fig:every}. We refer to \prettyref{sec:truncation} for detailed proofs and further monotonicity results concerning truncated neighborhoods.

\begin{figure}[ht]
\hspace{0.5cm}
\begin{subfigure}[c]{.5\linewidth}
    \centering
    \vspace{0.25cm}
    \begin{tikzpicture}[scale=0.54, every node/.style={scale=0.54}] 
        \node[circle, draw, fill=gray!30] (a1) at (-4,3) {$a_1$};
        \node[circle, draw, fill=gray!30] (a2) at (-4,2) {$a_2$};
        \node[circle, draw, fill=gray!30] (a3) at (-4,1) {$a_3$};
        \node[circle, draw, fill=gray!30] (a4) at (-4,0) {$a_4$};
        \node[circle, draw, fill=gray!30] (a5) at (-4,-1) {$a_5$};
        \node[circle, draw, fill=gray!30] (a6) at (-4,-2) {$a_6$};
        \node[circle, draw, fill=gray!30] (a7) at (-4,-3) {$a_7$};
        \node[circle, draw, fill=gray!30] (a8) at (-4,-4) {$a_8$};
        \node[circle, draw, fill=gray!30] (a9) at (-4,-5) {$a_9$};
        \node[circle, draw, fill=gray!30] (a10) at (-4,-6) {$a_{10}$};
        \node[circle, draw, fill=gray!30] (a11) at (-4,-7) {$a_{11}$};
        \node[circle, draw, fill=gray!30] (a12) at (-4,-8) {$a_{12}$};
        \node[circle, draw, fill=gray!30] (a13) at (-4,-9) {$a_{13}$};
        \node[circle, draw, fill=gray!30] (a14) at (-4,-10) {$a_{14}$};
        \node[circle, draw] (j1) at (0,3) {$j_1$};
        \node[circle, draw] (j2) at (0,2) {$j_2$};
        \node[circle, draw] (j3) at (0,1) {$j_3$};
        \node[circle, draw] (j4) at (0,0) {$j_4$};
        \node[circle, draw] (j5) at (0,-1) {$j_5$};
        \node[circle, draw] (j6) at (0,-2) {$j_6$};
        \node[circle, draw] (j7) at (0,-3) {$j_7$};
        \node[circle, draw] (j8) at (0,-4) {$j_8$};
        \node[circle, draw] (j9) at (0,-5) {$j_9$};
        \node[circle, draw] (j10) at (0,-6) {$j_{10}$};
        \node[circle, draw] (j11) at (0,-7) {$j_{11}$};
        \node[circle, draw] (j12) at (0,-8) {$j_{12}$};
        \node[circle, draw] (j13) at (0,-9) {$j_{13}$};
        \draw (a1) -- (j1);
        \draw (a1) -- (j2);
        \draw (a1) -- (j3);
        \draw (j1) -- (a2);
        \draw (j1) -- (a3);
        \draw (j2) -- (a4);
        \draw (j3) -- (a5);
        \draw (j3) -- (a6);
        \draw (a2) -- (j4);
        \draw (a2) -- (j5);
        \draw (a3) -- (j6);
        \draw (a3) -- (j7);
        \draw (a4) -- (j8);
        \draw (a5) -- (j9);
        \draw (a5) -- (j10);
        \draw (a6) -- (j11);
        \draw (a6) -- (j12);
        \draw (a7) -- (j4);
        \draw (a8) -- (j4);
        \draw (a9) -- (j5);
        \draw (a10) -- (j8);
        \draw (a11) -- (j11);
        \draw (a12) -- (j12);
        \draw (a13) -- (j12);
        \draw (a14) -- (j13);
        \draw (a10) -- (j11);
        \draw (a11) -- (j13);
        \draw (a14) -- (j12);
        \draw (a10) -- (j9);
        \end{tikzpicture}
    \caption
      {The bipartite graph $H$
        \label{fig:bipartite}
      }%
  \end{subfigure}\hspace{-3cm}
\begin{tabular}[c]{@{}c@{}}
    \begin{subfigure}[c]{0.4\linewidth}
      \centering
        \begin{tikzpicture}[scale=0.6, every node/.style={scale=0.6}]  %
    \node[circle, draw, fill=gray!30] (rho) at (0,2) {$a_1$};
    \node[circle, draw] (a) at (-2,0) {$j_1$};
    \node[circle, draw] (o) at (0,0) {$j_2$};
    \node[circle, draw] (b) at (2,0) {$j_3$};
    \node[circle, draw, fill=gray!30] (c) at (-3,-2) {$a_2$};
    \node[circle, draw, fill=gray!30] (d) at (-1,-2) {$a_3$};
    \node[circle, draw, fill=gray!30] (p) at (0,-2) {$a_4$};
    \node[circle, draw, fill=gray!30] (e) at (1,-2) {$a_5$};
    \node[circle, draw, fill=gray!30] (f) at (3,-2) {$a_6$};
    \draw (rho) -- (a);
    \draw (rho) -- (o);
    \draw (rho) -- (b);
    \draw (a) -- (c);
    \draw (a) -- (d);
    \draw (o) -- (p);
    \draw (b) -- (e);
    \draw (b) -- (f);
    \end{tikzpicture}
      \caption
        {%
         Truncation at depth $2$: $H_2(a_1)$.
          \label{fig:upper}%
        }%
    \end{subfigure}\\
    \noalign{\smallskip}%
    \begin{subfigure}[c]{0.5\linewidth}
      \centering
        \begin{tikzpicture}[scale=0.6, every node/.style={scale=0.6}]
        \node[circle, draw, fill=gray!30] (rho2) at (0,-4) {$a_1$};
        \node[circle, draw] (a2) at (-2,-6) {$j_1$};
        \node[circle, draw] (o2) at (0,-6) {$j_2$};
        \node[circle, draw] (b2) at (2,-6) {$j_3$};
        \node[circle, draw, fill=gray!30] (c2) at (-3,-8) {$a_2$};
        \node[circle, draw, fill=gray!30] (d2) at (-1,-8) {$a_3$};
        \node[circle, draw, fill=gray!30] (p2) at (0,-8) {$a_4$};
        \node[circle, draw, fill=gray!30] (e2) at (1,-8) {$a_5$};
        \node[circle, draw, fill=gray!30] (f2) at (3,-8) {$a_6$};
        \node[circle, draw] (g2) at (-4,-10) {$j_4$};
        \node[circle, draw] (h2) at (-3,-10) {$j_5$};
        \node[circle, draw] (i2) at (-2,-10) {$j_6$};
        \node[circle, draw] (j2) at (-1,-10) {$j_7$};
        \node[circle, draw] (q2) at (0,-10) {$j_8$};
        \node[circle, draw] (k2) at (1,-10) {$j_9$};
        \node[circle, draw] (l2) at (2,-10) {$j_{10}$};
        \node[circle, draw] (m2) at (3,-10) {$j_{11}$};
        \node[circle, draw] (n2) at (4,-10) {$j_{12}$};
        \draw (rho2) -- (a2);
        \draw (rho2) -- (o2);
        \draw (rho2) -- (b2);
        \draw (a2) -- (c2);
        \draw (a2) -- (d2);
        \draw (c2) -- (g2);
        \draw (c2) -- (h2);
        \draw (d2) -- (i2);
        \draw (d2) -- (j2);
        \draw (o2) -- (p2);
        \draw (p2) -- (q2);
        \draw (b2) -- (e2);
        \draw (b2) -- (f2);
        \draw (e2) -- (k2);
        \draw (e2) -- (l2);
        \draw (f2) -- (m2);
        \draw (f2) -- (n2);
        \end{tikzpicture}      \caption
        {%
           Truncation at depth $3$: $H_3(a_1)$.
          \label{fig:lower}%
        }%
    \end{subfigure}
  \end{tabular}
  \caption
    {{Given $H$ as a two-sided market} with applicants $\calA =\{a_1,a_2,\cdots, a_{14}\}$ and firms $\Long =\{j_1,j_2,\cdots, j_{13}\}$, $H_2(a_1)$ and $H_3(a_1)$ are the $2$-hop and $3$-hop neighborhoods of $a_1$ on $H$, respectively. In $H_2(a_1)$, $a_1$ is weakly worse off compared to $H$, while in $H_3(a_1)$, $a_1$ is weakly better off compared to $H$. For example, if $j_1$ is available to $a_1$ on $H_2(a_1)$, it must also be available to $a_1$ in $H$, and if $j_1$ is not available to $a_1$ in $H_3(a_1)$, $j_1$ must not be available to $a_1$ in $H$.
      \label{fig:every}%
    }
\end{figure}
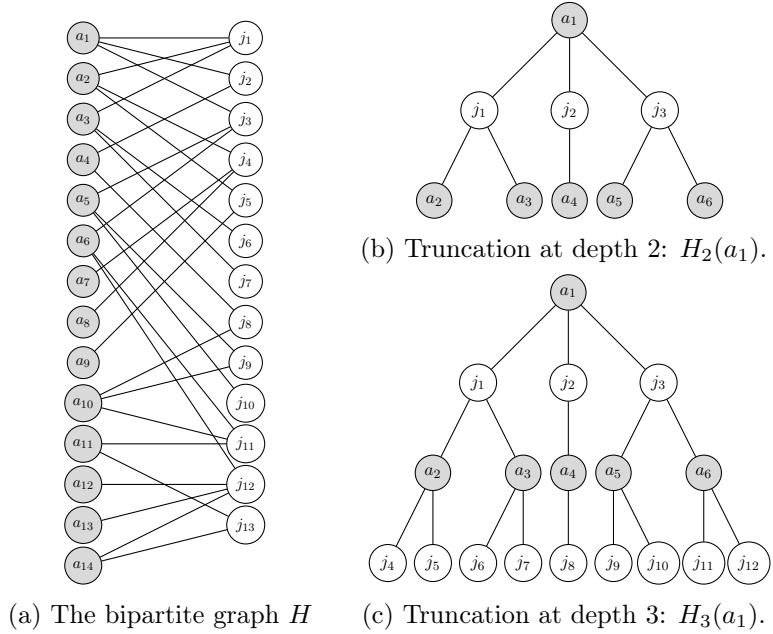

\paragraph{Message-passing algorithm on tree.}
When the local neighborhood around an agent is almost tree-like (i.e., contains only a constant number of cycles), we can obtain a tree by truncating the local neighborhood through the removal of a constant number of vertices. The stable matching on a tree is unique (\prettyref{lmm:tree_unique}), allowing us to apply the hierarchical proposal-passing algorithm (\prettyref{alg:proposal_passing_alg}) to find this matching and determine the availability of agents to the root node.

This algorithm consists of two phases: the proposing phase and the clearing matching phase. In the proposing phase, operating from the bottom to the top of the tree, each node may receive proposals from its child nodes and chooses to propose to its parent if it prefers the parent to all received proposals. In the clean-up matching phase, operating from the top to the bottom of the tree, each node that is not matched to its parent and has received proposals accepts the most favorable proposal and matches with the corresponding node.
An illustrative example of the hierarchical proposal-passing algorithm is shown in \prettyref{fig:both_phases_updated_7}. 
We refer to \prettyref{sec:proposal_passing_alg} for the complete algorithm description and proof of its correctness in finding the stable matching.

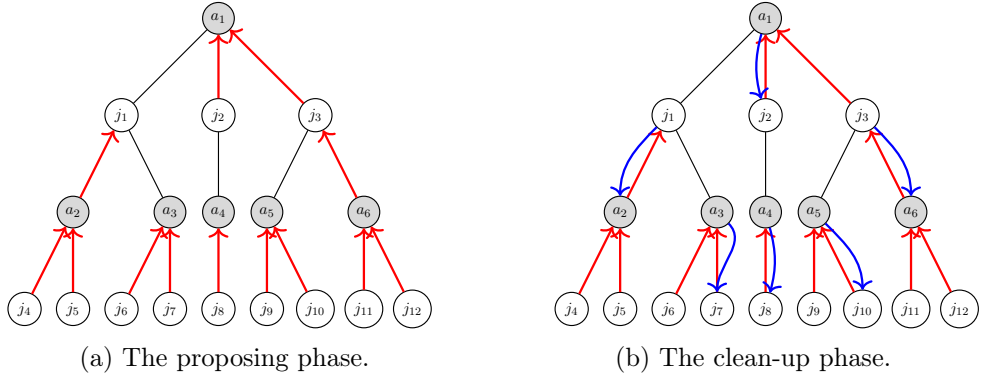
\begin{figure}[ht]
\centering
\begin{tabular}{C{.4\textwidth}C{.4\textwidth}}  %
\subcaptionbox{The proposing phase.}%
[0.9\linewidth]{
\resizebox{0.35\textwidth}{!}{%
\centering
\begin{tikzpicture}[scale=0.6, every node/.style={scale=0.5}] 
\node[circle, draw, fill=gray!30] (rho) at (0,0) {$a_1$};
\node[circle, draw] (a) at (-2,-2) {$j_1$};
\node[circle, draw] (o) at (0,-2) {$j_2$};
\node[circle, draw] (b) at (2,-2) {$j_3$};
\node[circle, draw, fill=gray!30] (c) at (-3,-4) {$a_2$};
\node[circle, draw, fill=gray!30] (d) at (-1,-4) {$a_3$};
\node[circle, draw, fill=gray!30] (p) at (0,-4) {$a_4$};
\node[circle, draw, fill=gray!30] (e) at (1,-4) {$a_5$};
\node[circle, draw, fill=gray!30] (f) at (3,-4) {$a_6$};
\node[circle, draw] (g) at (-4,-6) {$j_4$};
\node[circle, draw] (h) at (-3,-6) {$j_5$};
\node[circle, draw] (i) at (-2,-6) {$j_6$};
\node[circle, draw] (j) at (-1,-6) {$j_7$};
\node[circle, draw] (q) at (0,-6) {$j_8$};
\node[circle, draw] (k) at (1,-6) {$j_9$};
\node[circle, draw] (l) at (2,-6) {$j_{10}$};
\node[circle, draw] (m) at (3,-6) {$j_{11}$};
\node[circle, draw] (n) at (4,-6) {$j_{12}$};
\draw[<-, thick, red] (rho) -- (b);
\draw (rho) -- (a);
\draw[<-, thick, red] (rho) -- (o);
\draw[<-, thick, red] (a) -- (c);
\draw (a) -- (d);
\draw (b) -- (e);
\draw (b) -- (f);
\draw[<-, thick, red] (c) -- (g);
\draw[<-, thick, red] (c) -- (h);
\draw[<-, thick, red] (d) -- (i);
\draw[<-, thick, red] (d) -- (j);
\draw[<-, thick, red] (e) -- (k);
\draw[<-, thick, red] (e) -- (l);
\draw[<-, thick, red] (f) -- (m);
\draw[<-, thick, red] (f) -- (n);
\draw (o) -- (p);
\draw[<-, thick, red] (p) -- (q);
\end{tikzpicture}
}}&
\subcaptionbox{The clean-up phase.}[0.9\linewidth]{
\resizebox{0.35\textwidth}{!}{%
\centering
\begin{tikzpicture}[scale=0.6, every node/.style={scale=0.5}]  %
\node[circle, draw, fill=gray!30] (rho) at (0,0) {$a_1$};
\node[circle, draw] (a) at (-2,-2) {$j_1$};
\node[circle, draw] (o) at (0,-2) {$j_2$};
\node[circle, draw] (b) at (2,-2) {$j_3$};
\node[circle, draw, fill=gray!30] (c) at (-3,-4) {$a_2$};
\node[circle, draw, fill=gray!30] (d) at (-1,-4) {$a_3$};
\node[circle, draw, fill=gray!30] (p) at (0,-4) {$a_4$};
\node[circle, draw, fill=gray!30] (e) at (1,-4) {$a_5$};
\node[circle, draw, fill=gray!30] (f) at (3,-4) {$a_6$};
\node[circle, draw] (g) at (-4,-6) {$j_4$};
\node[circle, draw] (h) at (-3,-6) {$j_5$};
\node[circle, draw] (i) at (-2,-6) {$j_6$};
\node[circle, draw] (j) at (-1,-6) {$j_7$};
\node[circle, draw] (q) at (0,-6) {$j_8$};
\node[circle, draw] (k) at (1,-6) {$j_9$};
\node[circle, draw] (l) at (2,-6) {$j_{10}$};
\node[circle, draw] (m) at (3,-6) {$j_{11}$};
\node[circle, draw] (n) at (4,-6) {$j_{12}$};
\draw[<-, thick, red] (rho) -- (b);
\draw (rho) -- (a);
\draw[<-, thick, red] (rho) -- (o);
\draw[<-, thick, red] (a) -- (c);
\draw (a) -- (d);
\draw (b) -- (e);
\draw (b) -- (f);
\draw[<-, thick, red] (c) -- (g);
\draw[<-, thick, red] (c) -- (h);
\draw[<-, thick, red] (d) -- (i);
\draw[<-, thick, red] (d) -- (j);
\draw[<-, thick, red] (e) -- (k);
\draw[<-, thick, red] (e) -- (l);
\draw[<-, thick, red] (f) -- (m);
\draw[<-, thick, red] (f) -- (n);
\draw (o) -- (p);
\draw[<-, thick, red] (p) -- (q);
\draw[->, thick, blue] (rho) to[out=255, in=105] (o);
\draw[->, thick, blue] (a) to[out=225, in=90] (c);
\draw[->, thick, blue] (f) to[out=225, in=90] (m);
\draw[->, thick, blue] (d) to[out=315, in=90] (j);
\draw[->, thick, blue] (e) to[out=315, in=90] (l);
\draw[->, thick, blue] (p) to[out=285, in=75] (q);
\end{tikzpicture}
}}
\end{tabular}
\caption{An illustration of the hierarchical proposal-passing algorithm on the $3$-hop neighborhood of $a_1$ on $H$, denoted as $H_3(a_1)$ as shown in Figure \ref{fig:every}, with truncated preferences shown in \prettyref{tab:preferences_updated}. (a) The proposing phase (from bottom to top): all proposals are indicated by red arrows.
(b) The clean-up matching phase (from top to bottom): all accepted proposals are indicated by blue arrows.
}
\label{fig:both_phases_updated_7}
\end{figure}

\begin{table}[ht]
\centering
\begin{tabular}{c|c|c|c}
Applicant & Applicant preferences & Firm & Firm preferences  \\
$a_1$ & $j_1 \succ j_2 \succ j_3$ & $j_1$ & $a_2 \succ a_3 \succ a_1$\\
$a_2$ & $j_1 \succ j_4 \succ j_5$ &  $j_2$ & $a_4 \succ a_1$ \\
$a_3$ & $j_7 \succ j_6 \succ j_1$ & $j_3$ & $a_5 \succ a_6 \succ a_1$ \\
$a_4$ & $j_8 \succ j_2$\\
$a_5$ & $j_{10} \succ j_3 \succ j_9$\\
$a_6$ & $j_{11} \succ j_{12} \succ j_3$\\

\end{tabular}
\caption{Truncated preferences for each agent with respect to its neighbors in the 3-hop neighborhood of $a_1$ on $H$, denoted as $H_3(a_1)$, as shown in Figure \ref{fig:lower}. The preferences of leaf nodes $\{j_4, j_5, \ldots, j_{12}\}$ are omitted as they each have only one neighbor.}
\label{tab:preferences_updated}
\end{table}

{
On a tree, a neighbor is available to the root exactly when it proposes to the root in the bottom-up phase. Thus proposal probabilities give availability probabilities. We analyze these probabilities for almost regular trees and random branching trees in \prettyref{sec:message_passing_alg}.
}%

{
For a general bipartite graph, odd- and even-radius truncations give
upper and lower availability certificates by \prettyref{lmm:local_available}.
For random graphs, the sequential exposure in \prettyref{sec:stable}
controls collisions and branch dependence through
\eqref{eq:collision_factorial_bound} and \eqref{eq:all_cycles_joint_product}.
These estimates yield the availability and stable-matching results below.
}%

\section{Numerical Results}\label{sec:numerical}

We present numerical results on synthetic data to corroborate our theoretical findings in the single-tiered market setting. %

\begin{figure}[ht]
    \centering
    \subfloat[\centering $d=10$]
    {{\includegraphics[width=.48\linewidth]
    {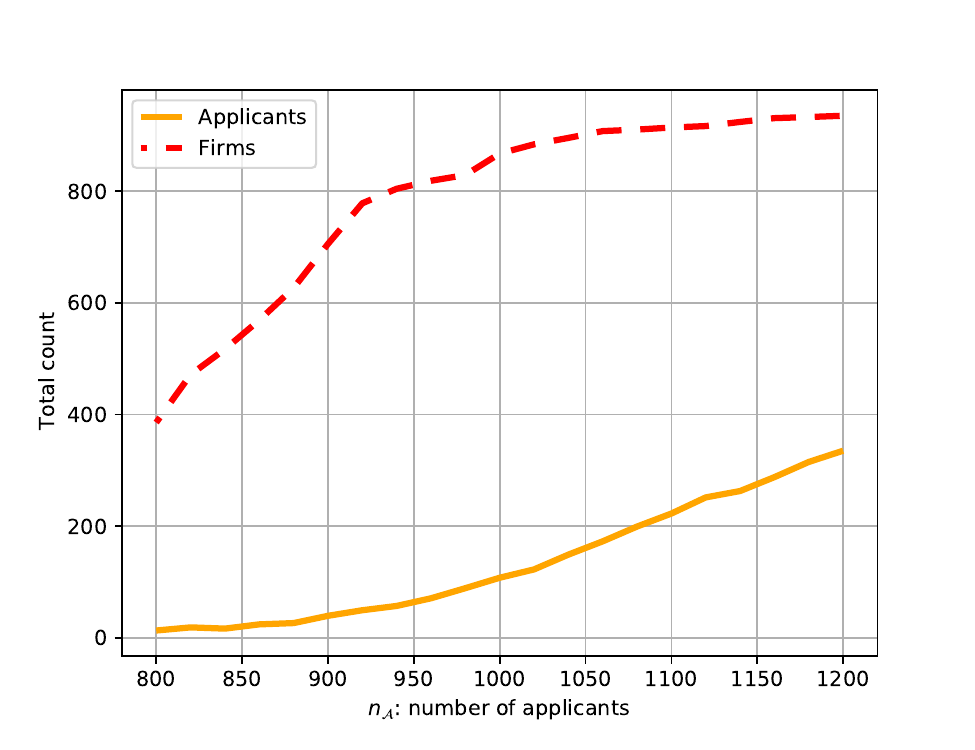} }}%
    \subfloat[\centering  $d=20$]
    {{\includegraphics[width=.48\linewidth]{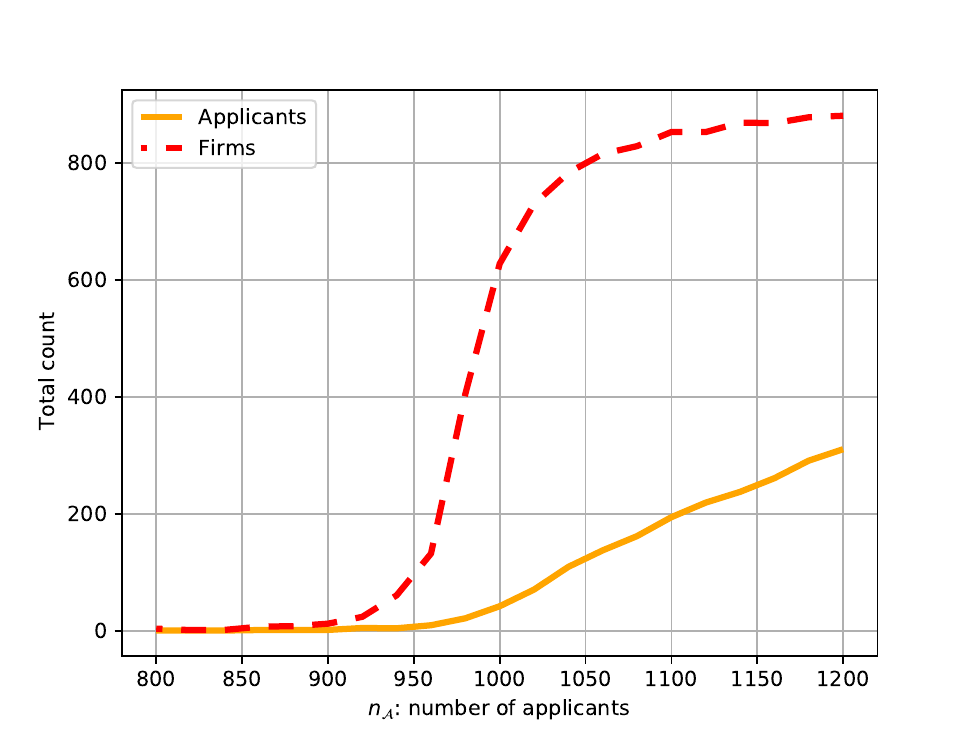} }}%
    \caption{The number of applicants and the number of firms involved in at least one interim blocking pair in the applicant-optimal stable matching for $d=10, 20$ respectively, with $800 \leq n_{\mathcal{A}} \leq 1200$, $n_{\mathcal{J}} = 1000$, $\DistB = \mathbb{N}(0,1)$, and $\DistA = \mathbb{U}[-1,1]$, where the interview graph is constructed by the applicant-signaling mechanism. Each data point represents the average over $10$ trials of simulations.}%
    \label{fig:increasing_n_A}%
\end{figure}

\prettyref{fig:increasing_n_A} illustrates how market imbalance influences the effectiveness of one-side signaling in achieving interim stability. 

As $n_{\Short}$ increases from 800 to 1200 (with $n_{\Long}$ fixed at 1000), the number of both applicants and firms involved in interim blocking pairs increases for both $d=10$ and $d=20$. This trend indicates the decreasing effectiveness of achieving interim stability as the market becomes more applicant-heavy, transitioning from short-side to long-side signaling.

Moreover, with fixed $n_{\Short}$ and $n_{\Long}$, the number of agents involved in interim blocking pairs is consistently higher for $d=10$ than for $d=20$. This comparison demonstrates that increasing the number of signals significantly enhances the effectiveness of achieving interim stability.
{For $d=20$, participation in interim blocking pairs is low when $n_{\Short}$ is well below $n_{\Long}$ and rises sharply near balance.} 

\begin{figure}[ht]
    \centering
    \includegraphics[width=0.7\linewidth]{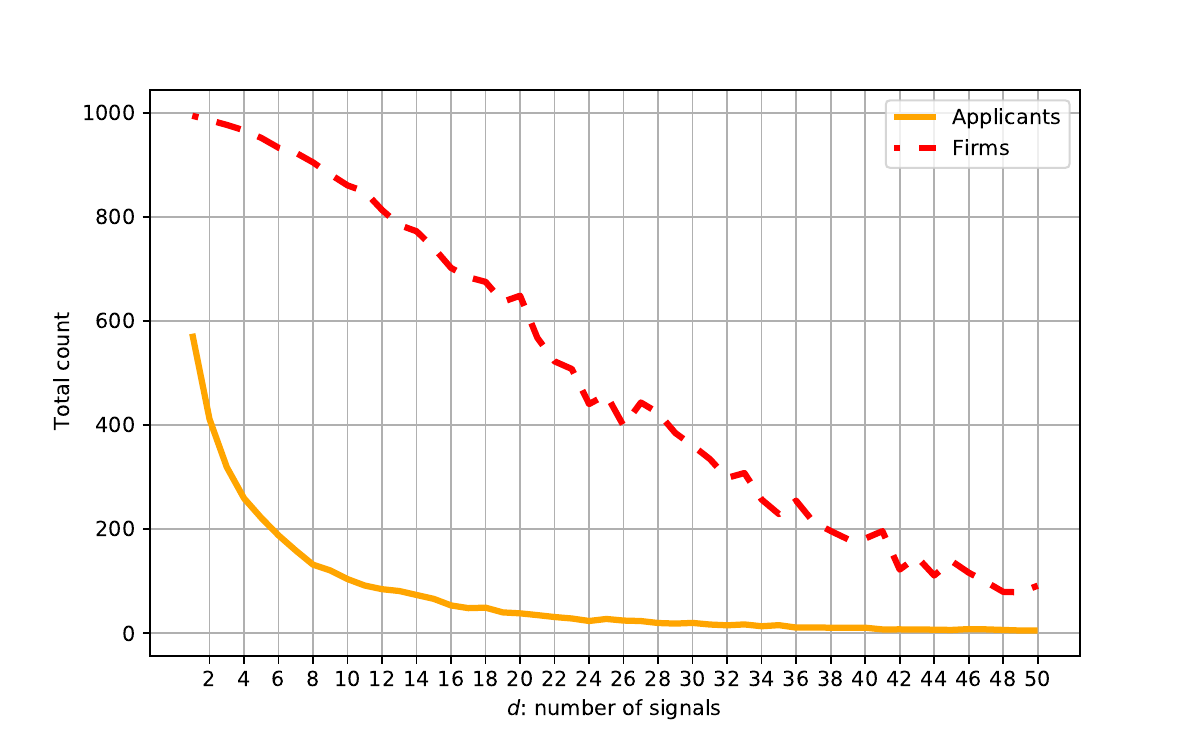}
    \caption{The number of applicants and the number of firms involved in at least one interim blocking pair in the applicant-optimal stable matching, $n_{\mathcal{A}}=n_{\mathcal{J}}=1000$, $1\leq d\leq 50$, $\DistB = \mathbb{N}(0,1)$, and $\DistA=\mathbb{U}[-1,1]$, where the interview graph is constructed by the applicant-signaling mechanism. Each data point represents the average over $10$ trials of simulations.}
    \label{fig:d_increase}
\end{figure}

{
\prettyref{fig:d_increase} plots the numbers of applicants and firms involved in at least one interim blocking pair as $d$ increases in a balanced market.
}%
As shown in the figure, both counts decrease significantly as $d$ increases, demonstrating the improved effectiveness of the signaling mechanism with more signals. 
{
These counts describe participation in blocking pairs. They do not equal the minimum number of agents whose removal yields perfect interim stability: blocking pairs may share endpoints, and deleting an agent may leave its former partner unmatched.
}%
 The number of firms involved in interim blocking pairs consistently exceeds that of applicants, which is a result of constructing the interview graph based on the applicant-signaling mechanism. 

\begin{figure}[!ht]
    \centering
    \subfloat[\centering $\DistB = \mathbb{N}(0,1)$ and $\DistA = \boldsymbol\delta_0 $]
    {{\includegraphics[width=.48\linewidth]
    {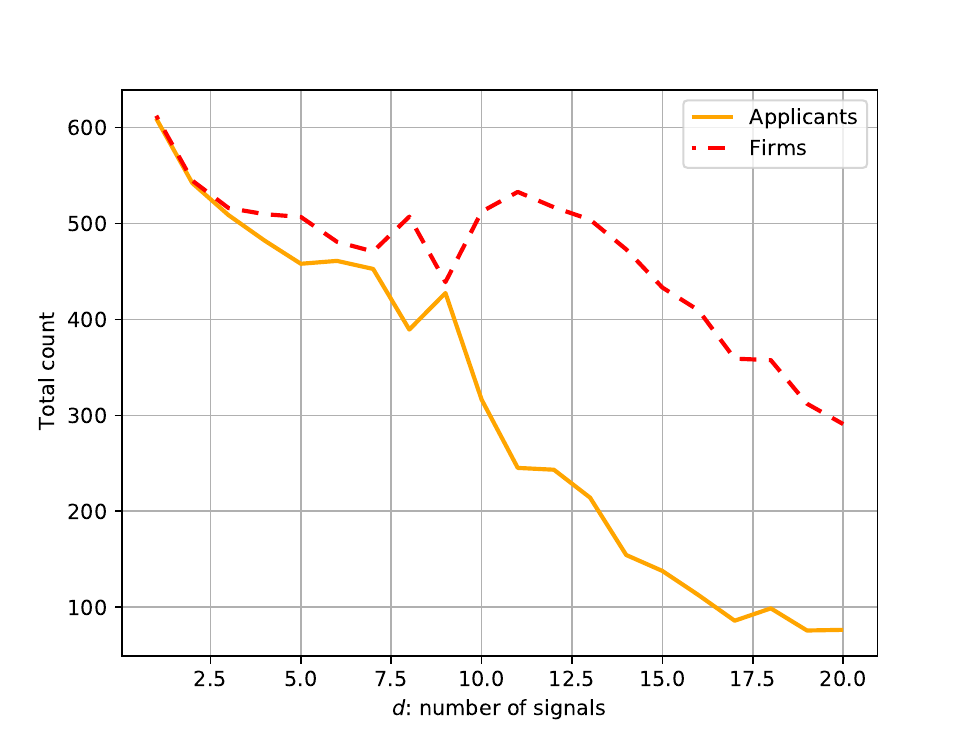} }}%
    \subfloat[\centering $\DistB = \boldsymbol\delta_0$ and $\DistA = {\mathbb{U}\left[-1,1\right]} $]
    {{\includegraphics[width=.48\linewidth]{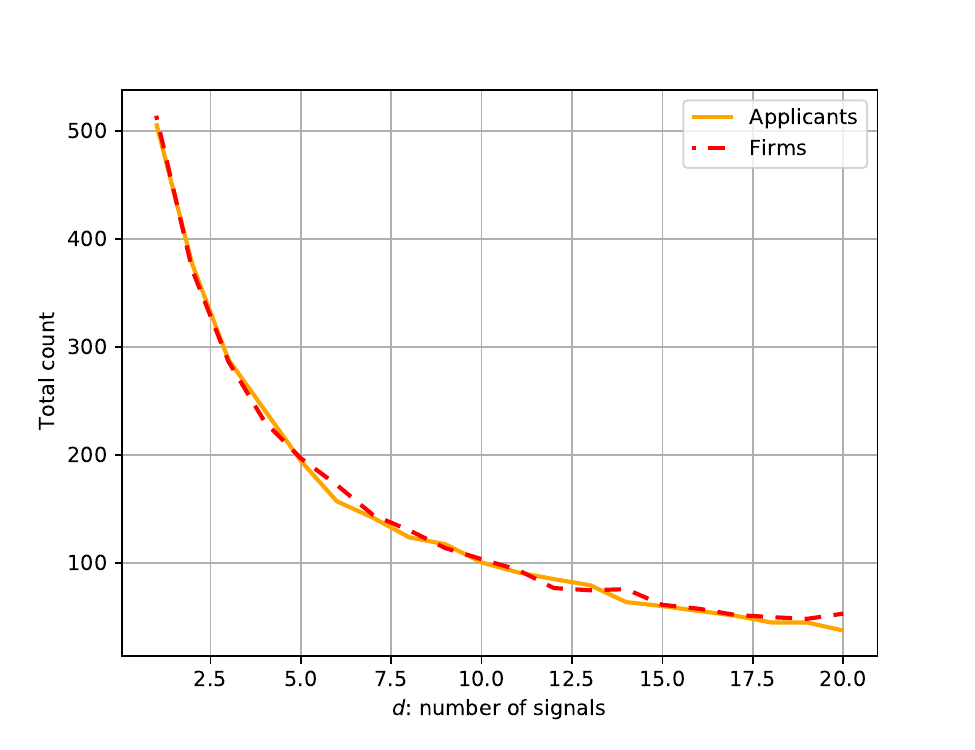} }}%
    \caption{The number of applicants and firms involved in at least one interim blocking pair in the applicant-optimal stable matching with $n_{\mathcal{A}} = n_{\mathcal{J}} = 1000$ and $1 \leq d \leq 20$, where the interview graph is constructed by the both-side signaling mechanism. Each data point represents the average over $10$ trials of simulations.}%
    \label{fig:both_increasing_n_A}%
\end{figure}

\prettyref{fig:both_increasing_n_A} demonstrates the effectiveness of both-side signaling as the number of signals increases from $1$ to $20$ in two extreme scenarios: (a) absent post-interview scores ($\DistB = \mathbb{N}(0,1)$ and $\DistA = \boldsymbol{\delta}_0$), and (b) absent pre-interview scores ($\DistB = \boldsymbol{\delta}_0$ and $\DistA = \mathbb{U}[-1,1]$). When post-interview scores are absent (a), the market fails to achieve interim stability with sparse signals ($d \leq 10$): a constant fraction of applicants and firms remains involved in at least one interim blocking pair, even as the number of signals increases. In contrast, when pre-interview scores are absent (b), both the fraction of applicants and the fraction of firms involved in at least one interim blocking pair decrease rapidly as $d$ increases. These observations align with theoretical insights suggesting that both-side signaling can fail to achieve almost interim stability when the impact of post-interview scores is negligible with sparse signals, but could succeed if the impact of pre-interview scores is negligible.

\section{Conclusion}
\label{sec:conclusions}

Signaling mechanisms  in two-sided matching  markets aim  to reduce the number of  interviews by signaling participants' interests. We study (single- and multi-tiered) random matching markets in which  interviews are conducted based on signals, after which a clearinghouse forms matches. We show that careful signaling mechanisms can give rise to interim stable matchings without many interviews. One insight for effective signaling is letting the short side of the market send signals to agents on the long side.

{
Almost interim stability permits substantially smaller signal budgets than the perfect-stability guarantees established here.
}%

{
Message passing gives exact availability on trees; truncation and probability bounds turn it into local certificates for random interview graphs.
}%

{
Further work could study richer vertical heterogeneity, repeated signaling rounds, and models in which signals or interviews change preferences in ways informed by empirical evidence.
}%

\clearpage
\bibliography{refs}

\clearpage
\appendix

\section{Preliminary facts}
\subsection{Concentration inequalities}
\begin{lemma}[{Chernoff lower-tail bound}]\label{lmm:chernoff}
Suppose  $X \sim \Binom(n,p)$ with mean $\mu=np$.

{
Then for any $\epsilon>0$,

}%
\begin{align}
\prob{ X \le (1-\epsilon) \mu } \le \exp\left( - \frac{\epsilon^2}{2} \mu \right). \label{eq:chernoff_binom_left}
\end{align}
\end{lemma}

\subsection{Facts on stable matching on random bipartite graph}
\begin{proposition}\label{prop:unmatched}
{
\cite[Theorem 1 and Appendix C]{kanoria2023competition}
      Given
}%
 a one-sided $d$-regular graph $H$ with uniformly generated strict preferences on $\Short \cup \Long$ with $n_{\Short} = \delta n_{\Long}$, where $\omega(1)\le d =o\left(\log^2 n\right)$ and  $1\le \delta \le 1 + n^{-\epsilon}$ for any constant $\epsilon>0$, for any stable matching $\Phi$ on $H$,   
        \begin{align}
        & \prob{  \exp\left(-d^{\frac{1}{2}} - 3d^{\frac{1}{4}} \right) \le  \frac{1}{n_{\Short}}\left|a\in\Short \text{ s.t. $a$ is unmatched on $\Phi$} \right| \le \exp\left(-d ^{\frac{1}{2}} +3d^{\frac{1}{4}} \right)   }  \nonumber  \\
        & \ge 1- O\left(\exp\left(-d/4\right)\right) \,,\label{eq:unmatched_fraction}
    \end{align}
    and {the applicants' average rank of firms in all stable matchings} on $H$ is $\Theta\left(\sqrt{d}\right)$ with high probability. 
\end{proposition}
{

\begin{proof}
Write $M=n_\Short$, $N=n_\Long$, and let $U_\Short,U_\Long$ denote
the unmatched counts. These counts are the same in every stable matching.
The estimates in \cite[Proposition 2 and Lemma 14]{kanoria2023competition}
give, for sufficiently large markets,
\[
 \mathbb E U_\Short\le N e^{-\sqrt d+1.4d^{1/4}},\qquad
 \Prob\{U_\Long\ge N e^{-\sqrt d-2d^{1/4}}\}
       \ge1-O(e^{-\sqrt N}).
\]
Since $U_\Short=M-N+U_\Long$ and $M/N=1+o(1)$, it follows that
\[
 M e^{-\sqrt d-2.5d^{1/4}}
  \le\mathbb E U_\Short
  \le M e^{-\sqrt d+1.5d^{1/4}}.
\]
Generate the market using independent applicant blocks: each block
contains that applicant's ordered list and its independent uniform score
at every firm. Changing one block changes $U_\Short$ by at most one.
Indeed, after deleting that applicant both instances are identical, and
inserting one applicant changes the stable matching size by either zero
or one, as seen by continuing applicant-proposing DA along its rejection
chain.
{
By \cite[Lemma~1.2]{mcdiarmid1989bounded},
}%

\[
 \Prob\{|U_\Short-\mathbb EU_\Short|\ge t\}
       \le2\exp(-2t^2/M).
\]
The expectation lies a distance at least
$cM e^{-\sqrt d-2.5d^{1/4}}$ from either boundary in
\prettyref{eq:unmatched_fraction}, for some fixed $c>0$ and sufficiently
large $d$. As $\sqrt d=o(\log N)$, the failure probability is at most
\[
 2\exp\!\left[-c' M e^{-2\sqrt d-5d^{1/4}}\right]
       =\exp[-N^{1-o(1)}]=O(e^{-d/4}).
\]
This proves the stated probability bound.

We now extend the applicant-optimal rank estimate to every stable
matching. Let $\mathcal F$ contain the full applicant-proposing DA
history, revealing applicant lists only as proposals are made, and write
$U=U_\Long$. Each of these $U$ unmatched firms has received no proposal.
Fix a matched firm $j$ and apply the single-firm divorce test of
\cite[Algorithm 4.1 and Section 4.2]{im2005}: reject its current
applicant and continue proposals. The test can find a second stable
partner for $j$ only if a proposal reaches $j$ before it reaches an
originally unmatched firm or an applicant exhausts its list.
If $U=0$, the desired bound below is trivial. Otherwise stop the test
at its first proposal to $j$ or to an originally unmatched firm, or at
list exhaustion. Conditional on $\mathcal F$ and each history before
this stopping time, the active applicant's next unexposed list entry is
uniform among its unproposed firms. All $U$ unmatched firms are still
in that set; $j$ is either also in it or has already been proposed to.
Thus the next-step probability of hitting $j$ is at most $1/U$ times
that of hitting an originally unmatched firm. If $q_j,q_U$ denote the
probabilities of these two first-hit events, summing over the disjoint
possible hitting times gives $q_j\le q_U/U$. Since $q_j+q_U\le1$,
\[
 \Prob\{j\hbox{ has more than one stable partner}\mid\mathcal F\}
       \le\frac1{U+1}.
\]
This argument does not require $d$ to be fixed. It also holds for unequal
population sizes. Let $C$ be the number of firms with multiple stable
partners, equivalently the number of applicants whose partners differ
between the two optimal stable matchings. With
$u=N e^{-\sqrt d-2d^{1/4}}$, the preceding conditional bound gives
\[
 \mathbb E C\le\frac N{u+1}+N\Prob\{U<u\}
       \le e^{\sqrt d+2d^{1/4}}+O(Ne^{-\sqrt N})=N^{o(1)}.
\]
Assign rank $d+1$ to an unmatched applicant. For every stable matching
$\Phi$, lattice extremality and the equality of unmatched sets imply
\[
 0\le\frac1M\sum_a\bigl(\operatorname{rank}_\Phi(a)
              -\operatorname{rank}_{\Phi^{\Short}}(a)\bigr)
     \le\frac{dC}{M}.
\]
Its expectation tends to zero, hence the difference is $o(\sqrt d)$
with high probability, simultaneously for all $\Phi$.
\cite[Theorem 1]{kanoria2023competition} gives the applicant-optimal
average $(1+o(1))\sqrt d$ with failure probability
$O(e^{-d^{1/4}})$, proving the stated conclusion for every stable
matching. In addition its expectation is $O(\sqrt d)$, since ranks are
at most $d+1$ and $(d+1)e^{-d^{1/4}}=o(\sqrt d)$; the same expectation
bound holds for the job-optimal average by the preceding estimate.
\end{proof}

}%

\section{Fixed point convergence}
{
For $d\ge0$ and $0\le p\le1$, define
\begin{equation}
 f_d(p)=
 \begin{cases}
 \displaystyle\frac{1-(1-p)^{d+1}}{(d+1)p},&p>0,\\[5pt]
 1,&p=0.
 \end{cases}
 \label{eq:f_d}
\end{equation}

}%

{
\begin{lemma}\label{lmm:stochastic_dominance}\label{lmm:inverse_binomial}
Let $d\ge0$ be an integer and let $Y_i\sim\Bern(p_i)$, $i\in[d]$, be independent, with $0\le\underline p\le p_i\le\overline p\le1$. Then
\[
 f_d(\overline p)\le
 \Expect\!\left[\frac1{1+\sum_{i=1}^dY_i}\right]
 \le f_d(\underline p).
\]
In particular, $\Expect[(1+X)^{-1}]=f_d(p)$ for $X\sim\Binom(d,p)$.
\end{lemma}
\begin{proof}
The identity $(1+k)^{-1}=\int_0^1t^k\,dt$ and independence give
\[
 \Expect\!\left[\frac1{1+\sum_{i=1}^dY_i}\right]
 =\int_0^1\prod_{i=1}^d(1-p_i+p_it)\,dt.
\]
For $0\le t\le1$, each factor lies between $1-\overline p+\overline pt$ and $1-\underline p+\underline pt$. Integrating these product bounds gives the result, since $\int_0^1(1-p+pt)^d\,dt=f_d(p)$, including $p=0$ and $d=0$.
\end{proof}
}%

\begin{lemma} \label{lmm:property_f_d}
$f_{d}(p)$ satisfies the following properties: 
\begin{enumerate}[label=(P\arabic*)]
    \item \label{P:1} 
{
$f_d$ is continuous on $[0,1]$ and takes values in $[0,1]$ for every real $d\ge0$;
}%

    \item\label{P:2} $f_d(p)$ is strictly decreasing on $0\le p\le 1$, for any $d\in \reals_+$;
    \item\label{P:3}  
{
$f_d(p)$ is nonincreasing in $d\ge0$ for every $0\le p\le1$;
}%

     \item\label{P:5}  
{
$f_d(p)$ is convex in $d\ge0$ for every $0\le p\le1$;
}%

{
    \item
{
\label{P:6} For $0<p\le1$ and $d>0$,
    \[
      f'_d(p)=\frac{(1-p)^d}{p}-\frac{f_d(p)}p
      \ge-\frac{f_d(p)}p,
    \]
    with strict inequality when $p<1$.

}%
}%
\end{enumerate}
\end{lemma}

\begin{proof}
{
The integral representation
\[
 f_d(p)=\int_0^1(1-pt)^d\,dt,
 \qquad f_d(0)=1,\qquad f_d(1)=\frac1{d+1},
\]
gives \ref{P:1}--\ref{P:5}: the integrand is strictly decreasing in $p$

{
for $d>0$ and $t>0$, and nonincreasing and convex in $d\ge0$.
Integration preserves these comparisons, including the endpoints by
continuity.
}%
 
{
Differentiating the defining quotient gives \ref{P:6}; its first
term is positive when $p<1$.

}%
}%
\end{proof}

\begin{lemma}\label{lmm:property_f_a_f_b}
    
{
For any $a,b\in\reals_+$, the function
}%
 $ \left(f_{a} \circ f_{b} \right)(x)$ satisfies the following properties:
    \begin{enumerate}[label=(P\arabic*)]
    
{
\setcounter{enumi}{5}
}%

    \item \label{P:7} $\left(f_{a} \circ f_{b} \right)(x)$ is continuous and strictly increasing on $0\le x\le 1$.
    \item
{
 \label{P:8} For every $0\le x\le1$,
    $0<(f_a\circ f_b)(x)<1$.

}%
    \item
{
 \label{P:9} For every $0<x\le1$,
    \[
       0\le (f_a\circ f_b)'(x)<\frac{(f_a\circ f_b)(x)}x.
    \]
    Consequently, $(f_a\circ f_b)(x)/x$ is strictly decreasing on $(0,1]$.

}%
    \item \label{P:10} $\left(f_{a} \circ f_{b} \right) (x)$ has a unique fixed point solution $x^*$ such that $0< x^* < 1$. 
    \end{enumerate}
\end{lemma}

\begin{proof}

{

{
Write $F=f_a\circ f_b$. Both factors are continuous and strictly
 decreasing, so $F$ is continuous and strictly increasing. Its endpoint
 values are $F(0)=1/(a+1)>0$ and $F(1)=f_a(1/(b+1))<1$, proving
 \ref{P:7} and \ref{P:8}.

}%

{
For $0<x\le1$, $0<f_b(x)<1$, so \ref{P:6} gives
\begin{align}
 0\le F'(x)=f'_a(f_b(x))f'_b(x)
 <\frac{F(x)}{f_b(x)}\frac{f_b(x)}x
 =\frac{F(x)}x.
 \label{eq:f_a_b_derivative}
\end{align}
Thus $(F(x)/x)'=(xF'(x)-F(x))/x^2<0$, proving \ref{P:9}.

}%

Finally, $F(0)>0$ and $F(1)<1$ imply that $F(x)-x$ has a zero in $(0,1)$. If it had two distinct zeros $x_1<x_2$, then $F(x_1)/x_1=F(x_2)/x_2=1$, contradicting the strict decrease just proved. This establishes \ref{P:10}.

}%
\end{proof}

{
For any $a,b\in\reals_+$ and $m\in\naturals$ (including $m=0$), define
\begin{align}
 g(a,b,m)\triangleq(f_a\circ f_b)^m(1).
 \label{eq:g_a_b_m}
\end{align}
}%

{
\begin{lemma}\label{lmm:g_a_b_x}
For $a,b>0$, let $F=f_a\circ f_b$ and let $r\in(0,1)$ be its unique fixed point. If $r<x'\le1$, $0<\gamma<1$, and an integer $m\ge0$ satisfies
\[
 m\ge\frac{\log\gamma}{\log(F(x')/x')},
\]
then
\[
 r\le g(a,b,m)\le x'+\gamma.
\]
Moreover, $g(a,b,m)\downarrow r$ as $m\to\infty$.
\end{lemma}
}%

{
\begin{proof}
Put $z_m=F^m(1)$, with $z_0=1$. Since $F$ is increasing, $F(r)=r$, and $F(1)<1$, the sequence $z_m$ decreases and is bounded below by $r$. Its limit is a fixed point of $F$, hence equals $r$ by uniqueness.

{
Set $q=F(x')/x'\in(0,1)$. If $z_m>x'$, monotonicity of
$z_j$ gives $z_j>x'$ for $j\le m$, and \ref{P:9} yields
$z_{j+1}\le qz_j$. Hence $z_m\le q^m\le\gamma$.
Otherwise $z_m\le x'$. In either case,
$r\le z_m\le\max\{x',\gamma\}\le x'+\gamma$.

}%
\end{proof}
}%

{
\begin{lemma} \label{lmm:fixed_point_convergence}
Suppose $a_n,b_n>0$, $a_nb_n=\omega(1)$, and $c_n=(a_n+1)/(b_n+1)$. Write $F_n=f_{a_n}\circ f_{b_n}$ and let $r_n\in(0,1)$ denote its unique fixed point. Define the comparison scale $x^*$ and $\Gamma_\epsilon$ as follows:
\begin{enumerate}[label=(F\arabic*)]
\item \label{F:1} 
{
If $c_n\le1-\Omega(1/(b_n+1))$, let
}%

\[
 x^*=-\frac{c_n}{\log(1-c_n)},\qquad
 \Gamma_\epsilon=\frac{1-(1-c_n)^{1/(1+\epsilon/2)}}{c_n},
 \qquad \epsilon=\omega\!\left(\frac{\log(b_n+1)}{b_n+1}\right).
\]
\item \label{F:2} If $c_n=1$, let
\[
 x^*=\frac1{\sqrt{b_n+1}},\qquad
 \Gamma_\epsilon=
 \frac{1-\exp\!\left(-\sqrt{b_n+1}/(1+\epsilon/2)\right)}
 {1-\exp\!\left(-(1+\epsilon)\sqrt{b_n+1}\right)},
 \qquad \epsilon=\omega\!\left(\frac1{\sqrt{b_n+1}}\right).
\]
In this case,
\[
 x^*\bigl[1-(1-x^*)^{b_n+1}\bigr]<r_n<x^*,
 \qquad r_n=(1+o(1))x^*.
\]
\item \label{F:3} 
{
If $c_n>1$, let
}%

\[
 x^*=-\frac{\log(1-1/c_n)}{b_n+1},\qquad
 \Gamma_\epsilon=\frac1{1+(\epsilon\wedge1)(1-1/c_n)},
 \qquad \epsilon>0.
\]

{
Here $r_n\le x^*$.
}%

\end{enumerate}
For all sufficiently large $n$, $0<\Gamma_\epsilon<1$ and, for every integer $m\ge0$ such that $m\ge\log(\epsilon x^*)/\log\Gamma_\epsilon$,
\begin{align}
 r_n\le g(a_n,b_n,m)\le(1+2\epsilon)x^*.
 \label{eq:g_epsilon}
\end{align}
\end{lemma}
}%

{
\begin{proof}
Suppress the subscript $n$, and put $B=b+1$, $c=(a+1)/B$, $F=f_a\circ f_b$, and $s=x^*$. The lower bound $r\le F^m(1)$ follows from \prettyref{lmm:g_a_b_x}. If $(1+\epsilon)s\ge1$, the upper bound is immediate from $F^m(1)\le1$. Otherwise put
\[
 x=(1+\epsilon)s<1,\qquad
 t=1-(1-x)^B,\qquad y=f_b(x)=\frac{t}{Bx}.
\]
Then
\[
 \frac{F(x)}x=\frac{1-(1-y)^{cB}}{ct}.
\]
We will show $F(x)/x\le\Gamma_\epsilon<1$ in each case. By \ref{P:9}, this also implies $x>r$. Since $0<\epsilon s<x<1$, \prettyref{lmm:g_a_b_x} applies with $x'=x$ and $\gamma=\epsilon s$, because $(F(x)/x)^m\le\Gamma_\epsilon^m\le\epsilon s$. This proves \prettyref{eq:g_epsilon} once the following bounds are established.

\begin{itemize}
\item In case \ref{F:1}, set $L=-\log(1-c)$, so $s=c/L$. Here $B\to\infty$, $L/c=O(\log B)$, and $s=\Omega(1/\log B)$. Since $x<1$ and $\epsilon=\omega(\log B/B)$,
\[
 1+\epsilon=O(\log B),\qquad
 y\le\frac{L}{Bc(1+\epsilon)}
 =o\!\left(\frac{\epsilon}{1+\epsilon}\right).
\]
For sufficiently large $n$, $y\le\epsilon/[4(1+\epsilon)]$, and hence
\[
 \frac{cBy}{1-y}
 =\frac{Lt}{(1+\epsilon)(1-y)}
 \le\frac{L}{1+3\epsilon/4}.
\]
Using $(1-y)^{cB}\ge\exp(-cBy/(1-y))$ gives
\[
 \frac{F(x)}x\le
 \frac{1-\exp(-L/(1+3\epsilon/4))}{ct}.
\]
The denominator correction $1-t\le e^{-Bs}$ is smaller than every fixed inverse power of $B$. It is dominated by the positive difference
\begin{align*}
 &e^{-L/(1+3\epsilon/4)}-e^{-L/(1+\epsilon/2)}\\
 &\qquad\ge
 \frac{(1-c)L\epsilon}{4(1+\epsilon)^2}
 \ge\frac{(1-c)c\epsilon}{4(1+\epsilon)^2}.
\end{align*}
Indeed, $1-c=\Omega(1/B)$, $c\ge1/B$, $\epsilon=\omega(\log B/B)$, and $1+\epsilon=O(\log B)$ make this difference at least an inverse power of $B$. Consequently, for sufficiently large $n$,
\begin{align}
 \frac{F(x)}x
 \le\frac{1-e^{-L/(1+\epsilon/2)}}c
 =\Gamma_\epsilon<1.
 \label{eq:delta_epsilon}
\end{align}

\item In case \ref{F:2}, $a=b$ and $B\to\infty$. The unique fixed point of the decreasing function $f_b$ is also the unique fixed point $r$ of $f_b\circ f_b$. At $s=B^{-1/2}$,
\[
 f_b(s)=s\bigl[1-(1-s)^B\bigr]<s.
\]
Thus $r<s$ and, since $f_b$ is strictly decreasing,
\[
 r=f_b(r)>f_b(s)=s\bigl[1-(1-s)^B\bigr].
\]
This proves the stated exact bracket and $r/s\to1$, because $(1-s)^B\le e^{-\sqrt B}$.
For the iteration bound, $y\le1/[\sqrt B(1+\epsilon)]$. Since $\epsilon\sqrt B\to\infty$, eventually $y\le\epsilon/[2(1+\epsilon)]$, so
\[
 \frac{By}{1-y}=\frac{t}{x(1-y)}
 \le\frac{\sqrt B}{1+\epsilon/2},
 \qquad
 t\ge1-e^{-(1+\epsilon)\sqrt B}.
\]
Therefore
\[
 \frac{F(x)}x
 \le\frac{1-e^{-\sqrt B/(1+\epsilon/2)}}
 {1-e^{-(1+\epsilon)\sqrt B}}
 =\Gamma_\epsilon<1.
\]

\item In case \ref{F:3}, write $q=1-1/c\in(0,1)$ and $s=-\log q/B$. If $s\ge1$ then $r\le s$ immediately. Otherwise
\[
 \frac{F(s)}s
 \le\frac1{c[1-(1-s)^B]}
 \le\frac1{c(1-e^{-Bs})}=1,
\]
so \ref{P:9} implies $r\le s$ as well. For $x=(1+\epsilon)s<1$,
\[
 \frac{F(x)}x
 \le\frac1{c[1-(1-x)^B]}
 \le\frac1{c(1-q^{1+\epsilon})}.
\]
Put $\alpha=\epsilon\wedge1$. Concavity of $u\mapsto u^\alpha$ gives $q^\epsilon\le q^\alpha\le1-\alpha/c$. Hence
\[
 c(1-q^{1+\epsilon})\ge1+\alpha q,
 \qquad
 \frac{F(x)}x\le\frac1{1+\alpha q}
 =\Gamma_\epsilon<1.
\]
This proves all three cases.
\end{itemize}
\end{proof}
}%

\section{Stability analysis via leveraging local neighborhood information}\label{sec:local}

{
We first compare stable outcomes under local truncation in
\prettyref{sec:truncation}, then derive proposal and probability
recursions on trees in \prettyref{sec:proposal_passing_alg} and
\prettyref{sec:message_passing_alg}. The sequential exposure in
\prettyref{sec:stable} extends these bounds to random bipartite graphs.

}%
\subsection{Truncation on local neighborhood} \label{sec:truncation}
For any bipartite graph $H$  with strict preferences, let $\Phi_{H}^{\Short}$ (resp. $\Phi_{H}^{\Long}$) denote the stable matching on $H$ resulting from $\Short$-proposing (resp. $\Long$-proposing) in the DA algorithm. 
For any vertex $i$ in $\calV(H)$, let $H_{-i}$ represent the subgraph of $H$ obtained by removing vertex $i$ and all of its incident edges. The following lemma, adapted from~\cite[{Theorems 1 and 2}]{crawford1991}, establishes that under the deferred acceptance algorithm, when a node is removed from {one side of a bipartite graph}, all nodes on the same side are weakly better off due to diminished competition. Conversely, nodes on the opposing side are weakly worse off because they are competing for a smaller set of opportunities. 

{
\begin{lemma}[{\cite[{Theorems 1 and 2}]{crawford1991}}]
\label{lmm:truncation} 
    Let $H$ be a bipartite graph with strict preferences. Fix $\Side \in \{\Short,\Long\}$. For any $i,j \in \calV(H)$ such that $i\neq j$, the following hold: 
    \begin{itemize}
        \item If $i$ and $j$ are on the same side, then $j$ weakly prefers $\phi_{H_{-i}}^{\Side}(j)$ to $\phi_H^{\Side} (j)$.
        \item If $i$ and $j$ are on different sides, then $j$ weakly prefers $\phi_H^{\Side} (j)$ to $\phi_{H_{-i}}^{\Side}(j)$.
    \end{itemize}
\end{lemma}
}%

{
Recall that for any graph $H$ and vertex $\rho\in\calV(H)$, $H_m(\rho)$ is its $m$-hop neighborhood. The following lemma compares the outcome of $\rho$ under the same side-proposing deferred acceptance algorithm on $H$ and on $H_m(\rho)$: the truncated outcome is weakly better for $\rho$ when $m$ is odd and weakly worse when $m$ is even.

}%

{
\begin{lemma}\label{lmm:local}
Let $H$ be a bipartite graph with strict preferences. Fix $\Side \in \{\Short,\Long\}$. For any $\rho\in \calV(H)$ and $m\in\naturals$,  the following hold: 
\begin{itemize}
    \item If $m$ is odd, $\rho$ weakly prefers $\phi_{H_m(\rho)}^{\Side}(\rho)$ to $\phi_{H}^{\Side}(\rho)$.
    \item If $m$ is even, $\rho$ weakly prefers $\phi_{H}^{\Side}(\rho)$ to $\phi_{H_m(\rho)}^{\Side}(\rho)$.
\end{itemize}
\end{lemma}
}%

\begin{proof}
{
Deleting the vertices at distance $m+1$ separates $H_m(\rho)$
from the rest of $H$. They lie on $\rho$'s side when $m$ is odd
and on the opposite side when $m$ is even. Apply
\prettyref{lmm:truncation} successively to these deletions; the
outcome of $\rho$ in its remaining component is its outcome on
$H_m(\rho)$.
}%
\end{proof}

Recall that for any $\rho \in \calV(H)$ and $i\in\calN(\rho)$, we say $i$ is available to $\rho$ on $H$,  if and only if $i$ weakly prefers $\rho$ to its match in every stable matching on $H$. We then present the following lemma, where an illustrated example is given in \prettyref{fig:every}. 

{
\begin{lemma}\label{lmm:local_available}
    Let $H$ be a bipartite graph with strict preferences. 
{
For any $\rho\in\calV(H)$, $i\in\calN(\rho)$ and integer $m\ge1$, the following hold:
}%
 
    \begin{itemize}
    \item If $m$ is odd and $i$ is available to $\rho$ on $H$, then $i$ is available to $\rho$ on $H_m(\rho)$. 
    \item If $m$ is even and $i$ is available to $\rho$ on $H_m(\rho)$, then $i$ is available to $\rho$ on $H$. 
    \end{itemize}
\end{lemma}
}%

\begin{proof}
{
Let $\Side_i$ be the side containing $i$. Availability is equivalent
to $i$ weakly preferring $\rho$ to $\phi_H^{\Side_i}(i)$, its best
stable partner. Delete the distance-$(m+1)$ layer as in
\prettyref{lmm:local}. By \prettyref{lmm:truncation}, this best
partner weakly worsens for $i$ when $m$ is odd and weakly improves
when $m$ is even. The two availability implications follow.
}%
\end{proof}

\subsection{Hierarchical proposal-passing algorithm on tree} \label{sec:proposal_passing_alg}

\begin{lemma}\label{lmm:tree_unique}
    Let $T$ be a tree graph with strict preferences. There is a unique stable matching on $T$.
\end{lemma}
\begin{proof}
{
Deferred acceptance supplies a stable matching. By unmatched-set
invariance~\citep{mcvitie1970stable}, two stable matchings have the
same matched vertices, so every nonempty component of their symmetric
difference is a cycle. A tree has no cycle; uniqueness follows.
}%
\end{proof}

We then demonstrate that by selecting an arbitrary vertex $\rho \in \mathcal{V}(T)$ as the root of $T$, we can determine the stable matching on $T$ using \prettyref{alg:proposal_passing_alg}. 

{
Write $P(i)$ and $\calC(i)$ for the parent and children of $i$,
with $P(\rho)=\emptyset$. In \prettyref{alg:proposal_passing_alg},
$\Set(i)$ records the child proposals received by $i$.
The algorithm passes proposals upward and chooses partners downward;
\prettyref{fig:both_phases_updated_7} illustrates the two phases.

}%
\begin{algorithm}[htp]
\caption{Hierarchical proposal-passing algorithm on tree}\label{alg:proposal_passing_alg}
\begin{algorithmic}[1]
\State{\bfseries Input:} A rooted tree $T$ with strict preferences, root $\rho$ and depth $m$. 
\State Let $\Set (i)$ denote the set of proposals received by $i$
and  initialize $\Set (i) = \emptyset$,
for each $i \in \calV(T)$. 
\For{{$\kappa = m, m-1, \cdots, 1$}}
\For{{each vertex $i$ at depth $\kappa$ in $T$}}
\If{$i$ prefers its parent node $P(i)$ over all vertices in $\Set(i)$}
\State {$i$ proposes to its parent node $P(i)$ and adds $i$} to $\Set(P(i))$.
\EndIf
\EndFor
\EndFor
\State Let $\Phi$ denote a matching on $T$, and initialize $\Phi = \emptyset$. 
\For{$\kappa = 0, \cdots, m-1$}
\For{{each vertex $i$ at depth} $\kappa$ in $T$}
\If{$\Set(i) \neq \emptyset$ and $(P(i),i)\not\in \Phi$}
\State  Let $i^*$ denote $i$'s  most preferred proposal in $\Set(i)$, and add $(i,i^*) $ to $\Phi$.  
\EndIf
\EndFor
\EndFor
\State{\bfseries Output:} $\Phi$. 
\end{algorithmic}
\end{algorithm}

The following lemma shows that \prettyref{alg:proposal_passing_alg} {outputs a stable matching on a tree.}
\begin{lemma}\label{lmm:proposal_passing_alg}
   For any rooted tree $T$ with strict preferences, the matching $\Phi$ returned in \prettyref{alg:proposal_passing_alg} is a stable matching. 
\end{lemma}
\begin{proof}

{
The clean-up phase constructs a matching: a vertex already matched to its
parent selects no child, and each child has only one parent. Moreover,
every vertex that receives a proposal weakly prefers its assigned partner to every proposal it
received. Indeed, it either chooses its best child proposal or is matched
to its parent, which it preferred to all child proposals when proposing.

Consider an edge $(i,j)$ with $j=P(i)$. If $i$ proposed to $j$, the
preceding observation rules out a blocking pair. Otherwise $j$ cannot
select $i$, so $i$ chooses its best child proposal, which it prefers to
$j$. Thus no edge blocks the returned matching.

}%
\end{proof}

\subsection{Message passing on tree with uniformly generated strict preferences}\label{sec:message_passing_alg}

In this subsection, {we consider the case in which the preference list} of each  node with respect to its neighbors is independently uniformly generated. {We say that such a tree has uniformly generated strict preferences.} 
Fix a rooted tree $T$ with root $\rho$, depth $ m \in \naturals_+$ and uniformly generated strict preferences. 
For any node $i \in \calV(T)$ such that $i\neq \rho$, let $\sfX_{i,P(i)} \left(T\right)$ be {an indicator of the event} that $i$ proposes to its parent node $P(i)$ following \prettyref{alg:proposal_passing_alg}. 
By taking expectation over the uniformly generated preferences on the tree $T$, we define
\[
\mu_{i,P(i)}\left(T\right) \triangleq \expect{\sfX_{i,P(i)}\left(T\right) \, | \, T}  \,.
\]

{
For a nonroot vertex $i$, its child messages depend on disjoint
subtrees and are independent of each other and of $i$'s uniform
preference order. Conditional on these messages, its parent is the
most preferred among the parent and the proposing children with
probability $(1+\sum_{v\in\calC(i)}\sfX_{v,i}(T))^{-1}$. Hence

}%
\begin{align}
    \mu_{i,P(i)} \left( T \right)
    & =\expect{ \expect{\sfX_{i,P(i)} \left( T \right)  \big | \{\sfX_{v,i}\left( T \right)\}_{v\in \calC(i) } \,, T}  \, \bigg | \, T} \nonumber \\
    & =  \expect{\frac{1}{1+ \sum_{v\in \calC(i)} \sfX_{v,i} \left( T \right)} \, \bigg | \, T }\, . 
    \label{eq:message_pass}
\end{align} {If} $\calC(i)= \emptyset$, then $ \mu_{i,P(i)} \left( T \right)=  1$. Otherwise, we have $\sfX_{v,i}\left( T \right) \overset{\mathrm{ind}}{\sim} \Bern\left(\mu_{v,i} \left( T \right)\right)$ for $v \in \calC(i)$, where $\mu_{v,i}\left( T \right)$ is determined in the previous iteration.

\begin{lemma} \label{lmm:f_inequality}
    Consider a rooted tree $T$ with root $\rho$, depth $ m \in \naturals_+$ and uniformly generated strict preferences. 
    
{
For any nonroot vertex $i\in\calV(T)\setminus\{\rho\}$ with degree $d_i$,
}%
 if $ \underline{\mu} \le \mu_{v,i} \left(T\right)\le \overline{\mu}$ for any $v \in \calC(i)$, we have
    \begin{align*}
        f_{d_i-1} \left(\overline{\mu }\right) \le \mu_{i,P(i)} \left( T \right)  \le f_{d_i-1} \left(\underline{\mu }\right) \,. 
    \end{align*}
\end{lemma}
\begin{proof}

{
Apply \prettyref{lmm:stochastic_dominance} to the independent child messages in \eqref{eq:message_pass}.
}%
\end{proof}

\begin{lemma} \label{lmm:rooted_tree}
Consider a rooted tree $T$ with root $\rho$, depth $ m \in \naturals_+$ and uniformly generated strict preferences. 
{

{
Let $0\le\dunderodd\le\doverodd$ and
$0\le\dundereven\le\dovereven$.
Suppose that for every nonroot vertex $i\in\calV(T)\setminus\{\rho\}$, if $i$ is at odd depth $<m$,
}%

}%
 
$\dunderodd \le  d_i -1  \le \doverodd$; {if $i$ is at even depth $<m$}, $\dundereven \le d_i -1\le \dovereven$. 

{
For every $j\in\calC(\rho)$, the following bounds hold:
}%
 
\begin{itemize}
    \item If $m$ is even: 
        \begin{align}
            \foverodd \circ \left( \fundereven \circ  \foverodd \right)^{m/2-1}(1) 
            \le \mu_{j,\rho} \left( T \right) \le \funderodd \circ \left( \fovereven\circ  \funderodd  \right)^{m/2-1}(1) \,. 
            \label{eq:even_iterative}
        \end{align}
    \item Otherwise:
        \begin{align}
             \left(\foverodd \circ  \fundereven \right)^{\left(m-1\right)/2}(1)  
             \le \mu_{j,\rho} \left( T \right) \le \left(\funderodd \circ  \fovereven \right)^{\left(m-1\right)/2}(1) \,. 
             \label{eq:odd_iterative}
        \end{align}
\end{itemize}

{

{
In particular, if every nonroot vertex at depth $<m$ has degree $d\ge1$, then
}%

}%

    \begin{align}
             \mu_{j,\rho}\left(T\right) = f_{d-1}^{m-1}(1) \,.  \label{eq:d_iterative}
    \end{align}
\end{lemma}

\begin{proof}

{
Set $L_m=U_m=1$ and, for $k=m-1,\ldots,1$, define
\[
 \begin{array}{lll}
 L_k=\foverodd(U_{k+1}),&
 U_k=\funderodd(L_{k+1}),& k\text{ odd},\\[3pt]
 L_k=\fovereven(U_{k+1}),&
 U_k=\fundereven(L_{k+1}),& k\text{ even}.
 \end{array}
\]
By \prettyref{lmm:f_inequality} and monotonicity in the degree
parameter from \ref{P:3}, backward induction gives
$L_k\le\mu_{i,P(i)}(T)\le U_k$ for every existing vertex $i$ at
depth $k\ge1$. This includes early leaves, for which
$d_i-1=0$ and $f_0\equiv1$.

{
Substituting the recursion at depth one gives
\eqref{eq:even_iterative} and \eqref{eq:odd_iterative}, with zero
iterations when $m=1$. When all nonroot vertices above depth $m$
have degree $d$, both recursions apply $f_{d-1}$ exactly $m-1$ times,
which gives \eqref{eq:d_iterative}.
}%

}%
\end{proof}

{
Next, define the random rooted tree model $\mathbb T_\ell(\kappa_1,\kappa_2,\fone,\ftwo)$ for real $\kappa_1,\kappa_2\ge1$, integer $\ell\ge1$, and $0<\fone,\ftwo<1$. Generate its offspring counts in depth-first order: reveal a vertex's count when it is first visited, and complete each child subtree before visiting the next child. Write $\mathcal F_i$ for the offspring information revealed before visiting an existing vertex $i$. At every depth strictly less than $\ell$, require, almost surely,
\begin{align}
 \Expect[O_i\mid\mathcal F_i]&\le\kappa_1
       &&\text{if the depth of $i$ is odd},\nonumber\\
 \Prob\bigl(O_i<(1-\fone)\kappa_2\mid\mathcal F_i\bigr)&\le\ftwo
       &&\text{if the depth of $i$ is even}.
 \label{eq:O_i_even}
\end{align}
These conditions include vertices with $O_i=0$: an early leaf at even depth is one of the lower-tail failures. Vertices at depth $\ell$ are forced leaves and are exempt from the offspring conditions. The counts need not be independent, but the displayed bounds must hold uniformly over the possible revealed histories. Conditional on the resulting tree, all preference lists are independent and uniform; the tree is generated without using these preference random variables. Consequently, revealing the preferences of completed subtrees does not alter the conditional offspring bounds. The usual independent branching model is a special case.

The following lemma retains the same explicit composition bounds while accounting for lower-tail failures directly.
Vertices are indexed by their ordered ancestor sequences. A conditional statement about a child refers to a fixed index, not to a child selected after observing its proposal message.

}%
\begin{lemma}\label{lmm:rooted_tree_expected_degree}
{

For $T\sim\mathbb T_\ell(\kappa_1,\kappa_2,\fone,\ftwo)$, suppose that $\ell\ge2$ is even and $2(\kappa_1\vee\kappa_2)\ftwo<1$. Then, for any child $j$ exposed at the root,
\begin{align}
 \Expect[\sfX_{j,\rho}(T)\mid j\in\calC(\rho)]
 &\ge f_{\paratwo\kappa_1}\circ
       (f_{\paraone\kappa_2}\circ f_{\paratwo\kappa_1})^{\ell/2-1}(1),
 \label{eq:sfT_even_iterative_new_1}
\end{align}
and, when $\ell\ge4$, for an exposed child $i$ of $j$,
\begin{align}
 \Expect[\sfX_{i,j}(T)\mid j\in\calC(\rho),\ i\in\calC(j),\ O_i\ge\paraone\kappa_2]
 &\le (f_{\paraone\kappa_2}\circ f_{\paratwo\kappa_1})^{\ell/2-1}(1),
 \label{eq:sfT_even_iterative_new_2}
\end{align}
whenever the conditioning event has positive probability, where
\begin{align}
 \paraone=1-\fone,\qquad
 \paratwo=\bigl(1-2(\kappa_1\vee\kappa_2)\ftwo\bigr)^{-1}.
 \label{eq:paraone_paratwo}
\end{align}
The bounds also hold conditional on the history preceding the respective vertex's offspring draw. For $\ell=2$, the second bound is replaced by the boundary identity $\sfX_{i,j}=1$ for every existing vertex $i$ at depth two, without conditioning on its offspring count. 

}%
\end{lemma}

\begin{proof}
{

Set $K=\kappa_1$, $L=(1-\fone)\kappa_2$, $r=\ftwo$, $M=\kappa_1\vee\kappa_2$, and $\eta=(1-2Mr)^{-1}$. We first establish the elementary comparison
\begin{equation}
 f_K\bigl(r+(1-r)x\bigr)\ge f_{\eta K}(x)
 \quad\text{for }\frac1{\kappa_2+1}\le x\le1.
 \label{eq:tree_inflation_repaired}
\end{equation}
For $x<1$, put $q=r+(1-r)x$. Since $2M\ge\kappa_2+1$ and $2Mr<1$,
\[
 (\eta-1)(-\log(1-x))
 \ge\frac{2Mr}{(1-2Mr)(\kappa_2+1)}
 \ge\frac r{1-r}\ge-\log(1-r).
\]
It follows that $(1-x)^\eta\le(1-r)(1-x)=1-q$. The function
$h(t)=1-qt-(1-xt)^\eta$ is concave on $[0,1]$, with $h(0)=0$ and $h(1)\ge0$. Hence $(1-xt)^\eta\le1-qt$ throughout this interval. Using
$f_d(p)=\int_0^1(1-pt)^d\,\mathrm dt$ gives \prettyref{eq:tree_inflation_repaired}. The endpoint $x=1$ follows from $f_K(1)=1/(K+1)\ge1/(\eta K+1)$.

Define deterministic bounds
\[
 b_0=1,\qquad a_h=f_{\eta K}(b_h),\qquad
 b_{h+1}=f_L(a_h),\qquad h\ge0.
\]
We prove uniformly over every entry history that an odd-depth vertex with $2h+1$ levels remaining sends a proposal to its parent with conditional probability at least $a_h$. An even-depth vertex with $2h$ levels remaining has conditional proposal probability at most $b_h$ when its offspring count is at least $L$; for $h=0$ its message equals $b_0=1$ without an offspring condition.

We use conditional stochastic domination. Process a vertex's children in the specified depth-first order. If each next child's proposal indicator has conditional probability at least $p$, uniformly over the previously completed subtrees and their preferences, their sum stochastically dominates a binomial random variable with parameters $O_i,p$. The reversed probability bound gives the reversed domination. To see this, generate the next indicator from a fresh independent uniform variable and its conditional success probability, and compare it to the Bernoulli indicator with threshold $p$. The conditional bounds of the model persist under this exposure order. The parent's independent uniform preference order then gives
\[
 \Prob(\sfX_{i,P(i)}=1\mid\text{all child messages})
 =\frac1{1+\sum_{v\in\calC(i)}\sfX_{v,i}}.
\]
Thus the binomial identity in \prettyref{lmm:inverse_binomial} applies to the dominating variables even when the actual offspring counts and messages are dependent.

For $h=0$, all children of an odd-depth vertex are terminal. Conditional on its entry history, Jensen's inequality and monotonicity in the degree parameter yield
\[
 \Expect[\sfX_{i,P(i)}\mid\mathcal F_i]
 =\Expect[f_{O_i}(1)\mid\mathcal F_i]
 \ge f_{\Expect[O_i\mid\mathcal F_i]}(1)
 \ge f_K(1)\ge a_0.
\]
Here and below the history may also include the preferences of completed subtrees. Suppose the odd-depth bound is established with $2h-1$ levels remaining. At an even-depth vertex with $2h$ levels remaining, conditional on its history and on $O_i\ge L$, the binomial domination just described gives
\[
 \Expect[\sfX_{i,P(i)}\mid\mathcal F_i,O_i]
 \le f_{O_i}(a_{h-1})\le f_L(a_{h-1})=b_h.
\]
On the lower-tail event the message is at most one, so, before drawing this vertex's count, its proposal probability is at most
$q_h=r+(1-r)b_h$. Consequently an odd-depth parent with $2h+1$ levels remaining satisfies
\begin{align*}
 \Expect[\sfX_{j,P(j)}\mid\mathcal F_j]
 &\ge\Expect[f_{O_j}(q_h)\mid\mathcal F_j]
 \ge f_K(q_h)\ge f_{\eta K}(b_h)=a_h.
\end{align*}
The middle inequality uses convexity and monotonicity of $f_d(q_h)$ in $d$; the last uses \prettyref{eq:tree_inflation_repaired}, since
$b_h=f_L(a_{h-1})\ge f_L(1)=1/(L+1)\ge1/(\kappa_2+1)$.
This completes the induction. Taking $h=\ell/2-1$ gives the two asserted compositions.

}%
\end{proof}

\subsection{Stability analysis in random bipartite graphs}\label{sec:stable}

{
Let $H$ be a random one-sided $d$-regular bipartite graph  with uniformly generated strict preferences, where each $a\in \Short$ is independently connected to $d$ randomly chosen $j\in \Long$. Let $\Short' \subset \Short$ and $\Long' \subset \Long$ be subsets chosen independently of both the connections in $H$ and all preference arrays, with $|\Short'|= \gamma_1 n_{\Short}$ and $|\Long'|= \gamma_2 n_{\Long}$ for some $0< \gamma_1 \le 1$ and $\Omega(1) \le \gamma_2 \le 1$. Let $H'$ denote the vertex-induced subgraph of $H$ on $\Short' \cup \Long'$. 
For any $i\in \calV(H')$, let $\calN(i)$ denote the neighbors of $i$ on $H'$. 

Throughout the neighbor-set estimates in this section and their two-sided analogues, a selected set $S\subseteq\calN(\rho)$ is determined by the root data: its incident edges, their generating layers when applicable, and the root's own preferences or marks that, conditional on those edges, are independent of other agents' preferences and unexplored rows. It cannot be chosen using branch messages or matching outcomes. The whole neighbor set, a prefix of the root's preference list, and neighbors with a nonnegative root-side shock satisfy this convention under the model's independence assumptions. When $|S|$ is random, the probability bounds average the displayed power over the root data. Retained populations are fixed or sampled independently of the graph and preference arrays.

}%

{
We use availability as defined in \prettyref{sec:message}.
The following results bound the probability that a selected neighbor
set contains an available vertex.

}%
{
\begin{proposition}\label{prop:remove_gamma_d_omega_1}
    Suppose $\omega(1) \le d \le O\left(\polylog n_\Long\right)$ and $\gamma_1 n_{\Short} = \delta \gamma_2 n_{\Long}$ for some $\Omega(1) \le \delta \le 1+\frac{1}{d^\lambda}$ where $\lambda \ge \omega\left(\frac{1}{\log d}\right)$. 
    Define 
    \begin{align}
        \nu \triangleq \frac{\left(1\wedge \lambda\right) \log d}{d} \,. \label{eq:nu}
    \end{align}
    The selected neighbor sets obey the root-data convention above.\begin{itemize}
        \item    For any $a\in\Short'$ and $\calN'(a) \subset \calN(a)$, we have
           \begin{align}
                 \prob{\forall\text{ $j\in \calN'(a)$, $j$ is unavailable to $a$ on $H'$}
                 } 
                 & \le \Expect\!\left[\left(1-\underline C\nu\right)^{|\calN'(a)|-2}\right]  + o\left(\frac{1}{n}\right) \,,  \label{eq:j_a_calN'_a}
            \end{align}
         where  $\underline{C}>0$ is some constant that only depends on $\frac{\log d}{\log \log n_{\Long}}$.
           \item 
            For any $j\in\Long'$ and $\calN'(j) \subset \calN(j)$, we have
           \begin{align}
                 \prob{\forall\text{ $a\in \calN'(j)$, $a$ is unavailable to $j$ on $H'$}}
                 & \ge  \Expect\!\left[\left(1-\frac{\overline C}{\nu d}\right)^{|\calN'(j)|}\right] -o(1)\,,  \label{eq:a'_j_calN'_j}
            \end{align}
            where $\overline{C} >0$ is some constant that only depends on  $\frac{\log d}{\log \log n_{\Long}}$.  
    \end{itemize}
\end{proposition}
}%

{
\begin{corollary}\label{cor:calN_a_omega_1} 
    Suppose $\omega(1) \le d \le O\left(\polylog n_\Long\right)$ and $\gamma_1 n_{\Short} = \delta \gamma_2 n_{\Long}$ for some $\Omega(1) \le \delta \le 1+\frac{1}{d^\lambda }$ where $\lambda \ge \omega \left(\frac{1}{\log d}\right)$. 
    Use root-data measurable selected sets as above, with the size conditions holding with probability $1-o(1)$.\begin{itemize}
        \item  For any $a\in\Short'$ and $\calN'(a) \subset \calN(a)$ such that $ |\calN'(a)| \ge \omega \left( \frac{1}{\nu}\right)$,
       \begin{align*}
             \prob{\forall \text{ $j\in \calN'(a)$, $j$ is unavailable to $a$ on $H'$}} \le  o(1) \,.
       \end{align*}
       \item  For any $j\in\Long'$ and $\calN'(j) \subset \calN(j)$ such that $|\calN'(j) | \le o\left(\nu d \right)$,
       \begin{align*}
             \prob{\forall \text{ $a\in \calN'(j)$, $a$ is unavailable to $j$ on $H'$}} \ge 1-o(1) \,.
       \end{align*}
    \end{itemize}
\end{corollary}
}%

{
\begin{proposition}\label{prop:remove_gamma_d_log_n}
Suppose $\omega(1)\le d\le O(\polylog n_\Long)$ and
$\gamma_1n_\Short=\delta\gamma_2n_\Long$, where
$\delta\le1-\Omega(1)$ or $\delta\ge1+\Omega(1)$.
Choose a retained root and define
\[
 (\rho,c)=
 \begin{cases}
 (a,\,\delta/\log(1/(1-\delta))),&a\in\Short',\quad\delta<1,\\
 (j,\,1/[\delta\log(\delta/(\delta-1))]),&j\in\Long',\quad\delta>1.
 \end{cases}
\]

{
Let $S\subseteq\calN(\rho)$ obey the root-data convention above.

}%
There is a deterministic $\epsilon_n=o(1)$, uniform over these selection
rules, such that
\begin{align*}
 (1-\epsilon_n)\Expect\!\left[
      [1-(1+\epsilon_n)c]_+^{|S|}\right]
 &\le\prob{\text{every }v\in S\text{ is unavailable to }\rho\text{ on }H'}\\
 &\le\Expect\!\left[(1-(1-\epsilon_n)c)^{|S|-2}\right]+o(1/n).
\end{align*}

{
Here $[x]_+=\max(x,0)$ and $0^0=1$.
}%

\end{proposition}
}%

{
\begin{corollary}\label{cor:calN_a_logn}
Under the assumptions of \prettyref{prop:remove_gamma_d_log_n}, suppose
$\delta\le1-\Omega(1)$. For each retained applicant $a\in\Short'$, choose
$\calN'(a)\subseteq\calN(a)$ by that proposition's root-data convention.
For a fixed $\epsilon>0$, put
\[
 \tau_n=\max\left\{1,\ (1+\epsilon)\frac1\delta
     \log\left(\frac1{1-\delta}\right)\log n_\Short\right\}.
\]
Let $\Phi^{\Short'}_{H'}$ and $\Phi^{\Long'}_{H'}$ be the applicant-optimal
and job-optimal stable matchings, and let $\phi_\Phi$ be the partner map
of $\Phi$. Then
\begin{align}
 \Prob\left\{\begin{gathered}
   \exists\Phi\in\{\Phi^{\Short'}_{H'},\Phi^{\Long'}_{H'}\},\quad
   \exists a\in\Short':\\
   |\calN'(a)|\ge\tau_n\quad\text{and}\quad
   \phi_\Phi(j)\succ_j a\ \text{for every }j\in\calN'(a)
 \end{gathered}\right\}\le o(1).\label{eq:delta_constant_union_1}
\end{align}
Thus, with high probability, every retained applicant whose selected
neighbor set meets the threshold has a willing selected neighbor in
each of these two optimal stable matchings.
\end{corollary}
}%

\begin{corollary}\label{cor:unmatched_constant}

{
Suppose $\gamma_1=\gamma_2=1$, $\omega(1)\le d\le O(\polylog n_\Long)$, and $\delta_0\le\delta\le1-\delta_0$ for some fixed $0<\delta_0<1/2$. Let $U=|\Short_U|$ be the number of unmatched applicants on $H$, and define
\[
    c_\delta=\frac{\delta}{\log(1/(1-\delta))},
    \qquad \beta_\delta=-\log(1-c_\delta)>0.
\]
There is a deterministic sequence $\eta_n=o(1)$ such that
\[
    n_\Short e^{-(1+\eta_n)\beta_\delta d}
    \le \expect{U}
    \le n_\Short e^{-(1-\eta_n)\beta_\delta d}+o(1).
\]
Moreover, for every $t>0$,
\[
    \prob{|U-\expect{U}|\ge t}
    \le2\exp\left(-\frac{2t^2}{n_\Short}\right).
\]
Consequently, for each fixed $0<\epsilon<1/2$, if
$d\le(1/2-\epsilon)\beta_\delta^{-1}\log n_\Short$, then
\begin{align}
    \prob{n_\Short e^{-(1+o(1))\beta_\delta d}
          \le U\le
          n_\Short e^{-(1-o(1))\beta_\delta d}}
    \ge1-o(1).
    \label{eq:unmatched_constant}
\end{align}
In particular, this two-sided estimate holds when $d=o(\log n_\Short)$. If instead $d\ge(1+\epsilon)\beta_\delta^{-1}\log n_\Short$ for a fixed $\epsilon>0$, every applicant is matched on $H$ with high probability. These conclusions concern every stable matching, whose unmatched-agent sets coincide.

}%
\end{corollary} 

When $H'$ is relatively dense, we introduce the following proposition that characterizes the conditions under which applicants are guaranteed to be matched with partners within a certain top range of their preference lists
in the stable matchings. 

\begin{proposition}\label{prop:perfect_stable}
     Suppose $\gamma_1 n_{\Short} = \delta \gamma_2 n_{\Long}$ for some $1-o(1) \le \delta \le 1$.
     Let 
     \[ 
     r_n \triangleq \frac{1}{\delta} \log \left( \frac{1}{1-\delta + \frac{\delta^2}{\gamma_2 n_\Long}}\right) \log n_\Short\,.
     \]
     Then, with high probability, for any constant $\epsilon>0$, we have 
     \begin{itemize}
         \item   
{
If $d > \frac{1+2\epsilon }{\gamma_2}r_n$,
}%
  in the $\Short'$-optimal stable matching on $H'$:
            \begin{itemize}
            \item every $a \in \Short'$ is matched with one of its top $\left(1+\epsilon\right) r_n$ preferred firms;

{
\item there exists some $a\in\calA'$ that is not matched with any of its top
            $\left(1-\epsilon\right)r_n\,\frac{\log|\Short'|}{\log n_\Short}$ preferred firms.
}%

            \end{itemize} 
         \item  If $d > \frac{2+2\epsilon }{\gamma_2} r_n$ and $\delta < 1$,  every $a \in \Short'$ is matched with one of its top $\left(2+\epsilon \right) r_n$ preferred firms in any stable matching on $H'$. 
     \end{itemize}
\end{proposition}

{
The applicant-optimal bounds follow from the sharp perfect-matching threshold of
{\citet{potukuchi2024unbalanced}, together with the prefix couplings below.}
The lower bound uses the size of the retained applicant population;
it has the original scale $r_n$ whenever $\log|\Short'|\sim\log n_\Short$.
For the upper bound in every stable matching, we complete the applicants'
lists and combine a first-moment estimate with the total-rank bound of
\citet{ashlagi2017unbalanced}. This argument applies to arbitrary $d$
satisfying the displayed budget.
}%
{
The proof of \prettyref{cor:calN_a_omega_1} follows directly from \prettyref{prop:remove_gamma_d_omega_1}. The remaining proofs are given below.
}%

\subsubsection{Proof of \prettyref{prop:remove_gamma_d_omega_1}}
\label{sec:post_prop_remove_gamma_d_omega_1}
For any $a\in \Short$ and $\ell \in \naturals_+$, let $H_\ell(a)$ (resp. $H'_{\ell}(a)$) denote the vertex-induced subgraph of $H$ (resp. $H'$) on its $\ell$-hop neighborhood of $a$. Let $T'_{\ell}(a)$ denote the spanning tree rooted at $a$ with depth $\ell$ explored by the {breadth-first search} exploration on $H'_{\ell}(a)$. 

{
\begin{lemma}\label{lmm:claim_1}
Use the assumptions of \prettyref{prop:remove_gamma_d_omega_1} and put $N=n_\Long$, $B=\lfloor N^{1/8}\rfloor$. For an even depth
\[
 \ell\asymp\frac{\log N}{\log d\vee\log\log N},\qquad
 \ell\le\frac{\log N}{16(\log d\vee\log\log N)},
\]
reserve the root star and use the stopped, completed depth-first exploration defined below. Its parameters are
\begin{align}
 \kappa_A=\gamma_2d,\quad \kappa_J=\delta\gamma_2d,\quad
 K_t=\frac{\kappa_t}{1-2B/(n_\Short\wedge n_\Long)},
 \quad t\in\{A,J\},\label{eq:claim_tree_1}\\
 f_t=\kappa_t^{-1/4},\qquad r_t=\exp(-\sqrt{\kappa_t}/8).
 \label{eq:fone_ftwo_1}
\end{align}
Uniformly at each branch-entry history with at most $B$ discovered vertices and at most $2B$ previously queried negative coordinates in each pending row and a reserved root star of size $O(d+\log N)$, the completed messages satisfy
\begin{align}
 \Expect[X^{\rm even}_{J\to A}\mid\mathcal F]\ge\underline C\nu,
 \label{eq:claim_1}\\
 \Expect[X^{\rm odd}_{A\to J}\mid\mathcal F]\le\frac{\overline C}{\nu d}.
 \label{eq:claim_2}
\end{align}
Here ``even'' uses an applicant-rooted depth-$\ell$ truncation and ``odd'' a firm-rooted depth-$(\ell-1)$ truncation. The completed exploration agrees with the graph on every branch having no collision or stopping failure. The constants have the dependence stated in the proposition.
\end{lemma}
}%

{
\paragraph{Conditional exploration shared by the one-sided and two-sided proofs.}
Put $N_A=n_\Short$, $N_J=n_\Long$, $m_A=\gamma_1N_A$,
$m_J=\gamma_2N_J$, and let $N_*=N_A\wedge N_J$ and
$m_*=m_A\wedge m_J$.  In the one-sided case set $\widehat d=0$.
In the two-sided case the independent applicant and firm rows have fixed
sizes $d$ and $\widehat d$, respectively. Define
\[
 \kappa_A=m_J\left(\frac d{N_J}+\frac{\widehat d}{N_A}\right),
 \qquad
 \kappa_J=m_A\left(\frac d{N_J}+\frac{\widehat d}{N_A}\right),
 \qquad D=d\vee\widehat d.
\]
The conditional exploration means may differ from these nominal means.

Fix the root's entire star, including the layer memberships of its edges
in the two-sided case, and reserve its neighbors before exploring branches.
Their order may be any order fixed from this information, and may in
particular put two prescribed branches first. Explore the remaining tree
in depth-first order. Only vertices not yet discovered are eligible as new
children. When a vertex is expanded, reveal its own outgoing row and all
incoming memberships from still neutral opposite-side vertices. Revealing
both layer memberships also detects duplicate edges. Never query the
unrevealed parts of a previously discovered vertex's own row merely because
another vertex is expanded. Such a vertex is already excluded from the
neutral set. A newly discovered row therefore has at most one revealed
positive entry, namely its parent, and at most as many revealed negative
entries as there have been expansions. Conditional on this exact history,
unrevealed row subsets remain uniform and different rows remain independent.
These assertions follow by induction: each query reveals coordinates of
one independent uniform row, and its unqueried coordinates, conditional on
all answered coordinates, form a uniform subset of the remaining positions.
Preferences are independent of the graph and are exposed only after the
corresponding child subtree is complete.

Reserve a budget $B$ of discovered vertices and use $s=2B<N_*/2$ in all row denominators, allowing for both discovery queries and deferred negative collision audits. Stop the real exploration at a size or offspring cutoff or at its first collision. At an own-row collision, draw that audit and its neutral-child count jointly from the pre-audit residual-row law; retain this triggering count and replace subsequent counts. Do not condition the current offspring draw on a negative audit. At any expansion before this stopping rule, an own-row
contribution has law
\[
\begin{gathered}
 X\sim\Hyper(S,N,r),\qquad S\ge m_{\rm opposite}-s,\\
 N_{\rm opposite}-s\le N\le N_{\rm opposite},\qquad
 (d_{\rm own}-1)_+\le r\le d_{\rm own}.
\end{gathered}
\]
The other layer contributes, conditional on $X$,
\[
 Y\mid X\sim\Binom(S-X,p),\qquad
 \frac{d_{\rm other}}{N_{\rm current}}
 \le p\le\frac{d_{\rm other}}{N_{\rm current}-s}.
\]
Indeed, every still neutral row has only negative answers and still has its
full selection budget. All these bounds hold conditional on the entire
previously completed-branch history and on the reserved root star.
For one-sided firm expansions only the binomial contribution is present.
For one-sided applicant expansions only the hypergeometric contribution is
present.  Thus, for either type $t\in\{A,J\}$, its offspring count $O_t=X+Y$
satisfies
\begin{align}
 \Expect[O_t\mid\mathcal F]
 &\le K_t:=\frac{\kappa_t}{1-s/N_*},
 \label{eq:conditional_exploration_mean_upper}\\
 \Expect[O_t\mid\mathcal F]
 &\ge \kappa_t(1-b_t),\qquad
 b_t:=\frac1{\kappa_t}+\frac{s+D}{m_*}.
 \label{eq:conditional_exploration_mean_lower}
\end{align}
For the lower bound, the possible parent consumes at most one expected
outgoing edge. The depletion of the neutral pool and the removal of the
at most $D$ outgoing endpoints from the incoming candidates cost at most
$\kappa_t(s+D)/m_*$. These conservative bounds include early leaves and
include the root-neighbor reservations.

The union count has a Chernoff lower tail despite the overlap between its
two layers. To verify this directly, put $v=1-e^{-u}$ for $u>0$ and
$A=1-pv$. The hypergeometric generating-function bound
$\Expect[z^X]\le(1-(S/N)(1-z))^r$, $0\le z\le1$, gives
\begin{align*}
 \Expect[e^{-u(X+Y)}\mid\mathcal F]
 &=A^S\Expect\left[\left(\frac{e^{-u}}A\right)^X\middle|\mathcal F\right]\\
 &\le\exp\left[-v\{Sp+(1-p)\Expect[X\mid\mathcal F]\}\right]
 =\exp[-v\Expect[O_t\mid\mathcal F]].
\end{align*}
The hypergeometric bound follows from the elementary-symmetric-mean
inequality for sampling without replacement. Exponential Markov inequality
therefore gives, for every $2b_t\le f_t<1$,
\begin{equation}
 \Prob\{O_t<(1-f_t)\kappa_t\mid\mathcal F\}
 \le r_t:=\exp(-\kappa_tf_t^2/8).
 \label{eq:conditional_exploration_tail}
\end{equation}
These bounds hold uniformly over the permitted histories.

For the asymptotic applications below, $f_t\kappa_t\to\infty$.
Complete the forest after the stopping rule using the deterministic
offspring count $\lfloor\kappa_t\rfloor$ at type $t$. This count is at
most $K_t$ and at least $(1-f_t)\kappa_t$ for all sufficiently large $n$,
so one completed forest satisfies both envelopes simultaneously. Retain
the actual count of the vertex whose draw triggers the stopping rule;
replace only subsequent draws. Continue to the prescribed depth and use
independent uniform preferences.
The completed forest obeys the hypotheses used in the proof of
\prettyref{lmm:rooted_tree_expected_degree} at every nonroot vertex. The root
law is immaterial for its child-message conclusions. An applicant-rooted
comparison uses $(\kappa_1,\kappa_2)=(K_J,\kappa_A)$ and $(\fone,\ftwo)=(f_A,r_A)$;
a firm-rooted comparison uses $(K_A,\kappa_J)$ and $(f_J,r_J)$. The same
choices apply after a parity shift. 

For applications on compact ranges of $\delta>0$, put $N=N_J$ and suppose
$m_A,m_J=\Theta(N)$, $N_A\ge m_A$, $d=\omega(1)$, and
$\widehat d=O(d)$ and $d=O(\polylog N)$. These conditions follow from
$\gamma_2=\Omega(1)$ and $\delta=m_A/m_J=\Theta(1)$; they do not
require $\gamma_1=\Omega(1)$ or comparability of the original populations.
Take $B=\lfloor N^{1/8}\rfloor$ and $s=2B$.
For either fixed $\alpha\in(0,1/2)$ and $f_t=\kappa_t^{-\alpha}$,
$b_t=o(f_t)$ and $r_t=\exp[-\Theta(d^{1-2\alpha})]$, while
$K_t=\kappa_t(1+O(N^{-7/8}))$.
Below we use $\alpha=1/4$ for both types.
These choices preserve $\max(K_A,K_J)\max(r_A,r_J)=o(1)$.

The cutoff is controlled in the joint argument below. On a collision-free neighborhood, the exploration tree agrees with the breadth-first tree of that neighborhood. On other histories the completion is only a comparison process.

\paragraph{Root data and sequential products.}
Condition on the root's entire star, including generating-layer indicators
in the two-sided model, and on any root-local marks independent of the
remaining rows and other agents' preferences given that star. Choose
$S\subseteq\calN(\rho)$ according to the root-data convention above, and
put $R=|\calN(\rho)|$ and $k=|S|$.

Reserve all root neighbors before exploring branches. The preceding
offspring comparison is uniform after complete earlier branches and their
preferences have been revealed. Therefore a conditional lower bound $p$
on each next-branch proposal probability implies, by iterated conditional
expectation,
\begin{equation}
 \Expect\left[\prod_{i=1}^k(1-X_i)\right]\le(1-p)^k.
 \label{eq:sequential_branch_product_upper}
\end{equation}
The inequality reverses for a conditional upper bound $p$. Real exploration is replaced by synthetic
continuation when its stopping rule is triggered. Retain the triggering
vertex's full neutral-child count as one draw from its entry-history
residual-row law, and replace only subsequent draws; do not condition
that count on avoiding the trigger.

\paragraph{Collision counts and discarded branches.}
Work first with root degree $R\le T$ and a degree cutoff $L\ge T$. Choose
a deterministic vertex bound $B\ge1+T\sum_{u=0}^{m-1}L^u$.
On $\Delta(H')\le L$, every exploration order has offspring at most $L$
and cannot reach its vertex cutoff. This common agreement event is used
only after the conditional estimates; the offspring laws are never
conditioned on it. Use the history bounds with $2B$ reserved coordinates to cover
both discovery queries and negative collision audits.

Sample a processed vertex's own row by its remaining selected endpoints.
There are at most $B(d+\widehat d)$ endpoint draws, each hitting an already
discovered vertex with probability at most
$B/(N_*-2B-d-\widehat d)$. Deferred coordinate tests between discovered
vertices use only rows not already sampled in full. There are at most
$C B^2$ such tests, each positive with probability at most
$(d+\widehat d)/(N_*-2B-d-\widehat d)$. Here $N_*$ is the smaller
original sampling population of the layers used; in the one-sided case
take $N_*=n_\Long$ and $\widehat d=0$. These bounds follow from the exact
uniform residual-row law. In particular, a row is not exposed by first
testing almost all of its coordinates negative and then treating its last
coordinate as a fresh small-probability test.

Pad the two types of tests with zero indicators up to their deterministic
count bounds. Independent uniform variables couple each type to
independent Bernoulli indicators with the corresponding displayed upper
probability. Provided $B+d+\widehat d=o(N_*)$, the total mean of these
dominating indicators is at most
$\varepsilon_N=C B^2(d+\widehat d)/N_*$. Their factorial moments give
\begin{equation}
 \Pr\{\text{at least }r\text{ collision edges before stopping}
                  \mid\text{root data}\}
 \le\varepsilon_N^r/r!.
 \label{eq:collision_factorial_bound}
\end{equation}
An already known tree edge is not counted again.

For a connected rooted neighborhood of excess $r$, consider its components
after the root is removed. A component with $s$ root neighbors and
internal cycle rank $b$ contributes $b+s-1$ to the rooted excess.
Call it exceptional if this contribution is positive. Since
$s\le2(b+s-1)$ for every exceptional component, their root-neighbor set
$D$ has size at most $2r$. Every retained component is a tree with a
unique root path; its depth-first and breadth-first truncations coincide.

For each fixed candidate $D$, explore its components first, using
multi-source breadth-first depths. This discarded prefix may contain
cycles; no message estimate is used inside it. Incoming membership
queries omit pending retained rows. Audit the remaining coordinates
between the prefix's own discovered vertices, including its terminal
rows. The audits concern discarded rows only. Consequently every
pending retained row still has at most its one parent-positive entry and
at most $2B$ negative entries. Restart the completed-forest comparison
for the retained branches from this prefix history, switching its future
offspring to synthetic ones at its first new failure.

On the event that $D$ is the actual exceptional set, this prefix exhibits
all $r$ collision edges and every retained branch is unchanged. Only
after this event inclusion may we drop the requirement that $D$ be the
actual exceptional set. Conditional on its prefix history, the product
for the selected retained branches is at most
$q^{(k-2r)_+}$, where $q=1-p$. Deleting $D$ removes agents on the side
opposite the root, so the even-depth local certificate applies to the
retained component. There are at most $(T+1)^{2r}$ candidate sets.
For $0<q<1$, summing over all excesses gives
\begin{align}
 &\Pr\{\text{no available vertex in }S,\ \Delta(H')\le L
                  \mid\text{root data}\}\nonumber\\
 &\quad\le\sum_{r\ge0}
   \frac{((T+1)^2\varepsilon_N)^r}{r!}\,q^{(k-2r)_+}
 \le q^k\exp\left(\frac{(T+1)^2\varepsilon_N}{q^2}\right).
 \label{eq:all_cycles_joint_product}
\end{align}
This also covers $2r\ge k$, since
$q^{(k-2r)_+}\le q^{k-2r}$. 

\paragraph{Cutoffs and averaging.}
Put $N=n_\Long$. In the sparse and compact-imbalance applications,
both retained populations are $\Theta(N)$, both nominal offspring means
are $\Theta(d)$, and $N_*=\Omega(N)$; the original applicant population
need not be comparable to $N$. Take
$T=A_0(d+\log N)$ and $L=A_1(d+\log N)$ with $A_1\ge A_0$.
Choose the depth so that the preceding deterministic vertex bound is at
most $B=N^{1/8}$, for example the even depth used in the sparse proof.
Conditional on the root star, each other vertex degree is bounded by its
own budget, a possible root edge, and a sum of independent Bernoulli
variables of mean $O(d)$. A Chernoff bound and a union bound over the
$\Theta(N)$ retained vertices give
$\Pr\{\Delta(H')>L\mid\text{root data}\}\le C N e^{-cL}$.
The same global event applies to all exploration orders for the candidate sets $D$. Its exponential constant can be made
sufficiently large by increasing $A_1$. The root-degree decomposition gives
$\Expect R=O(d)$ and
$\Pr\{R>T\}\le\exp[-cA_0(d+\log N)]$, with the analogous freedom in
$A_0$. These estimates concern full root data, not only the selected
cardinality.

Here $q\to1$ in the sparse application and $q$ is bounded away from zero
in the compact-imbalance dense application. Since $k\le R$, Jensen's
inequality gives
\[
 \Expect[q^k]\ge\Expect[q^R]\ge q^{\Expect R}\ge e^{-C d}.
\]
Choose the cutoff constants successively large. The two discarded tails
are then $o(p\Expect[q^k])$, while
$\xi_N=(T+1)^2\varepsilon_N/q^2=o(p)$; the latter follows from
$p\asymp\nu$ in the sparse case or $p\asymp1$ in the compact dense case.
Averaging \eqref{eq:all_cycles_joint_product}, and using
$q^{-2}-1\ge2p$, yields
\begin{equation}
 \Pr\{\text{no available vertex in }S\}
 \le(e^{\xi_N}+o(p))\Expect[q^k]
 \le\Expect[q^{k-2}].
 \label{eq:pruning_joint_product}
\end{equation}
Thus the original loss of two neighbors survives, with no additive error
requiring normalization by the original population $n$.
\paragraph{A relative lower bound from sequential good-branch events.}
Fix admissible root data with $R\le T$ and a selected neighbor set
$S=\{i_1,\ldots,i_k\}$.  Reserve all root neighbors at the outset, and
explore the selected branches first.  Suppose that the artificial completed
branch has conditional message mean at most $p_+$ at every entry history
under consideration.  Let $\epsilon_N$ uniformly bound the probability,
from any such history, of a collision or an offspring-cutoff failure while
exploring one branch and checking its edges to the vertices already seen.
The endpoint and deferred-coordinate estimates above, together with the
offspring Chernoff bound, give
\[
\epsilon_N\le C B^2(d+\widehat d)/N_*+B e^{-cL}=o(1).
\]
Use the stopped completion above, which retains the triggering count
(jointly drawn with its audit when applicable) from the entry-history
residual-row law and replaces only subsequent counts. Its conditional
offspring bounds hold on every history, and it agrees with the real branch
whenever no failure occurs.
Write $q_+=1-p_+$.  If $\mathcal F_{i-1}$ contains the complete preceding
branches, their preference coins, and the negative collision tests, then
\begin{align*}
 &\Pr\{X_i=0,\ \hbox{the $i$th branch is good}\mid\mathcal F_{i-1}\}\\
 &\qquad\ge
 \Pr\{X_i^{\rm completed}=0\mid\mathcal F_{i-1}\}
       -\Pr\{\hbox{the $i$th branch is bad}\mid\mathcal F_{i-1}\}
 \ge q_+-\epsilon_N.
\end{align*}
This holds also after conditioning on all preceding selected messages
being zero: those messages are measurable in $\mathcal F_{i-1}$.
Iterating the inequality therefore gives at least
$(q_+-\epsilon_N)^k$ for zero messages on every selected branch and no
failure during their exploration.

Next explore all unselected root branches and reveal any remaining edge
coordinates between discovered vertices.  Conditional on every surviving
history, the probability of another collision or a cutoff failure is
again at most $\epsilon_N$, by the same uniform bound applied to the whole
remaining exploration.  Consequently the probability that the entire
odd-depth neighborhood is a tree and all selected branch messages are
zero is at least
\begin{equation}
 (1-\epsilon_N)(q_+-\epsilon_N)^k.
 \label{eq:relative_sequential_lower}
\end{equation}
Odd-depth truncation turns this into a lower bound for no available
selected neighbor in the full graph.  In the compact-imbalance regime,
$p_+=(1+o(1))c_\delta$ and $q_+$ is bounded away from zero.  The additional
$\epsilon_N$ can be included in the $o(1)$ multiplying $c_\delta$, giving
the asserted relative lower bound.  The root-cutoff tail is removed by
the averaging argument above.

Deferred collision audits may involve unexpanded terminal rows of
completed branches. Their boundary messages are fixed, so these audits
reveal no additional preference information. Negative answers preserve
the uniform residual-row law, and positive answers declare failure.
Thus every surviving pending row still has at most one parent-positive
entry and at most $2B$ negative entries, as allowed in the shared mean
envelopes. These bounds remain valid after conditioning on the preceding
zero messages.

\paragraph{Applying the comparison to \prettyref{prop:remove_gamma_d_omega_1}.}
Put $N=n_\Long$ and choose
\begin{equation}
 m=2\left\lfloor\frac{\log N}{32(\log d\vee\log\log N)}\right\rfloor.
 \label{eq:m}
\end{equation}
At this depth a degree cutoff $L=C(d+\log N)$ bounds the explored size by $N^{1/16+o(1)}<B=N^{1/8}$. Lemma~\ref{lmm:claim_1}, proved below, supplies the entry-history bounds $p=\underline C\nu$ and $p_+=\overline C/(\nu d)$.
For an applicant root, even-depth local truncation and deletion of the exceptional root-neighbor components give the deterministic availability certificate used in \eqref{eq:all_cycles_joint_product}: no available selected firm in the graph implies zero messages on all retained selected branches. Deleting firms can only improve the remaining firms' best stable partners, so any retained firm available in the pruned tree is available in the unpruned local graph. Apply \prettyref{lmm:local_available} to pass from the even local graph to the full graph. The all-excess bound, followed by averaging root data, yields
\begin{equation}
 \Pr\{\text{every }j\in\calN'(a)\text{ is unavailable to }a\text{ on }H'\}
 \le\Expect[(1-\underline C\nu)^{|\calN'(a)|-2}].
 \label{eq:expect_j_calN'_a_backslash_calV_a}
\end{equation}
This is stronger than \eqref{eq:j_a_calN'_a} and avoids an error normalized by the possibly much larger original population $n$.

For a firm root, use odd depth $m-1$. By \prettyref{lmm:local_available}, zero selected messages on a collision-free odd neighborhood imply no selected applicant is available in the full graph. Apply the sequential good-branch lower bound \eqref{eq:relative_sequential_lower}. Here $k\le C(d+\log N)$ on the root cutoff event and $k\epsilon_N=o(1)$, so
\begin{equation}
 \Pr\{\text{every }a\in\calN'(j)\text{ is unavailable to }j\text{ on }H'\}
 \ge\Expect[(1-\overline C/(\nu d))^{|\calN'(j)|}]-o(1).
 \label{eq:expect_a'_calN'_j}
\end{equation}
The root cutoff has probability $1-o(1)$. This proves \eqref{eq:a'_j_calN'_j}.

}%
Lastly, we are left to prove \prettyref{lmm:claim_1}. 
\begin{proof}[Proof of \prettyref{lmm:claim_1}]

{

The conditional offspring estimates proved above apply with $\widehat d=0$ and $s=2B$. Put $f=f_A$, $L_A=(1-f)\kappa_A$, and enlarge the firm mean bound to
\[
 K_J^*=\frac{(1+d^{-\lambda})\kappa_A}{1-2B/(n_\Short\wedge n_\Long)},\qquad
 \eta=(1-2(K_J^*\vee\kappa_A)r_A)^{-1},\qquad b=\eta K_J^*.
\]
This is an upper bound on every conditional firm offspring mean throughout the allowed imbalance range. Enlarging this mean weakens both required comparisons. The conditional induction of \prettyref{lmm:rooted_tree_expected_degree} gives the applicant-message upper bound
\[
 r_A+(1-r_A)b_\ell,\qquad
 b_\ell=(f_{L_A}\circ f_b)^{\ell/2-1}(1),
\]
and the firm-message lower bound $f_b(b_\ell)$. These bounds hold after every previously completed branch and its preferences have been revealed.

Here is the scalar estimate, including the finite-population slack. Let
\[
 c=\frac{1-3f}{1+d^{-\lambda}},\qquad a_0=c(b+1)-1,\qquad
 \beta=\lambda\wedge\tfrac14.
\]
Since $b=\kappa_A(1+d^{-\lambda})(1+o(f))$, eventually $0<a_0\le L_A$ and
$1-c=\Theta(d^{-\lambda}+d^{-1/4})$. Hence
\[
 b_\ell\le(f_{a_0}\circ f_b)^{\ell/2-1}(1),\qquad
 x^*=-\frac c{\log(1-c)}=O\!\left(\frac1{\beta\log d}\right).
\]
The ratio used here is exactly $(a_0+1)/(b+1)=c$, as required by \ref{F:1} of \prettyref{lmm:fixed_point_convergence}. Take $\epsilon=1\vee(\log d/\log\log N)$. Its iteration bound is
$O(d^{2\beta/(2+\epsilon)}\log\log d)$: this follows by substituting
$1-c=\Theta(d^{-\beta})$ into its displayed contraction factor. This many iterations are $o(\ell)$, since $d=O(\polylog N)$ and $\beta\le1/4$. Thus
\begin{equation}
 b_\ell\le\frac{C}{(1\wedge\lambda)\log d}
 =\frac{C}{\nu d}.\label{eq:mu_j_rho_lower_bound_2}
\end{equation}
Since $r_A=o(1/(\nu d))$, this proves \eqref{eq:claim_2}, including the contribution of small offspring counts. Finally, $b=\Theta(d)$, $\nu d\to\infty$, and $\nu d\le\log d$; the formula for $f_b$ gives
\[
 f_b(b_\ell)\ge f_b(C/(\nu d))\ge\underline C\nu.
\]
This proves \eqref{eq:claim_1}. The same scalar argument applies to the two-sided nominal means whenever $\kappa_A=\Theta(d)$ and $\kappa_J/\kappa_A=\delta$.

}%
\end{proof}
{
\subsubsection{Proof of \prettyref{prop:remove_gamma_d_log_n}}\label{sec:post_prop_remove_gamma_d_log_n}
Write $N=n_\Long$, $K=\gamma_2N$, $M=\gamma_1n_\Short$, and
$\kappa_A=\gamma_2d$, $\kappa_J=\delta\kappa_A$.
Use the conditional exploration and root-data convention proved above.
The next lemma gives bounds for the completed comparison forests.

\begin{lemma}\label{lmm:claim_1_logn}
Suppose $\delta$ lies in a compact subset of $(0,1)$ or $(1,\infty)$.
Use the history-uniform parameters $K_t,f_t,r_t$ from
\eqref{eq:conditional_exploration_mean_upper} and
\eqref{eq:conditional_exploration_tail}, with $f_t=\kappa_t^{-1/4}$.
Let $\widetilde T_\ell(\rho)$ denote the resulting completed exploration
rooted at a retained vertex $\rho$, conditional on admissible root data.
For even $\ell=\omega(\log d)$ with
$\ell\le\log N/[32(\log d\vee\log\log N)]$, put
\[
 u_J=
 \begin{cases}
 \delta/\log(1/(1-\delta)),&\delta<1,\\
 \log(\delta/(\delta-1))/(\gamma_2d),&\delta>1,
 \end{cases}
 \qquad
 u_A=
 \begin{cases}
 \log(1/(1-\delta))/(\delta\gamma_2d),&\delta<1,\\
 1/[\delta\log(\delta/(\delta-1))],&\delta>1.
 \end{cases}
\]
At every branch entry, conditional on all earlier completed branches
and their preference coins, the following bounds hold uniformly:
\begin{align}
 \Expect[\sfX_{j,a}(\widetilde T_\ell(a))\mid\mathcal F]
 &\ge(1-o(1))u_J,\label{eq:claim_1_logn}\\
 \Expect[\sfX_{j,a}(\widetilde T_{\ell-1}(a))\mid\mathcal F]
 &\le(1+o(1))u_J,\label{eq:claim_2_logn}\\
 \Expect[\sfX_{a,j}(\widetilde T_\ell(j))\mid\mathcal F]
 &\ge(1-o(1))u_A,\label{eq:claim_1_logn_j}\\
 \Expect[\sfX_{a,j}(\widetilde T_{\ell-1}(j))\mid\mathcal F]
 &\le(1+o(1))u_A.\label{eq:claim_2_logn_j}
\end{align}
The compact-range restriction is only for this auxiliary comparison;
the proposition's cases $\delta\to0$ and $\delta\to\infty$ are proved
directly below.
\end{lemma}

\begin{proof}[Proof of \prettyref{lmm:claim_1_logn}]
We record the scalar convergence used for the dense case. Suppose
$a,b=\Theta(d)$ and $a/b=\rho+o(1)$, uniformly for $\rho$ in a compact
subset of $(0,1)$. Put $H=f_a\circ f_b$ and
$c_\rho=\rho/\log(1/(1-\rho))$. Then
\begin{equation}
 H^h(0)=c_\rho+o(1),\qquad H^h(1)=c_\rho+o(1)
 \quad\text{whenever }h/\log d\longrightarrow\infty.
 \label{eq:conditional_dense_scalar_limit}
\end{equation}
Here all errors are uniform over the stated parameter ranges. To prove
this, let $\rho_d=(a+1)/(b+1)$. For $x\ge1/(b+1)$, monotonicity and
\[
 f_b(x)\le\frac1{(b+1)x},\qquad
 f_a(y)\ge\frac{1-e^{-(a+1)y}}{(a+1)y}
\]
give
\[
 H(x)\ge g_{\rho_d}(x):=
 \frac{x}{\rho_d}(1-e^{-\rho_d/x}).
\]
Also $H(0)=1/(a+1)\ge1/(b+1)$ for all sufficiently large $d$.
For sufficiently small fixed $b_0>0$, the compact range of $\rho_d$
implies $g_{\rho_d}(x)\ge(1+c)x$ on $0<x\le b_0$, for some fixed
$c>0$. Thus the iteration from zero enters $[b_0,1]$ in $O(\log d)$
steps. On that interval $H\to g_\rho$ uniformly. The limiting map has
the unique positive fixed point $c_\rho$, is strictly above the diagonal
below $c_\rho$, and strictly below the diagonal above it. The iterations
from zero and one are monotone and bracket all other iterations.
Uniform drift away from any fixed neighborhood of $c_\rho$ therefore
puts both iterations into that neighborhood in a further bounded number
of steps. This proves the display.

Use the exploration parameters
\[
 L_A=(1-f_A)\kappa_A,\quad L_J=(1-f_J)\kappa_J,
 \quad
 \eta_A=(1-2(K_J\vee\kappa_A)r_A)^{-1},\quad
 \eta_J=(1-2(K_A\vee\kappa_J)r_J)^{-1}.
\]
The conditional tree induction gives, for even $\ell$, the following
bounds for firm-to-applicant messages at the two truncation parities:
\begin{align*}
 \Prob(X_{J\to A}^{\rm even}=1\mid\text{entry history})
 &\ge (f_{\eta_AK_J}\circ f_{L_A})^{\ell/2}(0),\\
 \Prob(X_{J\to A}^{\rm odd}=1\mid\text{entry history})
 &\le r_J+(1-r_J)
       (f_{L_J}\circ f_{\eta_JK_A})^{\ell/2-1}(1).
\end{align*}
The first bound is the lower composition in
\prettyref{lmm:rooted_tree_expected_degree}, rewritten by using
$f_{\eta_AK_J}(1)=(f_{\eta_AK_J}\circ f_{L_A})(0)$.
The second retains the bad-count contribution at the firm vertex.
Both bounds are uniform after previously completed branches, including
their preferences, have been revealed.

For $\delta=\kappa_J/\kappa_A$ in a compact subset of $(0,1)$,
the degree ratios in these two compositions converge to $\delta$,
since $f_A,f_J=o(1)$, $K_t/\kappa_t=1+o(1)$, and
$\max(K_A,K_J)\max(r_A,r_J)=o(1)$. Consequently, when
$\ell=\omega(\log d)$, the even-depth lower bound is
$c_\delta-o(1)$ and the odd-depth upper bound is $c_\delta+o(1)$.
For $\delta$ in a compact subset of $(1,\infty)$, interchange the
applicant and firm roles. The corresponding applicant-to-firm message
bounds are
\[
 c_{1/\delta}+o(1)
 =\frac{1+o(1)}{\delta\log(\delta/(\delta-1))}.
\]

For the other direction of a message, the conditional tree induction
places its probability between
$f_{(1+o(1))\kappa_t}(c+o(1))$ and
$r_t+(1-r_t)f_{(1-o(1))\kappa_t}(c-o(1))$, where $c$ is the
opposite direction's constant just obtained. Here $\kappa_t=\Theta(d)$,
$c$ is bounded away from zero, and $r_t=o(1/d)$.
Both expressions equal $(1+o(1))/(\kappa_tc)$, by the formula for $f$.
For $\delta<1$, this gives
$u_A=1/(\kappa_Ac_\delta)$; for $\delta>1$, it gives
$u_J=1/(\kappa_Jc_{1/\delta})$.
The same conditional induction with the parity interchanged proves
the remaining displayed bounds. 
\end{proof}

\begin{proof}[Proof of \prettyref{prop:remove_gamma_d_log_n}]
First suppose that $\delta$ is in either compact range of the lemma.
Choose an even $m$ such that
\begin{align}
 m=\omega(\log d),\qquad
 m\log(d+\log N)=o(\log N);\label{eq:m_logn}
\end{align}
for example, the nearest even integer to $(\log\log N)^2$.
Condition on admissible root data with full root degree
$R\le T=A_0(d+\log N)$, and write $k=|S|$.
Use a nonroot offspring cutoff $L=A(T+d+\log N)$.
The number of vertices before the cutoff is
$B\le1+T(1+L+\cdots+L^m)=N^{o(1)}$.
The collision bound is $\varepsilon_N=O(B^2d/N)$ and the global degree
cutoff tail is at most $CN(eCd/L)^L$.
All constants in this paragraph may be chosen uniformly on the
compact imbalance range.

Let $c=u_J$ for an applicant root and $c=u_A$ for a firm root.
By \prettyref{lmm:claim_1_logn}, the even-depth conditional proposal
mean is at least $p_-=(1-o(1))c$ and the odd-depth mean is at most
$p_+=(1+o(1))c$. Even-depth truncation and the cycle-prefix bound
\eqref{eq:all_cycles_joint_product} give, with $q_-=1-p_-$,
\begin{align}
 \prob{\text{no available vertex in }S\mid\text{root data}}
 &\le q_-^k
 \exp\!\left(\frac{(T+1)^2\varepsilon_N}{q_-^2}\right)
       +CN(eCd/L)^L.\label{eq:j_calN'_a_unavailable_logn}
\end{align}
This bound charges the probability of the deleted cycle branches
before conditioning on the remaining branches.
For the lower bound, \eqref{eq:relative_sequential_lower} and
odd-depth truncation give
\begin{align}
 \prob{\text{no available vertex in }S\mid\text{root data}}
 &\ge(1-\zeta_N)(1-p_+-\zeta_N)^k,
 \label{eq:a'_calN'_j_unavailable_logn}
\end{align}
where $\zeta_N=O(B^2d/N)+CN(eCd/L)^L=o(1)$ uniformly over entry histories.
The error is subtracted from the one-branch failure probability before
taking its power. Thus the lower bound is relative even if its right
side is exponentially smaller than $1/N$.

On the compact ranges $q_-$ and $1-p_+$ stay bounded away from zero.
Increasing $A$ makes the cutoff tail $o(q^k)$ uniformly for $k\le T$
and either failure base $q$. The exponential multiplier in the upper
bound is $1+o(1)$ and is absorbed by its two-neighbor exponent shift,
after an $o(1)$ decrease of the availability constant. For $k\le2$
the upper bound is trivial. The lower display has precisely the
claimed form after absorbing $\zeta_N$ into its $o(1)$ terms.

For an applicant root $R\le d$. For a firm root
$R\sim\Binom(M,d/N)$ with mean $O(d)$. Choose $A_0$ large enough that
$\prob{R>T}\le\exp[-C_0(d+\log N)]$ for any desired fixed $C_0$.
For every selected $S\subseteq\calN(\rho)$ and fixed $q$ in the
compact range of failure bases,
\[
 \Expect[q^{|S|}]\ge\Expect[q^R]\ge q^{\Expect R}\ge e^{-C_1d}.
\]
The root-star tail is consequently $o(\Expect[q^{|S|}])$.
Averaging the conditional estimates proves both asserted bounds with
expectations of the random powers. In fact the upper argument absorbs
all these errors and needs no residual additive error; it therefore
also implies the stated version with $o(1/n)$ for the paper's original
$n=n_\Short+n_\Long$.

\paragraph{Imbalance tending to zero or infinity.}
The preceding compact-range argument can be extended to the full
imbalance ranges in the proposition, but its lower bound becomes
vacuous at an extreme imbalance. Write $[z]_+=\max(z,0)$ when its
lower-bound base is used. If $\delta\to0$ or $\delta\to\infty$, the
corresponding availability constant tends to one. The unspecified
$o(1)$ in that lower bound can therefore be chosen so that its base
is zero. A nonempty selected set then has the trivial lower bound
zero; an empty set has probability one, with the convention $0^0=1$.

For $\delta\to0$, fix $\varepsilon>0$ and add independent applicant
rows until the active applicant population is
$\lceil\varepsilon K\rceil$, where $K=\gamma_2n_\Long$.
This comparison is used only once $\delta<\varepsilon$.
The root applicant's star and its selected set are unchanged.
Adding competing applicants weakly improves every firm's firm-optimal
partner, so it can only increase the event that all selected firms
are unavailable to the root applicant. Apply the compact-range upper
bound at imbalance $\varepsilon+o(1)$ in this enlarged market.
Since $c_\varepsilon\to1$ as $\varepsilon\downarrow0$, a diagonal
choice of $\varepsilon\downarrow0$ gives an upper bound whose
one-neighbor failure base tends to zero. This has the required form
$1-(1-o(1))c_\delta$, because $c_\delta\to1$ as well. 

For completeness, there is a direct bound when $\delta\to\infty$
which is uniform over every root star, including stars of atypically
large degree. Fix a firm $j$, condition on its root data, and let its
selected applicant set be $S$, $|S|=k$. Put $N=n_\Long$,
$K=\gamma_2N$, and $M=\gamma_1n_\Short=\delta K$.
If every member of $S$ is unavailable to $j$, each is matched to a
strictly preferred firm in the applicant-optimal matching.
Delete $j$ from that matching. If it was matched, free its former
partner and let that applicant propose in descending order only to
firms below its former partner. Continue the resulting rejection
chain in the same manner. Firms only improve their holdings; a
previously rejected or less-preferred proposer cannot create a
blocking pair. At termination the resulting matching is stable
without $j$. Every previously unmatched applicant remains unmatched,
and at most one previously matched applicant becomes unmatched.
The Rural Hospital Theorem then shows that at least $k-1$ members
of $S$ are matched also in the applicant-optimal matching of the
market without $j$.

Conditional on the root star, an applicant who was adjacent to $j$
has a uniform ordered $(d-1)$-list on the other $N-1$ firms; an
applicant who was not adjacent to $j$ has a uniform ordered $d$-list.
In the market without $j$, delete the final choice from the full
ordered list of each applicant in the latter group. If this choice
is an active firm, it is that applicant's last active choice;
otherwise this deletion has no effect on the reduced market.
None of the lists of applicants in $S$ is changed. Run deferred
acceptance first with these truncated lists and then activate the
deleted last choices of unmatched applicants. This is a valid run
with the untruncated lists. During the continuation, firms only
improve their holdings, so applicants in $S$ can only lose partners
or move to less-preferred partners. In particular, any member of
$S$ matched in the untruncated market was already matched in the
truncated market.

After this truncation all $M$ applicant lists are independent uniform
ordered $(d-1)$-lists on the same $N-1$ firms. Preferences at the
remaining firms are independent and uniform. Thus the matched
applicant set is exchangeable and has cardinality at most $K-1$.
For $k\ge2$, a union bound over the $k$ subsets of $S$ of size $k-1$
gives
\[
 \Pr\{\text{all }a\in S\text{ are unavailable to }j\mid
                 \text{root data}\}
 \le k\frac{(K-1)_{k-1}}{(M)_{k-1}}
 \le k\delta^{-(k-1)},
\]
where $(z)_r=z(z-1)\cdots(z-r+1)$. If $k\ge K$, the original event
is impossible, since its members would require $k$ distinct partners
among $K-1$ firms. For $\delta\ge2$ and $k\ge2$,
\[
 k\delta^{-(k-1)}\le(2/\delta)^{k-2},
\]
because $k/2^{k-2}\le2$. For $k\le2$, the proposed upper bound is
trivial. Finally,
\[
 2/\delta=1-(1-\eta_\delta)c_{1/\delta},\qquad
 \eta_\delta=
 \frac{2/\delta-(1-c_{1/\delta})}{c_{1/\delta}}=o(1)
 \quad(\delta\to\infty).
\]
This proves the required extreme-imbalance upper bound.
Together with the compact-range argument and the preceding lower-bound
convention, it covers every subsequence in either stated imbalance
range. If the selected cardinality is random, average its power over
the root data in the unconditional statement.

\end{proof}
\subsubsection{Proof of \prettyref{cor:calN_a_logn}}\label{sec:cor_calN_a_logn}
Write $c_\delta=\delta/\log(1/(1-\delta))$, so $0<c_\delta\le1$.
For a retained applicant $a$, let $E_a$ be the event that every selected
neighbor is unavailable. Applicant roots have full degree at most $d$.
Consequently the proof of \prettyref{prop:remove_gamma_d_log_n} gives
its stronger upper bound, with no additive error, conditional on every
admissible applicant root datum. This holds on compact imbalance ranges
and, by the augmentation argument, when $\delta\to0$.
Along a subsequence with $n_\Short\to\infty$, the threshold $\tau_n$
eventually exceeds two, and hence
\begin{align*}
 \prob{E_a,\ |\calN'(a)|\ge\tau_n}
 &\le\Expect\!\left[\indc{|\calN'(a)|\ge\tau_n}
       (1-(1-o(1))c_\delta)^{|\calN'(a)|-2}\right]\\
 &\le\exp\left(-(1-o(1))c_\delta(\tau_n-2)\right)
 \le n_\Short^{-1-\epsilon+o(1)}.
\end{align*}
The fixed exponent shift contributes only $e^{O(1)}$ in this exponential
bound. A union bound over $|\Short'|\le n_\Short$ makes the probability
of any such event tend to zero.
Along a subsequence with bounded $n_\Short$, choose one selected neighbor
whenever the threshold is met; this is possible since $\tau_n\ge1$.
Conditional on its root data, this firm has another retained applicant
neighbor with probability at most $(n_\Short-1)d/N=o(1)$, where
$N=n_\Long$. A firm with just its root neighbor is available to that
applicant. A union bound over the bounded number of applicants gives
the same conclusion. These two arguments cover every subsequence.
A firm is available to $a$ precisely when it weakly prefers $a$ to its
partner in the job-optimal matching, its most preferred stable partner.
Such a firm also weakly prefers $a$ to its partner in the applicant-optimal
matching. This proves the two stated optimal-matching conclusions for
all applicants whose selected sets meet the threshold.

}%

{
\subsubsection{Proof of \prettyref{cor:unmatched_constant}}
Write $M=n_\Short$. Since $\gamma_1=\gamma_2=1$, we have $H'=H$. An applicant is unmatched if and only if all its neighbors are unavailable. The forward implication follows from stability. Conversely, if an applicant $a$ is matched, its partner in the firm-optimal stable matching weakly prefers $a$ to every one of its stable partners, so this partner is available to $a$. The unmatched set is common to all stable matchings by the Rural Hospital Theorem.

Apply \prettyref{prop:remove_gamma_d_log_n} to the full neighborhood $\calN(a)$, which has size $d$. Applicant symmetry and linearity of expectation give
\begin{align}
    M\left(1-(1+o(1))c_\delta\right)^d(1-o(1))
    &\le\expect{U}\nonumber\\
    &\le M\left(1-(1-o(1))c_\delta\right)^{d-2}+o(1).
    \label{eq:expect_calA_U}
\end{align}
The additive error is $M\,o(1/n)=o(1)$. On the compact range $\delta\in[\delta_0,1-\delta_0]$, both $c_\delta$ and $1-c_\delta$ are bounded away from zero, and $\beta_\delta$ is bounded above and below by positive constants. Thus
\[
    \log\left(1-(1\pm o(1))c_\delta\right)
      =-\beta_\delta+o(1).
\]
Since $d\to\infty$, the factors $1-o(1)$ and the shift from $d$ to $d-2$ can be absorbed in an $o(1)$ relative error in the exponent. This proves the claimed expectation bounds with some $\eta_n=o(1)$. In particular, replacing $(1-c_\delta)^d$ by $e^{-c_\delta d}$ would lose the correct exponential rate.

For concentration, use the independent applicant-block representation
and the deletion-and-reinsertion argument in the proof of
\prettyref{prop:unmatched}. They apply to arbitrary population sizes
and show that changing one block changes $U$ by at most one.
The bounded-differences inequality~\citep[Lemma~1.2]{mcdiarmid1989bounded}
therefore gives
\begin{align}
    \prob{|U-\expect{U}|\ge t}
      \le2\exp\left(-\frac{2t^2}{M}\right),\qquad t>0.
    \label{eq:unmatched_bounded_differences}
\end{align}

If $d\le(1/2-\epsilon)\beta_\delta^{-1}\log M$, the expectation lower bound implies
\[
    \expect{U}\ge M^{1/2+\epsilon-o(1)}.
\]
Set $t=\sqrt{M\log M}$ in \prettyref{eq:unmatched_bounded_differences}. Then $t=o(\expect{U})$ and the failure probability is at most $2M^{-2}$. Hence $U=(1+o(1))\expect{U}$ with high probability, and \prettyref{eq:unmatched_constant} follows from the expectation bounds. In particular, the same argument applies to $d=o(\log M)$.

Finally, if $d\ge(1+\epsilon)\beta_\delta^{-1}\log M$, then
\[
    \expect{U}\le M^{-\epsilon+o(1)}+o(1)=o(1).
\]
Markov's inequality gives $\prob{U\ge1}\le\expect{U}=o(1)$. This proves the asserted sufficient condition for every applicant to be matched.

}%
\subsubsection{Proof of \prettyref{prop:perfect_stable}}\label{sec:post_prop_perfect_stable}
Fix any arbitrarily small but fixed constant $0<\epsilon<1$.
 For every $a\in\Short'$, {let $\calN(a)$ denote its neighbors on $H'$. Then}
\[
    |\calN(a)| \sim \Hyper\left(\gamma_2 n_{\Long} , n_{\Long}, d \right) 
\]
{

Hoeffding's sampling-without-replacement comparison~\citep[Theorem~4]{hoeffding1963probability}, followed by the binomial Chernoff bound, gives

}%
\begin{align}
    \prob{|\calN(a)| \le \left( 1- \frac{\epsilon}{4} \right)\gamma_2  d } \le  \exp\left( - \frac{\epsilon^2 \gamma_2^2 }{32}   d   \right) = o\left(\frac{1}{n_{\Short}}\right)\,, \label{eq:prob_calN_a}
\end{align}
where the last equality holds because $\gamma_2 d = \omega \left(\log n_\Short \right)$.

{
\paragraph{Case $1$: $d \ge \frac{1+2\epsilon}{\gamma_2}r_n$ and $\delta\le1$.}
We first prove the upper-rank assertion. Put $M=|\Short'|$, $N=|\Long'|$, and $\alpha=(N-M)/M=\delta^{-1}-1=o(1)$. In the notation of~
{
\cite[Lemma 8]{potukuchi2024unbalanced}
}%
, the threshold for a random market in which the short side has lists of length $k$ is
\[
    d_0=\log M\,\log\left(\frac{1+\alpha}{\alpha+1/(M(1+\alpha))}\right)
       =\log M\,\log\left(\frac1{1-\delta+\delta/N}\right).
\]
Because $M\le n_\Short$, $\delta\le1$, and $N=\gamma_2n_\Long$, the definition of $r_n$ gives $d_0\le r_n$.
Let $k=\lfloor(1+\epsilon)r_n\rfloor$. For all sufficiently large markets, $k>(1+\epsilon/2)d_0$. Also
\[
    (1-\epsilon/4)(1+2\epsilon)
       =1+\frac74\epsilon-\frac12\epsilon^2>1+\epsilon
\]
for $0<\epsilon<1$. Thus \prettyref{eq:prob_calN_a} and a union bound imply that every applicant in $\Short'$ has at least $k$ neighbors with probability $1-o(1)$.
On this event, retain for each applicant only its first $k$ neighbors in its preference ordering. Conditional on the degree vector, these ordered lists are independent and uniformly distributed over the ordered $k$-subsets of $\Long'$; the jobs' preference orders are independently uniform. Hence the retained market has precisely the random short-side-list law to which~
{
\cite[Lemma 8]{potukuchi2024unbalanced}
}%
 applies, after interchanging its job and candidate labels. With probability $1-o(1)$, short-side-proposing DA on this retained market matches every applicant.
Couple that run to applicant-proposing DA on $H'$ using the same proposal order. Until an applicant exhausts its first $k$ choices, the runs coincide. Since every applicant is matched when DA on the retained market terminates, no additional choice is needed in the run on $H'$. Therefore the two applicant-optimal outcomes coincide, proving that each applicant receives one of its first $k\le(1+\epsilon)r_n$ choices.

\paragraph{Lower bound in Case $1$.}
Set $r_n^\circ=r_n\log M/\log n_\Short$ and retain the notation $d_0$
above. Since $\delta\to1$, both logarithmic factors diverge and
$r_n^\circ=(1+o(1))d_0$. Fix
$k=\lfloor(1-\epsilon)r_n^\circ\rfloor$ and
$D=(1-\epsilon/2)d_0$. Eventually
$k<(1-3\epsilon/4)d_0<D$.
Consider an independent bipartite Erd\H{o}s--R\'enyi market on these $M+N$
agents with edge probability $D/M$. 
{
Since $d_0\sim\log M\log(N/(N-M+1))$, the below-threshold direction of
\cite[Theorem 6]{potukuchi2024unbalanced} says that it has no stable
matching covering all $M$ applicants with probability $1-o(1)$.
}%

Generate this market by giving each applicant an independent full uniform
order of the $N$ firms and retaining its first
$K_a\sim\operatorname{Binomial}(N,D/M)$ choices, independently across
applicants. Give each firm an independent full uniform preference order.
Since $\mathbb E K_a=ND/M\ge D$ and $d_0\gg\log M$, Chernoff's inequality
gives
\[
 \Prob\{\min_aK_a<k\}\le M\exp[-\Omega(\epsilon^2d_0)]=o(1).
\]
On the complementary event the market of exactly the first $k$ choices of
every applicant is a prefix truncation of this Erd\H{o}s--R\'enyi market.
If applicant-proposing DA matched all applicants in that truncation, the
same proposal sequence would terminate with all applicants matched in
the larger market. Thus the random $k$-choice market has no stable matching
covering all applicants with probability $1-o(1)$.

By \prettyref{eq:prob_calN_a}, every applicant on $H'$ has at least $k$
neighbors with high probability. Conditional on the degree vector, their
first $k$ neighbors have exactly the preceding independent uniform-list
law. If a stable matching on $H'$ gave every applicant a partner in these
prefixes, its restriction would be a stable matching covering every
applicant in the $k$-choice market: all matched edges are retained and
deleting other edges creates no blocking pair. This contradiction proves
the asserted lower bound, in particular for the applicant-optimal
matching.

\paragraph{Case $2$: the bound for every stable matching.}
Write $M=|\Short'|$, $N=|\Long'|$, $b=N-M\ge1$, and
$L=\log(N/b)\to\infty$. Complete the lists on $H'$ as follows. Conditional on its applicant degree vector, complete each
applicant's ordered neighbor list by a uniform ordering of the remaining
firms below all its neighbors. Complete each firm's ordering to a full
independent uniform ordering of the applicants, preserving its ordering
of neighbors. This generates exactly a complete market with independent
uniform preferences. Case $1$ and the rural hospitals theorem imply that
every applicant is matched in every stable matching on $H'$ with high
probability. Every such matching is stable in the completed market:
an added edge is below the applicant's present partner and cannot block.
Its applicant ranks are unchanged by completion. Since
\[
 \frac1\delta\log\!\left(\frac1{1-\delta+\delta^2/N}\right)\log M
   =(1+o(1))\frac NM L\log M\le(1+o(1))r_n,
\]
the bound from \prettyref{lmm:part_II_total}, applied with $\epsilon/2$, is at most $(2+\epsilon)r_n$ for all
sufficiently large markets. This proves Case $2$ without a restriction
on $d$ beyond the stated budget.

\begin{lemma}\label{lmm:part_II_total}
Consider a complete random matching market with $M<N$ applicants and firms,
independent uniform preferences, and $M/N\to1$. Write $L=\log(N/(N-M))$.
For every fixed $\epsilon>0$, with high probability, every applicant's rank
in every stable matching is at most $(2+\epsilon)(N/M)L\log M$.
\end{lemma}
}%

{
\begin{proof}
Let $b=N-M$ and fix $0<\epsilon<1$. For each applicant $i$
and firm $j$, generate independent uniform $[0,1]$ scores $X_{ij}$ and
$Y_{ij}$, with smaller scores preferred. If $\mu$ is a matching covering
all applicants, put
\[
 Q(\mu)=\sum_i\operatorname{rank}_i(\mu(i)),\qquad
 x_i=X_{i,\mu(i)},\qquad s=\sum_i x_i.
\]
The total-rank bound of \cite[Theorem 2]{ashlagi2017unbalanced} gives,
for each fixed $\eta>0$, with high probability,
\[
 Q(\mu)\le(1+\eta)NL
 \quad\hbox{simultaneously for every stable }\mu.
\]
To convert this into a bound on $s$ without conditioning individual
scores on the matching, let $X_{i,(r)}$ be the $r$th order statistic in
row $i$ and define
\[
 Z_i=\max_{1\le r\le N}[NX_{i,(r)}-(1+\eta)r]_+.
\]
Binomial Chernoff bounds, summed over $r$, give constants
$c_\eta,C_\eta>0$, independent of $N$, such that
$\Prob\{Z_i>z\}\le C_\eta e^{-c_\eta z}$ for $z\ge0$.
Indeed $NX_{i,(r)}>(1+\eta)r+z$ forces a binomial random variable of
mean $(1+\eta)r+z$ to be less than $r$, whose probability is at most
$\exp[-c_\eta(r+z)]$; thresholds above $N$ have probability zero.
The $Z_i$ are independent, so $N^{-1}\sum_iZ_i=O_\eta(1)$ with high
probability. Consequently, on one event that works for every rank vector,
\[
 s\le\frac{1+\eta}{N}Q(\mu)+\frac1N\sum_iZ_i
   \le(1+\eta)^2L+O_\eta(1)\le(1+3\eta)L=:S_0
\]
for every stable $\mu$, for sufficiently large $N$ and $0<\eta<1$.

{
For a fixed injection, let $x_i$ and $y_j$ be the matched scores,
$s=\sum_i x_i$, and $h(s)=(1-e^{-s})/s$, with $h(0)=1$.
The fixed-injection formula and integral estimate of
\citet[Lemma~1.2 and Section~3.1]{pittel2019likely} give
\begin{equation}
 \int_{[0,1]^M}\prod_i(1-x_i)^b
       \prod_{i\ne j}(1-x_iy_j)\,d\boldsymbol y
 \le e^2e^{-bs}h(s)^M.
 \label{eq:rank_integral_envelope}
\end{equation}

}%
For fixed $s$, enlarge the $\boldsymbol x$ domain to the positive simplex
$\sum_i x_i=s$. The simplex change of variables contributes the volume factor
$s^{M-1}/(M-1)!$, and the fraction
with $\max_i x_i>t$ is at most
$M(1-t/s)_+^{M-1}\le M e^{-(M-1)t/s}$.

{
There are $(N)_M=N!/b!$ injections. Equation~(3.6) of
\citet{pittel2019likely} gives
\[
 \frac{(N)_M}{(M-1)!}
 \int_0^\infty e^{-bs}\frac{(1-e^{-s})^M}{s}\,ds\le\frac Nb.
\]
Hence

}%
\begin{equation}
 \Prob\{\exists\hbox{ stable }\mu:\ s\le S_0, \max_i x_i>t\}
 \le e^2\frac{MN}{b}\exp\!\left(-\frac{(M-1)t}{S_0}\right).
 \label{eq:rank_simplex_tail}
\end{equation}
Choose $\eta>0$ sufficiently small depending on $\epsilon$ and set
$t=(2+\epsilon/3)L\log M/M$. Since $N/M\to1$ and $b\ge1$,
the right-hand side tends to zero. Also, a Chernoff bound and a union
bound over rows give
\[
 \max_i|\{j:X_{ij}\le t\}|\le(1+\eta)Nt
 \quad\hbox{with high probability},
\]
because $Nt\gg\log M$. Combining these events proves that every
applicant's rank in every complete-market stable matching is at most
$(2+\epsilon/2)(N/M)L\log M$, after decreasing $\eta$ if necessary.

\end{proof}
}%

\section{Analysis for single-tiered market}\label{sec:single_proof}

Before proving our main results, we first present several key lemmas. Consider a single-tiered two-sided market with applicants $\Short$ and firms $\Long$. Let $H$ denote an interview graph constructed based on the applicant-signaling mechanism, 
{
where each applicant signals its top $d$ firms based on its pre-interview utilities with respect to all firms in $\Long$.
}%

\begin{lemma}\label{lmm:interview_d_regular}
$H$ can be considered a randomly generated one-sided $d$-regular graph, where each applicant in $\Short$ is connected to $d$ randomly chosen firms in $\Long$, with every agent exhibiting uniformly generated strict preferences over their neighbors. 
\end{lemma}

\begin{proof}

{
    For each applicant, continuity and exchangeability of its pre-interview scores imply that the unordered set of its top $d$ firms is uniform among the $d$-element subsets of $\Long$. These sets are independent across applicants; the $d$ choices within a set are sampled without replacement.
    Conditional on $H$, the selected pre-interview scores for one applicant are exchangeable across the labels in its neighbor set. Adding independent, identically distributed post-interview scores preserves this exchangeability. Since the resulting utilities have no ties almost surely, every ordering of that applicant's neighbors is equally likely. The different applicants' score arrays are independent. The firm-side score arrays are independent of the applicant-side arrays that generated $H$; conditional on $H$, each firm's neighbor order is therefore uniform, and these orders are independent across firms and from the applicant orders. This proves the asserted conditional random-preference representation.

}%
\end{proof}
Let $\calN(a)$ denote the set of neighbors of $a$ on $H$, and define
\begin{align}
    \calN_+(a)\triangleq \{j \in \calN(a): A_{a,j} \ge 0\}  \, .  \label{eq:calC_+} 
\end{align}
{
\begin{lemma}\label{lmm:positive_available}\label{lmm:positive_available_addition}
Fix an applicant $a\in\Short$.
\begin{enumerate}
\item If some $j\in\calN_+(a)$ is available to $a$ on $H$, then $a$ is interim stable in every stable matching on $H$.
\item In any stable matching in which $a$ is matched to one of its top $|\calN_+(a)|$ post-interview neighbors, $a$ is interim stable.
\end{enumerate}
\end{lemma}
\begin{proof}
If $\calN(a)=\Long$, stability already excludes every blocking pair involving $a$. Otherwise, each $j\in\calN_+(a)$ satisfies
\[
 U_{a,j}^A\ge U_{a,j}^B>
 \max_{j'\notin\calN(a)}U_{a,j'}^B,
\]
because it is an original top-$d$ neighbor with a nonnegative interview shock.
In the first case, availability and stability imply that $a$ weakly prefers its partner to $j$ in every stable matching. In the second case, $m=|\calN_+(a)|>0$ and all $m$ such neighbors exceed every noninterviewed utility, so the $m$th best interviewed neighbor and the assigned partner do as well. Thus $a$ rejects all noninterviewed firms in either case, while interviewed pairs cannot block a stable matching.
\end{proof}
}%

\subsection{Proofs for Single-tiered market with sparse signals}\label{sec:single_sparse_proof}

\subsubsection{Proof of \prettyref{thm:single_signal_sparse}} \label{sec:single_signal_sparse}

{
If $n_\Short\le(1+o(1))n_\Long$ and $d\to\infty$, there is a sequence $\lambda>0$ with $\lambda\log d\to\infty$ such that $n_\Short\le(1+d^{-\lambda})n_\Long$. Indeed, when $\varepsilon_n=(n_\Short/n_\Long-1)_+>0$, choose $\lambda=\min\{1,-\log\varepsilon_n/\log d\}$, and choose $\lambda=1$ when $\varepsilon_n=0$. Thus it suffices to prove the result under $n_\Short\le(1+d^{-\lambda})n_\Long$ and $p(1\wedge\lambda)\log d\to\infty$. The latter condition holds in particular when $p\ge\Omega(1)$.

}%

{If $n_\Short<n_\Long$, first add independent applicants until the populations are balanced. Each original applicant keeps its neighbors and scores. By \prettyref{lmm:truncation}, removing the added applicants weakly worsens each firm's best stable partner, so any firm available to an original applicant in the augmented market remains available in the original market. We may therefore apply the following bound in the augmented market when necessary.}

{
Let $K_a=|\calN_+(a)|\sim\Binom(d,p)$ and let $\nu$ be defined in
\eqref{eq:nu}. The selected set obeys the root-data convention.
By \prettyref{prop:remove_gamma_d_omega_1}, with the preceding
augmentation when necessary,
\begin{align}
 \Pr\{\text{every }j\in\calN_+(a)\text{ is unavailable to }a\}
 &\le (1-\underline C\nu)^{-2}(1-p\underline C\nu)^d+o(1/n)
 =o(1),\label{eq:interim_H}
\end{align}
where the binomial generating function gives the first expression and
$pd\nu\to\infty$ gives the limit. By \prettyref{lmm:positive_available},
every applicant outside this event is interim stable in every stable matching.

}%
{
    Let $\Short_{\mathrm{bad}}$ be the set of applicants with no available neighbor in $\calN_+(a)$. The preceding uniform estimate gives $\mathbb E|\Short_{\mathrm{bad}}|=o(n_\Short)$, so Markov's inequality yields $|\Short_{\mathrm{bad}}|=o(n_\Short)$ with high probability. On this event, fix any stable matching $\Phi$ and delete $\Short_{\mathrm{bad}}$ together with the firms matched to these applicants in $\Phi$. At most $2|\Short_{\mathrm{bad}}|$ agents are deleted, and every retained agent keeps its original partner. No interviewed pair becomes blocking, and every retained applicant rejects all non-interviewed firms by \prettyref{lmm:positive_available}. The restriction of $\Phi$ is therefore perfect interim stable. The same good event works for every stable matching, although the deleted partners may depend on $\Phi$.

}%
{

The stronger common-market conclusion in
\prettyref{rmk:single_signal_sparse_identify} uses a different reduction:
delete only $\Short_{\mathrm{bad}}$ and recompute a stable matching.
As proved in \prettyref{sec:single_signal_sparse_identify}, availability
is preserved when applicants alone are removed, so every stable matching
on that same induced graph is perfect interim stable.

}%

\subsubsection{Proof of \prettyref{rmk:single_signal_sparse}}
\begin{itemize}
    \item $n_{\Short} \le \left(1+ d^{-\lambda}\right)n_{\Long}$ and $p = \omega\left(1/ \left(\left(\lambda\wedge 1\right)\log d\right)\right)$ for any $\lambda \ge \omega\left(1/\log d\right)$: it is proved in \prettyref{sec:single_signal_sparse}. 
    \item $\omega(1)\le d\le o\left(\log^2 n\right)$, $n_{\Short} \le \left(1+ n^{-\lambda}\right)n_{\Long}$ and $p = \omega(1/\sqrt{d})$ for any $\lambda \ge \Omega(1)$. By applying Chernoff bound \prettyref{eq:chernoff_binom_left} in \prettyref{lmm:chernoff}, for every $a\in\Short$, 
    
{
\[
     \prob{|\calN_+(a)| \ge  \frac{1}{2} p d } \ge 1-o(1)\,. 
    \]
}%

{
To include the entire stated range $n_\Short\le(1+n^{-\lambda})n_\Long$,
put $N=n_\Long$ and add independent applicants, when necessary, until
there are $M_+=\max\{n_\Short,N\}$ applicants. The resulting market
satisfies \prettyref{prop:unmatched}; its proof gives expected
job-optimal average rank $O(\sqrt d)$, with unmatched rank $d+1$.
Applicant-entry monotonicity for the job-optimal matching implies that
each original applicant's rank can only worsen upon these additions.
By exchangeability in the enlarged market,
\[
 \mathbb E\bigl|\{a\in\Short:
     \operatorname{rank}_{\Phi^{\Long}}(a)>pd/2\}\bigr|
       \le O\!\left(\frac{n_\Short}{p\sqrt d}\right)=o(n_\Short).
\]
This also bounds the same count in every stable matching, because the
job-optimal matching gives each applicant its worst stable partner.
The preceding Chernoff bound likewise implies that the expected number
of applicants with $|\calN_+(a)|<pd/2$ is $o(n_\Short)$.
Markov's inequality shows that the union of these two exceptional sets
has size $o(n_\Short)$ with high probability. Every other applicant is
interim stable in every stable matching by
\prettyref{lmm:positive_available_addition}. For each stable matching,
delete the exceptional applicants together with their matched firms,
as in the paired-deletion argument above. This proves the claimed
almost interim stability throughout the stated population range.
}%

{
    \item $\omega(1) \le d \le O\left(\polylog n\right)$, $n_{\Short}\le(1-\Omega(1))n_\Long$, and $p=\omega(1/d)$. Put $\delta=n_\Short/n_\Long$. The unmatched-firm bound \eqref{eq:degree_free_available_upper}, averaged over $|\calN_+(a)|\sim\Binom(d,p)$ as derived below, gives
    \[
    \Pr\{\text{no available firm in }\calN_+(a)\}
    \le(1-p)^d+dp[1-p(1-\delta)]^{d-1}=o(1),
    \]
    since $dp\to\infty$ and $1-\delta=\Omega(1)$. By \prettyref{lmm:positive_available}, every applicant outside this exceptional set rejects all non-interviewed firms in every stable matching. The expectation, Markov and paired-deletion argument above proves almost interim stability.

}%
\end{itemize}

\subsubsection{Proof of \prettyref{rmk:single_signal_sparse_identify}}\label{sec:single_signal_sparse_identify}

{
Let $\Short_{\rm bad}$ be the exceptional set in
\prettyref{sec:single_signal_sparse}; it has size $o(n_\Short)$
with high probability. Delete these applicants and retain all firms.
Successive applications of \prettyref{lmm:truncation} weakly worsen
each firm's best stable partner. Thus every nonnegative-shock neighbor
that was available to a retained applicant remains available in this
single reduced market. By \prettyref{lmm:positive_available}, every
stable matching of that market is perfect interim stable.

}%
\subsubsection{Proof of \prettyref{thm:single_firm_signal_sparse}}\label{sec:single_firm_signal_sparse}

{
\begin{proof}[Proof of \prettyref{thm:single_firm_signal_sparse}]
Write $N=n_\Long$ and $\delta=n_\Short/N$. By the population assumptions, there are fixed constants $c>0$ and $C<\infty$ such that $1+c\le\delta\le C$ for all sufficiently large markets. Thus $N=\Theta(n)$, $d=o(N)$, and
\[
 \kappa_1=\lceil2\delta d\rceil=o(N),
 \qquad \kappa_2=n_\Short-N-\kappa_1=\Theta(N)>0.
\]
Set
\[
 t=F_{\DistA*\DistB}^{-1}\!\left(\frac{\kappa_1}{\kappa_1+1}\right)
 <F_{\DistB}^{-1}\!\left(\frac{\kappa_2}{\kappa_2+1}\right).
\]
The continuity of $\DistB$ and independence of the scores imply that $\DistA*\DistB$ is continuous. Therefore
\begin{align*}
 F_{\DistA*\DistB}(t)&=\frac{\kappa_1}{\kappa_1+1},
 &\prob{B>t}&\ge\frac1{\kappa_2+1},\qquad B\sim\DistB.
\end{align*}
Expose the applicant-side scores and hence the interview graph $H$, but leave the firm-side scores unrevealed. The graph has exactly $n_\Short d=\delta Nd$ edges, so at least $N/2$ firms have degree at most $\kappa_1$. Denote this set by $J_0$. For $j\in J_0$, let $E_j$ be the event that every interviewed applicant gives $j$ an idiosyncratic post-interview utility at most $t$:
\[
 E_j=\{A_{j,a}+B_{j,a}\le t\text{ for every }a\in\calN(j)\}.
\]
Conditional on the information exposed so far, the events $E_j$ are independent across firms and
\[
 \prob{E_j\mid H,\text{ applicant-side scores}}
 =\left(\frac{\kappa_1}{\kappa_1+1}\right)^{|\calN(j)|}
 \ge\left(\frac{\kappa_1}{\kappa_1+1}\right)^{\kappa_1}>e^{-1}.
\]
A Chernoff bound gives, with probability $1-e^{-\Omega(N)}$, a set
\[
 G=\{j\in J_0:E_j\}\quad\text{with}\quad |G|\ge N/(4e).
\]
Expose all firm-side scores on interviewed pairs and run applicant-proposing deferred acceptance on $H$. Let $U$ be its unmatched applicants. By the Rural Hospital Theorem~\citep{mcvitie1970stable}, $U$ is the unmatched-applicant set in every stable matching on $H$, and cardinality gives
\[
 |U|\ge n_\Short-N\ge cN.
\]
Conditional on all this information, the unexposed firm-side scores
\[
 \{B_{j,a}:j\in G,\ a\in U,\ (a,j)\notin\calE(H)\}
\]
remain independent with distribution $\DistB$. Whenever such a score exceeds $t$, $(a,j)$ blocks every stable matching on $H$: applicant $a$ is unmatched and regards every firm as acceptable, while $j$ prefers $a$ to all its interviewed partners (and also to remaining unmatched). The common intrinsic value cancels in these comparisons.

We next find linearly many such pairs with no shared endpoints. Conditional on the exposed information, process any
\[
 m=\min\{|G|,\lfloor |U|/2\rfloor\}=\Omega(N)
\]
firms of $G$ in a fixed order. At each step, match the current firm to an unused applicant of $U$ with $B_{j,a}>t$ and $(a,j)\notin\calE(H)$, if one exists. At least $|U|/2$ applicants remain unused, and the current firm has at most $\kappa_1$ interviewed neighbors. For sufficiently large $n$, at least $cN/4$ of these applicants are nonneighbors. Moreover, $\kappa_2+1\le CN$. Since each step exposes a fresh firm's row of off-graph scores, its conditional probability of success is at least
\[
 1-\left(1-\frac1{\kappa_2+1}\right)^{|U|/2-\kappa_1}
 \ge 1-\exp\left(-\frac{c}{4C}\right)
 \triangleq\rho>0.
\]
Coupling each step with an independent uniform random variable shows that the number of successes stochastically dominates $\Binom(m,\rho)$. Another Chernoff bound yields $\Omega(N)$ vertex-disjoint pairs with probability $1-e^{-\Omega(N)}$. Every one of these pairs blocks every stable matching on $H$.

Fix any such stable matching and remove agents by restricting that matching to the retained agents. If both endpoints of one of the constructed pairs remain, they still block: the applicant remains unmatched, while the firm's assigned partner either stays the same or is removed. At least one endpoint of every constructed pair must therefore be deleted. The pairs have no common endpoints, so this requires $\Omega(N)=\Omega(n)$ deletions. Consequently no stable matching on $H$ is almost interim stable, with high probability.

Finally, the two examples satisfy the stated quantile condition. If $\DistA=\boldsymbol\delta_0$, then $\DistA*\DistB=\DistB$ and $\kappa_1<\kappa_2$ for sufficiently large $n$; continuity of $\DistB$ implies the required strict quantile ordering. For the second example, write $B\sim\mathbb N(\mu,\sigma^2)$ with fixed $\sigma>0$, and take a fixed finite $M$ such that $A\le M$ almost surely. With $u_i=\kappa_i/(\kappa_i+1)$,
\[
 F_{\DistA*\DistB}^{-1}(u_1)
 \le M+F_{\DistB}^{-1}(u_1)
 =\mu+M+O\!\left(\sqrt{\log\log n}\right),
\]
whereas
\[
 F_{\DistB}^{-1}(u_2)=\mu+\Theta\!\left(\sqrt{\log n}\right).
\]
Here $\kappa_1=O(\polylog n)$ and $\kappa_2=\Theta(n)$. These bounds give the strict inequality for sufficiently large $n$.
\end{proof}

}%

\subsubsection{Proof of \prettyref{thm:single_both_signal_sparse}} \label{sec:single_both_signal_sparse}

For every pair of applicant and firm, their pre-interview utility equals their post-interview utility. Let $H_1$ denote the interview graph constructed by the applicant-signaling mechanism, and let $H_2$ denote the interview graph constructed by the firm-signaling mechanism. Then, $H$ is the union graph of $H_1$ and $H_2$. For any $i\in\Short\cup\Long$, let $\calN_1(i)$ and $\calN_2(i)$ denote the set of neighbors of $i \in \Short\cup \Long$ on $H_1$ and $H_2$, respectively. Then, every applicant $a\in\Short$ strictly prefers $\calN_1(a)$ to $\calN_2(a)\backslash \calN_1(a)$, and every firm $j\in \Long$ strictly prefers {$\calN_2(j)$ to $\calN_1(j)\backslash \calN_2(j)$}.

Fix a stable matching $\Phi$ on $H$. 
Let $\Short_1 \subset \Short$ denote the set of applicants $a$ with $\phi(a) \in \calN_1(a)$, let $\Short_2 \subset \Short$ denote the set of applicants $a$ with $\phi(a) \in \calN_2(a)$, and let $\Short_U \subset \Short$ denote the set of applicants that are unmatched on $\Phi$. Then, we get $\Short_1\cup\Short_2 \cup \Short_U =\Short$. 
Similarly, let $\Long_1 \subset \Long$ denote the set of firms $j$ with $\phi(j) \in \calN_1(j)$, let $\Long_2 \subset \Long$ denote the set of firms $j$ with $\phi(j) \in \calN_2(j)$, and let $\Long_U \subset \Long$ denote the set of firms that are unmatched on $\Phi$. Then, we get $\Long_1\cup\Long_2 \cup \Long_U =\Long$. 

\begin{lemma}\label{lmm:both_applicant}
    {Suppose $\Phi$ denotes the applicant-optimal stable matching} on $H$. Then, we have
    \begin{align*}
        \prob{|\Short_2\cup\Short_U|,|\Long_1\cup\Long_U| \ge \Omega(n) } \ge 1- o(1) \,. 
    \end{align*}
\end{lemma}

{
The following turns the population estimate in \prettyref{lmm:both_applicant}, proved below, into a robust failure statement. Put $N=n_\Short=n_\Long$, and let $\Phi^\Short$ and $\Phi^\Long$ be the two optimal stable matchings on $H$. Define
\[
 D_A=\{a:\phi^\Short(a)\notin\calN_1(a)\},\qquad
 D_J=\{j:\phi^\Long(j)\notin\calN_2(j)\},
\]
including unmatched agents in these sets. By the lemma and its symmetric version, there is a fixed $c>0$ such that both sets have size at least $cN$ with probability $1-o(1)$. The overlap $\calE(H_1)\cap\calE(H_2)$ cannot be counted as incoming-only matches: its expected size is $d^2=o(N)$, so it affects these fractions only by $o(1)$ with high probability.

We work with the uniform quantiles $Q_{i,k}=F_{\DistB}(B_{i,k})$, which preserve the strict preferences almost surely. Set $\eta=c/8$ and $t=1-\eta/d$. Conditional on all firm signaling lists, the $Nd$ applicant-side scores on firm-signaled pairs are independent uniform variables. Consequently, with probability $1-e^{-\Omega(N)}$, at most $2\eta N$ applicants have any firm-signaled score above $t$. The symmetric assertion holds for firms. Remove those rows from $D_A$ and $D_J$, obtaining sets $G_A,G_J$ of size at least $cN/2$. Every $a\in G_A$ values its applicant-optimal partner at most $t$, and every $j\in G_J$ values its firm-optimal partner at most $t$; an unmatched partner is interpreted as having value $-\infty$. Optimality implies the same upper bounds for their partners in every stable matching on $H$.

Expose both signaling lists, their scores, and all remaining scores on $H$, but no off-graph scores. Write $T_i$ for agent $i$'s $d$th highest quantile. Conditional on this information, the off-graph entries in row $i$ remain independent uniform variables on $[0,T_i]$, independently across rows. Indeed, the row's top-$d$ labels and scores impose only the upper truncation $T_i$ on all other entries; revealing a subset of those entries on incoming edges does not alter the conditional law of the remaining entries. The sets $G_A,G_J$ and the two optimal matchings are measurable with respect to the exposed information. Moreover,
\[
 \Prob\!\left\{\min_iT_i<1-\frac{\eta}{2d}\right\}
 \le 2N\Prob\!\left\{\Binom\!\left(N,\frac{\eta}{2d}\right)<d\right\}=o(1),
\]
since $d=o(\log N)$. The union interview graph also has maximum degree $O(d+\log N)=o(N)$ with high probability, by a binomial upper-tail bound for its incoming degrees.

On these events, any unexposed nonedge $(a,j)\in G_A\times G_J$ has conditional probability at least
\[
 \left(\frac{\eta}{2d}\right)^2
\]
of satisfying $Q_{a,j}>t$ and $Q_{j,a}>t$. Such a pair blocks every stable matching on $H$. Process $\lfloor cN/4\rfloor$ distinct applicants of $G_A$ greedily, pairing each with an unused firm of $G_J$ satisfying both inequalities whenever possible. At every step at least $cN/4-o(N)$ eligible firms remain. The entries tested at that step are fresh independent truncated uniforms, even after earlier steps. Thus its conditional failure probability is at most
\[
 \exp\!\left[-\big(cN/4-o(N)\big)\frac{\eta^2}{4d^2}\right],
\]
and a union bound shows that every step succeeds with probability $1-o(1)$. This gives $\Omega(N)$ vertex-disjoint pairs that simultaneously block every stable matching on $H$. Restricting a matching after deleting agents cannot remove one of these blockers while retaining both its endpoints, because their assigned partners can only stay the same or be removed. Hence $\Omega(N)$ agents must be deleted. This proves \prettyref{thm:single_both_signal_sparse}. 
}%

{
\begingroup
\begin{proof}[Proof of \prettyref{lmm:both_applicant}]
Write $N=n_\Short=n_\Long$, and assume throughout $d\to\infty$ and
$d=o(\log N)$. Fix the polynomial failure exponent $M=20$ below. We prove a stronger conclusion: with high probability,
at least $N/3$ applicants and at least $N/3$ firms have no outgoing
partner in every stable matching. All preferences in this proof are the
pre-interview preferences, since $\DistA=\boldsymbol\delta_0$.

\paragraph{Unmatched agents and a supply of unexhausted lists.}
An applicant $a$ has the following sufficient certificate of being
unmatched in every stable matching: no firm signals to $a$, and for every
$j$ signaled by $a$ there is some $b\ne a$ whose first choice is $j$ and
whom $j$ prefers to $a$. Indeed, $b$ prefers $j$ to every other partner,
so $j$ cannot be matched to $a$ in a stable matching. Such an $a$ has no
other interviewed neighbors.

The first event has probability $(1-d/N)^N$. The numbers of other
applicants whose first choice is each of $a$'s $d$ signaled firms couple
to independent $\mathrm{Poisson}(1)$ variables with error $O(d^2/N)$.
Conditional on no firm signaling to $a$, its rank in each firm's row
exceeds $d$; this only increases the chance that a given nonempty set of
first-choice applicants contains someone preferred to $a$. For
$Z\sim\mathrm{Poisson}(1)$,
\[
 \mathbb E\frac{Z}{Z+1}=e^{-1}.
\]
The score rows of the $d$ firms are independent. Therefore the certificate
probability is at least $(1-o(1))e^{-2d}$.

We next bound the certificate count.
With probability at least $1-N^{-M}$, for any prescribed fixed $M$, every
incoming degree in either signaling graph is at most
$L=C_M(d+\log N)$. On this event, changing one applicant score row changes
only its own certificate and those using its former or new first-choice
firm, at most $1+2L$ certificates. Changing a firm score row changes at
most $2d$ zero-incoming conditions and at most $L$ certificates of
applicants signaling that firm. An $O(d+L)$ bound times the number of changed rows holds between any
two profiles satisfying the degree cap: splice rows of the two profiles,
so that every intermediate incoming degree is at most $2L$.
Extend the certificate count from this set as an $O(d+L)$-Lipschitz
function of the $2N$ independent score rows, and clip it to $[0,N]$.
The bounded-differences inequality gives, for
\[
 k=\lfloor N e^{-3d}\rfloor,
\]
that the number of certified unmatched applicants is at least $k$ except
with probability $N^{-M}+\exp[-N^{1-o(1)}]$. The two sides have equal
numbers of unmatched agents, by cardinality. Thus the unmatched set on
either side in every stable matching has size at least $k=N^{1-o(1)}$.

Let $s=\lfloor N e^{-2\sqrt d}\rfloor$. The one-sided unmatched estimate
in \prettyref{prop:unmatched}, applied to the applicant-signaling graph
$H_1$, gives at least $s$ initially unmatched applicants with high
probability. This estimate can be strengthened to failure probability
$\exp[-N^{1-o(1)}]$: its expectation exceeds $2s$ for large $N$, and the
number unmatched changes by at most one when a single preference row is
changed. To see the latter assertion, remove that agent in the two
instances and use their common reduced market; reinserting one agent
changes the stable matching cardinality by at most one. Bounded
differences applies to the independent rows.

Run applicant-proposing DA on $H_1$ and then continue on $H$. All the
initially unmatched applicants exhaust their outgoing lists. If $X$ is
the final number of applicants without an outgoing partner, $U$ the
number unmatched, $m$ the number of firms matched to an outgoing
neighbor, and $R$ the number of matched reciprocal-signal pairs, then
\[
 X\ge s,\qquad X=U+m-R.
\]
For every unmatched set $V$ of applicants, its $H_2$-neighborhood consists
only of firms with outgoing partners: otherwise its incident pair would
block. Uniformly over all $V$ of size at least $s/2$, this neighborhood
has size greater than $s/4$ with probability
$1-\exp[-N^{1-o(1)}]$. Indeed, for fixed sets of sizes
$r=\lceil s/2\rceil$ and $t=\lceil s/4\rceil$,
\[
 \Pr\{\Gamma_{H_2}(V)\subset T\}
 \le\left(1-\frac rN\right)^{d(N-t)},
\]
whereas the logarithm of the number of such pairs $(V,T)$ is
$O(s\sqrt d)=o(ds)$. If $m<s/4$, then $U\ge3s/4$, contradicting this
uniform cut bound. Thus $m\ge s/4$. The symmetric firm-proposing argument
gives at least $s/4$ applicants with outgoing partners in the
firm-optimal matching. Optimality extends both lower bounds to every
stable matching. In particular, in a DA run of either orientation at
least $s/4$ proposers never exhaust their outgoing lists.

These bounds also hold, with $k-1$ and $s/8$ in place of $k$ and $s/4$,
after removing any one agent. Stable matching cardinality and its
one-sided analogue change by at most one; the cut bounds tolerate the
loss of one vertex. A surviving outgoing row then has $d$ or $d-1$
active entries. The following estimates allow either budget.

\paragraph{A deferred exposure for incoming proposals.}
We describe applicant-proposing DA; reverse the sides for the other
orientation. Reveal an applicant's outgoing list only as it proposes.
Upon every arriving proposal, reveal whether the proposer belongs to
the receiver's outgoing group; record this membership as a hard positive
or negative entry in that receiver's row. After this list is exhausted, assign independent uniform priorities to
its potential incoming-only neighbors. Reveal only its next actual
incoming proposal and the priority of that proposal. If its revealed
priority prefix ends at $\tau_a$, every firm not yet proposed to has
been observed to have no positive incoming edge with priority at most
$\tau_a$. Conditional on its indicator $I_{j,a}$ of having signaled $a$,
this observation has likelihood $1-\tau_a I_{j,a}$.

Consequently, conditional on the entire adaptive transcript, different
firm rows remain independent. Row $j$ is a weighted fixed-size subset:
its known positive entries $P_j$ are included, its known negative entries
$E_j$ are excluded, and its remaining subset $S$ has probability
proportional to
\[
 \prod_{a\in S}w_a,\qquad w_a=1-\tau_a.
\]
The weights are common across all firms. The hard exclusions $E_j$ come
only from earlier outgoing proposals. An applicant with $\tau_a>0$ has
already proposed to all its outgoing neighbors, so no unfixed outgoing
entry is incorrectly assigned an incoming priority weight. Comparisons
of the ranks of known proposers reveal only their relative order within
the outgoing or incoming group, independently of the identities of the
unrevealed row entries. Conditioning on each successive revealed minimum
priority multiplies the row likelihoods separately, which verifies the
claimed factorization under adaptive DA and under its continued rejection
chains.

Suppose at least $W$ columns have weight one, and $|E_j|\le L$. In the
base weighted $d$-subset law with no forced entries, every coordinate has
marginal at most $d/(W-d)$. For example, swap a selected coordinate with
one of the at least $W-d$ unused unit-weight coordinates; the resulting
weight is no smaller. Conditioning additional coordinates to be included
decreases the marginal of a different coordinate. This follows from
pairwise negative correlation for weighted fixed-size subsets, obtained
from Newton's inequality for elementary symmetric polynomials, and can
be applied successively since the conditioned law has the same form.
Conditioning the coordinates of $E_j$ to be absent increases a marginal
by a factor at most
\[
 \left(1-\frac{Ld}{W-d}\right)^{-1}.
\]
Forced positive entries whose current weights are zero are handled by
positive limiting weights; their weights cancel from their conditional
row law.

An unmatched firm has received no proposals, and hence has no forced
positive entries and no hard exclusions. Therefore, as long as
$W\ge s/16$, the conditional incoming-proposal intensity to any specified
matched firm is at most $1+o(1)$ times the intensity to each unmatched
firm. The comparison also permits budgets $d$ and $d-1$.
If $e_r$ is the elementary symmetric polynomial of the weights other than
the coordinate in question and $Z$ their sum, then
\[
 Z-r+1\le\frac{r e_r}{e_{r-1}}\le Z.
\]
Together with the marginal-odds formula $p_r/(1-p_r)=w e_{r-1}/e_r$,
this gives
\[
 \frac{p_d}{p_{d-1}}
 \le\frac d{d-1}\bigl(1+O(d/W)\bigr)=1+o(1),
\]
uniformly when $W\ge s/16$. The priority denominator $1-\tau_a$ is the same for all
firms. Incoming proposals therefore hit a set of $u$ currently unmatched
firms with conditional probability at least $(1-o(1))u/N$ per proposal,
and their combined hazard is at least $(1-o(1))u$ times the hazard of a
proposal to any specified target firm. The same inequalities hold for
outgoing proposals, since their next labels are uniform among unexposed
positions and none of those unmatched firms has been proposed to.

All estimates are used with stopping rules: stop if a hard-exclusion
count exceeds $L$, or fewer than $s/16$ columns of weight one remain.
The high-probability structural events above ensure the first stop does
not occur. A continued chain has to make at least $s/16$ proposals before
the second stop can occur, apart from a harmless change of constants.
With at least $k-1$ initial unmatched sinks its length is dominated, up to
these stops, by a geometric variable of mean $O(N/k)$. In particular, for
\[
 T=C\frac Nk\log N=N^{o(1)},
\]
the probability that a test lasts more than $T$ can be made $N^{-M}$ by
increasing $C$. Since $T=o(s)$, one or two such tests preserve the weight
floor. We never condition the deferred row law on a future good event;
we stop the exploration and add its failure probability afterwards.

\paragraph{The core and its second moment.}
Let $C_H$ be the number of firms with more than one stable partner. Start
from the applicant-optimal matching and test a firm $j$ by divorcing its
partner, allowing the freed applicant and later displaced applicants to
continue down their preference lists. Firm $j$ rejects every applicant
it does not prefer to its original partner. The usual rejection-chain
test~\cite[Algorithm 4.1 and Section 4.2]{im2005} succeeds exactly when
$j$ has a better stable partner. Success
requires an improving return to $j$ before reaching an original unmatched
firm; the competing-hazard comparison gives probability $O(1/k)$.
A successful test produces a stable matching in which firms weakly
improve and applicants weakly worsen. Only the agents on its chain
change partners. Thus
\[
 \mathbb E C_H=O(N/k)+o(1)=N^{o(1)}.
\]
For the required second moment, retain the revealed transcript after a
successful test of $j$, rather than undoing it. If a second core firm
$h$ lies outside this chain, it still has its applicant-optimal partner
and hence still has a strictly better stable partner. Its own test must
therefore succeed. The first successful history has no proposals to the
unmatched sinks; it is a monotone continuation of the same DA exposure,
so the same conditional hazard bound applies to the second test. Truncate
both tests at $T$. Firms on the first chain contribute at most $T$ choices,
and firms outside it have joint success probability $O(k^{-2})$. Hence
\[
 \mathbb E C_H^2
 \le O(TN/k)+O(N^2/k^2)+o(1)=N^{o(1)}.
\]
Our choice $M=20$ makes the exceptional-event contributions negligible
after all sums over one or two target firms. The same bound holds in a market with one
agent removed, using the $d$ or $d-1$ row budgets already discussed.

\paragraph{Concentration and symmetry.}
Let $X_H$ denote the number of applicants without outgoing partners in
the applicant-optimal matching. We apply Efron--Stein to the $2N$
independent preference rows. If one applicant row is resampled, delete
that applicant first and start from the applicant-optimal matching of
the common market. Inserting the old or the new applicant by a single
proposing chain gives the corresponding applicant-optimal matching; it is
just a valid DA ordering that processes the other applicants first.
The two counts differ by at most the two chain lengths plus one.

For a resampled firm row, use the firm-optimal matching of the common
market and the two firm-proposing insertion chains. These give the two
firm-optimal outcomes. Switching each resulting market back to its
applicant-optimal outcome changes the count by at most its number of
core agents. A chain started by a newly inserted agent may have a predetermined
incoming list; exempt its own proposals from the hazard comparison.
Revealing this list hard-fixes at most the inserted coordinate in every
opposite-side row, leaving a budget of $d$ or $d-1$ and at most one extra
hard exclusion, so use the cap $L+1$. For old-agent hazards, remove the
newcomer coordinate from the row; the unmatched sink rows then have no
forced positive entries among the remaining old-agent coordinates. The
preceding marginal comparison covers this exposure.
There are at most $O(d+L)$ proposals by the inserted agent, whereas every
proposal by an old agent obeys the unmatched-sink hazard. The insertion
chain is therefore bounded by $O(d+L)$ geometric-length segments, with
geometric mean $O(N/k)$. Its second moment is $N^{o(1)}$. The same stopping
rules apply, since $N^{o(1)}=o(s)$. A rejection or insertion chain has at
most $N^2$ proposals, so its squared length is at most $N^4$; the choice
$M=20$ makes all exceptional-event penalties negligible, including the
subsequent sum over $2N$ resampled rows. Together with the
core second-moment estimate, this gives, for every resampled row $r$,
\[
 \mathbb E\big(X_H-X_H^{(r)}\big)^2{\le} N^{o(1)},\qquad
 \operatorname{Var}(X_H){\le} N^{1+o(1)}=o(N^2).
\]

Let $Y_H$ be the number of applicants without outgoing partners in the
firm-optimal matching, and let $Z_H$ be the number of firms without
outgoing partners in that matching.
Then $0\le Y_H-X_H\le C_H$. Side symmetry gives
$\mathbb E Z_H=\mathbb E X_H$, while exact edge counting gives
\[
 Y_H+Z_H=N+U-R_H,
\]
where $R_H$ is the number of matched reciprocal-signal pairs.
Since $\mathbb E R_H\le d^2=o(N)$ and $U\ge0$,
\[
 2\mathbb E X_H
 =N+\mathbb E U-\mathbb E R_H-\mathbb E(Y_H-X_H)
 \ge N-o(N).
\]
Concentration implies $X_H\ge N/3$ with high probability; by symmetry
$Z_H\ge N/3$ with high probability. Optimality then gives these lower
bounds in every stable matching. This proves the lemma for $d=o(\log N)$.
\end{proof}

\endgroup
}%

\subsubsection{Proof of \prettyref{rmk:single_both_signal_sparse_2}}
Analogous to \prettyref{lmm:interview_d_regular}, $H$ can be considered as a union graph of $H_1$ and $H_2$ with uniformly generated strict preferences, where $H_1$ is an interview graph in which each applicant randomly selects $d$ firms to signal, and $H_2$ is an interview graph in which each firm randomly selects $d$ applicants to signal.

{
When $n_\Short/n_\Long$ is bounded away from zero, use
\prettyref{prop:both_remove_gamma_d_omega_1} in place of
\prettyref{prop:remove_gamma_d_omega_1} in the proof of
\prettyref{thm:single_signal_sparse}, selecting the own-signal neighbors
with nonnegative shocks. Their count is $\Binom(d,p)$, so averaging as
in \eqref{eq:interim_H} gives an $o(1)$ failure probability per applicant.
Since pre-interview scores are absent, an available such neighbor gives
its applicant at least the utility of every noninterviewed firm.
The same Markov and paired-deletion argument proves the conclusion.
}%

{
For the remaining ratios, it suffices to consider $n_\Short/n_\Long\le1/2$. Remove an applicant $a$ and let $U^-$ be the unmatched firm set in the remaining market. Its size is at least $n_\Long-n_\Short+1$, and it is independent of $a$'s own outgoing signal list and post-interview scores. Reinsert $a$ last in applicant-proposing deferred acceptance. At most one member of $U^-$ becomes matched, and the Rural Hospital Theorem gives the same unmatched firm set in every stable matching. Thus two nonnegative-shock own-signal neighbors in $U^-$ guarantee an available such neighbor in every stable matching. The own-signal neighbors with nonnegative shocks have a $\operatorname{Bin}(d,p)$ count and uniform labels independent of $U^-$. The hypergeometric calculation in \eqref{eq:degree_free_available_upper}, averaged over this count, therefore bounds failure by
\[
 (1+dp)(1-p/2)^{d-1}=o(1).
\]
The same summation and deletion argument then proves the asserted almost interim stability.
}%

\subsection{Single-tiered market with dense signals}\label{sec:single_dense_proof}
\subsubsection{Proof of \prettyref{rmk:firm_optimal_signal}}

{
Put $N=n_\Short=n_\Long$ and take the fixed distributions
$\DistB=\mathbb U[0,b]$ and
$\Pr\{A=-b\}=\Pr\{A=0\}=1/2$, where $b>0$ is any fixed
constant. Thus $p=1/2$; choosing $b<1/4$ also satisfies the
bounded-support condition. Set $d=\lceil16\log^2N\rceil$.
All utilities below are measured relative to the common intrinsic value.
By \prettyref{thm:single_signal_dense}, the applicant-optimal
matching is perfect interim stable with probability $1-o(1)$.
In a balanced market an interim-stable matching cannot leave an
applicant and a firm unmatched. Hence this matching is perfect,
and the Rural Hospital Theorem implies that the firm-optimal matching
is perfect with probability $1-o(1)$ as well.

We first expose the interview graph and the firms' scores on its edges.
A Chernoff bound gives
$\max_j|\calN(j)|\le2d$ with probability $1-o(1)$.
Put $t=b/(4d)$ and
\[
 \calQ=\{j\in\Long:B_{j,a}\le b-t\text{ for all }a\in\calN(j)\}.
\]
Conditional on the graph, these events are independent across firms,
and, on the degree event,
\[
 \Pr\{j\in\calQ\mid H\}=(1-1/(4d))^{|\calN(j)|}
 \ge(1-1/(4d))^{2d}>1/2
\]
for all sufficiently large $N$. Consequently $|\calQ|\ge N/3$
with probability $1-o(1)$. Since every post-interview shock is
nonpositive, every firm in $\calQ$ has utility at most $b-t$
from every possible interviewed partner.

Run firm-proposing deferred acceptance, exposing an applicant's
post-interview shocks only when its proposals must be compared.
On the event that the algorithm produces a perfect matching, let
$a_*$ be the last applicant to receive its first proposal. Every other
applicant then holds a proposal, so acceptance of this first proposal
terminates the algorithm. In particular, $a_*$ receives exactly one
proposal. Its identity and the proposing firm are determined before
any of the shocks in its own row need to be inspected. Conditional
on this history, the shock on its matched edge equals $-b$ with
probability $1/2$. On that event its matched utility is negative,
whereas its utility from every non-interviewed firm is nonnegative.
The event $|\calQ|\ge N/3$ and the history up to this final proposal
also use no shocks in this row. Thus this unfavorable final shock
occurs jointly with a perfect matching and $|\calQ|\ge N/3$ with
probability at least $1/2-o(1)$.

There are at least $N/3-d$ firms in
$\calQ\setminus\calN(a_*)$. Their unused scores $B_{j,a_*}$ are
independent $\mathbb U[0,b]$ variables conditional on the graph,
all interviewed scores, and the entire deferred-acceptance history:
the graph is selected using applicants' scores, while these are firms'
non-interviewed scores. Therefore
\[
 \Pr\{B_{j,a_*}\le b-t\text{ for every }
          j\in\calQ\setminus\calN(a_*)\mid H,\text{interviewed data}\}
 \le (1-1/(4d))^{N/3-d}=o(1).
\]
With conditional probability $1-o(1)$, some such firm prefers $a_*$
to its current partner. When the final applicant's shock is $-b$,
the resulting non-interviewed pair blocks. The firm-optimal matching
therefore fails perfect interim stability with probability at least
$1/2-o(1)$, as claimed.

}%
\subsubsection{Proof of \prettyref{thm:single_signal_dense}}\label{sec:single_signal_dense}

{
Write
\[
L_\delta=\log\frac1{1-\delta},\qquad
r_n=\frac1\delta\log\frac1{1-\delta+\delta^2/n_\Long}\log n_\Short,
\qquad c=\frac{11}{40},
\]
where $L_\delta$ is used only when $\delta<1$. The assumed signal budget gives $dp\ge8r_n$. For each applicant, $|\calN_+(a)|\sim\Binom(d,p)$, and hence
\[
\prob{|\calN_+(a)|\le cdp}
\le\exp\left(-\frac{841}{3200}dp\right).
\]
Since $r_n\ge(1-o(1))\log n_\Short$, a union bound gives
\begin{align}
\prob{\exists a\in\Short:\ |\calN_+(a)|\le cdp}=o(1).
\label{eq:addition_union_calN_+}
\end{align}
In particular, on the complementary event every applicant has more than $cdp\ge2.2r_n$ non-negative-shock neighbors.

\paragraph{A complete-market rank bound away from balance.}
We first supplement the near-balanced bound in
\prettyref{prop:perfect_stable} by a bound for
$M/N=\delta\in[\delta_0,1-\delta_0]$, where $\delta_0>0$ is fixed.
Consider a complete market with $M$ applicants, $N$ firms, and independent
uniform preferences. Put $k=N-M$, $L=\log(N/k)$, and
$h(s)=(1-e^{-s})/s$. Use the independent uniform scores and the
fixed-injection integral from \eqref{eq:rank_integral_envelope}.
For a stable injection $\phi$, write $x_i$ for applicant $i$'s score of
its partner and $s=\sum_i x_i$. The integral, enlarged from the cube
to the positive orthant and sliced at $s$, gives
\[
 \Expect\#\{\phi:\ \phi\text{ stable},\ s>S_0\}
 \le {e^2N\over k}\Pr\{V<e^{-S_0}\},
 \qquad V\sim\operatorname{Beta}(k,M).
\]
Here $1/s\le1/(1-e^{-s})$ and the substitution $v=e^{-s}$
give the beta integral exactly. Take $S_0=11L/10$.
The beta order-statistic identity and a binomial Chernoff bound show
that the right side is $e^{-\Omega(N)}$, uniformly on the stated
compact interval of population ratios.

We also keep track of a single applicant's rank $R_i$.
Conditional on $x,y$ and stability of a fixed injection, the indicators
that its other matched firms rank above its partner are independent,
with probabilities
\[
 {x_i(1-y_j)\over1-x_iy_j}\le x_i.
\]
An unmatched firm cannot rank above the partner in a stable injection.
Consequently $\Expect[z^{R_i-1}\mid x,y,\phi\text{ stable}]
\le \exp((z-1)(M-1)x_i)$ for $z>1$.
On the uniform simplex $\sum_i x_i=s$, the moment identity
\[
 \Expect[x_i^j\mid s]
 ={j!s^j\over M(M+1)\cdots(M+j-1)}
\]
implies
\[
 \Expect[\exp((z-1)(M-1)x_i)\mid s]
 \le {1\over1-(z-1)s},\qquad (z-1)s<1.
\]
Choose $z=1+3/(4L)$; then $(z-1)S_0=33/40$.
Applying the same integral envelope with this extra factor yields
\[
 \Expect\sum_{\phi\text{ stable},\ s\le S_0}z^{R_i(\phi)-1}
 \le {40e^2N\over7k}.
\]
Thus a union bound over applicants and Markov's inequality give
\begin{equation}
 \Pr\{\exists\phi\text{ stable},\exists i:\ R_i(\phi)>r\}
 \le e^{-\Omega(N)}+C_{\delta_0}M z^{-(r-1)}.
 \label{eq:compact_complete_max_rank}
\end{equation}
Since $\delta=1-e^{-L}\le\min(L,1)$ and
$L\log(1+3/(4L))$ is increasing in $L$, we have
\[
 {2L\over\delta}\log\left(1+{3\over4L}\right)
 \ge2\log(7/4)>1.
\]
Taking $r=(2L/\delta)\log M$ in
\eqref{eq:compact_complete_max_rank} proves, with high probability,
the required simultaneous upper bound for every stable matching.

\paragraph{Transfer to lists of any length.}
Conditional on the retained degrees, each applicant's ordered neighbor
list is an independent uniform ordered subset. Append the unlisted
firms in an independent uniform order and extend the firms' independent
uniform orders consistently. This produces a complete uniform market.
If its stable partners all have rank at most $r$ and every retained
degree is at least $r$, its applicant-optimal matching is also stable
in the list market and covers all applicants. Every stable matching
in the list market therefore covers them all, by the Rural Hospital
Theorem. Such a matching is stable also in the completion: each appended
firm lies below that applicant's present partner. The complete-market
rank bound consequently applies to every stable matching in the list
market. This argument imposes no upper bound on the list lengths.

\paragraph{An unmatched-firm bound.}
Write $K=|\Long'|$, $M=|\Short'|$, and $\delta=M/K<1$.
Fix an applicant $a\in\Short'$ and an admissible selected set
$S\subseteq\calN(a)$, $|S|=k$. Remove $a$ and let $U^-$ be the set of
unmatched firms in the remaining market. This set is the same in every
stable matching and has cardinality at least $K-M+1$.

Upon adding $a$, firms weakly improve their partners in the firm-optimal
matching. Previously matched firms therefore remain matched. An old
unmatched applicant cannot become matched: every adjacent firm preferred
its old partner to that applicant by stability, and its new partner is
weakly better. Thus at most one additional firm becomes matched, and at
most one member of $U^-$ ceases to be unmatched. Every other member of
$U^-$ is unmatched in every stable matching of the full market and is
available to $a$ if it is a neighbor of $a$. Consequently, two members of
$S$ in $U^-$ guarantee an available selected neighbor.

The root data are independent of the market without $a$. By symmetry,
$U^-$ is a uniform subset of the firms conditional on its size $u$.
Conditional also on this size, the intersection count $|S\cap U^-|$ has law $\Hyper(u,K,k)$, using the paper's successes--population--draws convention. For
$k\ge1$, the event $|S\cap U^-|\le1$ requires some $(k-1)$-subset of
$S$ to lie among the at most $M-1$ matched firms of that market. Therefore
\begin{equation}
 \Pr\{\text{no available firm in }S\mid\text{root data}\}
 \le \min\left\{1,
       k\frac{(M-1)_{k-1}}{(K)_{k-1}}\right\}
 \le \min\{1,k\delta^{k-1}\},
 \label{eq:degree_free_available_upper}
\end{equation}
with the trivial bound one when $k=0$. This holds for every $d\le n_\Long$.
For random $k$, average the displayed bound over the root data.

For the full one-sided market and $S=\calN_+(a)$, $k\sim\Binom(d,p)$,
so averaging gives the explicit bound
\[
 \Pr\{\text{no available nonnegative-shock neighbor of }a\}
 \le (1-p)^d+dp\,[1-p(1-\delta)]^{d-1}.
\]
For a reduced market, set $\gamma_2=K/n_\Long$ and
$\gamma_2^-=(K-1)/(n_\Long-1)$. Hypergeometric sampling instead yields
\[
 \Pr\{\text{no available nonnegative-shock neighbor of }a\}
 \le(1-\gamma_2p)^d
    +d\gamma_2p\,[1-\gamma_2^-p(1-\delta)]^{d-1}.
\]
Indeed, distinguish one selected positive edge and use sampling without
replacement for the other $d-1$ choices. In particular, when
$\gamma_2=\Omega(1)$, $1-\delta=\Omega(1)$, and
$dp/\log(2+n_\Short)\to\infty$, these bounds and a union bound establish
the positive-neighbor criterion for every retained applicant, without
any upper restriction on $d$.

\paragraph{Completing all imbalance regimes.}
Suppose first that $\delta$ stays in a compact subinterval of $(0,1)$.
The complete-market bound \eqref{eq:compact_complete_max_rank} and
the list-completion argument imply that, with high probability,
every applicant has rank at most $(2+\eta)r_n$ in every stable matching
on $H$, for any fixed $0<\eta<0.1$. Here
$r_n=(1+o(1))(L_\delta/\delta)\log n_\Short$ and
$d\ge8r_n/p>(2+\eta)r_n$.
This rank is below $|\calN_+(a)|$ on the event
\eqref{eq:addition_union_calN_+}. By
\prettyref{lmm:positive_available_addition}, every stable matching
is therefore perfect interim stable. No upper bound on $d$ was used.

If $\delta\to0$, use the unmatched-firm bound above directly.
Writing $\mu=dp$, the per-applicant failure probability is at most
\[
 e^{-\mu}+e\mu e^{-(1-\delta)\mu}.
\]
Since $\mu\ge(8-o(1))\log n_\Short$ and $n_\Short\to\infty$,
this is $o(1/n_\Short)$, uniformly also when $\mu$ is much larger
than its lower bound. A union bound and
\prettyref{lmm:positive_available} finish this case.

Finally, suppose $\delta=1-o(1)$. Fix $0<\eta<0.1$.
The budget gives $d>(2+2\eta)r_n$.
By \prettyref{prop:perfect_stable}, for $\delta<1$
every applicant has rank at most $(2+\eta)r_n$ in every stable
matching; for $\delta=1$, the applicant-optimal rank is at most
$(1+\eta)r_n$. These ranks are below $2.2r_n<|\calN_+(a)|$
on \eqref{eq:addition_union_calN_+}. Apply
\prettyref{lmm:positive_available_addition} in each case.
Every sequence of population ratios has a subsequence in one of
the regimes just treated, so the stated conclusions follow in full.

}%

{
\paragraph{Proof of \prettyref{rmk:insuffcient_signal}.}
To obtain \prettyref{rmk:insuffcient_signal} from
\cite[Theorem~6 and Lemma~4]{potukuchi2024unbalanced}, put
$N=n_\Short$, $M=n_\Long$, $k=M-N=o(N)$, and
$D_0=\log N\log(M/(k+1))$. Its stated upper budget is
$(1-\epsilon)(1+o(1))(M/N)D_0$.
Consider the uniform-preference Erd\H{o}s--R\'enyi market with edge
probability $(1-\epsilon/2)D_0/N$.
It has an unmatched short-side agent with high probability.
Every short-side preference list has at least $d$ entries with high
probability, since its mean is $(M/N)(1-\epsilon/2)D_0$ and
$D_0/\log N\to\infty$. Retaining the first $d$ entries in each such
list gives the ordinal law of the applicant-signaling market.
Shortening the proposing side's lists to prefixes cannot turn a
nonperfect market into a perfect one: if deferred acceptance on the
shorter lists matched every proposer, it would make exactly the same
proposals on the longer lists. Consequently $H$ has an unmatched
applicant and an unmatched firm, which form an interim blocking pair.

}%
\subsubsection{Proof of \prettyref{thm:single_firm_signal_dense}}\label{sec:single_firm_signal_dense}

{

Write $N=n_{\Long}$, $M=n_{\Short}$, and $r=M-N\ge1$.
Expose the applicant-side scores and hence $H$, but leave the
firm-side scores unrevealed. For every firm $j$,
$|\calN(j)|\sim\Binom(M,d/N)$, and the Chernoff bound gives
\[
 \Pr\{\max_j|\calN(j)|\le\kappa_1\}=1-o(1),
 \qquad \kappa_1=(2dM/N)\vee\log^2n.
\]
Set
\[
 t=F_{\DistA*\DistB}^{-1}\!\left(\frac{\kappa_1}{\kappa_1+1}\right),
 \qquad
 \calQ=\{j:A_{j,a}+B_{j,a}\le t
                    \text{ for every }a\in\calN(j)\}.
\]
Continuity of $\DistB$ implies continuity of $\DistA*\DistB$.
On the degree event, conditional on $H$, the events defining
$\calQ$ are independent across firms and have probabilities at least
\[
 \left(\frac{\kappa_1}{\kappa_1+1}\right)^{\kappa_1}
 >e^{-1}>1/3.
\]
Thus $|\calQ|\ge N/3$ with probability $1-o(1)$.

Now expose all firm-side scores on interviewed pairs. Let $U$ be the
unmatched-applicant set, common to all stable matchings by the Rural
Hospital Theorem; necessarily $|U|\ge r$. Both $U$ and $\calQ$ are
determined by these exposed data. For each $a\in U$, there are at least
\[
 |\calQ\setminus\calN(a)|\ge N/3-d\ge N/12
\]
firms in $\calQ$ that did not interview $a$. On all these distinct
ordered pairs, the unused scores $B_{j,a}$ remain independent with
law $\DistB$. The quantile assumption gives
\[
 \Pr\{B_{j,a}>t\}\ge\frac1{\kappa_2+1}
 =\frac1{N-d}\ge\frac1N.
\]
Whenever one such score exceeds $t$, $(a,j)$ blocks every stable
matching: $a$ is unmatched and $j$ prefers $a$ to every interviewed
partner. Conditional on the exposed data and $|\calQ|\ge N/3$,
the probability that no such pair exists is at most
\[
 (1-1/N)^{|U|N/12}\le e^{-|U|/12}\le e^{-r/12}.
\]
Adding the $o(1)$ probability of the preceding exceptional events
proves the stated bound. This argument does not treat the maximum of
$\kappa$ samples as concentrating at its $\kappa/(\kappa+1)$ quantile.

For the first example, $\DistA=\boldsymbol\delta_0$ and
$\kappa_1<\kappa_2$ for sufficiently large $n$, so continuity of
$\DistB$ gives the required strict quantile ordering. For the second,
write $B\sim\mathbb N(\mu,\sigma^2)$ with fixed $\sigma>0$, and choose
a fixed finite $L$ with $A\le L$ almost surely. 
{
Mills' ratio~\cite[Equation~7.1.13]{abramowitz1968handbook} gives
}%

\begin{align*}
 F_{\DistA*\DistB}^{-1}\!\left(\frac{\kappa_1}{\kappa_1+1}\right)
 &\le L+\mu+\sigma\sqrt{2\log\kappa_1}\,(1+o(1)),\\
 F_{\DistB}^{-1}\!\left(\frac{\kappa_2}{\kappa_2+1}\right)
 &=\mu+\sigma\sqrt{2\log\kappa_2}\,(1+o(1)).
\end{align*}
If $d\le n^\alpha$ for a fixed $\alpha<1$, then
$\limsup\log\kappa_1/\log\kappa_2\le\max\{\alpha,0\}<1$.
The displayed upper quantile is therefore strictly smaller than the
lower one for all sufficiently large $n$.

}%
\subsubsection{Proof of \prettyref{thm:single_both_signal_dense}}\label{sec:single_both_signal_dense}
Let $H_1$ denote the interview graph constructed by the applicant-signaling mechanism, and let $H_2$ denote the interview graph constructed by the firm-signaling mechanism. Then, $H$ is the union graph of $H_1$ and $H_2$. For any $i\in\Short\cup \Long$, let $\calN(i)$ denote the set of neighbors of $i$ on $H$, and let $\calN_1(i)$ and $\calN_2(i)$ denote the set of neighbors of $i$ on $H_1$ and $H_2$, respectively.

{

{
{
Suppose the upper-atom alternative in the theorem holds.
}%
Write $M=n_\Short=n_\Long$ and
$q=\Pr\{A=c\}\ge p/2$, where $A\le c$ almost surely and $c\ge0$.
}%

For each applicant $a$, let $b_a$ be its $d$-th highest pre-interview
utility and let $K_a$ count its outgoing signal edges with $A_{a,j}=c$.

{
Then $K_a\sim\Binom(d,q)$. Since $dq\ge4\log^2M$, a Chernoff bound and
a union bound give
\[
 \Pr\{\min_{a\in\Short}K_a\le2\log^2M\}
 \le M\exp(-dq/8)
 \le M\exp[-\tfrac12\log^2M]=o(1).
\]
}%

The applicant-optimal rank bound used in
\prettyref{thm:single_signal_dense} implies that, with probability
$1-o(1)$, every applicant's partner in the applicant-optimal matching
$\Phi_1$ of $H_1$ has rank below $2\log^2M$. Therefore
\[
 U^A_{a,\phi_1(a)}\ge b_a+c
 >U^A_{a,j}\qquad
 \text{for every }j\in\calN_2(a)\setminus\calN_1(a),
\]
because such a firm has $U^B_{a,j}<b_a$ and every shock is at most $c$.
Thus no added edge blocks $\Phi_1$, and $\Phi_1$ is stable on $H$.
Applicant optimality on $H$ weakly improves every applicant relative to
this stable matching. Every applicant in the applicant-optimal matching
on $H$ consequently receives utility at least $b_a+c$, whereas every
firm outside its interview neighborhood has pre-interview utility
strictly below $b_a$. This proves perfect interim stability of the
applicant-optimal matching. The symmetric argument starting with $H_2$
proves the firm-optimal conclusion. A union bound gives both prescribed
endpoints simultaneously with high probability in this special case.

{
\paragraph{A general-law subcase near the complete interview graph.}
Let $n_\Short=n_\Long=N$, put $d=N-r$, and suppose
$r^2\log^2N=o(N)$ and $pd\ge8\log^2N$. Under the original iid
additive-score assumptions, both prescribed optimal matchings of the
bilateral interview graph are perfect interim stable with high
probability. No upper atom is needed for this subcase.

First consider one row of $N$ independent pairs $(B_i,A_i)$, with
$B_i$ continuous and the iid shocks $A_i$ independent of the pre-scores.
Let $R$ be the labels of the $r$ smallest $B_i$ and let $T$ be the labels
of the $k$ largest $B_i+A_i$, for an integer $1\le k\le N$. We claim
\[
 \Pr\{i\in R\cap T\}\le\frac{rk}{N^2}
 \qquad\text{for every fixed label }i.
\]
To see this, condition on the pre-score vector and write
$v_i=\Pr\{i\in T\mid(B_j)_{j=1}^N\}$. If $B_i\le B_j$, then
$v_i\le v_j$. Indeed, interchange the two base scores: this increases
$i$'s score and decreases $j$'s score, so membership of $i$ in the
top $k$ can only become more likely. Exchangeability of $A_i,A_j$
identifies the new probability with the original $v_j$.
Since $\sum_i v_i=k$, the sum of $v_i$ over the $r$ smallest base
scores is at most $rk/N$. Average over the base scores and use
exchangeability of the labels to obtain the claim. Strict rankings
hold almost surely under the original continuous pre-score law.

Now pre-sample interview shocks for every directed pair, including
pairs that do not ultimately interview, and consider the complete
market with utilities $B+A$ on every pair. It has independent uniform
preference orders. Set $k=\lceil2\log^2N\rceil$. For a fixed pair
$(a,j)$, the event that $j$ belongs to applicant $a$'s bottom $r$
pre-score labels and top $k$ complete post-score labels has probability
at most $rk/N^2$. The event that $a$ belongs to firm $j$'s bottom $r$
pre-score labels is independent of the applicant's entire row and has
probability $r/N$. A union bound therefore gives
\[
 \Pr\{\text{some applicant's complete top-$k$ list contains a
 nonedge of $H$}\}\le\frac{r^2k}{N}=o(1).
\]
The symmetric bound holds for firms' complete top-$k$ lists.

By the applicant-optimal part of \prettyref{prop:perfect_stable}, applied
with both retention fractions equal to one and complete lists,
the complete-market applicant-optimal maximum rank is
at most $2\log^2N$ with high probability. On the preceding list
inclusion event, its matched edges all belong to $H$, so this complete
applicant-optimal matching is stable on $H$ as well.
For each applicant let $b_a$ be its $d$th highest pre-score and let
$K_a$ count its outgoing signals with nonnegative shocks. Then
$K_a\sim\operatorname{Bin}(d,p)$, and
\[
 \Pr\{\min_a K_a\le k\}\le N e^{-pd/8}
 \le N e^{-\log^2N}=o(1),
\]
because $k\le pd/2$ for all sufficiently large $N$ and Chernoff's
inequality applies. Thus the complete-market
applicant-optimal matching gives every applicant utility at least
$b_a$. The applicant-optimal matching on $H$ weakly improves every
applicant relative to this stable matching of $H$. Every nonedge has
applicant pre-score below $b_a$, so this matching is interim stable.
Apply the same argument with the two sides interchanged and intersect
the events to obtain both prescribed endpoints simultaneously.

}%

}%

{
\paragraph{Conditional exposure of the signal graph.}
For the log-concave alternative, write $N=n_\Short=n_\Long$ and
$L=\log N$, and fix the positive lower bound $p_0$ from that alternative.
We suppress the common intrinsic values in within-tier
utility comparisons and work with the idiosyncratic scores $B+A$.
We bound proposal targets and incoming counts, then use retention estimates
to show that applicant-proposing DA finishes before any applicant proposes
below its pre-interview signal cutoff. A separate query argument handles
the largest degrees.

For $C L^2\le d\le N^{3/4}$, call an applicant's own signal a
\emph{source} option and a firm's signal without the reverse signal
an \emph{incoming} option. Color applicant-only, firm-only and
reciprocal pairs $O,I,R$, respectively; reciprocal proposals count
as source proposals. Fix $z<1$ and stop deferred acceptance before
any applicant uses more than $zd$ source options, before $1.01NL$
source proposals or $C_1NL$ incoming proposals in total, or before
any firm receives $100(1+C_1)L$ proposals. Additional earlier stops
are allowed. Uniformly before these stops, ordinary $O$ and $I$
proposals target each specified firm with probability at most
$(1+o(1))/N$. A source proposal targets a never-visited firm with
probability at least $(1-o(1))/N$, apart from the negligible
exceptions specified below. Reciprocal proposals have $o_p(N)$
total targets and future-hit probability $o(1)$ at a specified firm.
The load stop has probability $o(1)$; the total-source stop before
completion has probability $N^{-.01+o(1)}$.

Here and below these are estimates for stopped processes, without
conditioning on the event that future stops are avoided. We give the
conditional-graph justification, since uniform preferences on the
union graph cannot be assumed. Every outgoing signal set is an
independent uniform $d$-subset. All incoming degrees are
$(1+o(1))d$ with high probability. Conditional on the three colors'
row and column degrees, their graphs are uniform disjoint simple
bipartite graphs with those degrees. Conditional also on each row's
$d$th pre-score $b_i$, its $d-1$ strict source values are iid from
$B\mid B>b_i$, its one boundary value is $b_i$, and its incoming
values are iid from $B\mid B<b_i$; shocks are independent.
Assign unresolved labels by uniform bijections to these numerical
slots. Numerical preferences may be exposed in full or progressively.
At a receiving firm expose values only after the target and its color
are selected. An $O$ proposal receives a fresh nonsource value; an
$I$ or $R$ proposal receives a fresh source value. The boundary slot
has conditional probability at most $1/(d-O(L))$, and every other
fresh source value has the strict conditional law.

Every compatible residual graph has the same number of remaining
slot bijections, a product of factorials of residual degrees. The
likelihood of exposed receiving values depends only on already
revealed colors, cutoffs and boundary marks, not on unresolved
labels. Thus the uniform residual graph law is preserved after every
adaptive history. One may reveal each receiving $B,A$ pair, including
its boundary mark, immediately after its proposal. If reciprocal
labels are not initially revealed, the $R$ tags are a uniform subset
of each applicant's source slots, independent of its numerical values.
In particular, after at most $zd$ source positions,
\begin{equation}
 \Pr(\text{next source tag is }R\mid\mathcal F)
                 \le D_{R,i}/((1-z)d).
 \label{eq:lc_reciprocal_tag}
\end{equation}

Let a selected residual color have $M$ edges, maximum degree $D$,
forbidden-cell maximum degree $F$, and residual degrees $r_i,c_j$.
Conditional on the other colors, apply the single-edge bounds of
\citet[Lemma~5.3]{larkin2025subgraphs} to this uniform residual graph.
Zero residual degrees give probability zero directly.
Their degree-sum error terms are at most $2D(D+F)/M$, giving
\begin{equation}
 \Pr(ij\text{ is selected}\mid\mathcal F)
       \le\frac{r_ic_j}{M-O(D(D+F))},
 \label{eq:lc_switch_upper}
\end{equation}
and, for an allowed cell,
\begin{equation}
 \Pr(ij\text{ is selected}\mid\mathcal F,\text{other colors})
 \ge\frac{r_ic_j}{M}\{1-O(D(D+F)/M)\}.
 \label{eq:lc_switch_lower}
\end{equation}
The lower-bound conversion uses $r_ic_j\le D^2$.
 For $O,I$, $M=(1-o(1))Nd$ and
$D,F=O(d)$, so the relative error is $O(d/N)=o(1)$.
Divide the edge estimate by the proposing vertex's residual degree
to obtain the next-slot target estimate; that degree need not be
$(1+o(1))d$. A never-visited firm's target degree remains
$(1+o(1))d$. Averaging the lower bound over the other colors is
valid: their upper bounds show that an unqueried cell is forbidden
by them with probability $o(1)$.

For the reciprocal color $R$, put $r=d^2/N$.
A reciprocal vertex degree is $\operatorname{Bin}(d,d/N)$,
$\mathbb E|R|=d^2$, and $\operatorname{Var}|R|\le d^2$:
reciprocal indicators at disjoint cells are independent, and those
sharing a row or column are negatively correlated. If $r\le1/L$,
reveal the graph $R$ initially, but keep its labels' assignments to
applicant numerical source slots hidden until their actual proposals.
Those assignments are uniform and independent of the numerical values.
It has $o(N)$ incident firms and maximum degree
at most $L$ with high probability. Exclude its incident firms in
retention estimates. For a fixed never-visited firm, at most $L$
roots have a known reciprocal cell there; their at most $zdL=o(NL)$
source steps are the coupon exceptions. The remaining graphs retain
the exact uniform residual law just proved.

If $1/L<r\le L^2$, reveal only reciprocal degrees. With high
probability $D_R\le CL^2$ and $|R|=(1+o(1))d^2$.
Stop after $N/L^2$ reciprocal proposals. By
\eqref{eq:lc_reciprocal_tag} this stop has probability at most
$CL^5/d=o(1)$. The selected reciprocal color in
\eqref{eq:lc_switch_upper} has maximum degree $CL^2$, so its
switching error is $O(L^2/d)=o(1)$ and its next target probability
is at most $CL^2/d^2$. Together with the reciprocal-tag bound,
the conditional probability of any later reciprocal hit at a fixed
firm is at most $CNL^5/d^3=o(1)$. Its probability of occupying
an unqueried cell is at most $CL^4/d^2=o(1)$.
If $r>L^2$, all reciprocal degrees are $(1+o(1))r$. Their tags
have probability at most $C_zd/N$ per source step; stop after
$2C_zdL$ reciprocal proposals. A conditional Chernoff bound excludes
this stop, and reciprocal switching gives target probability at most
$(1+o(1))/N$. Their total count and future-hit bound are
$O(dL)=o(N)$ and $O(dL/N)=o(1)$, respectively. Their unqueried-cell
probability is $O((d/N)^2)=o(1)$. These estimates also imply the
uniform degree bound $D_{R,i}=O(d^2/N+L)=o(d)$.

The ordinary target bounds give the load stop by conditional binomial
domination. Tiny-regime reciprocal loads are at most $L$ deterministically;
in the other two regimes their per-firm compensators are respectively
$O(NL^5/d^3)=o(1)$ and $O(dL/N)=o(1)$. Conditional Chernoff bounds
therefore keep every reciprocal load below $L$ with high probability,
within the stated $100(1+C_1)L$ stop. At the total-source stop a fixed never-visited firm has,
after the stated reciprocal exceptions, $(1.01-o(1))NL$
opportunities with conditional hit probability $(1-o(1))/N$.
A union bound over firms gives $N^{-.01+o(1)}$. A visited firm
never becomes empty. All these estimates persist with earlier
numerical, occupancy or utility stops.

\paragraph{Numerical-prefix estimates.}
We need bounds that hold for every prefix of a numerical preference
row, so an adaptive stopping prefix is allowed. If a uniform word
has $d_0$ source and $m$ incoming marks, write $S_j,I_j$ for its
prefix counts. If $m/d_0\le(1-\eta)\lambda$ for a fixed
$\eta>0$, and $\lambda$ is bounded above, then
\begin{equation}
 \Pr\{\max_j(I_j-\lambda S_j)>t\}\le C_\eta e^{-c_\eta t},
 \qquad \mathbb E\max_j(I_j-\lambda S_j)_+=O_\eta(1).
 \label{eq:lc_prefix_fixed}
\end{equation}
The constants can be chosen uniformly as $\lambda\downarrow0$.
Index a crossing by its incoming count $\ell$: it requires at least
$\ell$ incoming marks among at most $\ell(1+1/\lambda)$ initial
marks. Their hypergeometric mean is at most $(1-c\eta)\ell$.
Chernoff's bound is $e^{-c_\eta\ell}$; sum over $\ell>t$.
When $\theta=m/d_0$ is much smaller than $0<\lambda\le1$, the
same argument, using the elementary hypergeometric exponential bound,
gives the sharper estimate
\begin{equation}
 \mathbb E\max_j(I_j-\lambda S_j)_+
 \le\frac{R}{(1-R)^2},\qquad
             R=e\theta(1+1/\lambda)<1.
 \label{eq:lc_prefix_vanishing}
\end{equation}
Indeed the crossing probability indexed by $\ell$ is at most
$R^\ell$, and integrating the sum of these bounds gives
$\sum_{\ell\ge1}\ell R^\ell$.

{These estimates apply when incoming priorities have a CDF
$G(v)\le a v$ on $[0,v_0]$, where $0<a\le1$, and ordinary source priorities are
independent uniforms. The case $a=0$ has no incoming priorities in this interval.
Couple $V=G^{-1}(U)$ to $W=U/a$ if}
$U\le a$, and $W=+\infty$ otherwise. Then $V\le v\le v_0$
implies $W\le v$. The finite count $M$ is binomial, and its
finite values merge uniformly with the source uniforms. If its mean
is $\theta d_0$ and $\theta/\lambda\to0$, let $E$ be the corresponding
nonnegative prefix excess. The conditional form
of \eqref{eq:lc_prefix_vanishing}, splitting at
$M=\lambda d_0/(8e)$, gives
\begin{equation}
 \mathbb E E=O(\theta/\lambda)+o(1)
 \quad\text{if }\lambda d_0\ge c\sqrt L.
 \label{eq:lc_prefix_binomial}
\end{equation}
The exceptional contribution is at most
$\mu(e\mu/(k-1))^{k-1}=o(1)$, where
$\mu=\theta d_0$, $k=\lambda d_0/(8e)$; this follows by
writing $\mu=(\theta/\lambda)\lambda d_0$.
For $\theta L\to0$, choose
$\lambda=\max\{\sqrt{\theta/L},L^{-3/2}\}$.
Then $\lambda NL=o(N)$ and the sum of the row excesses is
$o_p(N)$. Thus every collection of prefixes using $O(NL)$
source positions has $o_p(N)$ incoming positions.
All reciprocal source values remain in the numerical word before
their tags are revealed. Later omissions are charged by their actual
counts; no conditional iid assertion after those omissions is needed.

The case $\theta L=O(1)$ also admits the following sharper deferred
count. Suppose a row has $(1-o(1))d$ independent uniform source
priorities and at most $(1+o(1))d$ incoming priorities satisfying
$G(v)\le\theta v$ up to $v_0$. Call a root light before it uses
$\eta d$ source positions, with fixed $\eta>0$. Under any
history-based DA schedule, up to $O(NL)$ source proposals,
\begin{equation}
 I_{\rm light}\le
       \frac{1+o(1)}{1-\eta-o(1)}\theta T_{\mathrm{source}}+o_p(N).
 \label{eq:lc_light_flow}
\end{equation}
Here is the required filtration. Use the preceding $V,W$ coupling,
but expose numerical values progressively. Scan the merged fictitious
$W$/ordinary-source stream only through the next ordinary source.
Each encountered $W$ reveals its actual $V$, which is queued until
its turn in the actual preference list. Cache the next ordinary source;
additional source positions can be proposed while leaving that cache
unchanged. Labels remain hidden until an actual proposal. Thus the
DA run is unchanged and at most one source per row is looked ahead.
After scanning through level $v$, unseen source uniforms are
conditioned only to exceed $v$, and unseen $W$'s only to exceed $v$.
Their per-variable hazards are $1/(1-v)$ and
$\theta/(1-\theta v)$. As long as the root is light, their total
incoming/source hazard ratio is at most
$(1+o(1))\theta/(1-\eta-o(1))$. Queued values belong to already
exposed draws and impose no condition on unseen draws. Competing
hazards and conditional geometric domination therefore bound the
fictitious incoming count by a negative binomial with that mean
ratio times at most $T_{\mathrm{source}}+N$ source events. Its deviation
is $o_p(N)$, since its mean is $O(N)$; a Chernoff bound with
error $N^{2/3}$ is $e^{-\Omega(N^{1/3})}$ and also holds for all
prefixes. The cache cost is $O(\theta N)=o(N)$, proving
\eqref{eq:lc_light_flow}. If $v_0=2\eta$, all light rows remain
within this interval with high probability by source order-statistic
concentration. This argument does not condition on complete numerical
rows or future certificate events; those are intersected afterward.

\paragraph{Retention with few incoming proposals.}
If the stopped run has at most $C N$ incoming proposals, with fixed
$C$, every ordinary source proposal has a conditional permanent-
retention probability at least $c(C)/L>0$. To prove this, let $U_k$
count firms never targeted by incoming proposals after $k$ such
proposals. The target upper bound gives
$\mathbb E(U_{k+1}\mid\mathcal F_k)\ge
(1-(1+o(1))/N)U_k$. Multiplying by the inverse of this factor to
power $k$ gives a submartingale with bounded increments for $k\le CN$.
Azuma and a union bound yield $U_k\ge(e^{-C}-o(1))N$
simultaneously. Discarding reciprocal targets loses only $o(N)$ firms.
Add the occupancy lower bound and the low-priority count bound below
as stopping rules; the following retention estimates hold before
their first violation, whose probability is bounded separately.
Transform each firm's ordinary receiving scores by its own continuous
conditional nonsource law. Fresh priorities are uniform even when
these conditional laws differ between firms. Put $v_0=e^{-C}/(100L)$.
Among $1.01NL$ source proposals, at most $2v_0NL$ receiving
priorities are below $v_0$, with high probability. Consequently the
next ordinary proposal targets an uncontaminated firm holding no such
priority with probability at least $e^{-C}/2-o(1)$.
For its fresh uniform priority $v\le v_0$, future ordinary danger
has probability at most $(1+o(1))v/N$ per source proposal, and
future incoming danger at most $(1+o(1))/N$ per incoming proposal.
The conditional probability of avoiding both is at least
\begin{equation}
 (1-(1+o(1))v/N)^{1.01NL}
 (1-(1+o(1))/N)^{CN}\ge e^{-1.02Lv-C-o(1)}.
 \label{eq:lc_ordinary_retention}
\end{equation}
This product is valid under adaptive type scheduling by using
independent uniform tapes for the successive danger trials of each
type; the actual danger indicators are coupled below their conditional
Bernoulli upper bounds. Reciprocal future danger is $o(1)$.
Integrating over $v\le v_0$ gives $c(C)/L$.

If instead the incoming count is $o_p(N)$, choose a deterministic
$o(N)$ stop that contains it with high probability. The sharper
choice $v_0=1/(2L)$, a stop at $.51N$ ordinary priorities below
$v_0$, and the bound $.49(1-1.01Lv)$ give retention at least
$.18/L$ after integration. All occupancy and priority stops have
probability $o(1)$ by the preceding martingale and Chernoff bounds.

\paragraph{Degrees above $N^{3/4}$.}
With $p\ge p_0>0$, this range needs no shock-shape assumption. Put $\rho=d/N$
and choose deterministic quantiles $b_0<u$ with
$\Pr(B>b_0)=\rho/2$, $\Pr(B>u)=\rho/4$. All source cutoffs
are at most $b_0$ with probability $1-o(1)$. Every applicant has
more than $K=\lfloor pd/100\rfloor$ of its $N$ hypothetical
numerical scores $B+A$ above $u$, by
$\Pr(B+A>u)\ge p\rho/4$ and Chernoff.

Expose each applicant's full unlabeled array of $(B,A)$ pairs and
order it by $B+A$. A free applicant examines the next slot and its
uniform unexamined label. Its own-source slot is an interview;
otherwise it proposes only if the firm signals it, skipping a failed
membership query. For analysis query the firm's membership also at
own-source slots. This exactly executes DA on the union graph.
Expose only firm cutoffs initially, without reciprocal graphs or
incoming degrees. If $q_j$ labels have been queried at firm $j$,
$s_j$ of them sources, its next source probability is exactly
\begin{equation}
 \Pr(\text{firm source}\mid\mathcal F)=(d-s_j)/(N-q_j).
 \label{eq:lc_dense_query_membership}
\end{equation}
Receiving observations preserve this rule:
conditional on membership and cutoff, numerical values are independent
of the remaining label permutation. In the firm's conditional-source
priority coordinate, every fresh value has CDF at most $v$, and a
source value has a uniform component of weight at least
$1-1/(d-s_j)$. Values on failed membership queries may be exposed
for analysis; doing so only overcounts competitors.

Stop at total queries $T=\lceil10NL/\rho\rceil$, firm query load
$M=\lceil200L/\rho\rceil$, individual queries $K$, or $N/4$
receiving priorities below $q=\rho/(100L)$. A next label is uniform
among at least $N-K\ge .99N$ labels. Conditional binomial domination
excludes the load and low-priority stops. Since $M=o(d)$, source
membership has probability at least $\rho/2$. At least $.7$ of
target mass holds no priority below $q$, and future displacement of
$v\le q$ has probability at most $Tv/(N-K)\le .11+o(1)$.
Every query therefore has permanent-retention chance at least
$\rho^2/(1000L)$. The probability of an individual-query stop is
at most
\begin{equation}
 N\exp[-(1-o(1))pd^3/(10^5N^2L)]=o(1),
 \label{eq:lc_dense_query_retention}
\end{equation}
since $d>N^{3/4}$. A never-visited firm was previously queried by
at most $M$ applicants. Their total ineligible queries are at most
$MK\le .25T$. Every other query hits the firm and discovers source
membership with probability at least $\rho/(2N)$. Thus the
total-query stop has probability at most
$N(1-\rho/(2N))^{.75T}\le N^{-2.75+o(1)}$.
DA finishes before any applicant uses a score below $u>b_0$, so all
applicants reject noninterviewed options. Reversing the sides proves
both optimal conclusions in this degree range.

\paragraph{The log-concave alternative at sublinear degrees.}
Suppose now that $A$ has a log-concave density, is independent of the
continuous pre-score $B$, and $8p^{-1}L^2\le d\le N^{3/4}$.
The laws may depend on $N$; the estimates below use only $p\ge p_0>0$
and exact comparisons of conditional probability masses.
Take $\varepsilon_N=C\sqrt{L/d}$, with $C$ sufficiently large,
and choose deterministic quantiles $b_-,b_+$ with
\[
 \Pr(B>b_-)=(1+\varepsilon_N)d/N,\qquad
 \Pr(B>b_+)=(1-\varepsilon_N)d/N.
\]
All actual $d$th pre-score cutoffs $b_i$ lie in $[b_-,b_+]$
with high probability. Let $F_1,F_0$ be the laws of $B+A$
conditional on $B>b_+$ and $B\le b_+$, with upper tails $h_1,h_0$.
A strict source value in row $i$ has an $F_1$ component of weight
\begin{equation}
 \pi_i=\frac{\Pr(B>b_+)}{\Pr(B>b_i)}=1-O(\varepsilon_N),
 \label{eq:lc_common_source}
\end{equation}
and its remaining component is stochastically smaller. The actual
nonsource law $F_{0,i}$ is stochastically smaller than $F_0$, with
\begin{equation}
 \|F_{0,i}-F_0\|_{\rm TV}\le C\varepsilon_Nd/N=:\Delta_N,
                         \qquad \Delta_NL=o(1).
 \label{eq:lc_common_nonsource}
\end{equation}
The boundary source is also stochastically smaller than $F_1$.
These assertions follow directly by splitting the conditional
pre-score samples at $b_+$, and remain valid at fresh receiving slots.

The translation kernel $f_A(t-b)$ is totally positive of order two
by \cite[Proposition~2.3]{saumard2014logconcavity}.
Integrating over independent $B_1>b_+\ge B_0$ therefore gives
$F_0\le_{\mathrm{lr}}F_1$. The upper-tail ROC characterization
\citep[Corollary~2.7]{dumbgen2023stochastic}, with its coordinates
interchanged, yields
\begin{equation}
 g(v)=h_0(h_1^{-1}(v))\text{ is convex for }0\le v<1,\qquad
                         g(0)=0,\quad g(v)\le v.
 \label{eq:lc_convex_tail}
\end{equation}
Generalized inverses are used on nonterminal flat intervals, where
both tails are constant. At source priority one, ordinary scores
below the source support retain their actual utility ordering.

Fix
\begin{equation}
 x=.9p,\qquad x_+=.95p,\qquad t=h_1^{-1}(x),\qquad
 q=h_0(t),\qquad \theta=q/x.
 \label{eq:lc_utility_threshold}
\end{equation}
Since $h_1(b_+)\ge p$, we have $t\ge b_+$.
In every applicant row the source count above $t$ is at most $x_+d$
with high probability and has mean at least $(1-o(1))xd$.
Conditional on the cutoff and incoming degree $D_i=(1+o(1))d$,
its incoming count above $t$ is independently binomial with
probability $q_i$ satisfying
\begin{equation}
                     \theta x-\Delta_N\le q_i\le\theta x.
 \label{eq:lc_root_incoming}
\end{equation}
Stop a root immediately before it would use an option below $t$.
This respects the technical source cap $x_+d<d$.

By \eqref{eq:lc_convex_tail}, $g(v)\le\theta v$ for $0\le v\le x$.
The common-law source component in \eqref{eq:lc_common_source}
provides $(1-o(1))d$ independent uniform priorities per row; the
other source options can only add source positions. Thus the preceding
numerical-prefix estimates apply to every row up to the utility stop.
Use total-source budget $T_0=1.01NL$, with the graph and load stops
already proved. Reciprocal proposals are counted as contamination:
there are $o_p(N)$ of them and their conditional future-hit probability
is $o(1)$. Omitting a root's $O(d^2/N+L)$ reciprocal source slots
costs only $o(d/L)+O(1)$ in an exponent with retention rate $O(1/L)$.

It suffices to consider subsequences on which $\theta L$ converges.
If $\theta L\to\infty$, \eqref{eq:lc_prefix_fixed}, with
arbitrarily small fixed relative slack, gives total row excess
$O_p(N)=o_p(\theta NL)$. The binomial finite counts in the
thinning coupling concentrate uniformly here, since $\theta d/L\to\infty$.
Thus the incoming budget
\begin{equation}
                              T_1=1.02\theta NL
 \label{eq:lc_fine_incoming_budget}
\end{equation}
is not reached before another stop, with probability $1-o(1)$.

If $\theta L\to0$, \eqref{eq:lc_prefix_binomial} gives only
$o_p(N)$ incoming proposals before the other stops. The ordinary
retention calculation then gives $.18/L$ per ordinary source offer.
Each root's source count above $t$, followed by exact binomial
averaging, bounds the probability of any utility stop by
\[
 N\exp[-(.18-o(1))xd/L+o(d/L)+O(1)]=o(1),
\]
because $xd\ge7.2L^2$. This includes $\theta=0$ and proves
this regime without a weighted comparison of the receiving laws.

Suppose finally that $\theta L\to K\in(0,\infty)$.
A thinned row has more than $1.5\theta d$ finite incoming values
with probability $e^{-c\theta d}$. Charging its full $O(d)$ list,
the total contribution of such rows is $o_p(N)$ because
$d e^{-c\theta d}\to0$ uniformly for $d\ge8p^{-1}L^2$.
For all other rows slope $2\theta$ in \eqref{eq:lc_prefix_fixed}
has a uniform exponential excess tail. Conditional independence of
the numerical rows and a Chernoff bound therefore give a fixed
coarse incoming budget $C_KN$ with failure probability $o(1)$.
The occupancy argument \eqref{eq:lc_ordinary_retention} supplies
ordinary retention $c(C_K)/L$.

Fix $0<\eta<\min\{.3p_0,.003\}$ and, in the coarsely stopped run,
call a root heavy after $\eta d$ source positions. At least $\eta d/2$ of these are
ordinary. The return-time bound gives probability at most
$e^{-c(C_K)\eta d/(2L)}$ per root. Hence the heavy-root count $H$, set to zero off the initial
cutoff and degree good event, satisfies
\begin{equation}
                 \mathbb E[dH/N]\le
                    d e^{-c(C_K)\eta d/(2L)}=o(1).
 \label{eq:lc_heavy_rows}
\end{equation}
Heavy roots therefore contribute $o_p(N)$ incoming proposals in total.
For the light roots use \eqref{eq:lc_light_flow} with $v_0=2\eta<x$;
the required source order-statistic event holds with high probability.
It follows that
\[
 I_{\rm light}\le\frac{1+o(1)}{1-\eta}\theta T_0+o_p(N).
\]
Together with \eqref{eq:lc_heavy_rows}, this proves
\eqref{eq:lc_fine_incoming_budget}. The ghost-stream calculation
is made in its deferred filtration, not conditional on full numerical
rows or on the coarse thinning events. Their failure probabilities
are charged separately. Retain the coarse occupancy, priority and
load stops while introducing the fine incoming and grid stops below;
their already bounded failure events can only shrink. This handles
the intermediate regime without discarding an $O(N)$ error from a
budget of order $N$.

\paragraph{The weighted retention estimate.}
In the two regimes where \eqref{eq:lc_fine_incoming_budget} holds,
put $Q(s)=h_0(s)+\theta h_1(s)$. This continuous utility-tail
coordinate includes ordinary scores below the source support.
Set $s_0=+\infty$ and choose $m=1000$ thresholds $s_j$ with
$Q(s_j)=j/(1.03mL)$, using the intervals $(s_j,s_{j-1}]$.
The number of ordinary and incoming receiving scores above $s_j$
has conditional mean at most
\[
 T_0h_0(s_j)+T_1h_1(s_j)
                 \le1.02NLQ(s_j)<Nj/m.
\]
Martingale Chernoff bounds put it below $N(j/m+.001)$
simultaneously for all $j$, with probability $1-o(1)$.
The target bounds and the negligible reciprocal contamination give
at least $1-j/m-.002$ of available target mass with no better held
score. Conditional future dangerous-hit bounds give the same lower
bound for avoiding displacement. Therefore a fresh score in the
corresponding utility interval is retained before all stops with
probability at least $w_j=(1-j/m-.002)_+^2$.

Integrate this step weight against $F_0,F_1$, obtaining $a_0,b_0$.
The exact Stieltjes identity $dF_0+\theta dF_1=-dQ$ gives
\begin{equation}
 a_0+\theta b_0
       =\frac1{1.03mL}\sum_{j=1}^m w_j>\frac{.31}{L}.
 \label{eq:lc_weighted_retention}
\end{equation}
By \eqref{eq:lc_common_source}--\eqref{eq:lc_common_nonsource},
actual ordinary and incoming proposals have deterministic retention
bounds $a,b\ge0$ satisfying
\begin{equation}
        a+\theta b\ge .30/L,\qquad a=O(1/L),\qquad b\le w_1<1.
 \label{eq:lc_actual_retention}
\end{equation}
The ordinary additive loss is at most $\Delta_N=o(1/L)$;
the incoming loss is a $1-o(1)$ mixture factor. Fixed margins in
the grid absorb the graph and reciprocal errors. The weighted
identity is valid even when $\theta\to0$ and $b$ is of order one.

For the following root calculation, condition on the certified initial
cutoffs and graph degrees, and fix the applicant numerical words
whose source counts respect the cap. This does not condition on any
future event; all future occupancy, count and grid bounds are enforced
by stopping. The words fix source-versus-incoming positions, but need not reveal
reciprocal source tags. Reveal the next tag before its fresh target
and receiving value. At successive root returns multiply the
continuation indicator by $(1-a)^{-1}$ for an ordinary source,
$(1-b)^{-1}$ for an incoming proposal, and one for a reciprocal.
The conditional retention bounds make this a supermartingale.
Let $\tau_i$ be the root's utility stop and $\tau$ the first other
stop. If $S,I$ are its numerical source and incoming counts above $t$
and $D_R$ its reciprocal degree, the stopped supermartingale gives
\begin{equation}
 \Pr(\tau_i<\tau\mid\text{numerical words})
                \le(1-a)^{(S-D_R)_+}(1-b)^I.
 \label{eq:lc_root_product}
\end{equation}
Thus omitted reciprocal tags cost at most
$(1-a)^{-D_R}=\exp[o(d/L)+O(1)]$; no unexposed tags were conditioned on.

Average \eqref{eq:lc_root_product} over the independent source and
incoming samples, conditional on cutoffs and graph degrees. The
exact binomial product gives, on the common good event,
\begin{align*}
 \Pr(\text{root utility stop before all other stops},\mathcal G)
 &\le e^{o(d/L)+O(1)}(1-x_i a)^{d-1}(1-q_i b)^{D_i}\\
 &\le\exp[-(.30-o(1))xd/L+o(d/L)+O(1)],
\end{align*}
where $\mathcal G$ is the common initial structural and numerical
good event, $x_i\ge(1-O(\varepsilon_N))x$, and
$q_i\ge(\theta x-\Delta_N)_+$. We used
$D_i=(1+o(1))d$, $\Delta_NL=o(1)$ and
\eqref{eq:lc_actual_retention}. The nonnegative word product may
be integrated without its numerical-good-event indicator; the
exceptional probability is charged once globally. Since $xd\ge7.2L^2$,
a union bound over roots is $o(1)$.

All auxiliary stops have now been excluded. The source coupon bound
proves completion before the total-source stop. Every applicant is
matched above $t\ge b_+\ge b_i$, while every noninterviewed
pre-score is below its own cutoff $b_i$. The resulting full
applicant-optimal matching is therefore perfect interim stable.
Interchange applicants and firms and intersect the two events to
obtain both prescribed optimal outcomes simultaneously.

}%

\subsection{Market with correlated post-interview scores}\label{sec:correlated_proof}

\subsubsection{Proof of \prettyref{ex:good_bad_interviewers}}\label{sec:good_bad_interviewers}

{
Write $p=\Pr\{A_{a,j}\ge0\}$ for the applicants' own nonnegative-shock
probability. 
{
For the sparse conclusion, $p\ge p_0>0$ with fixed $p_0$.
}%

\paragraph{A rank estimate with finitely many applicant types.}
We first consider a balanced market with $N$ agents on each side,
$d\to\infty$, and $d\le N$. Conditional on the applicants' type labels,
let $N_s$ be the number of type-$s$ applicants, with
$N_s/N$ bounded away from zero, and let $F_s$ be the continuous
cumulative distribution of a firm's score for a type-$s$ applicant.
There are a fixed number $r$ of types; in the example $r=2$ and
$F_s$ is the distribution function of $\DistB*\DistA_s$.
The argument permits arbitrary $F_s$, without any stochastic-order
comparison between types.

Conditional on its $d$ signal labels, an applicant's post-interview
ordering is uniform and independent of the firms' score arrays.
Attach independent uniform order statistics to this actual rank
permutation. The resulting $[0,1]$ preference costs on the edges are
independent uniforms, with smaller costs better; this coupling
preserves both the actual ranking and its joint relation to the
nonnegative-shock count. No independence between the positive-shock
marks and these auxiliary costs is asserted or needed.
This auxiliary representation changes neither the stable matchings
nor their ranks. Let $X_i$ be an applicant's matched cost, with $X_i=1$
if it is unmatched. We show that, with high probability, simultaneously
for all stable matchings,
\begin{equation}
 \frac1N\sum_i X_i=o(1).
 \label{eq:heterogeneous_sparse_rank_cost}
\end{equation}

First, every stable matching has $u=o(N)$ unmatched applicants.
Indeed, unmatched applicants and unmatched firms span no interview
edges. For $v=\lceil d^{-1/3}N\rceil$, the probability that there are
sets of $v$ applicants and $v$ firms with no edges between them is at
most
\[
 \binom Nv^2\exp(-dv^2/N)=o(1).
\]
Thus we may restrict attention to $u\le\rho_NN$, where $\rho_N\to0$.

Fix a partial matching with $u_s$ unmatched type-$s$ applicants,
$m_s=N_s-u_s$ matched type-$s$ applicants, and
$m=N-u=\sum_s m_s$. Put
\[
 c=\frac{d-1}{N-1},\qquad
 S_s=u_s+\sum_{\substack{i\text{ matched}\\i\text{ of type }s}}X_i,
 \qquad \lambda_s=cS_s.
\]
For each matched firm $j$, let $V_j$ be its score for its assigned
applicant, and write $\overline F_s=1-F_s$. Conditional on the assigned
edges, costs, and scores, sampling the other $d-1$ interview labels
without replacement gives the following upper bound for the probability
of no blocking pair:
\begin{equation}
 \exp\!\left[-\sum_s\lambda_s
       \left(u+\sum_{j\text{ matched}}\overline F_s(V_j)\right)
       +c\sum_{i\text{ matched}}X_i
                    \overline F_{s(i)}(V_{\phi(i)})\right].
 \label{eq:heterogeneous_matching_envelope}
\end{equation}
To see this, the arithmetic-mean bound for a product sampled without
replacement gives, for a matched applicant $i$,
\[
 \left[1-\frac{X_i}{N-1}
       \left(u+\sum_{j\ne\phi(i)}\overline F_{s(i)}(V_j)\right)
 \right]^{d-1}.
\]
Use $1-z\le e^{-z}$ and multiply over applicants. An unmatched
applicant has no assigned edge and its analogous exponent has
coefficient $d/N\ge c$; replacing it by $c$ only enlarges the bound.
The probability that all $m$ assigned edges occur is $(d/N)^m$.

For a vector $\boldsymbol\lambda\ge0$ define
\[
 x_s(\boldsymbol\lambda)
 =\int \exp\!\left[-\sum_t\lambda_t\overline F_t(v)\right]dF_s(v).
\]
The Stieltjes chain rule for continuous distribution functions gives
\begin{equation}
 \sum_s\lambda_s x_s(\boldsymbol\lambda)
 =1-\exp\!\left(-\sum_s\lambda_s\right)\le1.
 \label{eq:heterogeneous_type_integral}
\end{equation}
This identity is what accounts for competition between different types.

The diagonal correction in
\eqref{eq:heterogeneous_matching_envelope} has only a subexponential
cost, uniformly for all $d\le N$. Set $K=\sqrt N$.
For a type with $S_s\le K$, simply bound its correction by
$\exp(cS_s)\le e^K$. For a type with $S_s>K$, under the probability
measure proportional to
$\exp[-\sum_t\lambda_t\overline F_t(v)]dF_s(v)$, the variable
$U=\overline F_s(v)$ is stochastically dominated by an exponential
variable of rate $\lambda_s$. Indeed, after the uniform probability
transform its density is proportional to
$e^{-\lambda_s U}$ times a nonincreasing function of $U$.
Consequently, for every $0\le X_i\le1$,
\[
 \frac{\int e^{cX_i\overline F_s(v)}
             e^{-\sum_t\lambda_t\overline F_t(v)}dF_s(v)}
      {x_s(\boldsymbol\lambda)}
 \le\frac{\lambda_s}{\lambda_s-c}
 =\frac1{1-1/S_s}.
\]
Multiplying these corrections over the matched firms bounds them by
$\exp[rK+N/(K-1)]=\exp(O(\sqrt N))$.
After integrating the firms' assigned scores and discarding the factor
$e^{-u\sum_s\lambda_s}\le1$, the remaining factor is
$\prod_s x_s^{m_s}$.

Since $x_s\le1$, identity \eqref{eq:heterogeneous_type_integral} also
implies $\sum_s(\lambda_s+1/N)x_s\le1+r/N$. Weighted
arithmetic--geometric mean therefore yields
\[
 \prod_s x_s^{m_s}
 \le (1+r/N)^m
       \prod_s\left[\frac{m_s}{m(\lambda_s+1/N)}\right]^{m_s}.
\]
For the unmatched counts under consideration, every $m_s$ is positive
for all sufficiently large $N$, since $N_s=\Theta(N)$ and $u=o(N)$.
The small regularization $1/N$ makes the ensuing integrals finite,
including when $u_s=0$. Writing $z_s$ for the sum of the $m_s$
matched costs of type $s$, its simplex slice has volume at most
$z_s^{m_s-1}/(m_s-1)!$. Hence each cost integral is bounded by
\[
 \frac1{(m_s-1)!}
 \int_0^{m_s}\frac{z_s^{m_s-1}}
 {(u_s+z_s+1/(Nc))^{m_s}}\,dz_s
 \le\frac{\log(1+N^2)}{(m_s-1)!}.
\]
If $z_s\ge\varepsilon m_s$ for some fixed $\varepsilon>0$, the
same bound acquires an additional factor $e^{-b_\varepsilon m_s}$
for some $b_\varepsilon>0$, up to a fixed multiplicative constant.
For completeness, the ratio of the cube-constrained simplex slice to
its full slice is
$\Pr\{\max_i E_i/\sum_kE_k\le1/z_s\}$ for independent unit
exponentials $E_i$. On $\sum_kE_k\le2m_s$, this event requires
$\max_iE_i\le2/\varepsilon$; the probabilities of this event and
of $\sum_kE_k>2m_s$ are both exponentially small in $m_s$.

There are
$\frac{N!}{u!}\prod_s\binom{N_s}{u_s}$ partial matchings with the
specified unmatched type counts. Combining the preceding bounds,
Stirling's inequalities, and
$d/(Nc)=d(N-1)/(N(d-1))\le e^{1/(d-1)}$, gives the uniform
first-moment bound
\begin{equation}
 N^{O(r)}(\log(1+N^2))^r
 \exp\!\left[O(\sqrt N)+\frac{N}{d-1}
                  +2N h(u/N)\right],
 \label{eq:heterogeneous_partial_matching_count}
\end{equation}
where $h(x)=-x\log x-(1-x)\log(1-x)$, with $h(0)=0$.
In obtaining this expression, the factor $N!/u!$ and the simplex
factorials leave at most $e^{Nh(u/N)}$ up to a polynomial factor;
$\prod_s\binom{N_s}{u_s}\le\binom Nu$ supplies the other entropy
factor. If $S_s\ge\varepsilon N_s$ for any type, then $u=o(N)$
implies $z_s\ge(\varepsilon/2)m_s$ for all sufficiently large $N$.
The additional simplex factor is then $e^{-bN}$ for some $b>0$.
It dominates \eqref{eq:heterogeneous_partial_matching_count}, whose
exponent is $o(N)$ uniformly for $u\le\rho_NN$ and $d\to\infty$.
Summing over the polynomially many unmatched type counts proves
\eqref{eq:heterogeneous_sparse_rank_cost} by the first-moment bound.

\paragraph{From ranks to interim stability.}
Except for an expected $o(N)$ applicants, each row has at least
$p_0d/2$ nonnegative-shock neighbors and at most $p_0d/2$ auxiliary
costs below $p_0/4$; both assertions follow from binomial concentration.
Equation \eqref{eq:heterogeneous_sparse_rank_cost} implies that, in
every stable matching, only $o(N)$ applicants have matched cost
exceeding $p_0/4$. Equivalently, take the firm-optimal matching, which
gives every applicant its worst stable rank, to obtain a common
exceptional set. Every other applicant's partner is ranked no worse
than its $p_0d/2$th interview option. Because there are at least that
many nonnegative-shock signal neighbors, its partner is better than
every non-interviewed firm.

For $n_\Short\le n_\Long$, add independent applicants of the same
model until the market is balanced. Applicant-entry monotonicity for
the firm-optimal matching can only worsen each original applicant's
worst stable partner. The type counts in the enlarged market have
nonvanishing proportions with high probability. By exchangeability,
the expected exceptional fraction in the balanced market is $o(1)$;
therefore the expected number of exceptional original applicants is
$o(n_\Short)$, even when the original populations are not comparable.
Markov's inequality gives a common exceptional set of size
$o(n_\Short)$ with high probability. For each stable matching, delete
these applicants and their matched firms. Every retained agent keeps
its partner, and every retained applicant rejects all non-interviewed
firms. The resulting restriction is perfect interim stable. This
proves the sparse conclusion for every stable matching, throughout
the stated range $d\to\infty$.

\paragraph{Dense signals.}
Write $N=n_\Long$, $M=n_\Short\le N$, and $\beta=1-\alpha$,
where $\alpha\in(0,1)$ is fixed. Conditional on the applicant types,
write $F_g,F_b$ for the continuous CDFs of the firms' post-interview
idiosyncratic scores for the two types. The assumed stochastic order
implies $F_g\le F_b$. Scores on distinct ordered pairs are independent
conditional on types.

Each applicant's post-interview ordering of its $d$ signaled firms
is a uniform ordered list of distinct labels, independently across
applicants and independently of all firm-side scores. Extend each
list to a complete uniform permutation of the $N$ firms, and give
the additional firm--applicant pairs independent scores from the
appropriate type law. We first bound applicant ranks in the
applicant-optimal matching of this complete comparison market.

Fix one applicant $a$ and run deferred acceptance with all other
applicants first, inserting $a$ last. Every other applicant is matched
before this insertion. Whenever $a$ is free during the ensuing chain,
the other $M-1$ applicants are again matched. The next proposal of
any free applicant is uniform among its unvisited firms. Allowing
repeated draws and ignoring repetitions couples the total number of
proposals with the coupon-collector time for $N$ labels. Consequently
the total is at most
\[
 T_0=6N\log n
\]
except with probability at most $n^{-5}$.

When the inserted applicant is a bad interviewer, put $\varepsilon=\beta/(128\log n)$ and
call a revealed bad-type score high if its $F_b$-quantile exceeds
$1-\varepsilon$. A newly proposed pair has a fresh score, so the
number of high bad-type scores among the first $T_0$ proposals is
stochastically bounded by $\Binom(\lceil T_0\rceil,\varepsilon)$.
With probability $1-e^{-\Omega(\beta N)}$, this number is at most
$\beta N/8$. If $M>N/2$, a Chernoff bound also gives at least
$\beta M/2$ bad interviewers with probability
$1-e^{-\Omega(\beta N)}$. Whenever $a$ is free, the number of firms
that are either unmatched or hold a bad interviewer is therefore at
least
\[
 (N-M+1)+(\beta M/2-1)\ge\beta N/2.
\]
If $M\le N/2$, the unmatched firms alone give this bound.
Except for the at most $\beta N/8$ firms holding high bad-type
scores, these firms accept a fresh bad-type score above quantile
$1-\varepsilon$. Before $a$ makes
\[
 L=\left\lceil8192\,\beta^{-2}\log^2 n\right\rceil
\]
proposals, at most $L=o(\beta N)$ of these firms have been visited
by $a$. Thus, for all sufficiently large markets, its next proposal
is accepted with conditional probability at least
\[
 (\beta/4)\varepsilon=\frac{\beta^2}{512\log n}.
\]
The preceding estimates are used with stopping when a proposal-count
or high-score bound is first exceeded; no conditioning on a future
favorable event is needed.

When the inserted applicant is a good interviewer, use $\varepsilon=1/(128\log n)$.
For a newly revealed score of either type, its $F_g$-quantile exceeds
$1-\varepsilon$ with probability at most $\varepsilon$.
The same argument bounds the number of such high scores by $N/8$.
At least $N/2$ unvisited firms then either are unmatched or hold a
score that a fresh good-type score beats with probability
$\varepsilon$. The acceptance bound above, with its smaller factor
$\beta^2$, is therefore valid for good interviewers as well.

After an acceptance at firm $j$, the rejection chain reaches an
unmatched firm before returning to $j$ with probability at least
$1/2$. Indeed, an unmatched firm has never received a proposal,
whereas $j$ may already have been visited by the currently proposing
applicant. At every subsequent draw, the conditional probability of
hitting any one fixed still-unmatched firm is at least the probability
of hitting $j$. The first of these two events therefore favors the
unmatched firm with probability at least $1/2$. On that event the
chain terminates and $a$ keeps $j$.

Each new proposal by a free $a$ consequently has probability at least
$\beta^2/(1024\log n)$ of being its last proposal, up to the stated
stopping events. Formally, count a stopping event as an artificial
success. Acceptance followed by either termination or stopping before
the chain returns to the held firm still has at least this probability.
The geometric bound can therefore be combined with the unconditional
probability of the stopping events. Hence
\[
 \Pr\{\operatorname{rank}(a)>L\}
 \le n^{-8}+n^{-5}+e^{-\Omega(\beta N)}.
\]
A union bound over applicants and the order independence of
applicant-proposing deferred acceptance prove that every applicant's
complete-market rank is at most $L$ with probability $1-o(1)$.

Take the absolute constant $C$ in the statement sufficiently large
that $dp\ge4L$. The $d$ outgoing interview shocks of each applicant
are independent, with nonnegative probability $p$. A Chernoff bound
and a union bound give at least $L$ nonnegative shocks for every
applicant, with probability $1-o(1)$. Also $L<d$, so the complete-market
deferred-acceptance run never uses an appended edge and gives exactly
the applicant-optimal matching on $H$. {
Each applicant's partner is among its first
$L\le|\calN_+(a)|$ post-interview choices. By
\prettyref{lmm:positive_available_addition}, the applicant-optimal
matching on $H$ is perfect interim stable.
}%

}%
\subsubsection{Proof of \prettyref{rmk:correlated}}\label{sec:thm_correlated}

{

{
Write $M=n_\Short$, $N=n_\Long$ and $b=N-M\ge1$. Conditional on all
type labels, the directed pre-interview scores are independent and
identically distributed, and the directed interview shocks are
independent, with the specified type-dependent laws. The type labels
are independent of all pre-interview scores.

Fix an applicant $a$. Expose all other applicants' score rows and all
firms' score rows, and remove $a$ and its incident edges from $H$;
in particular, do not redirect any signal after this removal.
Let $U$ be the unmatched firms in the applicant-optimal matching of the
remaining market. Then $|U|\ge b+1$. Conditional on the exposed data
and the type labels, $U$ is independent of $a$'s outgoing signal labels
and its own interview shocks.
Choose a set $W\subseteq U$ of exactly $b+1$ firms using only the
exposed data, for example the $b+1$ firms with the smallest labels.

{
The deterministic entry comparison used in
\eqref{eq:degree_free_available_upper} shows that at most one member of
$U$ becomes matched after $a$ is reinserted. Every other member remains
unmatched in every stable matching of $H$. Thus two outgoing signal
neighbors in $W$ with nonnegative own shocks guarantee that, in every
stable matching, $a$ weakly prefers its partner to at least one such
neighbor. Its partner's post-interview utility is at least $a$'s $d$th
highest pre-interview utility, and hence strictly exceeds the
pre-interview utility of every non-interviewed firm. This certificate allows all incoming interview edges.
}%

The own-shock probabilities along $a$'s outgoing row may differ by
firm type, but all are at least $p$. Couple their nonnegative-shock
indicators with independent Bernoulli($p$) lower marks. Let $K_a$ be
the number of lower-marked outgoing neighbors in $W$. The $d$ signal
labels are a uniform sample without replacement. Conditional on
$W$, the arithmetic-mean bound for without-replacement products
therefore gives
\begin{align*}
 \Pr\{K_a\le1\mid W\}
 &\le\left(1-p\frac{b+1}{N}\right)^d
   +d\frac{p(b+1)}{N}
      \left(1-p\frac{b}{N-1}\right)^{d-1}.
\end{align*}
For the second term in the first line, specify the sole marked
position, condition on its label belonging to $W$, and apply the
same product bound to the remaining $d-1$ sampled labels. Put
$\lambda=pd\,b/N$. Since $d\le N$ and $p\le1$,
\[
 \frac{dp(b+1)}N\le\lambda+1,
 \qquad
 \frac{p(d-1)b}{N-1}\ge\lambda-1.
\]
Hence $\Pr\{K_a\le1\mid W\}\le e(\lambda+2)e^{-\lambda}$.

A union bound over the $M$ applicants yields the explicit error bound
\[
 \Pr\{\text{some stable matching of $H$ is not interim stable}\}
 \le eM(\lambda+2)e^{-\lambda}.
\]
In particular, the conclusion holds whenever
\[
 \lambda-\log M-\log(\lambda+2)\longrightarrow\infty.
\]
The simpler sufficient condition
$pd(1-M/N)\ge(1+\epsilon)\log n$, together with the stated
$d\ge(8+\epsilon)\log^2n/(\gamma p)$, makes the error
$O(n^{-\epsilon}\log n)$, because $M+N=n$ and
$(\lambda+2)e^{-\lambda}$ is decreasing for $\lambda\ge0$.
}%

}%

\section{Analysis for multi-tiered market}\label{sec:multi_proof}

{
We process the highest remaining applicant and firm tiers successively.
Earlier matches change the residual population sizes, so the induction
tracks both their cardinalities and the conditional score laws.
The one-sided reduced-market bounds are
\prettyref{thm:single_tier_general_sparse} and
\prettyref{thm:single_tier_general_dense}; balanced bilateral blocks use
\prettyref{rmk:single_tier_general_both_dense}. These results are proved
in \prettyref{sec:supp_reduced}.

}%

\subsection{Multi-tiered market with sparse signals}\label{sec:multi_sparse_proof}

\subsubsection{Proof of \prettyref{thm:multi_signal_sparse}}\label{sec:multi_signal_sparse}
{
\paragraph{The deterministic order of the tier pairs.}
Write $A_s=\sum_{u\ge s}|\Short_u|$ and
$J_t=\sum_{v\ge t}|\Long_v|$. The successive overlaps of the intervals
$(A_{s+1},A_s]$ and $(J_{t+1},J_t]$, read from zero upwards, give the
usual greedy order of the tier pairs. There are at most $m+\ell-1$
pairs. Since the market is generally imbalanced, their right endpoints
never coincide. In the pair whose applicant endpoint comes first,
$\Short_s$ signals $\Long_t$; with the endpoints reversed,
$\Long_t$ signals $\Short_s$. We call this the sending tier of the pair.
It has no interview edges to tiers below its target. For example, if
$\calT(\Short_s)=\Long_t$, then $A_s\le J_t<J_v$ for $v<t$,
so no such $\Long_v$ can target $\Short_s$. The analogous statement
holds with the sides interchanged.

Process these original tier pairs in this fixed order. In each pair,
match the currently surviving agents by any stable block matching,
remove its matched agents, and set aside the unmatched agents of the
sending tier. These set-aside agents have no edges to later tiers and
are unmatched in every stable matching of the full interview graph.
Inductively, every stable matching of the full graph has the same
earlier block-matched identities. A previously set-aside sender has no
edge to a later block. Thus every stable matching restricts to a stable
matching of the current block: a remaining partner outside the block is
in a lower tier and cannot prevent a blocking pair inside the block.
The Rural
Hospital Theorem therefore makes the two sets matched inside the block,
and consequently all subsequent surviving identities, independent of
which stable matching is being considered. Empty blocks are skipped.
This argument concerns these matched identities; it does not claim that
deleting unmatched agents preserves the entire set of stable matchings.

\paragraph{Conditional laws and population sizes.}
Expose the full score arrays of an original tier pair only when that
pair is processed, including its scores involving already removed agents.
All surviving identities are determined by earlier pairs, so the current
arrays retain their original independent law conditional on that history.
Restrict each original top-$d$ signal list to the surviving recipients.
The order above processes each original pair at most once.

Let $U_r$ be the total number of unmatched senders set aside through the
first $r$ blocks. The actual residual cardinalities differ from the
ideal greedy cardinalities by at most $C_r U_r$, for constants depending
only on the number of tiers. Indeed, matching a sending population of
size $M$ with $u$ unmatched agents removes $M-u$ recipients instead of
$M$; the residual error in the other tier changes by at most the sum of
the two previous errors and $u$. Thus, at the next nonempty block, with
sending and receiving populations of sizes $M$ and $K$,
\[
 M-K\le C_rU_r.
\]
As long as $U_r=o(n)$, the retained fraction of the original receiving
tier is bounded below by a fixed positive constant. To see this first
when $n_\Short/n_\Long$ is bounded below, at least one tier in every
greedy block is fresh and has size $\Omega(n)$. If the receiving tier
is not fresh, the sending tier is fresh, and
$K\ge M-C_rU_r=\Omega(n)$. If
$n_\Short/n_\Long\to0$, all applicant tiers target the highest firm
tier, no firm tier signals, and at most $n_\Short=o(n_\Long)$ members
of that receiving tier have been removed. Its retained fraction is then
$1-o(1)$. The same reasoning, with a fixed small threshold for the global
ratio, gives a uniform positive lower bound. Every relevant original
receiving population is $\Theta(n_\Long)$.

If $M<K$, add
independent fictitious senders, with the same original-target sampling
law, until there are $K$ senders. This comparison market is balanced and
has the same receiving population and retained fraction. Removing these
extra competitors weakly improves every real sender's worst stable
partner, by the receiver-side monotonicity of deferred acceptance.
Hence a certificate giving every stable partner of a real sender utility
at least a specified level in the comparison market is also valid in
the actual block. If $M\ge K$, no extra senders are needed. Thus \prettyref{prop:remove_gamma_d_omega_1} applies to the comparison market,
without requiring a nonvanishing retained sending fraction in the original block.

\paragraph{A finite induction with polynomial unmatched errors.}
Conditional on each preceding history satisfying
$U_r\le n d^{-c_r}$ for a fixed $c_r>0$, the next comparison block has
imbalance at most $1+C_rd^{-c_r}$. After reducing the positive exponent,
this is at most $1+d^{-\lambda_r}$ for a fixed $\lambda_r>0$.
The first block has imbalance at most one. Apply
\prettyref{prop:remove_gamma_d_omega_1} with this fixed exponent and the
receiving-fraction bound just proved. Its full-neighbor estimate and
the hypergeometric degree tail give, uniformly for every real sender,
\[
 \Pr\{\text{sender unmatched in the block}\mid\text{past}\}
 \le d^{-b_r}
\]
for a fixed $b_r>0$. Here a sender with an available neighbor is matched
in every stable matching. Summing and applying Markov's inequality
bounds the number newly set aside by $n d^{-b_r/2}$ with probability
$1-o(1)$. Decrease $c_{r+1}>0$ to absorb this term and the preceding
error. Since there are only finitely many blocks, the induction gives
$U_r\le n d^{-c_r}=o(n)$ throughout, with probability $1-o(1)$.

In the same conditional comparison market, the number of retained
neighbors with a nonnegative sender-side interview shock has mean
$\Theta(dp)$ and a hypergeometric--binomial lower tail. The selected set
is measurable in the root data. Thus the positive-neighbor version of
\prettyref{prop:remove_gamma_d_omega_1} gives
\[
 \Pr\{\text{no positive available neighbor}\mid\text{past}\}
 \le \exp[-c\lambda_r p\log d]+o(1)=o(1),
\]
because $p\log d\to\infty$. The utility certificate transfers from the
comparison market to the actual block as above. Let $D_r$ contain the
real senders without such a certificate. Conditional expectation and
Markov's inequality show $\sum_r|D_r|=o(n)$ with high probability.
When $n_\Short/n_\Long\to0$, all these senders are applicants and the
same argument gives $\sum_r|D_r|=o(n_\Short)$.

\paragraph{The original stable matchings.}
Put $D=\bigcup_rD_r$. This is a common set for all stable matchings of
$H$. Every noninterviewed interim blocking pair meets $D$. To check this,
use the earliest block at which either endpoint retires, where retirement
means either matching in that block or being set aside as an unmatched
sender. A set-aside endpoint is in $D$. If a matched sending endpoint is
not in $D$, its certificate makes it reject every noninterviewed member
of its target tier, and strict tier ordering makes it reject lower tiers.
An opposite endpoint in a higher tier either retired earlier, contrary
to the choice of block, or was set aside earlier and belongs to $D$.
If the first retired endpoint is a receiving agent, it rejects every
member of a lower original sending tier. An opposite endpoint in its
current sending tier also retires in this block and is covered by the
sending certificate or belongs to $D$. A higher-tier opposite endpoint
was matched earlier to a better tier or was set aside in $D$.
These cases also cover a pair with two unmatched endpoints.
Interviewed blocking pairs are excluded by stability.

For any stable matching $\Phi$ of the original graph, remove $D$ and
all partners $\phi(D)$ that are matched to members of $D$. Every retained
agent keeps its original partner, and no interim blocking pair remains.
At most $2|D|=o(n)$ agents are removed, on the same high-probability event
for every $\Phi$. This proves \prettyref{thm:multi_signal_sparse}.

}%
\subsubsection{Proof of \prettyref{rmk:multi_signal_sparse_identify}}\label{sec:multi_signal_sparse_identify}

{
\paragraph{A uniform empty-rectangle event.}
We first treat markets in which $n_\Short/n_\Long$ is bounded below.
Write $n=n_\Short+n_\Long$ and $R=m+\ell$; every original tier then has
size at least $c n$ for a fixed $c>0$. Put
\[
 t=(pd)^{-1/4},\qquad h=\lceil tn\rceil,
 \qquad b=\lceil t^2n\rceil,\qquad u=\lfloor h/4\rfloor.
\]
Here $pd\to\infty$, so $t\to0$, $b=o(h)$, and $h=o(n)$.
Call an outgoing signal edge positive when its sender-side shock is
nonnegative. With probability $1-o(1)$, simultaneously for every original
sending tier and its target tier, every set of $b$ senders and every set
of $u$ recipients have a positive signal edge between them.
Indeed, for prescribed sets $S,V$ of these sizes, and an original target
population of size $N$, independent sender rows and sampling without
replacement give
\[
 \Pr\{E_+(S,V)=\varnothing\}
 \le (1-p|V|/N)^{d|S|}
 \le \exp(-pd\,bu/N).
\]
There are at most $R4^n$ pairs of subsets to consider, whereas
$pd\,bu/N=\Omega((pd)^{1/4}n)$. A union bound proves the assertion.
In particular, this event holds for subsets subsequently selected using
any of the graph, preference, or matching outcomes.

\paragraph{A small reserve separating the cumulative counts.}
Before examining the realized graph, delete $O(h)$ agents so that the
strict order of all cumulative applicant and firm counts is preserved
and consecutive retained counts differ by at least $h$.
Here is an explicit construction. List the $R$ positive cumulative counts
in increasing order. Cluster consecutive counts whose gap is at most
$10Rh$. For sufficiently large $n$, every cluster has at most two members:
three consecutive members would contain two counts from the same side
at distance at most $20Rh<cn$. In cluster $k$, prescribe cumulative
deletions $2kh$ at each count, adding another $h$ at the earlier member
when the cluster has two members. These prescribed deletions increase
along each side's cumulative counts, and their successive differences
are feasible tier deletion counts. Within a two-member cluster the
retained gap increases by $h$; between clusters the original gap exceeds
$10Rh$, while the total prescribed deletion is at most $(2R+1)h$.
Thus all retained gaps are at least $h$, and all original target tiers
and signaling directions remain unchanged. Choose the deleted labels
in advance. Keep the original signal lists, without refilling them.

Run the fixed block decomposition from
\prettyref{sec:multi_signal_sparse} on this reserved market. Write $U_r$
for the unmatched recipients in block $r$. Every ideal block has a
recipient surplus at least $h$. Suppose inductively that every preceding
block has fewer than $b$ unmatched senders. The accumulated discrepancy
from the ideal counts is then at most $Rb=o(h)$, so the current block
has recipient surplus at least $3h/4$ and hence
\[
 |U_r|\ge3h/4.
\]
If this block had $b$ unmatched senders, those senders and any $u$
unmatched recipients would have a positive edge on the empty-rectangle
event, contradicting stability. This proves the inductive claim,
including its initial case.

\paragraph{Pruning without losing the reserve.}
For a sender $i$, let $\theta_i$ be its original $d$-th highest
pre-interview utility in its target tier. At each stage recompute the
block decomposition of the current induced graph. If a currently
surviving sender has utility below $\theta_i$ in its worst stable block
outcome, delete that sender and repeat. An unmatched sender is included
among these agents. The worst outcomes are obtained by deferred
acceptance with the block's recipients proposing. The procedure uses the realized graph and scores; no independence
of the retained labels is assumed.

We bound the number of deletions as follows. In a fixed bipartite market, deleting one vertex changes the
matched status of at most one other vertex on either side. To see this,
use \prettyref{lmm:truncation} with a fixed deferred-acceptance
orientation and the Rural Hospital Theorem \citep{roth1986allocation}: remaining agents on the deleted side
weakly improve, agents on the other side weakly worsen, and the number
of matched pairs decreases by at most one. The same bound, with the
roles reversed, applies when a vertex is added. In particular, changing
$k$ input vertices changes the unmatched recipient set by at most $3k$
labels; this counts both agents entering and agents leaving the block.
In successive tier blocks, one tier is fresh and the other is the
preceding block's unmatched recipient set. After one global deletion,
the fresh tier changes by at most one label. If $x_r$ denotes the total
input symmetric difference in block $r$, then $x_1\le1$ and
$x_{r+1}\le3x_r+1$. Induction gives
a constant $C_R$, depending only on the number of tiers, such that
one global vertex deletion changes each block's unmatched recipient
set by at most $C_R$ labels. For example $C_R=3^{R+1}$ suffices. All
comparisons keep the original preferences and interview edges.

Consider the first $q\le Rb$ pruning steps. Cumulative-count gaps stay
positive, so the same target pairs and block order apply. For block $r$,
remove from its initial set $U_r$ every label that leaves that block's
unmatched recipient set at any of these steps. The resulting common
reservoir $V_r$ is present and unmatched at every preceding stage, and
\[
 |V_r|\ge 3h/4-C_Rq\ge h/2>u
\]
for all sufficiently large $n$. Each sender deleted from block $r$ has
no positive signal edge to $V_r$: at its deletion stage, such an edge
would give it utility at least $\theta_i$ and would block its worst
stable block matching, because its endpoint in $V_r$ is unmatched.
The empty-rectangle event therefore prevents $b$ senders from ever being
deleted from the same block. There are fewer than $R$ blocks. By the
pigeonhole principle the pruning must terminate before $Rb$ deletions.
Including the preliminary reserve, the total number removed is $O(h)=o(n)$,
which is $o(n_\Short)$ on the applicant side and $o(n_\Long)$ on the firm
side in the present population range.

At termination every surviving sender receives utility at least
$\theta_i$ in every stable outcome of its block. Consequently it rejects
every noninterviewed member of its target tier, and strict tier ordering
makes it reject all lower tiers. The deterministic block decomposition
and earliest-retirement argument in
\prettyref{sec:multi_signal_sparse} now apply to every stable matching
of this \emph{one} final induced graph. There are no exceptional senders,
so no noninterviewed interim blocking pair remains. Stability already
excludes interviewed blocking pairs. This proves the claimed simultaneous
conclusion in the comparable-population case.

\paragraph{A vanishing applicant population.}
For completeness, suppose $n_\Short\le |\Long_\ell|/2$. Then all applicants
signal the highest firm tier and no firm tier signals. For a fixed applicant, remove it and let $U$ be the unmatched highest-tier firms in the firm-optimal matching of the remaining graph. The set $U$ is independent of that applicant's outgoing row and satisfies $|U|\ge|\Long_\ell|-n_\Short$. The entry comparison used in \eqref{eq:degree_free_available_upper} shows that two positive edges into $U$ guarantee its original pre-interview threshold in every stable matching after reinsertion. The probability of fewer than two such edges is at most
\[
 (1-p/2)^d+dp(1-p/2)^{d-1}=o(1).
\]
The uniform bound and Markov's inequality give
$o(n_\Short)$ exceptional applicants. Delete only these applicants;
same-side deletion improves every remaining applicant's worst stable
partner, exactly as in
\prettyref{sec:single_signal_sparse_identify}. Every remaining applicant
therefore rejects every noninterviewed firm in every stable matching of
the reduced graph. No firms need be removed. In all other markets,
$n_\Short>|\Long_\ell|/2=\Omega(n_\Long)$, so the preceding construction
applies with uniform constants. The two cases complete the proof.
}%

\subsubsection{Proof of \prettyref{rmk:restricted_sparse}}\label{sec:restricted_sparse}
{

{
Use the proof of \prettyref{thm:multi_signal_sparse} in \prettyref{sec:multi_signal_sparse} with the same deterministic order of cumulative endpoints.
}%
 When two endpoints coincide, the two tiers are still mutual
targets; the restricted mechanism permits only one of them to signal.
Designate that tier as the sender in this block. Its actual sending to
receiving ratio is at most $1+O(U_r/n)$, so the same conditional
availability estimate applies. After matching the block, set aside the
unmatched agents on both sides and advance both tiers. Neither tier has
interview edges below this target. The number additionally set aside on
the receiving side is at most its population surplus plus the number of
unmatched senders, hence is $O(U_r)+n d^{-b_r/2}$ on the preceding good
event. The finite polynomial-error induction is unchanged. Add all these
unmatched agents to the common exceptional set $D$. The same blocking-pair
cover and partner-preserving deletion prove
\prettyref{rmk:restricted_sparse}. This establishes almost interim
stability of every original stable matching; it does not infer a claim
about every new stable matching after arbitrary subsequent deletions.

}%
\subsubsection{Proof of \prettyref{rmk:relaxed_bounded_sparse}}\label{sec:relaxed_bounded_sparse}

{
We first consider $n_\Short/n_\Long$ bounded below, so all original tiers have size $\Omega(n)$. \paragraph{Uniform bounds over surviving sets.}
Put $\theta=q(p\wedge q)$ and suppose that the market is
$\gamma$-generally imbalanced, with
\[
 L_n:=\frac{\gamma dq}{\log d}\longrightarrow\infty.
\]
The assumptions $p\log d,q\log d\to\infty$ then give
$d\theta\gamma=L_n[(p\wedge q)\log d]\to\infty$.
For a signal from $s$ to $r$, call the edge positive-certified when
$A_{s,r}\ge0$ and $A_{r,s}+B_{r,s}>h$, and call it high-certified when
both endpoint scores exceed $h=M_\DistA+M_\DistB-1$.
Write $R$ for the original receiving tier of a one-sided target pair.
For any prescribed sender set $U$ and recipient set $V\subseteq R$ in
that pair, the probability that no positive-certified edge
joins $U$ to $V$ is at most
\[
 \exp\{-dpq\,|U||V|/|R|\}.
\]
The corresponding bound for high-certified edges replaces $pq$ by
$q^2$. For positive certification this follows by exposing independent
uniform signal labels and independent shocks. For high certification,
the ordered top-$d$ pre-interview scores dominate an independent sample
of $d$ pre-interview scores; couple the sender shocks in this order.
The receiving scores and the uniform label permutation are independent
of these sender scores. Sampling without replacement gives the
same exponential upper bound.

Use a third certification, requiring only
$A_{r,s}+B_{r,s}>h$ at the receiving endpoint; its no-edge probability
is bounded by $\exp\{-dq|U||V|/|R|\}$.
Set
\[
 \eta=\frac{\gamma}{\sqrt{L_n}},\qquad
 \zeta=(d\theta\gamma)^{-1/2},\qquad
 b=\lceil\eta n\rceil,\qquad v=\lfloor\gamma n/2\rfloor.
\]
Then $\eta=o(\gamma)$ and $\zeta=o(1)$. Since $d$ is at most
polylogarithmic in $n$, both $\eta n$ and $\gamma n$ tend to infinity;
indeed $\eta\ge\sqrt{L_n}\log d/d$.
For the receiving-high condition, count only rectangles of sizes
$b$ and $v$. The logarithm of their number is at most
\[
 \log\left[\binom{n}{b}\binom{n}{v}\right]
 =O\bigl(n[\eta\log(e/\eta)+\gamma\log(e/\gamma)]\bigr)
 =O(\gamma n\log d),
\]
where $\gamma\ge L_n\log d/d$ and $L_n\le d/\log d$ justify
the last bound. The negative logarithm of the no-edge bound is at
least $c dq\eta\gamma n=c\gamma\sqrt{L_n}(\log d)n$, which
dominates this entropy and tends to infinity. Thus, with probability
$1-o(1)$, every sender set of size at least $\eta n$ and every recipient
set of size at least $\gamma n/2$ have a receiving-high edge: any larger
empty rectangle would contain one of the counted sizes. Integer
rounding only changes the displayed bounds by constant factors.

For either stronger certification, count sender sets of size
$\lceil\zeta n\rceil$ and recipient sets of size $v$. The no-edge
bound is at most $\exp\{-c d\theta\zeta\gamma n\}
=\exp\{-c n\sqrt{d\theta\gamma}\}$. A union bound over at most
$4^n$ pairs of subsets now suffices. Consequently every sender set
of size at least $\zeta n$ and every recipient set of size at least
$\gamma n/2$ have an edge of each stronger type, with probability
$1-o(1)$. Taking the intersection over the finitely many target pairs
gives an event uniform over sets chosen after the entire market is
observed.

Fix any stable matching $\Phi$ of the full interview graph. Process the
original cumulative-count staircase along $\Phi$: remove the current
sending tier and its matched partners, and set aside unmatched senders.
A current sender has no edges to later recipient tiers, independently
of any preference-ordering assumption. Every matched current sender
therefore has its partner in the current target tier. If earlier
unmatched errors total $O(\eta n)$, the current surplus of recipients
is at least $\gamma n-O(\eta n)\ge\gamma n/2$.
Choose as a reservoir the current recipients not matched to current
senders. Their partners, if any, belong to lower sender tiers.

Every receiver-high edge to this reservoir is preferred by its
receiving endpoint to its partner. First, any currently unmatched
sender accepts every such edge, so stability and the first rectangle
condition bound unmatched senders by $\eta n$. This closes the finite
induction on the residual surplus, since $\eta=o(\gamma)$.

Now apply the two stronger rectangle conditions to the \emph{entire
original sending tier}, including agents already matched to higher
tiers in earlier steps. A sender whose current utility lies below its
original $d$th pre-interview utility in its target tier has no
positive-certified edge to the reservoir. A sender whose current
utility is at most its target-tier value plus $h$ has no high-certified
edge to it. Otherwise the certified edge would block $\Phi$.
Each exceptional set has fewer than $\zeta n$ agents.

For any applicant--firm tier pair, cumulative-count dominance implies
that at least one endpoint's target tier is the other endpoint's tier
or a higher tier. If that endpoint has both certificates, it rejects
a noninterviewed partner in its target tier by the original top-$d$
threshold, and rejects a partner in a lower tier because
$M_\DistA\ge0$ and its high-score certificate exceeds the maximum
pre-interview utility in every lower tier. Thus every noninterviewed
blocking pair meets the union $D_\Phi$ of the exceptional sending
agents. The number of tiers is fixed, so $|D_\Phi|=O(\zeta n)=o(n)$.
Delete $D_\Phi$ and its $\Phi$-partners. Every surviving partner is
unchanged and every interim blocking pair is removed. The event was
uniform over $\Phi$, so every original stable matching is almost interim
stable simultaneously. This does not assert a common deletion set
working for all newly created stable matchings.

If instead $n_\Short\le|\Long_\ell|/2$, all applicants target the highest firm tier and no firm signals. The unmatched-reservoir argument below applies with $\varepsilon=1/2$, yielding $o(n_\Short)$ exceptional applicants. Deleting these applicants and their partners proves the same almost-stability conclusion. In the remaining population range $n_\Short>|\Long_\ell|/2$, the preceding comparable-population argument applies uniformly.
\paragraph{All signals target one firm tier.}
Suppose that all applicants target the highest firm tier $R=\Long_\ell$
and that $n_\Short\le(1-\varepsilon)|R|$ for a fixed $\varepsilon>0$.
No firm sends a signal. Fix an applicant $a$, remove it, and let $U$ be the unmatched firms in $R$ under the firm-optimal matching of the remaining graph. The set $U$ is independent of $a$'s complete score row and original top-$d$ list, and satisfies $|U|\ge\varepsilon|R|$. The entry comparison used in \eqref{eq:degree_free_available_upper} leaves at most one member of $U$ matched after reinsertion. Every stable partner of $a$ is therefore at least as good as its second-best signaled member of $U$, since otherwise an unmatched member would block.

Its count of signaled members of $U$ having $A_{a,j}\ge0$ stochastically
dominates a binomial thinning of a hypergeometric sample of mean
$\varepsilon dp$. Its count with $A_{a,j}+B_{a,j}>h$ has the analogous
lower comparison with mean $\varepsilon dq$: the top-$d$ pre-interview
scores dominate an independent sample of $d$ pre-interview scores, and
uniform signal labels are independent of their ordered values.
Writing $r=p$ or $q$, the probability of at most one witness is at most
\[
 (1-\varepsilon r)^d
 +dr(1-\varepsilon r/2)^{d-1}
 \le (1+d r)\exp[-\varepsilon r(d-1)/2].
\]
Thus $d(p\wedge q)\to\infty$ gives an $o(1)$ failure probability for
an applicant, and $d(p\wedge q)\ge C_\varepsilon\log n$ gives
$o(n^{-1})$ uniformly, for a sufficiently large fixed constant.
The two witnesses make the applicant reject every noninterviewed firm
in $R$ and every lower firm tier. Summation and Markov's inequality prove
almost interim stability, after deleting exceptional applicants and
their partners. A union bound in the dense case proves perfect interim
stability of every stable matching. No preference-ordering assumption,
conditioning on adaptively surviving labels, or extra factor $q$ is
used here.

}%

\subsection{Multi-tiered market with dense signals}\label{sec:multi_dense_proof}
\subsubsection{Proof of \prettyref{thm:multi_signal_dense}}\label{sec:multi_signal_dense}

{
Define the cumulative populations
$A_s=\sum_{u\ge s}|\Short_u|$ and
$J_\kappa=\sum_{v\ge\kappa}|\Long_v|$, with
$A_{m+1}=J_{\ell+1}=0$. Write $t_s,t'_\kappa$ for the indices
of the nonempty target tiers of $\Short_s,\Long_\kappa$.
The greedy tier-pair order of \prettyref{sec:multi_signal_sparse}
has, at a targeted pair $(s,\kappa)$, ideal residual sizes
\[
 e_{s\kappa}=\max\{A_{s+1},J_{\kappa+1}\},\qquad
 a_{s\kappa}=A_s-e_{s\kappa},\qquad
 j_{s\kappa}=J_\kappa-e_{s\kappa}.
\]
These sizes are positive for a nonempty target pair. Define its
parameters in each signaling orientation by
\[
 \begin{array}{lll}
 \calT(\Short_s)=\Long_\kappa:
 &\displaystyle\gamma_s=\frac{j_{s\kappa}}{|\Long_\kappa|},
 &\displaystyle\delta_s=\frac{a_{s\kappa}}{j_{s\kappa}},\\[6pt]
 \calT(\Long_\kappa)=\Short_s:
 &\displaystyle\gamma'_\kappa=\frac{a_{s\kappa}}{|\Short_s|},
 &\displaystyle\delta'_\kappa=\frac{j_{s\kappa}}{a_{s\kappa}}.
 \end{array}
\]
Thus $\gamma_s,\gamma'_\kappa$ are retained receiving fractions
relative to the original target populations, and the corresponding
$\delta$ is at most one. In a mutual-target pair $A_s=J_\kappa$, so
both residual sizes equal
$L_{s\kappa}=\min\{|\Short_s|,|\Long_\kappa|\}$ and both ratios
equal one; neither original tier need be wholly retained.
The ideal sizes equal the actual residual sizes on preceding dense
success events, as proved below. Signal lists and cutoffs always use
the original populations. {Tiers with empty targets contribute no
parameters. Fix $\epsilon>0$ and set}

}%
{

{
\begin{align}
    C_{p,\boldsymbol{\alpha},\boldsymbol{\beta},\delta} 
    & = \max\Bigg\{\max_{\ess:\,\calT(\Short_\ess)\ne\emptyset}
    \left\{\frac{16+\epsilon}{\gamma_\ess p},
    -\frac{8+\epsilon}{\gamma_\ess\delta_\ess p}
    \log\left(1-\delta_\ess+\frac{\delta_\ess^2}{\gamma_\ess\beta_{t_\ess}n_\Long}\right)\right\},\nonumber\\
    &\qquad\max_{\kappa:\,\calT(\Long_\kappa)\ne\emptyset}
    \left\{\frac{16+\epsilon}{\gamma_\kappa' p},
    -\frac{8+\epsilon}{\gamma_\kappa'\delta_\kappa' p}
    \log\left(1-\delta_\kappa'+\frac{(\delta_\kappa')^2}{\gamma_\kappa'\alpha_{t_\kappa'}n_\Short}\right)\right\}\Bigg\}.
    \label{eq:C_alpha_beta}
\end{align}
}%

}%

{
{
In \prettyref{eq:C_alpha_beta}, an inner maximum with an empty index set is omitted.
}%

Fix any feasible $d\ge\ceil{C_{p,\boldsymbol{\alpha},\boldsymbol{\beta},\delta}\log n}$.
The bounds below apply separately at each such $d$. By
\prettyref{eq:C_alpha_beta} and the population checks below, the sufficient
threshold $\underline d$ can be chosen of order $\log^2 n$ if the market
is not generally imbalanced, and of order $\log(1/\gamma)\log n$ in a
$\gamma$-generally imbalanced market.

}%
{
{
\paragraph{Conditional exposure and a fixed choice of optimal matching.}
Let $A_*,J_*$ be the highest remaining tiers. Strict tier ordering
implies that every stable matching restricts to a stable matching of
$H[A_*\cup J_*]$. Its matched-agent sets are independent of the
chosen stable matching by the Rural Hospital Theorem
\citep[Lemma~1]{roth1986allocation}, applied in both orientations.
Conversely, any stable block matching can be united with a stable
matching after its matched agents are removed: a crossing pair is
ruled out either by strict tier ordering or by block stability when
the residual agent was unmatched in the block. The same argument
includes noninterviewed pairs if both matchings are interim stable.
The stable outcomes therefore decompose over a common residual
population. Choosing the applicant-optimal outcome in both parts
gives the applicant-optimal union; the firm-optimal statement is
analogous. Fix that orientation throughout the induction.
}%
Expose the complete score arrays of an original applicant-tier/firm-tier
pair only when that pair is processed. These arrays
include scores involving any previously removed agents of the two
original tiers. Before processing a new pair, the surviving identities
are measurable with respect to the arrays of previously processed
pairs. Independence of the score arrays across different tier pairs
therefore leaves the new pair's arrays unaltered conditionally on this
history. Target tiers are determined by population counts, so the
unexposed signals into the new pair are generated from the original
target-tier populations. Restricting them to
the surviving identities gives precisely the reduced-market law;
one does not resample $d$ signals from the smaller population.

On the event that each processed block is perfect interim stable, all
agents on its smaller side are matched, because otherwise an unmatched
pair would be an interim blocking pair. Hence at least one current tier
is exhausted, and a new original tier pair is processed at the next
step. On the preceding success events, the residual cardinalities are
the corresponding deterministic population-count differences, and
exchangeability makes the conditional new-block law depend only on
these cardinalities, not on the surviving labels. Stop the exposure at
the first block for which the required local conclusion fails. If each applicable conditional local failure
probability is at most $\varepsilon_n=o(1)$ uniformly over the preceding
histories, a union bound over at most $m+\ell-1$ blocks is valid.

{
The retained-population and degree hypotheses are checked below,
including blocks with a vanishing sending fraction.

}%
 This exposure argument does not apply
to a vertex subset selected using the same tier pair's preferences,
nor when the strict ordering of tiers is relaxed.

}%
\paragraph{{Case $1$: the market is generally imbalanced.}}

{
At every block of the deterministic tier-pair order, exactly the smaller
residual side signals. Condition on the score arrays of all preceding
blocks, stopping if any earlier block fails to be perfect interim stable.
The surviving populations then have deterministic cardinalities and the
fresh reduced-market law established above.

\paragraph{Population checks in the conditional dense induction.}
Let $c_*>0$ be a lower bound for the sizes of all original tiers divided
by $n=n_\Short+n_\Long$. Such a bound follows from
Assumption~\ref{assump:non_vanishing} and $n_\Short/n_\Long=\Omega(1)$.
At each greedy tier pair at least one current tier is a fresh original
tier. Write $M\le K$ for the current sending and receiving sizes, and
$N$ for the original receiving-tier size. If the receiving tier is
fresh, $K/N=1$. Otherwise the sending tier is fresh, so
$K\ge M\ge c_*n$, and again $K/N\ge c_*$. Consequently all receiving
fractions $\gamma_\ess,\gamma'_\kappa$ in
\eqref{eq:C_alpha_beta} are bounded below, even when the surviving
fraction of a sending tier tends to zero.

If the local ratio $\rho=M/K$ is bounded below, the two retained
populations satisfy all population hypotheses of
\prettyref{thm:single_tier_general_dense}, conditional on the preceding
original tier pairs. For the remaining small-ratio cases, say
$\rho\le1/2$, use instead the reduced-market version of the direct
unmatched-recipient argument in \eqref{eq:degree_free_available_upper}.
Write $\gamma_2=K/N$, $\gamma_2^-=(K-1)/(N-1)$ and
$\mu=\gamma_2dp$. For each real sender it gives
\[
\begin{aligned}
 \Pr\{\text{no positive available neighbor}\mid\text{past}\}
 &\le (1-\gamma_2p)^d
       +d\gamma_2p[1-\gamma_2^-p(1-\rho)]^{d-1}\\
 &\le e^{-\mu}+\mu e^{-(1/2-o(1))\mu}.
\end{aligned}
\]
Here $K=\Omega(n)$, so $\gamma_2^-/\gamma_2=1+O(1/n)$.
The budget in \eqref{eq:C_alpha_beta} implies
$\mu\ge(16+\epsilon)\log n$, making this bound $o(n^{-2})$
uniformly in $M$, including bounded $M$. Thus the conditional union bound applies without a lower bound on the
surviving sending fraction. The conclusion here concerns one-sided signaling
blocks.

In a $\gamma$-generally imbalanced market the gap in every such block
is a nonzero cumulative-count difference. Therefore
$1-\rho=(K-M)/K\ge\gamma n/K\ge\gamma$.
Since $K=\Omega(n)$ and the receiving fractions are bounded below,
the maximum in \eqref{eq:C_alpha_beta} is
$O(\log(1/\gamma)/p)$. The constants depend only on the fixed original
tier-size bounds and the global population-ratio bound. 
{
When cumulative
counts coincide, the local gap is zero and the logarithmic term has
order $\log n$. The mutual-target parameters specified before
\eqref{eq:C_alpha_beta} retain any depleted receiving fraction; the
prescribed-endpoint argument for these balanced blocks is given in
Case~$2$.
}%

Apply the reduced one-sided theorem, or the small-ratio estimate just
proved, to the present block. The original-target sampling budgets are
exactly those checked by \eqref{eq:C_alpha_beta}. This gives conditional
probability $1-o(1)$ that every stable block matching is perfect interim
stable, uniformly over preceding successful histories. At least one
current tier is then exhausted. Iterate through at most $m+\ell-1$
blocks and use a union bound. The deterministic decomposition then gives perfect interim stability
of every stable matching of the original graph on this common event.
All signal lists remain those of the original populations.

}%
\paragraph{{Case $2$: the market is not generally imbalanced.}}

{
Use the same conditional exposure and deterministic decomposition,
and assume the shock-distribution hypothesis in the first case of
\prettyref{thm:multi_signal_dense}. Fix applicant-proposing DA for the
entire market. Every imbalanced block is covered by the one-sided
argument in Case~$1$, whose conclusion holds for every stable matching.

At a balanced mutual-target block, at least one tier is fresh. Since
the cumulative counts at the two target tiers coincide, the common
residual size is
$L_{\ess\kappa}=\min\{|\Short_\ess|,|\Long_\kappa|\}$.
The receiving fractions for applicant-signaling and firm-signaling
are therefore respectively $L_{\ess\kappa}/|\Long_\kappa|$ and
$L_{\ess\kappa}/|\Short_\ess|$, exactly as specified before
\eqref{eq:C_alpha_beta}. Both are bounded below by the original-tier
size lower bound $c_*$, and the local ratio is $1$ in either
orientation. Hence the maximum in \eqref{eq:C_alpha_beta}, multiplied
by $\log n$, implies \eqref{eq:d_needed} separately for both
orientations; the population logarithms in that equation are at most
$\log n$. Its order remains $\log^2 n/p$.

{
Apply \prettyref{rmk:single_tier_general_both_dense} to this fresh
balanced block, using its log-concave or upper-atom branch as appropriate.
The effective source degrees are $d_A=dL_{\ess\kappa}/|\Long_\kappa|$
and $d_J=dL_{\ess\kappa}/|\Short_\ess|$; the preceding budget checks
are exactly \eqref{eq:lc_reduced_budget} in the two orientations.
The local applicant-optimal conclusion is uniform over preceding
successful histories by the conditional exposure above. The
fixed-orientation decomposition and a union bound over at most
$m+\ell-1$ blocks prove perfect interim stability of the full
applicant-optimal matching.
}%

Run the same argument with firm-proposing DA fixed throughout. At
balanced blocks the interchanged budget just checked gives the
firm-optimal local conclusion, while imbalanced blocks again require
no change. A second union bound proves that both prescribed optimal
matchings are perfect interim stable simultaneously with high
probability. No proposing side is changed during either induction.
}%

\subsubsection{Proof of \prettyref{rmk:relaxed_bounded_dense}}\label{sec:relaxed_bounded_dense}

{
The unmatched-reservoir argument in \prettyref{sec:relaxed_bounded_sparse} applies to every applicant with the fixed surplus $\varepsilon$. Its two witness failure probabilities are $o(n^{-1})$ under the displayed budget. A union bound therefore gives both utility certificates for every applicant in every stable matching. Each applicant rejects every noninterviewed firm, 
{
proving the remark.
}%

\paragraph{{A counterexample beyond the surplus-tier setting.}}

Let $L=\log N$, let $B$ be uniform on $[-1/10,1/10]$, and take
\[
 \Pr(A=1)=p=L^{-6},\qquad \Pr(A=-1)=1-p,
 \qquad d=\lceil L^{10}\rceil.
\]
Then $q=p$, $dp\asymp L^4$, $dp^2\to0$, and
$d\gg\log^2N/(p\wedge q)$. First use two applicant tiers of sizes
$N,N$ and two firm tiers of sizes $N+1,N-1$, listed from high to low.
Their target graph is
$\Short_H\to\Long_H\to\Short_L\leftrightarrow\Long_L$.

In firm-proposing deferred acceptance, first process the positive
firm-side shock edges between the two high tiers. These are a
$p$-thinning of the high applicants' signal lists. Conditional on
retention their ordinal preferences are independent and uniform.
Every high applicant has at least $k=\lfloor dp/2\rfloor$ retained
neighbors with probability $1-o(1)$. Its first $k$ retained choices are
independent uniform ordered lists. The proposer saturation bound in
\prettyref{prop:perfect_stable}, followed by continuation of these
initial segments and the Rural Hospital Theorem, shows that all $N$
high applicants are matched, since $k\gg\log^2N$.
Exactly one high firm $j$ remains free. Its identity is independent of
all high firms' score rows toward the low applicants.

Let $a_0$ be $j$'s first outgoing low-applicant choice with positive own
shock. Its label is uniform and independent of the low-firm market.
The probability that $a_0$ receives any low-firm outgoing edge with
\emph{both} shocks positive is at most $dp^2=o(1)$.
Except on this event, $a_0$ prefers $j$ to every positive-shock outgoing
proposal of a low firm. Engage $(j,a_0)$ next. On deleting this pair,
the $N-1$ low firms have independent uniform positive-outgoing lists on
$N-1$ low applicants; their first $k-1$ choices suffice for a perfect
proposer-optimal matching by the same bound. The removal label $a_0$
is independent of these rows; possible failure of the preceding
condition on $a_0$ is charged separately, without conditioning the
random-list assertion on that condition.
This is a valid full deferred-acceptance run: every firm's accepted
edge has positive own shock; low firms rank every incoming-only
positive edge below their accepted outgoing top-$d$ choice; and
negative own-shock edges are below these positive accepted edges.
Thus this is the firm-optimal matching on the full interview graph.

Only $o(N)$ high applicants can have positive own shocks in it, because
the expected number of doubly-positive high-tier edges is
$Ndp^2=o(N)$. A simultaneous degree bound gives $O(d)=o(N)$ high-tier
interviews for every firm, including $j$.
The pre-interview rank of $a_0$ in $j$'s low-applicant row is geometric
with parameter $p$, truncated at $d$. With probability $1-o(1)$ this
rank is at least $p^{-1/2}$, and uniform order statistics imply
\[
 B_{j,a_0}\le 1/10-c/(N\sqrt p)
\]
for a fixed $c>0$. Conditional on the exposed data, the scores
$B_{j,a}$ toward the $N-o(N)$ noninterviewed high applicants with
negative own shocks remain independent uniforms. Their maximum exceeds
$B_{j,a_0}$ with probability $1-o(1)$. Such an applicant $a$ and $j$
block: $a$ prefers a noninterviewed high firm to its negative-shock
match, while
$v_{\Short_H}+B_{j,a}>v_{\Short_L}+1+B_{j,a_0}=U^H_{j,a_0}$.

Finally, in descending tier order take applicant sizes
$(N,N,N+1,N-1)$ and firm sizes $(N+1,N-1,N,N)$.
The upper two tiers give the counterexample above and the lower two
its applicant--firm reversal. Cumulative counts agree at the boundary,
so the interview graph has no edges between the components. The upper
component blocks the global firm-optimal outcome and the lower one
blocks the global applicant-optimal outcome, with probability $1-o(1)$.

{
All tier proportions are bounded below, and $d\gg\log^2n/(p\wedge q)$.
}%

}%

\section{Analysis for incentive compatibility}\label{sec:incentive_proof}

{
\subsection{Proof of \prettyref{thm:incentive}}\label{sec:thm_incentive}
\paragraph{Exact best response on the proposing side.}

{
Use the signaling game and bounded-payoff convention of
\prettyref{sec:incentive}. Fix applicant-proposing deferred acceptance;
reversing the roles gives the firm-proposing claim. We prove the exact
applicant best response, allowing either one-side or both-side signaling.
The same argument applies to choices within a prescribed target tier,
without asserting that the target tier itself is optimal.

}%
To prove the assertion, fix an applicant $a$ and sample all potential interview shocks before comparing strategies. Keep all other agents' signals and preferences fixed. Consider the auxiliary market in which $a$ may list any firm; a firm's value for $a$ in this auxiliary market is its potential post-interview value, and all preferences and edges not incident to $a$ are those generated by the other agents' signals.

By strategy-proofness of applicant-proposing deferred acceptance~\citep{dubins1981machiavelli,roth1982}, there is a menu $\mathcal M_a$ independent of $a$'s report such that $a$ receives its most preferred element of $\mathcal M_a$ among the firms it lists, or remains unmatched if this intersection is empty. Indeed, define the menu as all partners obtainable by some report. If a report can obtain $j$, strategy-proofness applied to preferences listing only $j$ implies that this singleton report also obtains $j$; applying strategy-proofness once more gives the stated menu representation for every available list.

Let $I_a$ be the firms that signal $a$. This random incoming set is not observed when $a$ chooses its signals; all comparisons integrate over its joint law with the menu. If $S$ is $a$'s outgoing signal set and $b_j=U^B_{a,j}$, its realized payoff is therefore
\[
 V_a(S)=\max\left\{u_{a,\emptyset},\ 
          b_j+A_{a,j}: j\in\mathcal M_a\cap(I_a\cup S)\right\}.
\]
Conditional on $a$'s pre-interview row, the pairs
$\bigl(\mathbf1\{j\in\mathcal M_a\},\mathbf1\{j\in\mathcal M_a\cap I_a\}\bigr)$
are exchangeable over firms in the same tier and are independent of $a$'s shocks. Thus each firm has one of three states: unavailable, available only if signaled, or available regardless of $a$'s signal.

Consider a signal set containing $k$ but not $j$, where $b_j\ge b_k$, and condition on the states and shocks of the other firms. Pair each realization of the two remaining states with its label-swapped realization. If neither firm is available only when signaled, exchanging the signal has no effect. If exactly one is in that state, the average difference reduces to a nonnegative multiple of
\[
 \mathbb E\bigl[\max\{C,b_j+A\}-\max\{C,b_k+A\}\bigr]\ge0,
\]
where $C$ is the best remaining payoff; when the other firm is available regardless of the signal, the two terms in which both firms are available cancel. If both firms require a signal, the same inequality applies directly. Independence and identical distribution of the two own-side shocks justify these cancellations. Repeated exchanges show that the top-$d$ set maximizes conditional expected payoff. By the menu formula, adding a signal cannot lower an applicant's realized payoff. Thus using fewer than $d$ signals cannot improve this best response. 
{
The same exchange proof works inside a target tier after conditioning on the other firms' states. The single-side equilibrium claims follow because every permitted sender then best responds.
}%

\paragraph{{A counterexample to exact two-side equilibrium.}}

Let there be $k$ applicants and $N=\lceil e^{k^2}\rceil$ firms in one tier, with $d=k$, $B\sim\mathbb U[-10^{-3},10^{-3}]$, and $A\in\{0,0.05\}$ with equal probabilities, independently across directed pairs. All intrinsic values are $1$ and the unmatched payoff is zero. The prescribed rule has applicants signal their top $k$ firms and firms send no signals, since their target is empty. By symmetry, a specified firm $j$ is truthfully matched with probability at most $k/N$. It can deviate by signaling all $k$ applicants. For a specified applicant $a$, with probability
\[
 (1-k/N)\,2^{-k-1},
\]
$j$ is absent from $a$'s original list, all $k$ original shocks in $a$'s evaluation are zero, and its fresh shock toward $j$ is $0.05$. Then $a$ strictly prefers $j$ to every original neighbor. Firm $j$ must be matched in every stable matching after this deviation, since otherwise $(a,j)$ blocks. Its expected gain is therefore at least
\[
 0.999(1-k/N)2^{-k-1}-1.051\,k/N>0
\]
for all sufficiently large $k$. This example has fixed bounded score laws, $p=1$, both populations growing, and $d\to\infty$. The positive gain tends to zero, so it does not contradict the approximate conclusion. The restriction to proposing-side signals excludes this deviation; an agent with an empty prescribed target may otherwise send at most $d$ signals.

\paragraph{An endpoint estimate for a strictly imbalanced block.}
Here the distributions $\DistA$ and $\DistB$ are fixed as the market grows; write $a_* = \sup\operatorname{supp}(\DistA)$ and $b_* = \sup\operatorname{supp}(\DistB)$. Consider a one-sided signaling block whose sending side is smaller by at least $cN$ agents, where $N$ is the size of the original receiving tier and $c>0$ is constant. The same estimate permits one other agent to have been deleted, with all target tiers and top-$d$ lists still computed in the original population; the loss of at most one listed label does not change the exponential bound. Remove a specified sender $j$, and form the receiving-optimal stable matching of the remaining block. Let $U$ be its unmatched receivers. Then $|U|\ge cN-O(1)$ and $U$ is independent of $j$'s score row in this fresh tier pair. 
{
The entry comparison used to prove
\eqref{eq:degree_free_available_upper} shows that, after $j$ is
reinserted, at most one member of $U$ becomes matched; every other
member remains unmatched in every stable matching.
}%
 Every stable partner of $j$ therefore has utility at least the second-largest score among $j$'s signaled neighbors in $U$: otherwise one of the two better neighbors is unmatched and blocks. A missing second neighbor is an exceptional event.

The labels of $j$'s top $d$ scores form a uniform sample, independently of their ordered values. Hypergeometric concentration gives at least $k=\lfloor cd/3\rfloor$ such labels in $U$, except with probability $e^{-\Omega(d)}$. If $T$ is the $d$th pre-interview score among the other agents in the tier, the second-largest score of these neighbors is at least $T+A_{(k-1)}$, where $A_{(k-1)}$ denotes the second-largest of $k$ independent $\DistA$ draws.

Let $(B,X)$ be the potential pre- and post-interview scores of a fresh candidate who was not signaled by $j$. The event that this is a fresh candidate entails $B<T$, and $B$ is independent of the row determining $T$ and of the benchmark shocks. Let
\[
 \omega_B(\eta)=\sup_x\Pr\{x-\eta<B<x\}.
\]
Continuity and bounded support imply $\omega_B(\eta)\to0$ as $\eta\downarrow0$. For fixed $\eta>0$, put $q_\eta=\Pr\{A\ge a_*-\eta\}>0$. On $T-\eta<B<T$, exchangeability of $X$ and the $k$ benchmark shocks gives
$\Pr\{X>A_{(k-1)}\}\le2/(k+1)$, also when $\DistA$ has atoms. On $B\le T-\eta$, a candidate exceeding the benchmark requires at most one of the benchmark shocks to be at least $a_*-\eta$. Consequently,
\[
 \Pr\{B<T,\ B+X>T+A_{(k-1)}\}
 \le \frac{2\omega_B(\eta)}{k+1}
       +(k+1)(1-q_\eta)^{k-1}.
\]
Including the hypergeometric exception, first taking $d\to\infty$ and then $\eta\downarrow0$ gives a bound $o(1/d)$ for a fixed fresh candidate. This estimate is uniform over labels and over histories for which the block is fresh and its size gap is at least $cN$.

The argument requires fixed score laws (or the corresponding uniform tail conditions); the single quantity $p=\Pr\{A\ge0\}$ does not give these conditions for a triangular array of shock laws.

The same unmatched-set argument gives an endpoint payoff bound for a sender that survives to its target block. For any fixed $\eta>0$, at least a positive fraction (depending on $\eta$) of its $d$ highest pre-interview values exceed $b_*-\eta$, with probability tending to one. A positive fraction of these labels hit $U$, and independently a positive fraction of their shocks exceed $a_*-\eta$. There are at least two such neighbors with probability tending to one. Every stable matching consequently gives that sender a target-tier utility at least $v_{\calT(j)}+b_*+a_*-2\eta$. Bounded payoffs imply that its expected deficit from this upper endpoint is $o(1)$. This is an ex ante estimate over the sender's own type, not a bound uniform over all realized score rows.

\paragraph{Fixed score laws and general imbalance.}

{
Under the fixed-law and constant-imbalance hypotheses of
\prettyref{thm:incentive}, we prove the ex ante expected-payoff bound
specified in \prettyref{sec:incentive}, uniformly over agents and
admissible deviations.

}%
Use the strict-tier peeling of the original targeted signal graph, with applicant-optimal local matchings throughout. Constant general imbalance and non-vanishing original tiers give a fixed positive gap between the sides in every macroscopic target block. The same holds after one specified agent is deleted. In all leave-one-out comparisons, keep the original target tiers, signaling roles, quotas, and top-$d$ lists fixed. A list loses at most its edge to the deleted agent, which changes the estimates below only by a constant factor. Conditional on the histories of earlier, distinct tier pairs, the current pair's score arrays retain their independent original law.

The unmatched-set argument above bounds the chance that a specified sender is unmatched by $e^{-c_1d}$, for a constant $c_1>0$ depending only on the tier proportions and imbalance. Indeed, at least two of its selected labels hit the unmatched receiving set except with that probability. Therefore, at each of the finitely many blocks, the expected number of unmatched senders is at most $n e^{-c_1d}$. Markov's inequality, applied at a fixed small fraction of the population, and induction over the blocks show that the probability of losing any of the required constant gaps is $O(e^{-c_2d})=o(1/d)$. These statements hold in each leave-one-out market used below. Thus the endpoint and fresh-edge estimates apply with an error $o(1/d)$, uniformly over fixed labels.

Fix a possible deviator $i$. Delete $i$ and construct the applicant-optimal matching of the prescribed market without it. Consider a newly created edge from $i$ to a tier strictly above $i$'s prescribed target tier, treating an empty target as below every tier. By the dominance inequalities defining targets, the other endpoint $j$ targets either $i$'s tier or a tier strictly above it.

If $j$ is matched above $i$'s tier in the leave-one-out market, strict tier ordering already rules out acceptance. Otherwise, except on the exponentially small bad-history event, $j$ reaches a constant-imbalance sending block. If its target is strictly above $i$'s tier, its target match again rules out acceptance. If its target is $i$'s tier, the $o(1/d)$ fresh-edge estimate bounds the probability that $j$ prefers $i$'s potential post-interview value to its leave-one-out partner. Since $i$ chooses its recipients using only its own independent pre-interview row, a union bound over at most $d$ new edges makes {the probability that any such acceptance test succeeds} $o(1)$, uniformly over the deviating strategy.

This comparison also rules out indirect gains through old higher-tier edges. If $i$ is an applicant, any firm that would reject $i$ from its leave-one-out applicant-optimal partner is outside $i$'s deferred-acceptance menu: starting from that leave-one-out matching, a singleton proposal is rejected. Adding firms outside the menu cannot change $i$'s chosen partner.

If $i$ is a firm, adding it to the leave-one-out market weakly improves every applicant under applicant-proposing deferred acceptance. An applicant who prefers its leave-one-out partner to $i$ never proposes to $i$ in the enlarged market. The new higher-tier edges failing the acceptance test are therefore unvisited and can be deleted without changing the outcome.

Finally, altering $i$'s target-tier or lower-tier edges cannot affect its outcome in higher tiers. For an applicant this is immediate from the menu representation and strict tier ordering. For a firm, run applicant-proposing deferred acceptance with applicants in descending tier order. Every firm strictly prefers higher-tier applicants, so all higher-tier proposals and matches are determined before target-tier or lower-tier applicants are processed and cannot be displaced by them. This argument applies even when the deviation adds edges outside the original target pairs. Hence outside an event of probability $o(1)$, the deviation leaves every higher-tier outcome available to $i$ unchanged.

If $i$ was truthfully assigned above its own target, it consequently cannot gain, apart from the exceptional event. If it survives to a nonempty target, the endpoint estimate gives truthful payoff at least the target-tier upper endpoint minus a nonnegative random deficit with expectation $o(1)$. No match in that tier can exceed the endpoint, and every lower tier is strictly worse. If $i$ has an empty target, all possible newly created edges have already been excluded by the preceding argument. Bounded payoffs convert the exceptional probability and endpoint deficit into the stated uniform $o(1)$ ex ante expected gain. There are finitely many non-vanishing tiers, so the constants and errors can be taken uniformly over agents.

}%
\section{Supplementary materials}
\subsection{Supplementary materials for \prettyref{sec:single_proof}}
{Let $H_1$ and $H_2$ be independent random one-sided regular bipartite graphs, where each $a\in \Short$ is connected to $d$ randomly chosen $j\in \Long$ on $H_1$, and each $j \in \Long$ is connected to $\widehat{d}$ randomly chosen $a\in \Short$ on $H_2$, where $\widehat{d} \le O\left(d\right)$.} Let $H$ denote the union graph of $H_1$ and $H_2$ with uniformly generated strict preferences. 

{
Let $H'$ denote the vertex-induced subgraph of $H$ on fixed sets $\Short'\cup\Long'$, where $|\Short'|=\gamma_1n_\Short$ and $|\Long'|=\gamma_2n_\Long$, with $0<\gamma_1\le1$ and $\gamma_2=\Omega(1)$. The sets may also be chosen independently of both graph layers and the preference arrays. For any $i\in\calV(H')$, let $\calN(i)$ denote its neighbors on $H'$. Use the root-data neighbor-selection convention in the proof of \prettyref{prop:remove_gamma_d_omega_1}: a selected set may depend on the root's star and independent root-side marks, but not on the other branches or on the matching.
}%

{
\begin{proposition}\label{prop:both_remove_gamma_d_omega_1}
Suppose $\gamma_1n_\Short=\delta\gamma_2n_\Long$ and the
hypotheses on $d,\delta,\lambda$ in
\prettyref{prop:remove_gamma_d_omega_1} hold. For the two-sided graph
$H'$ above, the bounds \eqref{eq:j_a_calN'_a} and
\eqref{eq:a'_j_calN'_j} hold with constants
$\underline C',\overline C'>0$ in place of
$\underline C,\overline C$. These constants may depend on the fixed
parameter bounds, including the upper bound for $\widehat d/d$ and
for $\log d/\log\log n_\Long$. The selected sets obey the preceding
root-data rule; the expectations average their possibly random
cardinalities, and the respective errors remain $o(1/n)$ and $o(1)$.
In particular, the applicant-root probability is $o(1)$ if
$\nu|\calN'(a)|\to\infty$ in probability, and the firm-root probability
is $1-o(1)$ if $|\calN'(j)|/(\nu d)\to0$ in probability, where
$\nu$ is defined in \eqref{eq:nu}.
\end{proposition}
}%

{
\subsubsection*{Proof of \prettyref{prop:both_remove_gamma_d_omega_1}}
\begin{proof}
Put $N=n_\Long$, $m_A=\gamma_1n_\Short$, $m_J=\gamma_2N$, and
\[
 \kappa_A=\gamma_2d+\frac{\gamma_1}{\delta}\widehat d,
 \qquad \kappa_J=\delta\kappa_A,
 \qquad N_*=n_\Short\wedge n_\Long,
 \qquad B=\lfloor N^{1/8}\rfloor.
\]
The hypotheses give $m_A,m_J=\Theta(N)$, $N_*=\Omega(N)$ and
$\kappa_A,\kappa_J=\Theta(d)$, without requiring
$\widehat d\le d$, $\gamma_1=\Omega(1)$, or comparable original
populations. Use the shared completed exploration with $s=2B$,
reserving the root star before exploring its branches. No offspring
condition is imposed on the conditioned root.

The conditional estimates
\eqref{eq:conditional_exploration_mean_upper}--\eqref{eq:conditional_exploration_tail}
apply with
\[
 K_t=\frac{\kappa_t}{1-2B/N_*},\qquad
 f_t=\kappa_t^{-1/4},\qquad r_t=e^{-\sqrt{\kappa_t}/8}.
\]
Indeed, the relative lower-mean loss is
$O(1/d+(B+d+\widehat d)/(m_A\wedge m_J))=o(\kappa_t^{-1/4})$.
The completion retains the triggering count and replaces subsequent
counts by $\lfloor\kappa_t\rfloor$, preserving both envelopes on
every permitted history.

Choose
\[
 m=2\left\lfloor\frac{\log N}{32(\log d\vee\log\log N)}\right\rfloor.
\]
Put $L_A=(1-f_A)\kappa_A$ and
\[
 K_J^*=\frac{(1+d^{-\lambda})\kappa_A}{1-2B/N_*},\qquad
 \eta=(1-2(K_J^*\vee\kappa_A)r_A)^{-1},\qquad
 b_m=(f_{L_A}\circ f_{\eta K_J^*})^{m/2-1}(1).
\]
Then $K_J\le K_J^*$ and
$2B/N_*+2(K_J^*\vee\kappa_A)r_A=o(f_A)$.
Since $\kappa_A=\Theta(d)$, the scalar comparison proving
\eqref{eq:mu_j_rho_lower_bound_2} gives $b_m\le C/(\nu d)$.
The conditional induction of \prettyref{lmm:rooted_tree_expected_degree}
therefore gives, for the completed trees $\widetilde T$,
\begin{align*}
 \Expect[\sfX_{j,a}(\widetilde T_m(a))\mid\mathcal F_{\rm entry}]
 &\ge f_{\eta K_J^*}(b_m)\ge\underline C'\nu,\\
 \Expect[\sfX_{a',j}(\widetilde T_{m-1}(j))\mid\mathcal F_{\rm entry}]
 &\le r_A+(1-r_A)b_m\le\frac{\overline C'}{\nu d},
\end{align*}
where $K_J^*=\Theta(d)$ and $r_A=o(1/(\nu d))$.
These bounds hold uniformly for a child specified before its branch
is explored, after previous completed branches and their preferences
have been revealed.

For an applicant root, even-depth truncation, the all-cycle bound
\eqref{eq:all_cycles_joint_product}, and its root-cutoff averaging
argument yield
\[
 \Pr\{\text{every }j\in\calN'(a)\text{ is unavailable to }a\}
 \le\Expect[(1-\underline C'\nu)^{|\calN'(a)|-2}],
\]
after reducing the constant if necessary. This is stronger than the
asserted upper bound and requires no error normalized by the possibly
much larger original population $n_\Short+n_\Long$.
For a firm root, odd-depth truncation and
\eqref{eq:relative_sequential_lower} give the lower bound
$(1-\epsilon_N)(1-\overline C'/(\nu d)-\epsilon_N)^k$,
where $k=|\calN'(j)|$. On the root cutoff event,
$k\le C(d+\log N)$ and $k\epsilon_N=o(1)$, so this is at least
$(1-\overline C'/(\nu d))^k-o(1)$. Averaging the root data and
discarding its vanishing tail gives the asserted lower bound. The two probability
limits follow from these bounds by bounded convergence.
\end{proof}
}%

\subsection{Supplementary materials for \prettyref{sec:multi_proof}}\label{sec:supp_reduced}
Consider a single-tiered two-sided market with applicants $\Short$ and firms $\Long$. Let $H$ denote an interview graph constructed based on the applicant-signaling mechanism. 

{
Let $\Short' \subset \Short$ and $\Long' \subset \Long$ be subsets chosen independently of both the connections in $H$ and the relevant preference arrays, with $|\Short'|= \gamma_1 n_{\Short}$,  $|\Long'|= \gamma_2 n_{\Long}$ and $\gamma_1 n_{\Short} =  \delta \gamma_2 n_{\Long} $ for some $\Omega(1) < \gamma_1,\gamma_2 \le 1$ and $\delta \ge \Omega (1)$. Let $H'$ denote the vertex-induced subgraph of $H$ on $\Short' \cup \Long'$.
}%
 
\begin{theorem}\label{thm:single_tier_general_sparse}

{
Suppose $\omega(1) \le d \le O\left(\polylog n_{\Long}\right)$ and $ \Omega (1) \le \delta  \le 1+ d^{-\lambda} $ for some $\lambda>0$. 
}%

{
If $p = \omega\left(\frac{1}{\left(1\wedge \lambda \right)\log d } \right)$, every stable matching on $H'$ is almost interim stable with high probability.
}%
 
\end{theorem}

{
\begin{remark}\label{rmk:single_tier_general_sparse_unmatched}
Under the hypotheses of \prettyref{thm:single_tier_general_sparse}, there is a constant $c>0$ such that, in every stable matching on $H'$, the fraction of unmatched applicants in $\Short'$ is at most $d^{-c(1\wedge\lambda)}$ with high probability. The constant $c$ depends only on a positive lower bound for $\gamma_2$ and a fixed upper bound for $\log d/\log\log n_{\Long}$. In particular, the exponent can be chosen as a positive constant when $\lambda$ is bounded below by a positive constant.
\end{remark}
}%

\begin{remark}\label{rmk:single_tier_general_sparse_identify}
    {Under the hypotheses of \prettyref{thm:single_tier_general_sparse}, there} exists a subset $\Short''\subset \Short'$ such that every stable matching on the vertex-induced subgraph of $H$ on $\left(\Short'\backslash \Short''\right)\cup \Long'$ is perfect interim stable with high probability, where $|\Short''| = o \left( n_{\Short} \right)$.  
\end{remark}

{
\begin{theorem}\label{thm:single_tier_general_dense}
{Fix $\epsilon>0$.} Suppose

{
\begin{align}
   d \ge \max \left\{ \frac{16+\epsilon}{\gamma_2 p}\,, - \frac{8+\epsilon}{ \gamma_2 \delta p} \log \left( 1-\delta + \frac{\delta^2}{\gamma_2 n_\Long} \right) \right\} \log n_{\Short} \,. \label{eq:d_needed}
\end{align}
}%

\begin{itemize}
    \item If $\delta<1$, every stable matching on $H'$ is perfect interim stable with high probability.
    \item If $\delta = 1$, the applicant-optimal stable matching on $H'$ is perfect interim stable with high probability.
\end{itemize}
\end{theorem}
}%

\begin{remark}\label{rmk:single_tier_general_both_dense}

{
Let $H$ be constructed by both-side signaling with \prettyref{eq:d_needed} satisfied, and let $H'$ be its induced subgraph on $\Short'\cup\Long'$, with the retained sets independent of the graph and score arrays as above. The applicant-optimal stable matching on $H'$ is perfect interim stable with high probability in either of the following cases:
\begin{itemize}
\item $\delta=1$, $p\ge p_0$ for some fixed $p_0>0$, and $\DistA$ has a log-concave density;
\item $0<\delta\le1$ and, for some $c\ge0$, $A\le c$ almost surely with $\Pr\{A=c\}\ge p/2$.
\end{itemize}
If $\delta=1$ and the budget also holds with the two populations and retained fractions interchanged, both prescribed optimal matchings have this property simultaneously with high probability.
}%

\end{remark}

By \prettyref{lmm:interview_d_regular}, $H$ can be considered as a randomly generated one-sided $d$-regular graph, and then $H'$ can be viewed as a vertex-induced subgraph of $H$ on $\Short'\cup\Long'$. 
{
Let $\calN(a)$ be the neighbors of $a$ on $H'$ and define
}%

{
\begin{align}
    \calN_+(a)\triangleq \{j \in \calN(a): A_{a,j} \ge 0\}  \, .  \label{eq:calC_+_H'} 
\end{align}
}%

{

{
Condition on the independent retained sets. For independent
Bernoulli-$q$ marks on an applicant's outgoing edges, let $K_a(q)$
count retained marked neighbors. Then
$R_a=|\calN(a)|\sim\Hyper(\gamma_2n_\Long,n_\Long,d)$ and
$K_a(q)\mid R_a\sim\Binom(R_a,q)$. Hoeffding's comparison for
sampling without replacement~\citep[Theorem~4]{hoeffding1963probability},
applied to the convex function $s\mapsto(1-q+qz)^s$, gives
\begin{equation}
 \mathbb E[z^{K_a(q)}]
 =\mathbb E[(1-q+qz)^{R_a}]
 \le(1-\gamma_2q+\gamma_2qz)^d,
 \qquad 0<z\le1.
 \label{eq:retained_mark_pgf}
\end{equation}
Thus the binomial Chernoff lower tail applies with mean $\gamma_2dq$.
In particular, $K_a(p)=|\calN_+(a)|$ satisfies
\begin{align}
 \Pr\{|\calN_+(a)|\le\tfrac14\gamma_2dp\}
 \le e^{-9\gamma_2dp/32}\le e^{-\gamma_2dp/16}.
 \label{eq:calN_+_1}
\end{align}

}%
}%

\begin{proof}[Proof of \prettyref{thm:single_tier_general_sparse}]

{
Let $\nu$ be as in \eqref{eq:nu}. The hypothesis gives
$p\nu d\to\infty$, and hence $\lambda\log d\to\infty$ because
$p\le1$. The set $\calN_+(a)$ obeys the root-data selection rule.
Applying \prettyref{prop:remove_gamma_d_omega_1} and then
\eqref{eq:retained_mark_pgf} with $q=p$ yields, for some $C>0$,
\begin{align}
 \Pr\{\text{every }j\in\calN_+(a)\text{ is unavailable to }a\}
 &\le (1-C\nu)^{-2}(1-C\gamma_2p\nu)^d+o(1/n)
 =o(1),
 \label{eq:prob_j_available_omega_1}
\end{align}
uniformly over $a\in\Short'$. Let $\Short'_{\mathrm{bad}}$ be the applicants without an available neighbor in $\calN_+(a)$. Markov's inequality gives $|\Short'_{\mathrm{bad}}|=o(|\Short'|)$ with high probability. {
The paired-deletion argument in \prettyref{sec:single_signal_sparse}, applied on $H'$, then establishes almost interim stability of every stable matching on the same event, using at most $2|\Short'_{\mathrm{bad}}|=o(n_\Short)$ deletions.
}%

}%
\end{proof}
{

For \prettyref{rmk:single_tier_general_sparse_identify}, take
$\Short''=\Short'_{\mathrm{bad}}$ and delete only these applicants.
The availability-preservation argument in
\prettyref{sec:single_signal_sparse_identify} applies to $H'$ and proves
the conclusion for every stable matching on the same reduced graph.
Independence is required for the initial retained sets $\Short',\Long'$,
not for this subsequent deterministic deletion.

}%

{
\begin{proof}[Proof of \prettyref{rmk:single_tier_general_sparse_unmatched}]

{
Put $t=(1\wedge\lambda)\log d\to\infty$. Apply
\prettyref{prop:remove_gamma_d_omega_1} to the full neighborhood and
use \eqref{eq:retained_mark_pgf} with $q=1$ and $z=1-Ct/d$:
\[
 \Pr\{a\text{ is unmatched}\}
 \le (1-Ct/d)^{-2}(1-C\gamma_2t/d)^d+o(1/n)
 \le e^{-(C\gamma_2-o(1))t}+o(1/n).
\]
Since $t\le\log d$, $d=O(\polylog n_\Long)$ and
$\gamma_2=\Omega(1)$, this is at most $e^{-2ct}$ for a sufficiently
small constant $c>0$, uniformly over retained applicants. If $U$ is
their unmatched count, Markov's inequality gives
\[
 \Pr\{U/|\Short'|>d^{-c(1\wedge\lambda)}\}
 \le e^{-ct}=o(1).
\]
The unmatched set is common to all stable matchings
\citep[Lemma~1]{roth1986allocation}, so the conclusion is simultaneous.

}%
\end{proof}
}%

\begin{proof}[Proof of \prettyref{thm:single_tier_general_dense}]

{
Let
\[
r_n=\frac1\delta\log\frac1{1-\delta+\delta^2/(\gamma_2n_\Long)}\log n_\Short,
\qquad k_n=\frac14\gamma_2dp.
\]

{
The assumed bound \prettyref{eq:d_needed} gives
\[
k_n\ge(2+\epsilon/4)r_n,\qquad
\gamma_2dp\ge(16+\epsilon)\log n_\Short.
\]
Using \prettyref{eq:calN_+_1} and summing over $a\in\Short'$ yields
\begin{align}
\Pr\{\exists a\in\Short':\ |\calN_+(a)|\le k_n\}
\le n_\Short e^{-\gamma_2dp/16}=o(1).
\label{eq:addition_union_calN_+_general}
\end{align}

}%

The retained populations are fixed or independent of both the graph
and preference arrays. An outcome-dependent choice within the same
score block requires its own conditional-law argument.

For $\delta$ in a compact subinterval of $(0,1)$, apply
\eqref{eq:compact_complete_max_rank} with $M=|\Short'|$ and
$N=|\Long'|$. Its upper-rank threshold is
$(2L_\delta/\delta)\log M\le(2+o(1))r_n$.
On the positive-neighbor event above, every retained degree exceeds
$k_n\ge(2+\epsilon/4)r_n$. The list-completion argument in the proof
of \prettyref{thm:single_signal_dense} therefore applies to $H'$.
Every stable matching has all applicant ranks below $k_n$ with
high probability, at every degree allowed by \eqref{eq:d_needed}.
The rank version of the positive-neighbor criterion proves interim
stability on $H'$.

If $\delta=1-o(1)$, choose
$0<\eta<\min(\epsilon/4,1)$ sufficiently small that
$2+2\eta<8+\epsilon$. The budget implies
$d>(2+2\eta)r_n/\gamma_2$. The retained-population scale in
\prettyref{prop:perfect_stable} is at most this $r_n$, because
$\log|\Short'|\le\log n_\Short$. That proposition gives ranks
at most $(2+\eta)r_n$ in every stable matching when $\delta<1$,
and at most $(1+\eta)r_n$ in the applicant-optimal matching when
$\delta=1$. Again these ranks lie below $k_n$.
The supplementary setting has $\delta=\Omega(1)$, so the compact
and near-balanced cases cover its entire stated range.

}%

\end{proof}

{
\begin{proof}[Proof of \prettyref{rmk:single_tier_general_both_dense}]
{
First consider the upper-atom alternative.
}%
Let $H_A'$ retain only applicant-sent edges of $H'$, and let
$\Phi_A'$ be its applicant-optimal stable matching. Put
$q=\Pr\{A=c\}\ge p/2$. For each retained applicant $a$, write $b_a$
for the $d$th largest pre-interview utility in its original target
population $\Long$, and let
\[
 K_a=|\{j\in\Long':a\hbox{ signals }j,\ A_{a,j}=c\}|,
 \qquad k_n=\tfrac14\gamma_2dp.
\]

{
Apply \eqref{eq:retained_mark_pgf} with mark probability $q$.
Its Chernoff bound, with $\mu=\gamma_2dq\ge\gamma_2dp/2$, gives
\[
 \Pr\{K_a\le k_n\}
 \le\Pr\{K_a\le\mu/2\}
 \le e^{-\mu/8}\le e^{-\gamma_2dp/16}.
\]
Since \eqref{eq:d_needed} implies
$\gamma_2dp\ge(16+\epsilon)\log n_\Short$, a union bound gives
$\min_{a\in\Short'}K_a>k_n$ with probability $1-o(1)$.
The cutoffs and signal lists remain those of the original population.

}%

The rank argument in the proof of
\prettyref{thm:single_tier_general_dense}, applied to $H_A'$, gives
$\operatorname{rank}_a(\phi_A'(a))<k_n$ simultaneously for all retained
applicants with probability $1-o(1)$. This uses the applicant-optimal bound when $\delta=1$ and the
all-stable bound when $\delta<1$. Hence
\[
 U^A_{a,\phi_A'(a)}\ge b_a+c.
\]

{
The added-edge argument in \prettyref{sec:single_both_signal_dense}
applies with these original-population cutoffs: every added edge has
applicant utility below $b_a+c$, so $\Phi_A'$ is stable on $H'$.
Applicant optimality on $H'$ preserves the bound $b_a+c\ge b_a$,
which excludes every noninterviewed blocking pair.

}%

When the retained populations are balanced, interchanging the two
sides proves the corresponding firm-optimal assertion under the
interchanged budget.

{
\paragraph{The balanced log-concave case.}
Suppose $\delta=1$ and the log-concave alternative holds, with
$p\ge p_0>0$. Put
\begin{equation}
 N=|\Short'|=|\Long'|=\gamma_1n_\Short=\gamma_2n_\Long,
 \qquad L=\log N,\qquad
 d_A=\gamma_2d,\quad d_J=\gamma_1d.
 \label{eq:lc_reduced_degrees}
\end{equation}
The retained fractions are bounded below, so $d_A/d_J$ is bounded
above and below. The last term of \eqref{eq:d_needed} gives
\begin{equation}
 d_A\ge\frac{8+\epsilon}{p}\log N\log n_\Short
      \ge\frac{8+\epsilon}{p}L^2.
 \label{eq:lc_reduced_budget}
\end{equation}
We apply the proof of \prettyref{thm:single_both_signal_dense},
checking the changes caused by retention and unequal source degrees.

Condition on the retained identities, independently of this pair's
complete score arrays. Retained outgoing counts are
$\operatorname{Hyper}(N,n_\Long,d)$ and
$\operatorname{Hyper}(N,n_\Short,d)$, with means $d_A,d_J$;
incoming counts have the interchanged means. All concentrate uniformly
with relative error $O(\sqrt{L/\min(d_A,d_J)})=o(1)$.
Conditional only on outgoing counts, the rows are uniform subsets;
further conditioning on color degrees gives the same disjoint-graph
law as before. Conditional also on the original cutoffs, strict-source
and nonsource values have the previous iid truncated laws.
At most one retained source is a boundary value;
its receiving-query probability is $O(1/d_J)$ before the load stop.
Numerical marks remain independent of unresolved labels.

First let $\max(d_A,d_J)\le N^{3/4}$. Use separate reference cutoffs
$b_A^+,b_J^+$ with original tail probabilities
$(1-\varepsilon_N)d/n_\Long$ and $(1-\varepsilon_N)d/n_\Short$,
where $\varepsilon_N=C\sqrt{L/\min(d_A,d_J)}$.
Write $h_{r,A},h_{r,J}$ for the upper tails of the corresponding laws
$F_{1,A},F_{0,A},F_{1,J},F_{0,J}$. These laws satisfy
\eqref{eq:lc_common_source}--\eqref{eq:lc_common_nonsource} with
source weight $1-O(\varepsilon_N)$ and nonsource error
$\Delta_N=O(\varepsilon_N\max(d_A,d_J)/N)$, so $\Delta_NL=o(1)$.
In \eqref{eq:lc_switch_upper}--\eqref{eq:lc_switch_lower}, the
ordinary color masses and target degrees are $(Nd_A,d_A)$ and
$(Nd_J,d_J)$ asymptotically, so both target estimates remain
$(1+o(1))/N$, with the previous unvisited-target lower bound.
Replace $d^2/N$ by $d_Ad_J/N$ in the reciprocal estimates:
bounded degree ratios preserve $o_p(N)$ reciprocal proposals,
$o(1)$ future-hit probability, and root loss $O(d_Ad_J/N+L)$.

The required flow substitution is
\begin{equation}
 x=.9p,\qquad t=h_{1,A}^{-1}(x)\ge b_A^+,
 \qquad \theta=\frac{h_{0,A}(t)}x,\qquad
 w=\frac{d_J}{d_A}\theta.
 \label{eq:lc_reduced_weight}
\end{equation}
With $x_+=.95p$, the number of source scores above $t$ is at most
$x_+d_A$ in every row with high probability, giving a fixed
fractional cap on the retained source lists. By
\eqref{eq:lc_convex_tail}, $g_A(v):=h_{0,A}(h_{1,A}^{-1}(v))\le\theta v$
for $v\le x$.
Apply the three flow regimes with source and incoming reservoirs
$(1-o(1))d_A$ and $(1+o(1))d_J$. In the intermediate regime the
light-prefix odds in \eqref{eq:lc_light_flow} become
$(1+o(1))\theta d_J/[(1-\eta-o(1))d_A]$, and heavy roots contribute
$O(d_JH)=o_p(N)$ by \eqref{eq:lc_heavy_rows}.
Thus $T_0=1.01NL$ and $T_1=1.02wNL$ apply whenever $wL$ has
a positive extended limit. If $wL\to0$, incoming and reciprocal
offers are instead $o_p(N)$ contamination, giving ordinary retention
$.18/L$.

For the positive-limit regimes, repeat
\eqref{eq:lc_weighted_retention} with
$Q_J(s)=h_{0,J}(s)+w h_{1,J}(s)$. The identity
$dF_{0,J}+w\,dF_{1,J}=-dQ_J$ yields $a+wb\ge .30/L$ and
$a=O(1/L)$ after the same mixture and total-variation errors.
Conditional on retained counts and cutoffs, discarding the possible
boundary source slot, the source and incoming numerical
counts above $t$ are independent, with means at least
$(1-o(1))xd_A$ and $(1-o(1))d_J(h_{0,A}(t)-\Delta_N)_+$.
The exact binomial averaging of \eqref{eq:lc_root_product} therefore
gives
\begin{equation}
\begin{aligned}
 &\Pr\{\text{a specified root reaches utility below }t
            \text{ before the other stops},\mathcal G\}\\
 &\hspace{2em}\le \exp\{-(.30-o(1))xd_A/L+o(d_A/L)+O(1)\}.
\end{aligned}
 \label{eq:lc_reduced_root}
\end{equation}
The cancellation is $d_Jh_{0,A}(t)=wxd_A$; the additive error is
$\Delta_Nd_Jb=o(d_A/L)$ even for small $h_{0,A}(t)$.
As before, numerical good-event indicators are removed only after
the pointwise product bound. Equation~\eqref{eq:lc_reduced_budget}
excludes all root stops by a union bound, also with $.18$ in the
contamination regime. The source coupon argument completes DA
above every applicant's original cutoff.

For the remaining degrees both $d_A,d_J=\Omega(N^{3/4})$.
Restart the exposure conditional only on retained outgoing counts
and original cutoffs, without incoming or reciprocal degrees.
Apply \eqref{eq:lc_dense_query_membership}--\eqref{eq:lc_dense_query_retention}
with $K=\lfloor pd_A/100\rfloor$ and $\rho_J=d_J/N$.
The applicant-safe quantiles have original $B$ tails $d/(2n_\Long)$
and $d/(4n_\Long)$. The receiving membership probability is exactly
$(K_J-s)/(N-q)\ge d_J/(2N)$, where $K_J=(1+o(1))d_J$.
The same budgets $T=\lceil10NL/\rho_J\rceil$ and
$M=\lceil200L/\rho_J\rceil=o(d_J)$ give root union bound
$N\exp[-cpd_Ad_J^2/(N^2L)]=o(1)$ and coupon exclusion
$MK/T\le pd_A/(5N)+o(1)\le .2+o(1)$.
These estimates are uniform over retained fractions bounded below
and the stated $p_0$ and score-law conditions. The interchanged
budget proves the firm-optimal conclusion by reversing the entire
argument and intersecting the two high-probability events.

}%
\end{proof}
}%

\end{document}